\documentclass[%
  reprint,
  superscriptaddress,
  nobibnotes,longbibliography,
  amsmath,amssymb,
  aps,
  floatfix,
  twocolumn,nofootinbib
]{revtex4-2}

\usepackage{graphicx}
\usepackage{color}
\usepackage{natbib}
\usepackage{amsmath,amssymb,amsthm,mathrsfs,amsfonts,dsfont}
\usepackage{subfigure, epsfig}
\usepackage{braket}
\usepackage{bm}
\usepackage{bbm}
\usepackage{enumerate}
\usepackage{enumitem}
\usepackage{comment}
\usepackage{appendix}
\usepackage[normalem]{ulem}
\usepackage{makecell}
\usepackage{multirow}
\usepackage[utf8]{inputenc}
\usepackage{array}       
\usepackage{cellspace}
\usepackage[ruled,vlined]{algorithm2e}
\usepackage[colorlinks,linkcolor=blue,citecolor=blue,urlcolor=blue]{hyperref}
\usepackage{todonotes}
\usepackage{orcidlink}
\usepackage{adjustbox}

\newtheorem{theorem}{Theorem}
\newtheorem*{theorem*}{Theorem}
\newtheorem{corollary}{Corollary}
\newtheorem{lemma}{Lemma}
\newtheorem{proposition}{Proposition}

\newtheorem{fact}{Fact}
\newtheorem{task}{Task}
\theoremstyle{definition}
\newtheorem{dfn}{Definition}
\theoremstyle{remark}
\newtheorem{remark}{Remark}

\newcommand{\Hilb}[1]{\mathcal{H}_{#1}}
\newcommand{\Vect}{\mathbb{C}^{d}}
\newcommand{\Linear}[1]{\mathsf{L}(#1)}
\newcommand{\Herm}[1]{\mathsf{H}(#1)}
\newcommand{\Density}[1]{\mathsf{D}(#1)}
\newcommand{\Unit}[1]{\mathsf{U}(#1)}
\newcommand{\pure}[1]{\ket{#1}\bra{#1}}

\allowdisplaybreaks[4]

\makeatletter
\let\origaddcontentsline\addcontentsline
\newif\ifTOCsuppress
\TOCsuppressfalse
\renewcommand{\addcontentsline}[3]{
  \ifTOCsuppress
    \edef\tempa{#1}\def\tempb{toc}
    \ifx\tempa\tempb
    \else
      \origaddcontentsline{#1}{#2}{#3}%
    \fi
  \else
    \origaddcontentsline{#1}{#2}{#3}%
  \fi
}
\newcommand{\TOCwriteoff}{\TOCsuppresstrue}
\newcommand{\TOCwriteon}{\TOCsuppressfalse}
\makeatother

\newcommand{\beginsupplement}{%
  \setcounter{section}{0}\renewcommand{\thesection}{S\arabic{section}}%
  \setcounter{figure}{0}\renewcommand{\thefigure}{S\arabic{figure}}%
  \setcounter{table}{0}\renewcommand{\thetable}{S\arabic{table}}%
}

\begin{document}
\title{Learning Enables Exponential-to-Polynomial Sampling Overhead Scaling in\\ Quantum Divide-and-Conquer for Tree-Structured Circuits}

\author{Hiroyuki Harada\orcidlink{0009-0007-6873-4915}}
\email{hiro.041o\_i7@keio.jp}
\affiliation{Graduate School of Science and Technology, Keio University, Hiyoshi 3-14-1, Kohoku, Yokohama 223-8522, Japan}

\author{Kaito Wada\orcidlink{0000-0003-4976-4530}}
\affiliation{Graduate School of Science and Technology, Keio University, Hiyoshi 3-14-1, Kohoku, Yokohama 223-8522, Japan}
\affiliation{International Center for Elementary Particle Physics, University of Tokyo, 7-3-1 Hongo, Bunkyo-ku, Tokyo 113-0033, Japan}

\author{Naoki Yamamoto\orcidlink{0000-0002-8497-4608}}
\affiliation{Department of Applied Physics and Physico-Informatics, Keio University, Hiyoshi 3-14-1, Kohoku, Yokohama 223-8522, Japan}
\affiliation{Quantum Computing Center, Keio University, Hiyoshi 3-14-1, Kohoku, Yokohama 223-8522, Japan}

\author{Suguru Endo\orcidlink{0000-0002-2317-3751}}
\email{suguru.endou@ntt.com}
\affiliation{NTT Computer and Data Science Laboratories, NTT Inc., Musashino 180-8585, Japan}
\affiliation{NTT Research Center for Theoretical Quantum Information, NTT Inc. 3-1 Morinosato Wakanomiya, Atsugi, Kanagawa, 243-0198, Japan}

\date{\today}


\begin{abstract}
Quantum circuit cutting and knitting are divide-and-conquer methods that enable large-scale quantum computations on hardware with limited qubit resources and connectivity by decomposing a target computation into smaller local experiments.
Existing methods, however, typically incur a sampling overhead that grows exponentially with the number of cut locations, leaving open whether this exponential-in-cuts barrier is intrinsic even for nontrivial structured circuit families.
Here we show that this barrier is not universal by introducing a learning-based cutting protocol tailored to the target observable.
At each cut, the protocol locally learns a Heisenberg-picture effective observable that captures the downstream information relevant to the final measurement and uses it to construct an observable-adaptive cut.
This replaces the multiplicative variance amplification of conventional cutting with additive bias accumulation controlled by local learning accuracy.
We apply this framework to finite tree-structured circuits, a fundamental setting for both tree-based quantum divide-and-conquer and classical tensor-network simulation. 
For any finite rooted tree with $K$ cut edges (equivalently, cut wires) and cut-system dimension at most $d$, the protocol estimates the target expectation value within additive error $\epsilon$ with high probability using $\tilde{\mathcal{O}}(d^3K^3/\epsilon^2)$ measurements, including the local learning cost, thereby achieving polynomial measurement scaling in both the number of cuts and the cut dimension.
Moreover, for two-layer trees with $R$ cut edges, we prove an information-theoretic exponential separation between our learning-based protocol and learning-free wire-cutting protocols based on pre-specified randomized cutting rules: even with arbitrary classical post-processing, any such learning-free protocol requires $\Omega((d+1)^R/\epsilon^2)$ measurements, whereas our protocol uses $\tilde{\mathcal O}(d^3R^3/\epsilon^2)$.
These results identify local learning, rather than the tree structure alone, as the key mechanism driving the exponential-to-polynomial reduction in sampling overhead, opening a path toward more scalable quantum divide-and-conquer computation and virtual quantum simulation.
\end{abstract}

\maketitle

\TOCwriteoff

\section{Introduction}

\subsection{Background}

Quantum computing has emerged as a promising paradigm capable of addressing computational problems that are intractable for classical computers~\cite{10.1145/237814.237866,doi:10.1137/S0097539795293172}.
The anticipated applications span a wide range of domains, including quantum chemistry~\cite{peruzzo2014variational,kandala2017hardware}, materials science~\cite{bauer2020quantum,ma2020quantum}, machine learning~\cite{Moll_2018,Havlíček2019,Amaro2022,fuller2021approximate}, and even finance~\cite{orus2019quantum,herman2022survey}. 
Despite these prospects and rapid progress in quantum hardware and algorithms, near-term quantum processors remain severely constrained by physical noise and engineering limitations, which restrict usable qubit resources, connectivity, coherence times, and gate performance. 
In particular, limited qubit resources are expected to persist into the early stages of fault-tolerant quantum computing, as quantum error correction will introduce substantial overhead in the number of physical qubits required per logical qubit.

A widely discussed route to scalability is modular and distributed quantum computation~\cite{doi:10.1126/science.aam9288,RevModPhys.87.1379,jnane2022multicore}. 
Instead of scaling a single monolithic processor, one interconnects multiple modules and realizes global computations by coordinating local quantum processing across modules. 
In the long run, high-fidelity inter-module entanglement links and fast quantum communication could enable powerful distributed architectures; in the near term, however, such interconnects may not be available at scale. 
In this interim regime, classically assisted quantum divide-and-conquer protocols, collectively known as \textit{circuit knitting}, can play an important role for near-term devices~\cite{BARRAL2025100747}.

Historically, quantum hardware limitations have spurred the development of a broad range of circuit knitting techniques~\cite{PhysRevX.6.021043,PRXQuantum.3.010309,PRXQuantum.3.010346,PhysRevLett.127.040501,peng2021simulating}.
The common idea is to decompose a quantum computational task into multiple smaller subcircuits that fit on size-limited quantum hardware, execute them locally while allowing classical communication between processors, 
and reconstruct the desired quantity via classical post-processing~\cite{scholten2024assessing}. 
For instance, quantum circuit cutting methods have been actively explored, including \textit{wire cutting (timelike cuts)}~\cite{peng2021simulating} and \textit{gate cutting (spacelike cuts)}~\cite{Mitarai_2021}. 
These cutting methods typically rely on the quasiprobability decomposition (QPD) framework~\cite{PhysRevX.8.031027,PhysRevLett.115.070501,PhysRevLett.119.180509}: first, one represents the target nonlocal operations $\mathcal{T}$, such as an identity channel (i.e., a wire) or an entangling gate, as a linear combination of implementable local operations $\mathcal{E}_i$;
\begin{equation}
    \mathcal{T}= \sum_i c_i \mathcal{E}_i.
\end{equation}
Based on this equation, one replaces each target operation $\mathcal{T}$ with randomly sampled local operations $\mathcal{E}_i$, and then reweights the final measurement outcomes with $\gamma=\sum_i|c_i|>1$ to reconstruct the desired quantity.
By cutting selected nonlocal operations in this way, circuit cutting enables the emulation of larger quantum circuits on limited quantum processors. Recent experiments have also demonstrated practical benefits of circuit cutting, e.g., alleviating the effect of noise~\cite{bechtold2023investigating,PhysRevLett.130.110601} and relaxing connectivity constraints~\cite{carrera2024combining}.

Despite their appeal, existing divide-and-conquer approaches typically face a key practical bottleneck: the number of measurements required to achieve a desired accuracy grows exponentially as the computation is decomposed into more pieces.
In the case of conventional quantum circuit cutting, for example, rescaling factors $\gamma$ introduced in the reconstruction step inflate the estimator variance. Moreover, because these factors compound multiplicatively across cut locations, the sampling overhead typically grows exponentially with the number of cuts~\cite{PRXQuantum.6.010316}.
More precisely, let $K$ denote the number of cut locations, and let $\gamma_k$ denote the rescaling factor associated with the decomposition used at the $k$-th cut. 
A standard analysis then implies that, to estimate the expectation value of a bounded observable to additive error $\epsilon$ with high probability, the sufficient number of measurements scales as
\begin{equation}
    \mathcal{O}\left(\frac{\prod_{k=1}^{K}\gamma_k^2}{\epsilon^2} \right).
\end{equation}

This exponential-in-cuts scaling has motivated extensive efforts to reduce the measurement overhead of quantum circuit cutting. 
Prior work has largely pursued two main directions: reducing the number of cut locations $K$ via classical optimization~\cite{10.1145/3445814.3446758,tomesh2023divide,brandhofer2023optimal}, and reducing the rescaling factor $\gamma$ by developing improved decompositions.
For wire cutting, Ref.~\cite{peng2021simulating} introduces an early protocol using mid-circuit Pauli measurements with $\gamma=d^2$, where $d$ is the dimension of the wire. 
Subsequent works developed decompositions with smaller $\gamma$ and showed that allowing classical communication across the cut can further reduce the rescaling factor to $\gamma=2d-1$~\cite{Lowe2023fastquantumcircuit,brenner2023optimal,harada2025doubly,pednault2023alternative,PRXQuantum.6.010316}. Similar progress has also been made for gate cutting: following early prototypes for decomposing specific two-qubit gates (e.g., CNOT)~\cite{Mitarai_2021}, numerous studies have sought low-overhead decompositions for general two-qubit gates~\cite{Mitarai2021overheadsimulating,10236453,Ufrecht2023cuttingmulticontrol,PRXQuantum.6.010316,schmitt2024cutting}. 

Nevertheless, these developments also clarify the intrinsic limitations of QPD-based circuit cutting.
For wire cutting, it has been proven that $\gamma=d^2$ is optimal when classical communication (CC) is not allowed, whereas $\gamma = 2d-1$ is optimal when it is allowed~\cite{brenner2023optimal}.
Consequently, the standard analysis inevitably leads to an exponential-in-cuts sample-complexity upper bound, e.g., when cutting $K$ wires of dimension $d$, the sufficient number of measurements scales as
\begin{eqnarray}
     \mathcal{O}\left(\frac{d^{4K}}{\epsilon^2}\right)~\text{[without CC]},~~~ \mathcal{O}\left(\frac{(2d-1)^{2K}}{\epsilon^2}\right)~\text{[with CC]}.\nonumber
\end{eqnarray}

This naturally raises a simple but fundamental question: \textit{Is an exponential dependence on the number of cuts an unavoidable cost of divide-and-conquer approaches?}
This question is closely related to the longstanding belief that circuit cutting/knitting protocols incur a measurement overhead that is exponential in the number of cuts, as discussed in several recent reviews~\cite{BARRAL2025100747,scholten2024assessing,knorzer2025distributed}.
Indeed, such an exponential barrier would likely be unavoidable for arbitrary cut locations in arbitrary circuit structures.
A general cutting procedure that removed this exponential dependence for arbitrary circuits, while incurring only polynomial total overhead, would allow one to recursively decompose arbitrary quantum circuits into small local pieces, leading to an efficient classical simulation strategy (see Ref.~\cite{marshall2023all} for a more detailed discussion).
Importantly, however, this argument does not rule out the possibility that the exponential-in-cuts measurement overhead can be avoided for specific cut locations in structured circuit families, with any additional measurements required to construct the cuts counted in the total measurement count.
This leads to the following refined question:

\vspace{10pt}
\noindent\textit{Can one provably avoid exponential-in-cuts measurement overhead for nontrivial structured circuit families, within the standard circuit-cutting/knitting setting of local quantum experiments with classical communication between processors?}
\vspace{10pt}

This question is also natural in light of the historical development of quantum divide-and-conquer methods. 
In fact, this structural perspective was already present in the early theoretical work on quantum circuit cutting~\cite{peng2021simulating}, which connected circuit graph structure to the cost of cluster simulation through a tensor-network contraction viewpoint.
Beyond circuit cutting, a variety of other circuit-knitting strategies, including hybrid tensor networks (HTNs)~\cite{PhysRevLett.127.040501,PRXQuantum.6.010320,Harada2025densitymatrix,kanno2024quantum}, Deep VQE~\cite{PRXQuantum.3.010346,PhysRevResearch.3.043121}, and entanglement forging~\cite{PRXQuantum.3.010309,huembeli2022entanglement}, have also pursued divide-and-conquer strategies primarily for tree-structured circuits or ansatzes, rather than arbitrary circuit structures
\footnote{More broadly, several subsequent works have also challenged this issue by substantially relaxing the device requirements~\cite{huang2025estimating,PhysRevX.13.041057}}. 
Despite this structural progress, it remains unclear whether nontrivial structured circuit families can avoid exponential-in-cuts measurement overhead, even in tree-structured settings.
From another perspective, tree-like structures also arise naturally in settings beyond circuit knitting. 
For example, recent coherent parallel quantum algorithms use star-like circuit structures in which a global quantum task is decomposed into several lower-depth branches~\cite{oshio2026nearheisenberglimitedparallelamplitudeestimation,Martyn2025parallelquantum}.
Taken together, these observations make tree-structured circuits a natural setting for
asking whether the exponential-in-cuts overhead is an intrinsic feature of quantum divide-and-conquer methods.
Moreover, trees are the simplest nontrivial family in which the number of cuts can grow while each rooted subtree is connected to the rest of the computation through a single edge.

\subsection{Overview of results}

\begin{figure*}[t]
    \centering
    \includegraphics[width=\textwidth]{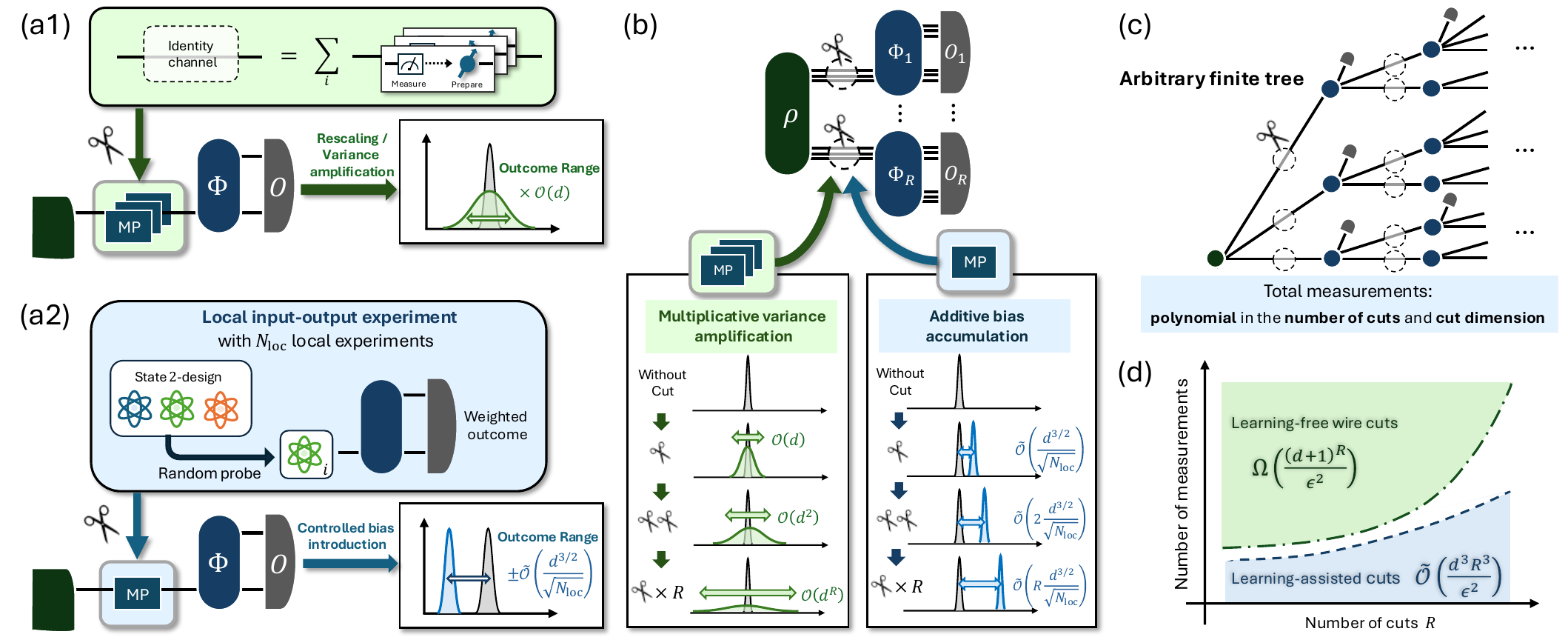}
    \vspace{-10pt}
    \caption{
    Overview of the learning-assisted circuit-cutting framework.
    (a1) Conventional wire cutting replaces the identity channel by a fixed measure-and-prepare decomposition; the associated rescaling enlarges the outcome range and leads to variance amplification.
    (a2) Our approach uses local input-output experiments to learn the downstream information relevant to the observable and constructs an observable-adaptive cut, replacing rescaling by a controlled additive bias.
    (b) Across multiple cuts, conventional rescaling factors accumulate multiplicatively, whereas the learning-assisted construction yields additive bias accumulation controlled by local learning accuracy. 
    The displayed ranges schematically indicate single-shot outcome ranges for conventional cutting and bias scales for the learning-assisted construction.
    (c) The same bottom-up idea extends to arbitrary finite tree-structured circuits, giving total measurement cost polynomial in the number of cuts and the cut dimension.
    (d) For two-layer trees, this yields polynomial scaling in the number of cuts, in contrast to the exponential lower bound for learning-free wire cutting.}
    \label{fig:introduction_figure}
\end{figure*}

In this work, we show that the exponential-in-cuts barrier is not universal. 
Specifically, we answer the refined question above in the affirmative by
giving a constructive learning-based protocol for estimating expectation
values of tree-structured quantum circuits. This is achieved within the
standard circuit-cutting/knitting setting: the algorithm uses only local
quantum experiments and classical communication, without assuming full
classical descriptions of the local components. Its total measurement
cost, including the local experiments used to construct the cuts, is polynomial in the actual number of cuts and in the cut dimension. 
Fig.~\ref{fig:introduction_figure} summarizes the main mechanism and the resulting scaling separation.

To achieve polynomial scaling, we first introduce a learning-enhanced cutting primitive. 
The key observation is that conventional cutting methods are designed to be generic and observable-agnostic: they use a fixed decomposition of the target channel, chosen independently of the final observable of interest (Fig.~\ref{fig:introduction_figure}(a1)).
However, our goal is only to estimate the expectation value of a fixed observable, not to simulate the target channel in full generality.
This naturally motivates an observable-adaptive cutting strategy, in which one locally learns the downstream information relevant to the final observable and uses it to adaptively construct the cut (Fig.~\ref{fig:introduction_figure}(a2)).
\begin{enumerate}
\setcounter{enumi}{0}
    \item \textbf{\textit{An observable-adaptive wire cut with unit rescaling.}}
    Consider a cut followed by a downstream quantum channel $\Phi$ and a
    target observable $O$ measured at its output.
    We show that, once the Heisenberg-evolved observable
    $O_\Phi:=\Phi^\dagger(O)$ is known, the cut wire can be replaced by an
    observable-adaptive measure-and-prepare channel with unit rescaling,
    while preserving the relevant downstream expectation value for arbitrary
    inputs. Moreover, if only an approximate classical description of
    $O_\Phi$ with operator-norm error at most $\eta$ is available, the
    resulting cut still has unit rescaling and incurs an additive
    bias of at most $2\eta$ in that expectation value
    (Theorems~\ref{thm:existence_of_no_cost_cut}
    and~\ref{thm:zero_cost_timelike_cut_with_an_approximated_unitary}).
\end{enumerate}
Intuitively, Result~1 shows that the more accurately we know the downstream information relevant to the final observable, the smaller the bias introduced by the cut becomes. Of course, the primitive is not yet
constructive, since the required Heisenberg-evolved observable is typically not known in advance. This naturally leads to our second
ingredient.
\begin{enumerate}
\setcounter{enumi}{1}
    \item \textbf{\textit{Randomized learning of Heisenberg-evolved observables.}}
    To make the above primitive constructive, we develop randomized learning protocols that estimate $O_\Phi$ directly from local input-output experiments on the unknown downstream channel $\Phi$, with finite-sample guarantees in operator norm.
    Specifically, using probe states drawn from a state 2-design, when the cut system has dimension $d$, $O_\Phi$ can be estimated
    within operator-norm error $\eta$ with high probability using
    \begin{equation}
        N_{\rm loc}
        =
        \tilde{\mathcal{O}}\left(\frac{d^3}{\eta^2}\right)
    \end{equation}
    local input-output experiments
    (Theorem~\ref{thm:error_bounds_for_each_x}).
\end{enumerate}

Taken together, Results~1 and 2 show that one can construct an
observable-adaptive cut whose bias can be reduced by increasing the
number of local input-output experiments. In particular, after
$N_{\rm loc}$ such experiments, the bias introduced by the resulting cut scales as $\tilde{\mathcal O}(\sqrt{d^3/N_{\rm loc}})$.
A key advantage of this construction is that it replaces the multiplicative variance amplification of conventional cutting methods with a locally controllable additive bias; see Fig.~\ref{fig:introduction_figure}(b).
Since the bias at each cut can be reduced by improving the corresponding local learning accuracy, we can distribute the total error budget over the cuts and keep the accumulated
bias under control, without introducing a multiplicative accumulation of rescaling factors. 
This mechanism enables the following polynomial-scaling circuit-knitting protocol for tree-structured quantum
circuits.
\begin{enumerate}
\setcounter{enumi}{2}
    \item \textbf{\textit{Polynomial-scaling circuit knitting for tree-structured circuits.}}
    For two-layer trees with maximum branching factor $R$ and cut
    dimension at most $d$ on each edge, we prove that the expectation value of a product observable can be estimated within additive error
    $\epsilon$ with high probability using
    \begin{equation}
        \tilde{\mathcal{O}}\left(\frac{d^3R^3}{\epsilon^2}\right)
    \end{equation}
    total measurements, including the local learning cost
    (Theorem~\ref{thm:2-layer}).
\end{enumerate}
The mechanism behind this scaling is as follows. For each cut location, we perform local learning so that the resulting observable-adaptive cut introduces only a small bias $\eta$. 
On two-layer trees, these local biases accumulate additively, so it suffices to choose $\eta=\mathcal{O}(\epsilon/R)$. 
Summing the local learning cost over the $R$ cuts and the final estimation cost on the local devices gives the total complexity above. 
We also complement the analytical bounds with numerical runtime estimates in the two-layer setting; see Sec.~\ref{sec:numerics}.

We then extend the same idea to arbitrary finite tree-structured circuits (Fig.~\ref{fig:introduction_figure}(c)). 
Let $K$ be the actual number of cut edges (i.e., wires), equivalently the number of non-root clusters, and let $d$ be an upper bound on the cut-system dimension, or equivalently the bond dimension of each cut edge.
The learning-based cutting protocol estimates the expectation value of a product observable within additive error $\epsilon$ with high probability using
\begin{equation}
    \tilde{\mathcal O}\!\left(
    \frac{d^3K^3}{\epsilon^2}
    \right)
\end{equation}
measurements (Theorem~\ref{thm:tree_simulation_general}). More generally, if the cut dimensions are non-uniform, the bound becomes
\begin{equation}
    \mathcal O\!\left(
    \frac{K^2}{\epsilon^2}
    \sum_{e\in E_{\rm cut}}d_e^3
    \ln\!\left(
    \frac{Kd_e}{\delta}
    \right)
    \right)
\end{equation}
where $E_{\rm cut}$ is the set of cut edges and $d_e$ is the bond dimension of cut edge $e$
\footnote{In the formal statement, we use the equivalent vertex notation $d_v$, where each non-root vertex $v$ is identified with its incoming cut edge.}.
Thus, this result shows that arbitrary finite tree-structured circuits can be cut with total measurement cost polynomial in the actual number of cuts $K$ and in the cut dimension.
The protocol achieves this by recursively learning, at each cluster, only the downstream information relevant to the target observable, thereby compressing each processed subtree into a single Heisenberg-evolved observable on the parent edge. 
By allocating the local learning accuracies across the tree, the accumulated bias remains controlled without introducing a multiplicative rescaling factor across the cuts.

One may still wonder whether this polynomial scaling is a consequence of the tree structure alone, and whether a conventional learning-free cutting method could achieve comparable scaling on the same family. 
We rule this out by proving an information-theoretic lower bound on the measurement cost of learning-free wire cutting, as illustrated in Fig.~\ref{fig:introduction_figure}(d).
\begin{enumerate}
\setcounter{enumi}{3}
    \item \textbf{\textit{Exponential separation from learning-free wire cutting.}}
    For the same two-layer tree task as in Result~3, we prove an information-theoretic lower bound for learning-free wire-cutting methods. 
    Here, learning-free means that each cut is performed using a fixed, pre-specified randomized wire-cutting rule; this class includes standard QPD-based wire-cutting schemes, including optimal wire cutting with classical communication. We show that any such method requires
    \begin{equation}
        \Omega\left(\frac{(d+1)^R}{\epsilon^2}\right)
    \end{equation}
    measurements to achieve constant success probability
    (Theorem~\ref{thm:learning-free-wire-cut-lower-bound}).
\end{enumerate}
This shows that the advantage of our protocol is not merely a consequence of restricting attention to tree-structured circuits. 
Rather, it comes from the learning-based, observable-adaptive construction of the cuts. In standard learning-free wire cutting, the measurement and preparation rules at each cut are fixed in advance rather than chosen from learned downstream information.
In our protocol, the relevant downstream information is first learned and then used to choose the cut basis adaptively. 
Thus, the lower bound shows that the polynomial scaling of Result~3 is not obtained from the tree structure alone, but from the adaptive use of learned downstream information.

Finally, we discuss how our results relate to several major lines of prior work.
\begin{enumerate}
\setcounter{enumi}{4}
    \item \textbf{\textit{Relation to prior work and broader applications.}}
    We clarify how our expectation-value-level construction differs from standard QPD-based cutting and why it is consistent with channel-level QPD lower bounds and recent circuit-knitting no-go results~\cite{PhysRevA.111.012433,marshall2023all}. 
    We also compare our observable-adaptive learning protocol with the graph-dependent cluster-simulation framework of Ref.~\cite{peng2021simulating}.
    We further contrast our local-experiment access model with classical treewidth-based tensor-network simulation~\cite{doi:10.1137/050644756}. 
    Finally, we show that our randomized learning framework gives finite-shot guarantees for a representative hybrid tree tensor-network contraction~\cite{PhysRevLett.127.040501,Harada2025densitymatrix,PRXQuantum.6.010320}, suggesting applications beyond circuit cutting (Sec.~\ref{sec:related_works}).
\end{enumerate}
As part of this comparison, we examine the cut-number dependence of the measurement overhead in the graph-dependent cluster-simulation framework of Ref.~\cite{peng2021simulating}.
By revisiting the finite-shot analysis underlying that framework, we derive an explicit measurement-overhead bound for representative rooted-tree instances in terms of the number of cut edges and the cut dimensions.
This reformulation makes the cut-number dependence explicit: the graph-dependent approach of Ref.~\cite{peng2021simulating} exhibits polynomial dependence on the actual number of cuts for specific tree families, such as paths and bounded-branching trees, but this polynomial-in-cuts behavior is not uniform over all finite trees. High-degree trees, such as stars, can still lead to exponential dependence.

The remainder of this paper is organized as follows.
In Sec.~\ref{sec:preliminaries}, we briefly review the background of conventional quantum circuit cutting.
In Sec.~\ref{sec:existence_theorems}, we introduce the observable-adaptive wire-cut primitives with unit or near-unit rescaling.
In Sec.~\ref{sec:learning_protocol}, we develop randomized learning protocols for estimating the Heisenberg-evolved observables needed to construct these cuts.
In Sec.~\ref{sec:poly_circuit_knitting_tree}, we apply the resulting framework to tree-structured quantum circuits and prove polynomial measurement scaling in the number of cuts.
In Sec.~\ref{sec:exponential_separation}, we prove an exponential separation in the number of cuts between our learning-based protocol and conventional learning-free cutting methods in the two-layer setting.
In Sec.~\ref{sec:numerics}, we numerically compare quantum runtime estimates in the two-layer case with prior approaches.
In Sec.~\ref{sec:related_works}, we discuss the relationship between our results and the major lines of related work. 
Finally, in Sec.~\ref{sec:coclusion}, we present our conclusions and discuss future directions.

\section{Preliminaries}\label{sec:preliminaries}

\subsection{Notation}\label{sec:notation}
We use the following notation in this work. 
For a quantum system $A$, $\Hilb{A}$ denotes the finite-dimensional Hilbert space of $A$. 
Throughout this paper, we consider qubit systems, and when writing the dimension as $d$, we assume $d=2^n$ for some number of qubits $n$.
For a $d$-dimensional complex vector space $\Vect$, we write $\Linear{\Vect}$ for the set of linear operators acting on $\Vect$, $\Herm{\Vect}$ for the set of Hermitian operators in $\Linear{\Vect}$, $\Density{\Vect}$ for the set of density operators in $\Herm{\Vect}$, and $\Unit{\Vect}$ for the unitary group in $\Linear{\Vect}$. 
We denote by $I_d$ the identity operator on $\mathbb{C}^d$, and by $\mathrm{id}^n$ the identity channel on $\Linear{\Vect}$; when the dimension is clear from context, we simply write $I$ or $\mathrm{id}$. 
For $d \in \mathbb{N}$, we use the shorthand $[d] := \{1,2,\dots,d\}$. All matrix norms $\|\cdot\|_{p}$ are taken to be Schatten $p$-norms.

\subsection{Quantum circuit cut}\label{sec:time_like_cut}

A quantum circuit cut is a technique that partitions a large quantum circuit into a collection of smaller subcircuits that can be executed on available quantum hardware~\cite{Mitarai_2021,peng2021simulating}. 
Such a partition can be implemented through the technique called quasiprobability simulation, which has been extensively employed across a wide range of quantum information and computation tasks, including quantum error mitigation~\cite{RevModPhys.95.045005,PhysRevLett.119.180509,PhysRevX.8.031027,PhysRevLett.119.180509,PhysRevX.8.031027}, classical simulation of quantum systems~\cite{PhysRevLett.115.070501, PhysRevLett.118.090501, doi:10.1098/rspa.2019.0251, PRXQuantum.2.010345}, and quantum circuit cutting. 
The key technical idea of the quasiprobability simulation is to express the target operation as a linear combination of implementable ones, referred to as a quasiprobability decomposition (QPD). 
Suppose that, due to physical limitations or experimental costs, the available set of quantum operations is restricted to a family $F$, while the desired operation $\mathcal{T}$ lies outside this set. 
Then, we represent $\mathcal{T}$ as a linear combination of operations $\{\mathcal{E}_i\}_i\subseteq F$:
\begin{equation}\label{eq:QPD}
    \mathcal{T} = \sum_i a_i \mathcal{E}_i,
\end{equation}
where the coefficients $a_i$ are real numbers that may take negative values. 
Building on this decomposition, quasiprobability simulation provides an unbiased estimator of the circuit output under $\mathcal{T}$ using only the circuit outputs under $\{\mathcal{E}_i\}_i$ (see Appendix~\ref{sec:QPD} for more detailed background). 
In the case of a quantum circuit cut, two types of QPDs are applied to entangling operations.
The first technique replaces the multi-qubit identity channel (i.e., wire) with a linear combination of channels consisting of a measurement followed by a state preparation, possibly together with classical post-processing.
This procedure is referred to as a wire cut, or equivalently, a timelike cut, as it effectively separates quantum operations along the temporal direction of the circuit~\cite{peng2021simulating,Lowe2023fastquantumcircuit,brenner2023optimal,harada2025doubly,pednault2023alternative,PRXQuantum.6.010316}.
The second technique, known as a gate cut or spacelike cut, decomposes an entangling unitary operation into a linear combination of local unitaries acting on spatially separated subsystems~\cite{Mitarai_2021,Mitarai2021overheadsimulating,10236453,Ufrecht2023cuttingmulticontrol,PRXQuantum.6.010316,schmitt2024cutting}.
In the following, we provide the key background of the wire cut, as it is more directly relevant to our main result.

In the wire cut, the target operation to be replaced is the $n$-qubit identity channel, denoted by $\mathrm{id}^n$.
The identity channel is defined as the completely positive and trace-preserving (CPTP) map that transmits any input exactly to itself without alteration. Formally, $\mathrm{id}^n: \Linear{\mathbb{C}^d} \rightarrow \Linear{\mathbb{C}^d}$ is defined by
\begin{equation}\label{eq:def_of_identity_channel}
    \mathrm{id}^n(X)=X \qquad \forall X \in \Linear{\mathbb{C}^d}.
\end{equation}
In contrast, the implementable operations are restricted to quantum channels that can be simulated using only local operations and classical communication (LOCC) between devices.
Specifically, we consider a special class of quantum channels that can be realized through the following sequence of steps: 
(i) the sender performs a measurement on the input state $\rho$, 
(ii) communicates the measurement outcome to the receiver via a classical channel, and 
(iii) the receiver prepares a quantum state conditioned on the received outcome. 
Such channels are known as \textit{measure-and-prepare} (MP) channels (or equivalently, \textit{entanglement-breaking} channels)~\cite{doi:10.1142/S0129055X03001709}, and they take the general form $\sum_{\mu} \mathrm{tr}(E_\mu \rho) \sigma_\mu$ where $\sigma_\mu$ are quantum states, and $\{E_\mu\}_\mu$ denotes a positive operator-valued measure (POVM) satisfying $\sum_\mu E_\mu=I$ and $E_\mu\geq0$ for all $\mu$. 
Moreover, when classical post-processing is allowed, the measurement
outcome $\mu$ can be assigned a scalar weight $c_\mu$, typically $c_\mu\in\{\pm1\}$. 
It is then convenient to absorb this classical weight into the linear representation
\begin{equation}\label{def:virtual_mp}
    \mathcal M(\rho)
    =
    \sum_\mu
    c_\mu \operatorname{tr}(E_\mu\rho)\sigma_\mu .
\end{equation}
Strictly speaking, when some $c_\mu$ are negative, $\mathcal M$ is not itself a CPTP channel.
Operationally, it is implemented by an MP channel followed by classical post-processing. 
We refer to such linear representations as post-processed MP maps.
In wire cutting, this channel is used as the decomposition basis.
In summary, a wire cut represents the identity channel as a linear combination of post-processed MP maps:
\begin{equation}\label{eq:QPD_of_timelike_cut}
    \mathrm{id}^n = \sum_i a_i \mathcal{M}_i.
\end{equation}

In this way, the QPD-based circuit cut provides an unbiased estimator of the desired expectation value, but the statistical error is amplified by a rescaling (multiplicative) factor of $\gamma:=\sum_{i}|a_i|$. According to Hoeffding's inequality~\cite{hoeffding1963probability}, to estimate the expectation value with an additive error $\epsilon$ with high probability, it suffices to perform 
\begin{equation}
    N=\mathcal{O}\left(\frac{\gamma^2}{\epsilon^2}\right)
\end{equation}
iterations. Hence, the sampling overhead scales quadratically with the $\ell_1$ norm of the coefficients. In practice, partitioning a quantum circuit to fit within available hardware typically requires multiple applications.
Let $\gamma_k$ denote the rescaling factor associated with the $k$-th replacement.
Since the estimator weights multiply across independent decompositions, the overall sampling overhead scales as $\prod_{k=1}^{K}\gamma_k^2$.

Thus, it is of practical importance to minimize the rescaling factor $\gamma$, and a number of approaches have been proposed to reduce the $\gamma$ in wire-cutting protocols~\cite{Lowe2023fastquantumcircuit,brenner2023optimal,harada2025doubly,pednault2023alternative,PRXQuantum.6.010316}. 
However, in standard channel-level wire cutting, it is known that
$\gamma$ is fundamentally lower bounded by $2d-1$, as established in the context of quantum circuit cutting~\cite{brenner2023optimal} and independently in the resource theory of quantum memory~\cite{yuan2021universal}. 
Similar nonunit rescaling factors also arise in QPD-based gate cuts for entangling operations~\cite{10236453,schmitt2024cutting,PRXQuantum.6.010316}. 
Consequently, cutting a quantum circuit at multiple locations based on QPDs inevitably induces an exponential measurement overhead, because the rescaling factors multiply across cuts. For example, when $K$ $d$-dimensional wires are
cut and an expectation value is estimated within additive error $\epsilon$, Hoeffding's inequality gives the standard sufficient measurement scaling $\mathcal{O}(d^{2K}/\epsilon^2)$~\cite{Lowe2023fastquantumcircuit,brenner2023optimal,harada2025doubly,pednault2023alternative,PRXQuantum.6.010316}. 
This exponential overhead poses a major bottleneck for the practical application of quantum circuit cutting in near-term quantum computation.

\section{Observable-adaptive rescaling-free circuit cuts}\label{sec:existence_theorems}

As reviewed in Sec.~\ref{sec:time_like_cut}, QPD represents a target operation as a linear combination of experimentally implementable operations, $\mathcal{T}=\sum_{i}a_i \mathcal{E}_i$, with the rescaling factor quantified by $\sum_i |a_i|$.
However, considering that the goal is merely to recover the expectation value of an observable, i.e., $\mathrm{tr}[O\mathcal{T}(\rho)]=\sum_i a_i \mathrm{tr}[O\mathcal{E}_i(\rho)]$, such a full channel-level decomposition is unnecessarily expressive in general. 
In this section, we demonstrate that in wire cutting the rescaling factor can be reduced dramatically by relaxing the requirement from a channel-level decomposition to an expectation-value-level decomposition; indeed, we prove the existence of a decomposition whose rescaling factor is unity in Sec.~\ref{subsec:existence}.
To connect this existence result to practical protocols, in Sec.~\ref{subsec:approximate} we further investigate approximate decompositions that inherit similar structural properties.

\subsection{Existence of a rescaling-free wire cut}\label{subsec:existence}

In Sec.~\ref{sec:time_like_cut}, we introduced the identity channel as the map that transmits any quantum state to itself without alteration, i.e., $\mathrm{id}^n(X)=X$.
While this operational definition is straightforward, it is often more insightful to characterize the identity channel in terms of its invariance under the Hilbert-Schmidt inner product.
Specifically, a linear map $\mathcal{E}$ coincides with the identity channel $\mathrm{id}^n$ defined in Eq.~\eqref{eq:def_of_identity_channel} if and only if 
\begin{equation}\label{eq:identity_equiv_def}
    \mathrm{tr}[YX] = \mathrm{tr}[Y \mathcal{E}(X)]\quad\text{$\forall X,Y \in \Linear{\mathbb{C}^d}$}
\end{equation}
\begin{proof}
For any fixed $X\in \Linear{\mathbb{C}^d}$, define the linear functional 
\begin{equation}
    f_X(Y):=\mathrm{tr}[Y(\mathcal{E}(X)-X)].
\end{equation}
By the assumption, $f_X(Y)=0$ for all $Y\in \Linear{\mathbb{C}^d}$. Since the Hilbert-Schmidt inner product is nondegenerate, this implies $\mathcal{E}(X)-X=0$. 
As this holds for every $X\in \Linear{\mathbb{C}^d}$, we conclude $\mathcal{E}=\mathrm{id}^n$.
\end{proof}

\begin{figure}[t]
\centering
\begin{center}
 \includegraphics[width=85mm]{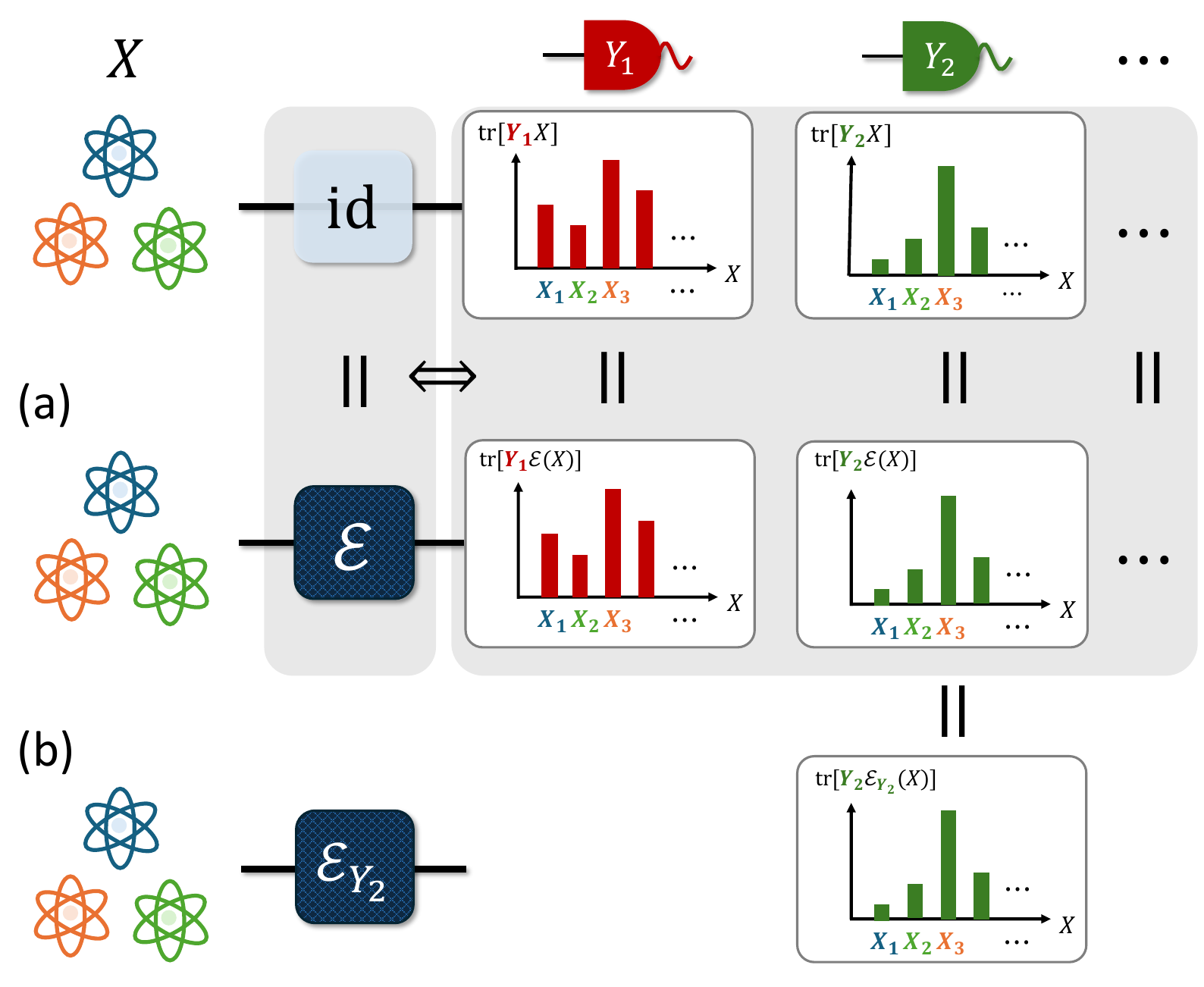}
\end{center}
\vspace{-10pt}
\caption{(a) Characterization of the identity channel $\mathrm{id}$ in terms of the invariance under the Hilbert-Schmidt inner product (Eq.~\eqref{eq:identity_equiv_def}). A linear map $\mathcal{E}$ coincides with $\mathrm{id}$ if and only if it preserves the expectation value of any observable $Y$ for any input $X$. (b) Schematic figure of the map $\mathcal{E}_{Y_2}$ defined in Eq.~\eqref{eq:identity_constrained}. We focus on the decomposition of the map $\mathcal{E}_{Y_2}$ that preserves the expectation value of a fixed observable $Y_2$ for any input $X$. \label{fig:def_of_identity}}
\end{figure}

Eq.~\eqref{eq:identity_equiv_def} thus indicates that $\mathrm{id}^n$ is the unique channel that preserves the expectation value of every observable $Y$ for any (possibly unknown) input operator $X$ (Fig.~\ref{fig:def_of_identity}(a)).
From the perspective of wire cuts, this interpretation clarifies the origin of the sampling overhead $\gamma^2 = (2d - 1)^2$, which scales quadratically with the input dimension $d$ of the system on which $\mathrm{id}^n$ acts: reproducing the outcome statistics of all observables requires an MP decomposition that recovers the full information of the input quantum state, and such completeness inevitably entails a large sampling overhead $\gamma^2$.

This observation naturally raises the following question: 
\textit{If one is only interested in reproducing the expectation values of a fixed observable for arbitrary unknown inputs, can the rescaling factor $\gamma$ be reduced below $2d-1$?}
Formally, consider an MP decomposition of the form $\mathcal{E}_{Y} = \sum_{i} a_i \mathcal{M}_i$ with $\mathcal{M}_i \in \mathrm{MP}$ and $a_i\in \mathbb{R}$.  
We ask whether, for a fixed $Y \in \Linear{\mathbb{C}^d}$,
\begin{equation}\label{eq:identity_constrained}
    \mathrm{tr}[YX] = \mathrm{tr}[Y \mathcal{E}_Y(X)] \quad \forall\, X \in \Linear{\mathbb{C}^d},
\end{equation}
the minimal $\gamma:= \sum_i |a_i|$ can be made smaller than $2d - 1$; see Fig.~\ref{fig:def_of_identity}(b) for the graphical representation of the map $\mathcal{E}_{Y}$.
The following theorem provides an affirmative answer to this question.

\begin{theorem}[Existence of a rescaling-free wire cut]\label{thm:existence_of_no_cost_cut}
Let $\Phi:\Linear{\mathbb{C}^{d_{\rm in}}} \to \Linear{\mathbb{C}^{d_{\rm out}}}$ be a CPTP map, and let $O \in \Herm{\mathbb{C}^{d_{\rm out}}}$ be a Hermitian operator. 
Then there exists a MP channel $\mathcal{M}_{\rm id}:\Linear{\mathbb{C}^{d_{\rm in}}} \to \Linear{\mathbb{C}^{d_{\rm in}}}$ such that, for any $X\in \Linear{\mathbb{C}^{d_{\rm in}}}$,
\begin{equation}
    \mathrm{tr}\left[O \,\Phi (X)\right] = \mathrm{tr}\left[ O\, (\Phi \circ \mathcal{M}_{\rm id}) (X) \right].
\end{equation}
An explicit construction of $\mathcal{M}_{\rm id}$ is given by
\begin{equation}\label{eq:gamma_optimal_timelike_cut}
    \mathcal{M}_{\rm id}(X) = \sum_{j} \mathrm{tr}\left[V \pure{j} V^{\dagger} X\right] V\pure{j} V^{\dagger},
\end{equation}
where $\{\ket{j}\}$ denotes the computational basis, and $V$ is a unitary that diagonalizes $O_{\Phi}:=\Phi^{\dagger}(O)$, i.e., \begin{equation}\label{eq:eigenvalue_decomposition_of_O}
    O_{\Phi}:= \Phi^{\dagger}(O) = V D V^{\dagger}.
\end{equation}
\end{theorem}
The proof of Theorem~\ref{thm:existence_of_no_cost_cut} is provided in Appendix~\ref{sec:proof_of_existence_of_no_cost_cut}. We note that, although $\Phi$ is defined as a CPTP map in the theorem statement for generality, a concrete instance of $\Phi$ can take the form
\begin{equation}
    \Phi(X) = U( X \otimes \pure{0^m} ) U^{\dagger}
\end{equation}
where $U$ is a unitary and $m \in \mathbb{Z}_{\geq 0}$ denotes the number of ancilla qubits; see Fig.~\ref{fig:exist_theorems} (a).

Theorem~\ref{thm:existence_of_no_cost_cut} establishes the existence of a rescaling-free wire cut:
for any observable $O$, there exists a MP channel $\mathcal{M}_{\rm id}$ that leaves the expectation value $\mathrm{tr}[O\,\Phi(X)]$ invariant for all input operators $X$. 
In other words, there exists a wire cut whose rescaling factor is unity, i.e., $\gamma=1$, in contrast to the conventional wire cut whose overhead scales as $\gamma = 2d - 1$ (Fig.~\ref{fig:exist_theorems} (b)(c)).
The key ingredient underlying the construction of the MP channel $\mathcal{M}_{\rm id}$ is the unitary $V$ that diagonalizes the effective observable $O_{\Phi} := \Phi^{\dagger}(O)$. 
Operationally, $\mathcal{M}_{\rm id}$ projects the input onto the eigenbasis of $O_{\Phi}$, measures in that basis, and then re-prepares the corresponding eigenstate; thereby transmitting precisely the information relevant to $O_{\Phi}$.

\subsection{Approximate version of the rescaling-free wire cut}\label{subsec:approximate}

It should be emphasized that Theorem~\ref{thm:existence_of_no_cost_cut} is an existence result. 
The explicit construction of 
$\mathcal{M}_{\rm id}$ requires knowledge of the unitary $V$ that diagonalizes 
$O_\Phi$. 
However, exact knowledge of $O_{\Phi}$, i.e., an error-free classical description, is typically unavailable unless $\Phi$ admits a classical 
description or possesses a known structure (e.g., Clifford circuits).
Nevertheless, even without complete knowledge of $O_\Phi$, having only an imperfect approximation $\tilde{O}_\Phi$ with bounded error is sufficient to construct a rescaling-free wire cut
that faithfully reproduces the desired expectation values up to a bounded bias.
Theorem~\ref{thm:zero_cost_timelike_cut_with_an_approximated_unitary} 
formalizes this statement.

\begin{figure}[t]
\centering
\begin{center}
 \includegraphics[width=85mm]{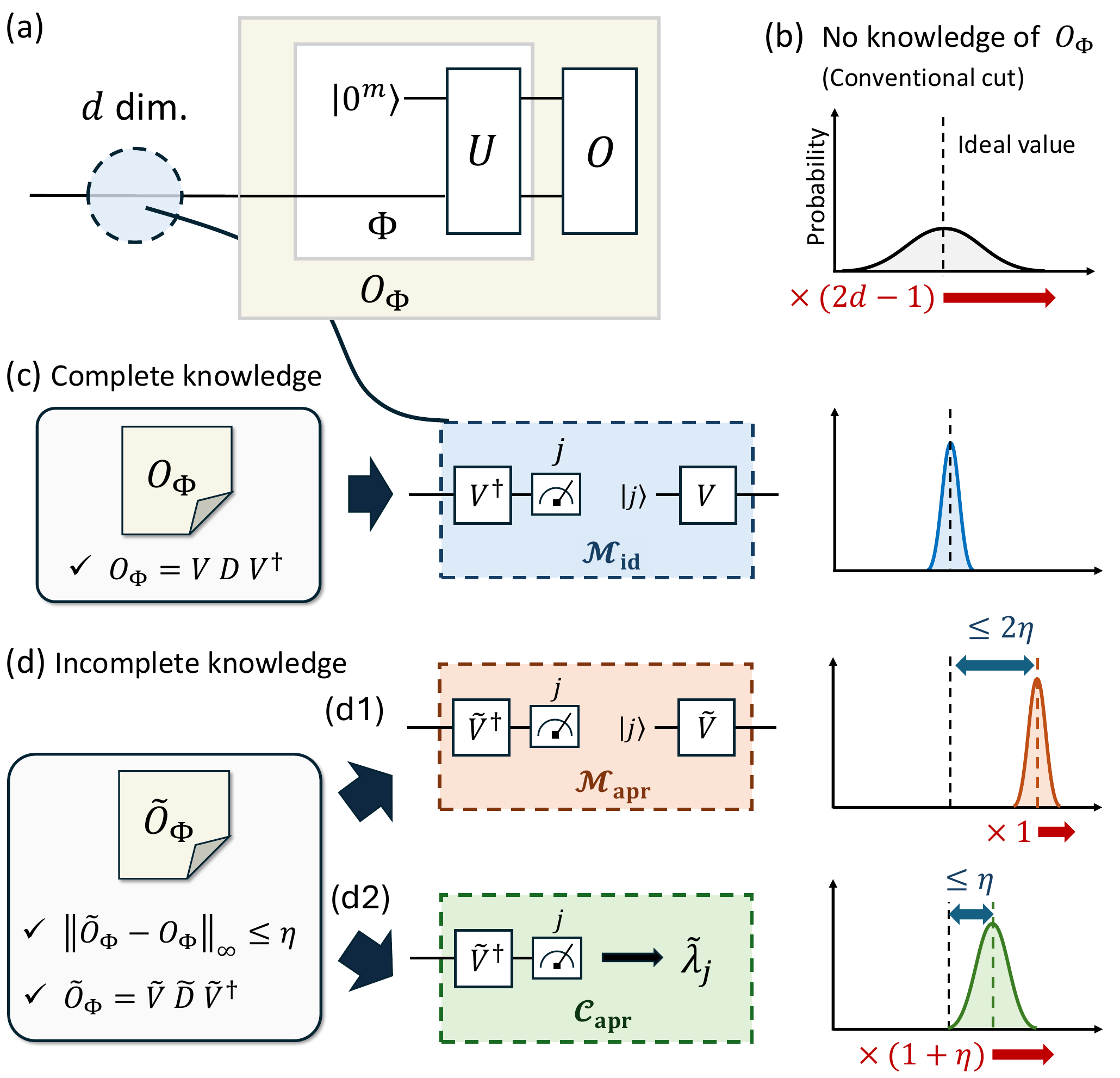}
\end{center}
\vspace{-15pt}
\caption{\label{fig:exist_theorems}
(a) Original quantum circuit before decomposition. 
Here, $\Phi$ denotes a CPTP map, and $O$ is a Hermitian operator satisfying $\|O \|_{\infty}\leq1$. 
We consider a decomposition of the identity channel enclosed by the blue dashed circle. 
(b) When the conventional wire cut in Eq.~\eqref{eq:QPD_of_timelike_cut} is applied to the dashed circuit in (a), a rescaling factor $2d-1$ is required, which consequently increases the variance.
(c) There exists an MP channel $\mathcal{M}_{\rm id}$ that preserves the unbiasedness of the ideal expectation value without introducing any additional rescaling factor, i.e., without increasing the variance (Theorem~\ref{thm:existence_of_no_cost_cut}). 
(d1) When the dashed circle is replaced by an MP channel $\mathcal{M}_{\rm apr}$ constructed from an approximate effective observable $\tilde{O}_{\Phi}$ satisfying $\| \tilde{O}_{\Phi} - O_{\Phi} \|_{\infty}\leq \eta$, the bias in the expectation value is bounded by $2\eta$, and no additional rescaling factor is required (Theorem~\ref{thm:zero_cost_timelike_cut_with_an_approximated_unitary}).
(d2) When a classical post-processing function $\mathcal{C}_{\rm apr}$ constructed from $\tilde{O}_{\Phi}$ is used, the bias is bounded by $\eta$, and the additional rescaling factor is at most $1+\eta$ (Theorem~\ref{thm:zero_cost_proccessing}).
}
\end{figure}

\begin{theorem}
[Approximate rescaling-free wire cut]
\label{thm:zero_cost_timelike_cut_with_an_approximated_unitary} 
Let $\Phi:\Linear{\mathbb{C}^{d_{\rm in}}} \to \Linear{\mathbb{C}^{d_{\rm out}}}$ be a CPTP map, and let $O \in \Herm{\mathbb{C}^{d_{\rm out}}}$ be a Hermitian operator.
Denote by $O_{\Phi}:=\Phi^{\dagger}(O)$ the corresponding effective observable, and let $\tilde{O}_{\Phi}$ be its Hermitian approximation satisfying $\|\tilde{O}_{\Phi} - O_{\Phi} \|_{\infty} \leq \eta$. 
Define the MP channel $\mathcal{M}_{\rm apr}:\Linear{\mathbb{C}^{d_{\rm in}}} \to \Linear{\mathbb{C}^{d_{\rm in}}}$ given by 
\begin{equation} 
    \mathcal{M}_{\rm apr}(X) := \sum_{j} \mathrm{tr} \left[ \tilde{V} \pure{j} \tilde{V}^{\dagger} X \right] \tilde{V} \pure{j} \tilde{V}^{\dagger}, 
\end{equation} 
where $\{\ket{j}\}$ denotes the computational basis, and $\tilde{V}$ is a unitary that diagonalizes $\tilde{O}_{\Phi}$ 
(i.e., $\tilde{O}_{\Phi}=\tilde{V}\tilde{D} \tilde{V}^{\dagger}$). 
Then for any  $X\in \Linear{\mathbb{C}^{d_{\rm in}}}$,
\begin{equation}
    \Bigl| \mathrm{tr}[O (\Phi \circ \mathcal{M}_{\rm apr}) (X)] - \mathrm{tr}[O\Phi(X)] \Bigr| \leq 2\eta \|X \|_{1}.
\end{equation}
\end{theorem}
The proof of Theorem~\ref{thm:zero_cost_timelike_cut_with_an_approximated_unitary} can be found in Appendix~\ref{sec:proof-zero_cost_timelike_cut_with_an_approximated_unitary}.
In particular, for any input state $\rho \in\Density{\mathbb{C}^{d_{\rm in}}}$, this theorem guarantees the bounded bias as
\begin{equation}
    \left|  \mathrm{tr}[O (\Phi \circ \mathcal{M}_{\rm apr}) (\rho)] - \mathrm{tr}[O\Phi(\rho)] \right| \leq 2\eta.
\end{equation}
This argument is schematically depicted in Fig.~\ref{fig:exist_theorems}(d1).

We remark that the above strategy assumes that one has access to a reasonably accurate classical description of the effective observable $O_{\Phi}$. Hence, in order for Theorem~\ref{thm:zero_cost_timelike_cut_with_an_approximated_unitary} to work meaningfully, it is required to construct a learning protocol for $O_{\Phi}$, with controlled accuracy in the operator norm. In the next section, we establish a practical learning procedure that enables such estimation. 

Theorem~\ref{thm:zero_cost_timelike_cut_with_an_approximated_unitary} naturally motivates a constructive strategy for realizing rescaling-free wire cuts in practice. 
For each cut location satisfying the conditions of Theorem~\ref{thm:zero_cost_timelike_cut_with_an_approximated_unitary}, one may proceed as follows: (i) estimate the effective observable $O_{\Phi}^{(k)}$ of the target subcircuit and obtain an estimator $\tilde{O}_{\Phi}^{(k)}$, (ii) diagonalize $\tilde{O}_{\Phi}^{(k)}$ to obtain the corresponding unitary $\tilde{V}^{(k)}$ such that $\tilde{O}_{\Phi}^{(k)}=\tilde{V}^{(k)} \tilde{D}^{(k)} \tilde{V}^{(k)\dagger}$, and (iii) construct the local MP channel $\mathcal{M}_{\rm apr}^{(k)}$ from this basis. 
By controlling the estimation error of each $\tilde{O}_{\Phi}^{(k)}$ so that the accumulated bias remains negligible, one can achieve a faithful global reconstruction of the circuit’s output expectation value while maintaining a rescaling factor of unity. 

From an implementation point of view, however, constructing such an observable-adaptive cut requires additional resources. 
To build $\mathcal M_{\rm apr}$, one first needs an approximate classical description $\tilde O_\Phi$ of the effective observable $O_\Phi$.
In Sec.~\ref{sec:main_performance_guarantee}, we show how to obtain such an approximation from local input-output experiments on $\Phi$, with finite-sample operator-norm guarantees. 
Once $\tilde O_\Phi$ has been obtained, one must diagonalize the $d_{\rm in}\times d_{\rm in}$ matrix $\tilde O_\Phi$, which costs $\mathcal O(d_{\rm in}^3)$ classical time. 
Implementing the resulting basis-change unitary may require $\mathcal O(d_{\rm in}^2)$ elementary gates in the worst case~\cite{PhysRevLett.92.177902,app12020759}.
Thus, when only a single cut is applied, standard optimal wire-cutting methods may remain preferable, since they are tomography-free and can be implemented with $\mathcal O(\mathrm{polylog}(d_{\rm in}))$ elementary gates~\cite{harada2025doubly,pednault2023alternative,PRXQuantum.6.010316}.

However, the distinctive advantage of the present construction is that, by removing the per-cut rescaling factor, it shifts the main trade-off from multiplicative amplification of the variance to additive bias control.
This change becomes particularly significant when multiple cuts are applied.
Conventional QPD-based wire cuts accumulate rescaling factors $\gamma=\mathcal O(d_{\rm in})$, so the variance typically scales as $\gamma^{2K}=\mathcal O(d_{\rm in}^{2K})$, leading to an exponential increase in the required number of measurements. 
In contrast, the observable-adaptive construction introduces local biases whose sizes can be controlled through learning. 
Using the finite-sample operator-norm guarantees for the learning protocols established in Sec.~\ref{sec:main_performance_guarantee}, we can choose the learning accuracy at each cut, or equivalently assign local bias budgets, so that the accumulated bias remains below the target additive error, i.e., the error parameter $\epsilon$ in the final expectation-value estimate. 
In Sec.~\ref{sec:poly_circuit_knitting_tree}, we carry out this bias allocation for tree-structured quantum circuits and obtain polynomial
measurement scaling of the form $\tilde{\mathcal{O}}(\mathrm{poly}(d_{\rm in},K)/\epsilon^2)$. 

Up to this point, we have considered constructing an approximate rescaling-free wire cut based on Theorem~\ref{thm:zero_cost_timelike_cut_with_an_approximated_unitary}. 
In this framework, we utilize only the diagonalizing unitary $\tilde{V}$ obtained from the approximate observable, while we have not exploited the corresponding approximate eigenvalues $\tilde{\lambda}_j$ in $\tilde{D}$.
Nevertheless, these eigenvalues may contain additional information about the effective observable $O_{\Phi}$ that can be leveraged to improve estimation accuracy.
Indeed, as shown in the following theorem, the approximate eigenvalues $\tilde{\lambda}_j$ allow us to construct an estimator with a smaller bias than that obtained in Theorem \ref{thm:zero_cost_timelike_cut_with_an_approximated_unitary}.

\begin{theorem}[Post-processing from an approximate effective observable]\label{thm:zero_cost_proccessing}
Let $\Phi:\Linear{\mathbb{C}^{d_{\rm in}}} \to \Linear{\mathbb{C}^{d_{\rm out}}}$ be a CPTP map, and let $O \in \Herm{\mathbb{C}^{d_{\rm out}}}$ be a Hermitian operator.
Denote by $O_{\Phi}:=\Phi^{\dagger}(O)$ the corresponding effective observable, and let $\tilde{O}_{\Phi}$ be its Hermitian approximation satisfying $\|\tilde{O}_{\Phi} - O_{\Phi} \|_{\infty} \leq \eta$. Define the classical post-processing function
\begin{equation}\label{eq:post-processing_functional}
    \mathcal{C}_{\rm apr}(X) := \sum_{j} \mathrm{tr} \left[ \tilde{V} \pure{j} \tilde{V}^{\dagger} X \right] \tilde{\lambda}_{j},
\end{equation}
where $\{\ket{j}\}$ denotes the computational basis, $\tilde{V}$ is a unitary that diagonalizes $\tilde{O}_{\Phi}$ (i.e., $\tilde{O}_{\Phi}=\tilde{V}\tilde{D} \tilde{V}^{\dagger}$), and $\tilde{D}:=\sum_j \tilde{\lambda}_j \pure{j}$. 
Then, for any  $X\in \Linear{\mathbb{C}^{d_{\rm in}}}$,
\begin{eqnarray}
    \Bigl|  \mathcal{C}_{\rm apr}(X) - \mathrm{tr}[O\Phi(X)] \Bigr| 
    &\leq& \eta \|X \|_{1},\\
    \max_{j} |\tilde{\lambda}_j| \leq \|O \|_{\infty}+\eta.\label{eq:maximal_norm} 
\end{eqnarray}
\end{theorem}
The proof of Theorem~\ref{thm:zero_cost_proccessing} is provided in Appendix~\ref{sec:proof_of_zero_cost_processing}.
Similar to Theorem~\ref{thm:zero_cost_timelike_cut_with_an_approximated_unitary}, for any input state $\rho \in \Density{\mathbb{C}^{d_{\rm in}}}$ and a bounded observable $O$ satisfying $\|O\|_{\infty}\leq 1$, the theorem ensures that 
\begin{equation}
    \Bigl| \mathcal{C}_{\rm apr}(\rho) - \mathrm{tr}[O\Phi(\rho)] \Bigr| 
    \leq \eta
\end{equation}
and
\begin{equation}
    \max_{j} |\tilde{\lambda}_j| \leq 1+\eta.
\end{equation}
Fig.~\ref{fig:exist_theorems} (d2) shows the graphical expression of this statement.

Now, comparing Theorem~\ref{thm:zero_cost_timelike_cut_with_an_approximated_unitary} and Theorem~\ref{thm:zero_cost_proccessing}, we can see the trade-off between the two approaches.
Using the approximate MP channel $\mathcal{M}_{\rm apr}$ of Theorem~\ref{thm:zero_cost_timelike_cut_with_an_approximated_unitary}, one can estimate the expectation value with a unit rescaling factor (i.e., without additional multiplicative factor) but with an increased bias $2\eta$.
In contrast, the post-processing function $\mathcal{C}_{\rm apr}$ of Theorem~\ref{thm:zero_cost_proccessing} achieves a smaller bias $\eta$, albeit with a non-unity rescaling factor bounded by $1+\eta$, i.e., a non-unity multiplicative factor. 
At first sight, applying $\mathcal C_{\rm apr}^{(k)}$ at many cut locations may appear to reintroduce a multiplicative overhead, because each local post-processing function takes values bounded in magnitude by $1+\eta$. 
However, as shown in Sec.~\ref{sec:poly_circuit_knitting_tree}, choosing the local learning accuracy of order $\eta/R$ keeps the resulting product of ranges bounded by a
constant. 
Thus the two approaches have comparable sample efficiency when the effective observables are learned to the same accuracy.

In terms of ease of implementation, the post-processing construction of Theorem~\ref{thm:zero_cost_proccessing} 
offers a practical advantage.
Once a sufficiently accurate approximation $\tilde{O}_{\Phi}$ is available, the remaining estimation can be performed by measuring in the learned eigenbasis followed by classical post-processing, without implementing the downstream channel $\Phi$ or directly measuring $O$.
In addition, whereas the MP-channel construction of Theorem~\ref{thm:zero_cost_timelike_cut_with_an_approximated_unitary} requires a feedforward operation to realize $\mathcal{M}_{\rm apr}$, the post-processing construction of Theorem~\ref{thm:zero_cost_proccessing} does not.

Lastly, we note that Appendix~\ref{apsec:relation_to_wire_cut} provides a detailed discussion of how Theorem~\ref{thm:existence_of_no_cost_cut} relates to existing QPD-based wire-cutting methods.

\section{Learning Heisenberg-evolved observables under unknown quantum channels}\label{sec:learning_protocol}

In the previous section, Theorems~\ref{thm:zero_cost_timelike_cut_with_an_approximated_unitary} and~\ref{thm:zero_cost_proccessing} established that, given an accurate estimate of $O_\Phi$, one can construct either an approximately rescaling-free wire cut $\mathcal{M}_{\rm apr}$ or a post-processing function $\mathcal{C}_{\rm apr}$ with a unit or near-unit rescaling factor.
These constructions fundamentally rely on the knowledge of the effective observable $O_{\Phi}$, which corresponds to the Heisenberg-evolved version of the observable $O$ under the unknown quantum channel $\Phi$. 
In general, however, the channel $\Phi$ is available only through local quantum experiments, and a full classical description of $\Phi$ is not assumed.
In this section, we develop randomized learning protocols for estimating $O_{\Phi}$ directly from input-output experiments on the unknown channel.
The goal is not to learn the full channel $\Phi$, but to learn the single effective observable needed to construct the observable-adaptive cut.
We formulate this learning task in Sec.~\ref{sec:main_problem_setup}, introduce three concrete learning protocols in Sec.~\ref{sec:main_leaerning_protocol}, and provide finite-sample operator-norm guarantees in Sec.~\ref{sec:main_performance_guarantee}.

\subsection{Problem setup}\label{sec:main_problem_setup}

Suppose we are given a known Hermitian operator $O\in \Herm{\mathbb{C}^{d_{\rm out}}}$ and an unknown CPTP map $\Phi: \Linear{\mathbb{C}^{d_{\rm in}}} \to \Linear{\mathbb{C}^{d_{\rm out}}}$.
When the input system consists of $n$ qubits, we write $d_{\rm in}=2^n$.
We do not assume access to a full classical matrix representation of $\Phi$. Instead, $\Phi$ is available through a local quantum device: one can prepare an input state, apply the channel, and measure the output.
The observable $O$ admits the spectral decomposition
\begin{equation}\label{eq:svd_of_O}
    O = W \Lambda W^{\dagger} 
\end{equation}
where $\Lambda := \sum_{j} \nu_j\ket{j}\bra{j}$ is diagonal in the computational basis $\{\ket{j}\}$, and $W$ is a unitary operator.
Our goal is to obtain a classical description of $O_{\Phi}=\Phi^{\dagger}(O) \in \Herm{\mathbb{C}^{d_{\rm in}}}$, where $\Phi^{\dagger}$ denotes the adjoint map of $\Phi$. 
Formally, given a target accuracy $\eta > 0$, we aim to construct an estimator $\hat{O}_{\Phi}$ satisfying
\begin{equation}\label{eq:main_definition_of_effective_observable}
    \| \hat{O}_{\Phi} - O_{\Phi}\|_{\infty} \leq \eta
\end{equation}
with probability at least $1-\delta$. 

The learning primitive needed here is more targeted than full channel
tomography~\cite{chuang1997prescription,PhysRevLett.78.390} or measurement
tomography~\cite{PhysRevLett.93.250407,zambrano2025fast}.
To see this, set $\Pi_j := W|j\rangle\langle j|W^\dagger$, so that $O=\sum_j \nu_j\Pi_j$.
Applying $\Phi$ and measuring $O$ induces an
effective POVM $\{F_j:=\Phi^\dagger(\Pi_j)\}_j$ on the input system, but
the observable-adaptive cut only requires the single effective observable
\begin{equation}
    O_\Phi=\Phi^\dagger(O)=\sum_j \nu_j F_j .
\end{equation}
Thus, we learn $O_\Phi$ directly from local input-output experiments on
$\Phi$, without reconstructing $\Phi$ or the entire induced POVM. 
The essential requirement for circuit cutting is an operator-norm guarantee:
it is uniform over all upstream input states and can be inserted directly
into the cutting-error analysis. 

\subsection{Learning protocols}\label{sec:main_leaerning_protocol}

We introduce three protocols for learning the unknown observable $O_{\Phi}$.
Before presenting the specific procedures, we outline the common structure shared by all of them.
To obtain an estimator $\hat{O}_{\Phi}$ satisfying Eq.~\eqref{eq:main_definition_of_effective_observable}, we need to collect measurement data from the circuit composed of the unknown quantum channel $\Phi$ followed by the measurement of $O$.
We obtain such data by probing the circuit with an appropriately chosen ensemble of input states (i.e., a probability distribution over pure states):
\begin{equation}
    \mathcal{S} := \{ (p_i,\ket{\psi_i}) \}, 
\end{equation}
where $p_i$ denotes the sampling probability associated with the pure states $\ket{\psi_i}$.
As will be discussed below and in Appendix~\ref{sec:quantum-t-design}, the probe states $\ket{\psi_i}$ in all of our schemes can be prepared efficiently on a quantum device, and their classical descriptions can be stored efficiently in classical memory.
Given $\mathcal{S}$, we perform the following randomized procedure consisting of random sampling of input states, circuit measurement, and data storage:
\begin{enumerate}
    \item \textbf{Random sampling}: 
    Draw an input state $\ket{\psi_i}$ from $\mathcal{S}$ according to the probability distribution $p_i$.
    
    \item \textbf{Measurement}: 
    Prepare $\ket{\psi_i}$ on the input register, evolve it under $\Phi$, and measure the output using the observable $O$. Operationally, this corresponds to applying the unitary rotation $W^{\dagger}$ immediately before measurement in the computational basis $\{\ket{j}\}$.
    Denote the observed outcome by $j$, associated with the eigenvalue $\nu_j$ of $O$.
    
    \item \textbf{Data storage}: 
    Store the pair $(\ket{\psi_i}, \nu_j)$ in classical memory. 
\end{enumerate}
Repeating the above procedure $N$ times, we obtain an empirical dataset
\begin{equation}
    \mathcal{D}_N := \{(\ket{\psi_{i_k}}, \nu_{j_k})\}_{k=1}^N
\end{equation}
where $i_k$ and $j_k$ respectively denote the indices of the sampled input state and the corresponding measurement outcome in the $k$-th trial. 
With a suitable choice of the ensemble $\mathcal S$ and a sufficiently
large sample size $N$, the empirical dataset $\mathcal{D}_N$ encodes enough information to reconstruct $O_{\Phi}$.
Thus, by performing classical post-processing on $\mathcal{D}_N$, which depends on the choice of ensemble, we obtain an estimate of the classical description of $O_{\Phi}$. 

In the following, we present three concrete realizations of this framework.
The first two protocols employ ensembles that form (i) a state 2-design (Sec.~\ref{sec:main_2-design}), and (ii) the uniform tensor product of single-qubit stabilizer states (Sec.~\ref{sec:main_stab}).
Both protocols can be described within the framework presented above. 
We present these two realizations because each has its own strengths and weaknesses in terms of sample efficiency and ease of implementation.
Protocol (ii) is more hardware-friendly, as its input-state-preparation can be executed with a circuit depth of at most two, using only single-qubit Clifford gates.
In contrast, implementing the protocol (i) using a Clifford ensemble requires $\mathcal{O}(n^2/\log n)$ Clifford gates.
However, protocol (i) offers better sample efficiency, as established in Theorem~\ref{thm:error_bounds_for_each_x}.
We also introduce a third protocol in Sec.~\ref{sec:main_pauli}, which instead employs the uniform distribution over $n$-qubit Pauli operators and their eigenstates.
The input states are similar to those of the stabilizer protocol, but the inversion step is different.

We briefly describe how the unbiased estimators and their error bounds are derived.
To derive the unbiased estimators, we view Steps 1--3 as a randomized linear map from $O_{\Phi}$ to a classical matrix-valued output. By computing the expectation of this map and applying an appropriate linear inversion~\cite{Guta_2020,zambrano2025fast}, we obtain unbiased estimators of $O_{\Phi}$.
For the error analysis, we use the matrix Bernstein inequality~\cite{tropp2015introduction}.
The necessary background, tools for the derivation, and detailed proofs are provided in Appendix~\ref{apsec:learning_observable}.

\subsubsection{2-design ensemble}\label{sec:main_2-design}

Let $\mathcal{S}_{\text{2-dgn}}:=\{ (p_i,\ket{\psi_i}) \}$ be a state ensemble forming a state 2-design, i.e., it satisfies
\begin{equation}
    \sum_{i} p_i (\pure{\psi_i})^{\otimes t} = \int_{\phi} d\phi (\pure{\phi})^{\otimes t},\quad t=1,2
\end{equation}
where the integral is taken with respect to the normalized Haar measure on the unit sphere of $\mathbb{C}^{d_{\rm in}}$. 
As concrete examples, maximal mutually unbiased bases (MUBs) provide state 2-designs in prime-power dimensions, including $d_{\rm in}=2^n$. 
SIC-POVMs, when available, also give state 2-designs.
Moreover, the set $\{U\ket{0^n} : U\in \mathrm{Cl}(2^n)\}$, where $\mathrm{Cl}(2^n)$ denotes the $n$-qubit Clifford group, also forms a state 2-design, since $\mathrm{Cl}(2^n)$ is known to form a unitary 3-design~\cite{webb2016clifford3design} (see Appendix~\ref{sec:quantum-t-design} for additional background). 
In addition, their classical descriptions can be stored and sampled efficiently.

Using $\mathcal{S}_{\text{2-dgn}}$, we perform the following procedure for $k=1,...,N$:
\begin{enumerate}
    \item \textbf{Random sampling}: 
    Sample a pure state $\ket{\psi_{i_k}}$ from $\mathcal{S}_{\text{2-dgn}}$.
    
    \item \textbf{Measurement}: 
    Prepare $\ket{\psi_{i_k}}$, apply $\Phi$, and measure the observable $O$.
    Let $j_k$ be the observed outcome associated with eigenvalue $\nu_{j_k}$.
    
    \item \textbf{Data storage}: 
    Store the pair $(\ket{\psi_{i_k}},\nu_{j_k})$ in classical memory.
\end{enumerate}
After repeating this procedure $N$ times, we classically post-process each pair $(\ket{\psi_{i_k}},\nu_{j_k})$ as
\begin{equation}
    \hat{\omega}_{\text{2-dgn}}^{(k)}:=\nu_{j_k} \left( d_{\rm in}(d_{\rm in}+1)\pure{\psi_{i_k}} - d_{\rm in}I_{d_{\rm in}} \right).
\end{equation}
By averaging these matrices, we obtain
\begin{eqnarray}
    \hat{O}_{\Phi,\text{2-dgn}} 
    &:=& \frac{1}{N}\sum_{k=1}^{N} \hat{\omega}_{\text{2-dgn}}^{(k)}\\
    &=& \frac{1}{N}\sum_{k=1}^{N} \nu_{j_k} \left( d_{\rm in}(d_{\rm in}+1)\pure{\psi_{i_k}} - d_{\rm in}I_{d_{\rm in}} \right),\notag\\
\end{eqnarray}
which is an unbiased estimator of $O_{\Phi}$, i.e., $\mathbb{E}[\hat{O}_{\Phi,\text{2-dgn}}]=O_{\Phi}$. 
A proof of its unbiasedness is given in Appendix~\ref{apsec:const_2-design}.

\subsubsection{Tensor product single-qubit stabilizer states}\label{sec:main_stab}

We consider the input ensemble $\mathcal{S}_{\text{stab}}$, defined as the uniform distribution over all $n$-fold tensor products of single-qubit stabilizer states. 
Let $\ket{\pm}$ and $\ket{\pm i}$ denote the eigenstates of the Pauli operators $\rm X$ and $\rm Y$, respectively, and define the set of single-qubit stabilizer states as  
\begin{equation}
    \mathcal{Q}_{\text{stab}} :=\left\{\ket{0},\ket{1},\ket{+},\ket{-},\ket{+i},\ket{-i}\right\}.
\end{equation}
The ensemble $\mathcal{S}_{\text{stab}}$ is then formally given by
\begin{equation}
    \mathcal{S}_{\text{stab}}:=\left\{\left(\frac{1}{6^n}, \bigotimes_{s=1}^{n}\ket{u^s}\right): \ket{u^s} \in \mathcal{Q}_{\rm stab}\right\}
\end{equation}

Using $\mathcal{S}_{\text{stab}}$, 
we perform the following protocol for $k=1,\ldots,N$:
\begin{enumerate}
    \item \textbf{Random sampling}: 
    Sample a product state $\bigotimes_{s=1}^{n}\ket{u^s_k}$ from $\mathcal{S}_{\text{stab}}$.
    
    \item \textbf{Measurement}: 
    Prepare the state $\bigotimes_{s=1}^{n}\ket{u^s_k}$, apply $\Phi$, and measure the observable $O$.
    Let $j_k$ be the observed outcome corresponding to eigenvalue $\nu_{j_k}$.
    
    \item \textbf{Data storage}: 
    Store the pair $\left(\bigotimes_{s=1}^{n}\ket{u^s_k},\nu_{j_k}\right)$ in classical memory.
\end{enumerate}
After $N$ iterations, we perform classical post-processing on each pair $\left(\bigotimes_{s=1}^{n}\ket{u^s_k},\nu_{j_k}\right)$ and define
\begin{equation}\label{eq:estimator_stab}
     \hat{\omega}_{\text{stab}}^{(k)}:=\nu_{j_{k}}\bigotimes_{s=1}^{n} \left(6\pure{u_{k}^s}-2I \right).
\end{equation}
The empirical estimator is obtained as the sample average
\begin{eqnarray}
    \hat{O}_{\Phi,\text{stab}} 
    &:=& \frac{1}{N}\sum_{k=1}^{N} \hat{\omega}_{\text{stab}}^{(k)}\\
    &=& \frac{1}{N} \sum_{k=1}^{N} \nu_{j_{k}}\bigotimes_{s=1}^{n} \left(6\pure{u_{k}^s}-2I \right),
\end{eqnarray}
which is an unbiased estimator of $O_{\Phi}$.
A proof of its unbiasedness can be found in Appendix~\ref{apsec:const_stab}.

\subsubsection{$n$-qubit Pauli operator}\label{sec:main_pauli}

In the protocol based on a 2-design (Sec.~\ref{sec:main_2-design}) or a uniform tensor product of single-qubit stabilizer states (Sec.~\ref{sec:main_stab}), the unknown observable $O_{\Phi}$ is estimated by random sampling of quantum states from the corresponding ensemble. 
Here, we instead consider an alternative model where we randomly sample an $n$-qubit Pauli operator and one of its eigenstates as the input state. 

We begin by recalling the single-qubit Pauli operators and their eigenvalue decompositions:
\begin{equation}
    P_{i} = \sum_{e\in[2]} c_{i,e}\,\pure{v_{i,e}} \qquad i\in [4],
\end{equation}
where we set $P_1,P_2,P_3,P_{4}$ to be the identity and the Pauli operators $X$, $Y$, and $Z$, respectively. Their eigenstates and corresponding eigenvalues are given by
\begin{align}
\begin{alignedat}{10}
\ket{v_{1,1}} &= \ket{0}, &\quad c_{1,1} &= +1, &\quad \ket{v_{1,2}} &= \ket{1}, &\quad c_{1,2} &= +1, \\
\ket{v_{2,1}} &= \ket{+}, &\quad c_{2,1} &= +1, &\quad \ket{v_{2,2}} &= \ket{-}, &\quad c_{2,2} &= -1, \\
\ket{v_{3,1}} &= \ket{+i}, &\quad c_{3,1} &= +1, &\quad \ket{v_{3,2}} &= \ket{-i}, &\quad c_{3,2} &= -1, \\
\ket{v_{4,1}} &= \ket{0}, &\quad c_{4,1} &= +1, &\quad \ket{v_{4,2}} &= \ket{1}, &\quad c_{4,2} &= -1.
\end{alignedat}
\end{align}
From these local definitions, we define the $n$-qubit Pauli ensemble as
\begin{equation}
    \mathcal{S}_{\rm Pauli} := \left\{ \left( \frac{1}{4^n},\, 
    P_{\bm{i}} := \bigotimes_{s=1}^{n} P_{i_s} \right) : \bm{i}\in [4]^n
    \right\},
\end{equation}
where $\bm{i}:=(i_1,\dots,i_n)$. 
Unlike $\mathcal{S}_{\text{2-dgn}}$ and $\mathcal{S}_{\text{stab}}$, $\mathcal{S}_{\text{Pauli}}$ is an operator ensemble, namely the uniform ensemble over $n$-qubit Pauli operators, rather than a state ensemble.
Each $n$-qubit Pauli operator admits the eigenvalue decomposition
\begin{equation}
\begin{split}
\label{eq:eigenvalue_decompositions_of_n_pauli}
    P_{\bm{i}}&:=\sum_{\substack{\bm{e}
    :=(e_1,...,e_n)\in[2]^n}} c_{\bm{i},\bm{e}}\pure{v_{\bm{i},\bm{e}}},\\
    c_{\bm{i},\bm{e}} 
    &:= \prod_{s=1}^{n} c_{i_s,e_s}, \qquad
    \ket{v_{\bm{i},\bm{e}}}:= \bigotimes_{s=1}^{n} \ket{v_{i_s,e_s}}.
\end{split}
\end{equation}

We now consider the following sampling procedure, repeated for $k=1,...,N$, where a Pauli operator and one of its eigenstates are chosen at random:
\begin{enumerate}
    \item \textbf{Random sampling}: 
    Sample a Pauli operator $P_{\bm{i}_k}$ uniformly at random, and sample one of its eigenstates $\ket{v_{\bm{i}_k,\bm{e}_k}}$ uniformly.
    \item \textbf{Measurement}: 
    Prepare $\ket{v_{\bm{i}_k, \bm{e}_k}}$, evolve it under $\Phi$, and measure the observable $O$. Denote the measurement outcome by $j_k$.
    \item \textbf{Data storage}: 
    Store the triple $\left( P_{\bm{i}_k},\ket{v_{\bm{i}_k,\bm{e}_k}}, j_k \right)$ in classical memory.
\end{enumerate}
After repeating the above procedure $N$ times, we classically post-process each stored triple $\left( P_{\bm{i}_k},\ket{v_{\bm{i}_k,\bm{e}_k}}, j_k \right)$ as
\begin{equation}\label{eq:estimator_pauli}
    \hat{\omega}_{\text{pauli}}^{(k)} := 4^n \nu_{j_k} c_{\bm{i}_k,\bm{e}_k} P_{\bm{i}_k}.
\end{equation}
By averaging over these matrices, we obtain the empirical estimator
\begin{eqnarray}
    \hat{O}_{\Phi,\text{pauli}}
    &:=& \frac{1}{N}\sum_{k=1}^{N} 
    \hat{\omega}_{\text{pauli}}^{(k)} \\
    &=& \frac{1}{N} \sum_{k=1}^{N} 4^n \nu_{j_k} c_{\bm{i}_k,\bm{e}_k} P_{\bm{i}_k}.
\end{eqnarray}
The resulting estimator is unbiased with respect to $O_{\Phi}$; see Appendix~\ref{apsec:const_pau} for a detailed proof.

\subsection{Performance guarantees}\label{sec:main_performance_guarantee}
We establish a performance guarantee for the empirical estimator $\hat{O}_{\Phi,x}$ for each $x\in\{\text{2-dgn, stab, pauli}\}$, introduced in Sec.~\ref{sec:main_leaerning_protocol}.

\begin{theorem}[Performance guarantee for $\hat{O}_{\Phi,x}$]
\label{thm:error_bounds_for_each_x}
Let $\Phi:\Linear{\mathbb{C}^{d_{\rm in}}} \to \Linear{\mathbb{C}^{d_{\rm out}}}$ be an unknown CPTP map, and let $O \in \Herm{\mathbb{C}^{d_{\rm out}}}$ be a known Hermitian operator.
Define $O_\Phi:=\Phi^\dagger(O)$.
Let $\hat{O}_{\Phi,x}$ be the estimator defined in Sec.~\ref{sec:main_leaerning_protocol} for $x\in\{\text{2-dgn},\text{stab},\text{pauli}\}$.
Then, for each $x$, it suffices to take $N$ at least as follows to guarantee $\|\hat{O}_{\Phi,x}-O_{\Phi}\|_{\infty}\le \eta$ with probability at least $1-\delta$.

\noindent{$x=\text{2-dgn}$:}
\begin{equation}
    \frac{2 \max\{1,\|O\|^2_{\infty}\}(d_{\rm in}^2+1) (d_{\rm in}+1+\eta/3) } {\eta^2} \ln\left(\frac{2d_{\rm in}}{\delta}\right),\notag
\end{equation}
\noindent{$x=\text{stab}$:}
\begin{equation}
    \frac{2 \max\{1,\|O\|_{\infty}^2\}(d_{\rm in}^{3.33}+1+(\eta/3)(d_{\rm in}^2+1)) } {\eta^2} \ln{\left(\frac{2d_{\rm in}}{\delta}\right)},\notag
\end{equation}
\noindent{$x=\text{pauli}$:}
\begin{equation}
    \frac{2 \max\{1,\|O\|_{\infty}^2\}(d_{\rm in}^4+1+(\eta/3)(d_{\rm in}^2+1)) } {\eta^2} \ln{\left(\frac{2d_{\rm in}}{\delta}\right)}.\notag
\end{equation}
\end{theorem}
For a detailed proof, see Appendix~\ref{appsec:performance_guarantee}.

\section{Polynomial-scaling circuit knitting for tree quantum circuits}\label{sec:poly_circuit_knitting_tree}

In this section, we turn the single-cut constructions of Theorems~\ref{thm:zero_cost_timelike_cut_with_an_approximated_unitary} and~\ref{thm:zero_cost_proccessing} into full circuit-knitting protocols for tree-structured circuits. 
Using the observable-learning procedures of Sec.~\ref{sec:main_leaerning_protocol}, we show that the measurement cost is polynomial in the actual number of cut edges.
In Sec.~\ref{sec:main_tree_setup}, we formulate the general problem setup.
In Sec.~\ref{sec:2-layer-tree}, we first analyze the two-layer star case, which illustrates how local learning turns the multiplicative rescaling overhead of wire cutting into a controllable additive bias.
In Sec.~\ref{sec:main-multi-layer}, we extend the analysis to arbitrary finite rooted trees, including nonuniform branching and nonuniform cut dimensions.

\subsection{General problem setup}\label{sec:main_tree_setup}

\begin{figure}[t]
\centering
\begin{center}
 \includegraphics[width=85mm]{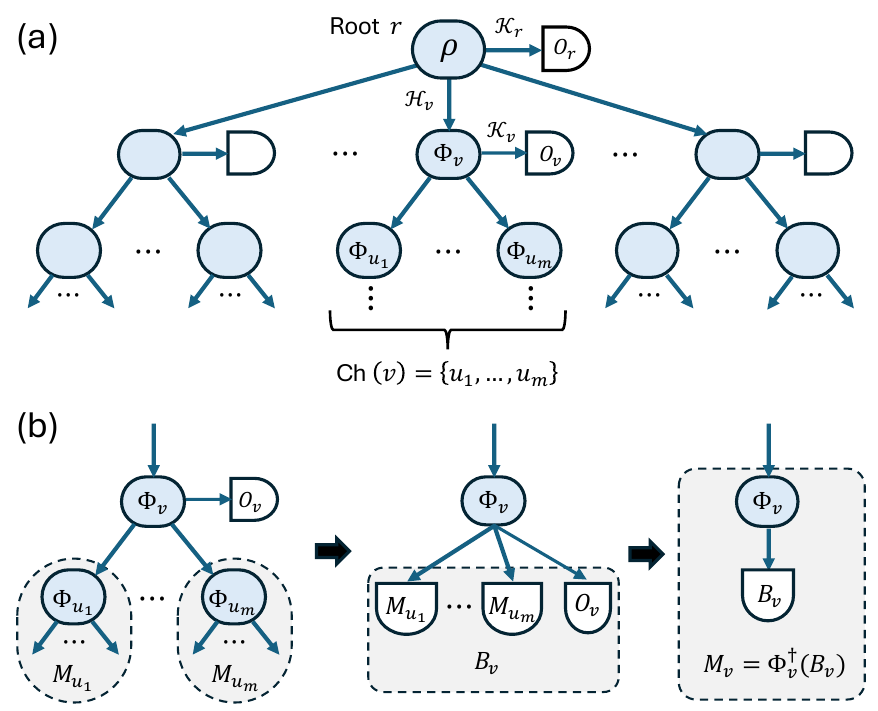}
\end{center}
\vspace{-15pt}
\caption{
General tree setting and local bottom-up recursion.
(a) A rooted tree of local components: the root carries the input state $\rho$, while each non-root vertex $v$ represents a local CPTP map $\Phi_v$ with incoming cut system $\mathcal H_v$.
The children of $v$ are denoted by $\mathrm{Ch}(v)=\{u_1,\ldots,u_m\}$.
(b) Local Heisenberg-picture recursion at a vertex $v$.
}\label{fig:notation}
\end{figure}

We formulate tree-structured circuit cutting on a finite rooted tree.
Let $T=(V,E)$ be a finite rooted tree with a unique root $r$, and write
\begin{equation}
    V^\circ:=V\setminus\{r\},\qquad K:=|V^\circ|.
\end{equation}
Here, $T$ represents the coarse-grained tree associated with the chosen cut locations. 
Each edge of $T$ carries a cut system connecting two neighboring local components.
Thus, $K$ is the number of selected cuts, equivalently the number of non-root local components.
For each non-root vertex $v\in V^\circ$, let $\mathcal H_v$ denote the Hilbert space carried by the edge from the parent
of $v$ to $v$, and let
\begin{equation}
    d_v:=\dim\mathcal H_v,\qquad d_v\le d .
\end{equation}
For each vertex $v\in V$, we allow a terminal output register $\mathcal K_v$ equipped with a known Hermitian observable
\begin{equation}
    O_v\in \Herm{\mathcal{K}_v},
    \qquad 
    \|O_v\|_\infty \leq 1 .
\end{equation}
A typical example of $O_{v}$ is a Pauli string such as $X \otimes Y \otimes \cdots \otimes Z$.
If no terminal observable is attached to $v$, we set $\mathcal K_v=\mathbb C$ and $O_v=1$.
This convention allows branches to
terminate at different depths and also allows observables to be attached to intermediate clusters.

Let $\mathrm{Ch}(v)$ denote the set of children of $v$; in particular, $\mathrm{Ch}(v)=\emptyset$ when $v$ is a leaf.
Each non-root
vertex $v\in V^\circ$ (circuit node in Fig.~\ref{fig:notation}(a)) is associated with an unknown CPTP map
\begin{equation}
    \Phi_v: \mathsf{L}\!\left(\mathcal{H}_v\right) 
    \to
    \mathsf{L}
    \left(
    \mathcal K_v\otimes \bigotimes_{u\in{\rm Ch}(v)}\mathcal{H}_u \right).
\end{equation} 
The empty tensor product is understood as $\mathbb C$.
The root node corresponds to an unknown quantum state
\begin{equation}
    \rho\in \mathsf{D}\!\left( \mathcal K_r\otimes \bigotimes_{u\in{\rm Ch}(r)}\mathcal H_u \right).
\end{equation} 

Let $\rho_{\rm tree}$ be the state generated by composing these local maps along the tree, before the final measurement of the terminal observables. 
Let
\begin{equation}
    O:=\bigotimes_{v\in V}O_v
\end{equation}
be the corresponding product observable. The goal is to estimate
\begin{equation}
    \mu:={\rm tr}[O\rho_{\rm tree}]
\end{equation}
under the constraint that we can access only local devices implementing $\rho$ and the maps $\{\Phi_v\}_{v\in V^\circ}$, together with local measurements required by the protocol.

Here we focus on a single product observable, which is a standard setting in circuit-cutting analyses of expectation-value estimation and includes individual Pauli strings. More general observables of the form $H=\sum_\alpha a_\alpha O^{(\alpha)}$, with $O^{(\alpha)}=\bigotimes_{v\in V}O_v^{(\alpha)}$, can be handled by linearity: one first allocates the measurement budget across the relevant product terms, for example in proportion to $|a_\alpha|$, then runs the protocol for each term and combines the resulting estimates with the coefficients $a_\alpha$.

For later use, we introduce the corresponding downstream and effective observables in the Heisenberg picture; see Fig.~\ref{fig:notation}(b). For each non-root vertex $v\in V^\circ$, define recursively the downstream observable
\begin{equation}
    B_v
    :=
    O_v\otimes
    \bigotimes_{u\in{\rm Ch}(v)} M_u ,
\end{equation}
with the convention that the empty tensor product is $1$, and set
\begin{equation}
    M_v:=\Phi_v^\dagger(B_v).
\end{equation}
Here, $B_v$ acts on the output space of $\Phi_v$, whereas $M_v$ acts on the incoming edge Hilbert space $\mathcal H_v$.
At the root, define
\begin{equation}
    B_r
    :=
    O_r\otimes
    \bigotimes_{u\in{\rm Ch}(r)} M_u .
\end{equation}
Then, the target expectation value can equivalently be written as
\begin{equation}
    \mu={\rm tr}[\rho B_r].
\end{equation}

\subsection{2-layer tree-structured quantum circuit}\label{sec:2-layer-tree}

\subsubsection{Estimation problem (2-layer)}

As a simple example illustrating the basic idea, we consider the two-layer, or star-shaped, special case. 
In this case, the root $r$ has $R$ children, which we label by $k=1,...,R$. Thus,
\begin{equation}
    V^{\circ}=\left\{1,\dots,R \right\},
    \qquad
    K=R.
\end{equation}
For simplicity, as in the standard two-layer setting, we assume that no terminal observable is attached to the root, i.e., $\mathcal{K}_r=\mathbb{C}$ and $O_r =1$. Each child node $k$ is a leaf and corresponds to a local CPTP map
\begin{equation}
    \Phi_k: \Linear{\mathcal{H}_k} \rightarrow \Linear{\mathcal{K}_k},
\end{equation}
where $\dim \mathcal H_k \leq d$. 
The root corresponds to an unknown quantum state $\rho\in \Density{\bigotimes_{k=1}^{R} \mathcal H_k}$ and the state generated by the two-layer tree is 
\begin{equation}
    \rho_{\rm tree} = \left( \Phi_1 \otimes \cdots \otimes \Phi_R \right)(\rho).
\end{equation}
We consider a product observable 
\begin{equation}
    O:=O_1 \otimes \cdots \otimes O_R,
    \qquad
    \|O_k\|_{\infty}\leq1\quad (k=1,...,R),
\end{equation}
and our goal is to estimate $\mu:=\mathrm{tr}[O\rho_{\rm tree}]$ within additive error $\epsilon\in(0,1]$ with success probability at least $1-\delta$. 
We assume that the available quantum hardware is strictly local in the circuit-knitting sense: we can access devices implementing the root state $\rho$ and each local map $\Phi_k$, but we do not require a coherent device implementing the full state $\rho_{\rm tree}$.

We note that, in the notation of Sec.~\ref{sec:main_tree_setup}, each child $k$ is a leaf. Thus, the downstream observable at $k$ is simply $B_k=O_k$
and the corresponding effective observable on the incoming edge is
\begin{equation}
    M_k:=\Phi_k^\dagger(O_k).
\end{equation}
Therefore, the root downstream observable is $B_r=\bigotimes_{k=1}^R M_k$ and
\begin{equation}\label{eq:mu_another}
    \mu=
    {\rm tr}\!\left[
    \left(
    \bigotimes_{k=1}^R M_k
    \right)\rho
    \right].
\end{equation}

\subsubsection{Estimation strategies}
We now analyze how this estimation task can be solved using observable-adaptive cutting protocols. 
We first examine Protocol A, the measure-and-prepare cut protocol derived from Theorem~\ref{thm:zero_cost_timelike_cut_with_an_approximated_unitary}. Recall that Theorem~\ref{thm:zero_cost_timelike_cut_with_an_approximated_unitary} guarantees the existence of an approximate MP channel $\mathcal{M}_{\rm apr}$ whose induced bias is controlled by the operator-norm error of an estimated effective observable. 
Combining this with the learning protocol in Sec.~\ref{sec:main_leaerning_protocol}, we obtain the following procedure.

\vspace{0.6em}
\noindent\textbf{Protocol A: measure-and-prepare implementation.}

\begin{description}
    \item[(a-1)] Estimate the effective observable $M_{k}:=\Phi_{k}^{\dagger}(O_{k})$ for each local subsystem $k=1,...,R$ using $N_{\mathrm{a},k}$ shots, and obtain an estimator $\tilde{M}_{k}$.
    \item[(a-2)] Diagonalize each estimator as 
    \begin{equation}
        \tilde{M}_k = \tilde{V}_k \tilde{D}_k \tilde{V}_k^{\dagger},
    \end{equation} 
    and construct the MP channel
    \begin{equation}
        \mathcal{M}_{\mathrm{apr}}^{(k)}(\bullet) := \sum_j \mathrm{tr} \left[ \tilde{V}_k \pure{j} \tilde{V}_k^{\dagger} \bullet \right] \tilde{V}_k \pure{j} \tilde{V}^{\dagger}_k,
    \end{equation}
    where $\{\ket{j}\}$ is the computational basis of $\mathcal{H}_k$.
    \item[(a-3)] Replace the identity channels connecting the root state $\rho$ and the local map $\Phi_k$ 
    by $\mathcal{M}_{\mathrm{apr}}^{(k)}$ for each $k=1,...,R$. 
    Execute the resulting quantum circuit for $N_{\mathrm{a},0}$ shots, and take the empirical mean $\hat{\mu}_{\rm a}$ as the estimator of $\mu$.
\end{description}
Step~\textbf{(a-3)} is implemented by measure-and-prepare operations across the cuts and therefore does not require coherent operations between the root device and the devices implementing $\Phi_k$. The estimator $\hat\mu_{\rm a}$, however, is biased because the MP channels are constructed from approximate effective observables. 
Thus the learning accuracies in Step~\textbf{(a-1)} must be chosen so that the accumulated bias is at most $\mathcal O(\epsilon)$.

Next, we consider Protocol~B, the post-processing protocol derived from Theorem~\ref{thm:zero_cost_proccessing}.
In this approach, the approximate eigenvalues of $\tilde M_k$ are also used, and the final estimation is performed by measuring the root state $\rho$ directly with classical post-processing.

\vspace{0.6em}
\noindent\textbf{Protocol B: post-processing implementation.}

\begin{description}
    \item[(b-1)] Estimate the effective observable $M_{k}:=\Phi_{k}^{\dagger}(O_{k})$ for each local subsystem $k=1,...,R$ using $N_{\mathrm{b},k}$ shots, and obtain an estimator $\tilde{M}_{k}$.
    \item[(b-2)] Diagonalize each estimator as
    \begin{equation}
        \tilde{M}_k:=\tilde{V}_k \tilde{D}_k \tilde{V}_k^{\dagger},\qquad
        \tilde{D}_k:=\sum_j \tilde{\lambda}_{k,j} \pure{j},
    \end{equation}
    and construct the measurement with classical post-processing
    \begin{equation}
        \mathcal{C}_{\mathrm{apr}}^{(k)}(\bullet) := \sum_j \mathrm{tr} \left[ \tilde{V}_k \pure{j} \tilde{V}^{\dagger}_k \bullet \right] \tilde{\lambda}_{k,j}.
    \end{equation}
    \item[(b-3)] Measure the root state $\rho$ in the product basis specified by $\{\tilde V_k\}_{k=1}^{R}$, weight each outcome $(j_1,...,j_R)$ by
    \begin{equation}
        \prod_{k=1}^{R} \tilde\lambda_{k,j_k},
    \end{equation}
    and take the empirical mean $\hat{\mu}_{\rm b}$ over $N_{\rm b,0}$ shots as the estimator.
\end{description}
The conditional mean of the estimator in Step~\textbf{(b-3)} is
\begin{equation}
    {\rm tr}\!\left[
    \left(
    \bigotimes_{k=1}^R \tilde M_k
    \right)\rho
    \right].
\end{equation}
Thus, this protocol also has a bias, now coming from replacing $M_k$ by $\tilde M_k$.
In addition, the range of the single-shot output is governed by $\|\bigotimes_k \tilde M_k\|_\infty$. 
The analysis below shows that by choosing the local learning accuracy as $\mathcal{O}(\epsilon/R)$, this range remains bounded by a constant, so no exponential sampling overhead appears in the two-layer protocol.

\subsubsection{Measurement complexity analysis}
In what follows, we upper bound the number of measurements required by Protocols A and B.
Let
\begin{equation}
    N_{{\rm a,tot}}
    =
    \sum_{k=0}^R N_{{\rm a},k},
    \qquad
    N_{{\rm b,tot}}
    =
    \sum_{k=0}^R N_{{\rm b},k},
\end{equation}
denote the total number of measurements used by Protocols A and B, respectively.

\begin{theorem}[Two-layer tree circuits]
\label{thm:2-layer}
Protocols A and B estimate $\mu={\rm tr}[O\rho_{\rm tree}]$ up to additive error $\epsilon\in(0,1]$ with probability at least $1-\delta$, provided that the
numbers of measurements satisfy, for $k=1,\ldots,R$,
\begin{equation}
\begin{split}
    N_{{\rm a},k}
    =
    N_{{\rm b},k}
    &=
    \mathcal O\!\left(
    \frac{d^3R^2}{\epsilon^2}
    \ln\!\left(
    \frac{Rd}{\delta}
    \right)
    \right),\\
    N_{{\rm a},0}
    =
    N_{{\rm b},0}
    &=
    \mathcal O\!\left(
    \frac{1}{\epsilon^2}
    \ln\!\left(
    \frac{1}{\delta}
    \right)
    \right).
\end{split}
\end{equation}
Consequently,
\begin{equation}
    N_{{\rm a,tot}}
    =
    N_{{\rm b,tot}}
    =
    \mathcal O\!\left(
    \frac{d^3R^3}{\epsilon^2}
    \ln\!\left(
    \frac{Rd}{\delta}
    \right)
    \right).
\end{equation}
\end{theorem}

We remark that the above measurement complexity assumes the use of the learning procedure
based on a $2$-design ensemble (Sec.~\ref{sec:main_2-design}). If we replace
it with tensor-product single-qubit stabilizer states
(Sec.~\ref{sec:main_stab}) or $n$-qubit Pauli operators
(Sec.~\ref{sec:main_pauli}), the total number of measurements becomes
\begin{equation}
    N_{{\rm a,tot}}
    =
    N_{{\rm b,tot}}
    =
    \mathcal O\!\left(
    \frac{d^{3.33}R^3}{\epsilon^2}
    \ln\!\left(
    \frac{Rd}{\delta}
    \right)
    \right)
\end{equation}
and
\begin{equation}
    N_{{\rm a,tot}}
    =
    N_{{\rm b,tot}}
    =
    \mathcal O\!\left(
    \frac{d^4R^3}{\epsilon^2}
    \ln\!\left(
    \frac{Rd}{\delta}
    \right)
    \right),
\end{equation}
respectively.

Both Protocols A and B fit the circuit-knitting/cutting setting. They require access only to the local device preparing the root state $\rho$ and to the local devices implementing $\{\Phi_k\}_{k=1}^R$. 
No coherent implementation of the
full state $\rho_{\rm tree}$ is required; the devices are combined only through classical communication and classical post-processing.

In the following, we sketch the proof of Theorem~\ref{thm:2-layer}; a detailed
version is provided in Appendix~\ref{sec:proof_2_layer}.

\begin{proof}[Proof sketch of Theorem~\ref{thm:2-layer}]
We first analyze Protocol A and then Protocol B.

\vspace{0.5em}
\noindent\textbf{1. Analysis of Protocol A.}
\vspace{0.1em}

Assume that the learned effective observables in Step~\textbf{(a-1)} satisfy 
\begin{equation}
    \|\tilde M_k-M_k\|_\infty\le \eta_{\rm a},
    \qquad
    k=1,\ldots,R.
\end{equation}
Under this assumption, the bias of Step~\textbf{(a-3)} is
\begin{equation}
    {\rm Bias}_{\rm a}
    :=
    \left|
    {\rm tr}\!\left[
    O
    \left(
    \bigotimes_{k=1}^R
    \Phi_k\circ \mathcal M_{\rm apr}^{(k)}
    \right)(\rho)
    \right]
    - \mu \right|.
\end{equation}
Using the telescoping identity
(Lemma~\ref{lemma:bound_the_op_norm_of_products}) and
Theorem~\ref{thm:zero_cost_timelike_cut_with_an_approximated_unitary}, we can express the bias solely in terms of the quantity $\eta_{\rm a}$ as
\begin{equation}
    {\rm Bias}_{\rm a} \leq 2R\eta_{\rm a}.
\end{equation}
Thus, by choosing $\eta_{\rm a}=\frac{9\epsilon}{20R}$, the total bias ${\rm Bias}_{\rm a}$ is at most $\frac{9\epsilon}{10}$.
As a consequence, after implementing the MP channels $\mathcal{M}_{\rm apr}^{(k)}$ at $R$ locations, an additional estimation with accuracy $\frac{\epsilon}{10}$ suffices to ensure that the overall estimation error is at most $\epsilon$:~\footnote{The $\frac{9\epsilon}{10}$--$\frac{\epsilon}{10}$ split is only for concreteness and for consistency with the numerical estimates in Sec.~\ref{sec:numerics}; any fixed constant split gives the same asymptotic scaling.}
\begin{eqnarray}
    \left| \hat{\mu}_{\rm a} - \mu \right|
    \leq \mathrm{Bias}_{\rm a} + \frac{\epsilon}{10} \leq \epsilon.
\end{eqnarray}
We have identified sufficient events under which the final estimator
$\hat\mu_{\rm a}$ achieves additive error at most $\epsilon$: 
the sufficient event for Step~\textbf{(a-1)} is
\begin{equation}
    \bigcap_{k=1}^R
    \left\{
    \|\tilde M_k-M_k\|_\infty\le \eta_{\rm a}
    \right\},\quad \eta_{\rm a}=\frac{9\epsilon}{20R}
\end{equation}
and the sufficient event for Step~\textbf{(a-3)} is
\begin{equation}
    E_{0}^{\rm a}:=\left\{\left|
    \hat\mu_{\rm a}
    -
    {\rm tr}\!\left[
    O
    \left(
    \bigotimes_{k=1}^R
    \Phi_k\circ \mathcal M_{\rm apr}^{(k)}
    \right)(\rho)
    \right]
    \right|
    \le
    \frac{\epsilon}{10} \right\}.
\end{equation}
The last statistical event $E_{0}^{\rm a}$ is conditioned on the learned observables $\{\tilde M_k\}$, but each single-shot outcome in Step~\textbf{(a-3)} is
bounded in $[-1,1]$, because the final measured observable is $O=\bigotimes_{k=1}^{R} O_k$ with $\|O_k\|_\infty\leq1$. 
Thus, Hoeffding's inequality holds uniformly over the learning outcomes, and the conditioning can be removed by the tower property.

By assigning failure probability $\frac{\delta}{2R}$ to each of the $R$ learning
events and $\frac{\delta}{2}$ to the final statistical event, Theorem~4 gives
\begin{equation}
    N_{{\rm a},k}
    =
    \mathcal O\!\left(
    \frac{d^3R^2}{\epsilon^2}
    \ln\!\left(
    \frac{Rd}{\delta}
    \right)
    \right),
\end{equation}
and Hoeffding's inequality gives
\begin{equation}
    N_{{\rm a},0}
    =
    \mathcal O\!\left(
    \frac{1}{\epsilon^2}
    \ln\!\left(
    \frac{1}{\delta}
    \right)
    \right).
\end{equation}

\vspace{0.5em}
\noindent\textbf{2. Analysis of Protocol B.}
\vspace{0.1em}

Assume that the learned effective observables in Step~\textbf{(b-1)} satisfy
\begin{equation}
    \|\tilde M_k-M_k\|_\infty\le \eta_{\rm b},
    \qquad
    k=1,\ldots,R.
\end{equation}
Under this assumption, the bias of Step~\textbf{(b-3)} is
\begin{equation}
    {\rm Bias}_{\rm b}
    :=
    \left|
    {\rm tr}\!\left[
    \left(
    \bigotimes_{k=1}^R\tilde M_k
    \right)\rho
    \right]
    -
    \mu
    \right|.
\end{equation}
Since
\begin{equation}
    \mu
    =
    {\rm tr}\!\left[
    \left(
    \bigotimes_{k=1}^R M_k
    \right)\rho
    \right],
\end{equation}
the telescoping identity (Lemma~\ref{lemma:bound_the_op_norm_of_products}) gives
\begin{equation}
    {\rm Bias}_{\rm b}
    \leq
    R\eta_{\rm b}(1+\eta_{\rm b})^{R-1}.
\end{equation}
Using the inequality $(1+x)^r\le 1+(e-1)rx$ for
$0\le x\le 1/r$ (Lemma~\ref{lemma:bound_exponential}), we find that choosing $\eta_{\rm b}=\epsilon/[2R(e-1)]$ yields
\begin{equation}
    {\rm Bias}_{\rm b}\le \frac{\epsilon}{2}.
\end{equation}
Thus, it is sufficient to perform the estimation with accuracy $\epsilon/2$ in Step~\textbf{(b-3)} to ensure that the overall estimation error is at most $\epsilon$:
\begin{eqnarray}
    \left| \hat{\mu}_{\rm b} - \mu \right| 
    &\leq& \mathrm{Bias}_{\rm b}+\frac{\epsilon}{2} \leq \epsilon.
\end{eqnarray}

We have thus identified sufficient events under which the final estimator $\hat{\mu}_{\rm b}$ achieves additive error at most $\epsilon$: the sufficient event for Step~\textbf{(b-1)} is
\begin{equation}
    T
    :=
    \bigcap_{k=1}^R
    \left\{
    \|\tilde M_k-M_k\|_\infty\le \eta_{\rm b}
    \right\},\quad \eta_{\rm b}=\frac{\epsilon}{2R(e-1)}
\end{equation}
and the final statistical event for Step~\textbf{(b-3)} is
\begin{equation}
    E_0^{\rm b}
    :=
    \left\{
    \left|
    \hat\mu_{\rm b}
    -
    {\rm tr}\!\left[
    \left(
    \bigotimes_{k=1}^R\tilde M_k
    \right)\rho
    \right]
    \right|
    \le
    \frac{\epsilon}{2}
    \right\}.
\end{equation}
On the event $T$, since
$\|M_k\|_\infty=\|\Phi_k^\dagger(O_k)\|_\infty\le \|O_k\|_\infty\le 1$, we have $\|\tilde M_k\|_\infty \leq \|M_k\|_\infty+\|\tilde M_k-M_k\|_\infty\leq1+\eta_{\rm b}$.
Therefore,
\begin{equation}
\label{eq:norm_upper_bound_two_layer_revised}
    \left\|
    \bigotimes_{k=1}^R\tilde M_k
    \right\|_\infty
    \le
    (1+\eta_{\rm b})^R
    \le
    1+\frac{\epsilon}{2}
    \le
    \frac{3}{2}.
\end{equation}
Thus, conditioned on the learning outcomes and on $T$, each single-shot
output in Step~\textbf{(b-3)} is bounded by a constant. Hoeffding's inequality and the tower property yield
\begin{equation}
    {\rm Pr}\!\left[
    (E_0^{\rm b})^c\cap T
    \right]
    \le
    2\exp\!\left(
    -\frac{N_{{\rm b},0}\epsilon^2}{18}
    \right).
\end{equation}

By assigning failure probability $\frac{\delta}{2R}$ to each of the $R$ learning
events and $\frac{\delta}{2}$ to the final statistical event on $T$, Theorem~\ref{thm:error_bounds_for_each_x} gives
\begin{equation}
    N_{{\rm b},k}
    =
    \mathcal O\!\left(
    \frac{d^3R^2}{\epsilon^2}
    \ln\!\left(
    \frac{Rd}{\delta}
    \right)
    \right),
\end{equation}
and Hoeffding's inequality gives
\begin{equation}
    N_{{\rm b},0}
    =
    \mathcal O\!\left(
    \frac{1}{\epsilon^2}
    \ln\!\left(
    \frac{1}{\delta}
    \right)
    \right).
\end{equation}
This proves the claimed measurement bounds for both protocols.
\end{proof}

Finally, we comment on the implementation cost of $\mathcal{C}_{\rm apr}^{(k)}$ and $\mathcal{M}_{\rm apr}^{(k)}$.
Both protocols require a classical diagonalization of the learned effective observable $\tilde M_k$, which acts on $\mathcal H_k$ and has dimension at most $d$. 
This diagonalization determines the basis unitary $\tilde V_k$, which is used to implement the MP channel $\mathcal{M}^{(k)}_{\rm apr}$ in Protocol A and to perform the measurement with classical post-processing $\mathcal C^{(k)}_{\rm apr}$ in Protocol B.
Since the classical diagonalization cost is $\mathcal O(d^3)$ per subsystem, the total cost is $\mathcal O(d^3R)$ for the two-layer tree. 
In the worst case, the basis-change unitary $\tilde V_k$ can be implemented using $\mathcal O(d^2)$ elementary gates~\cite{PhysRevLett.92.177902,app12020759}.

\subsection{General tree-structured quantum circuits}\label{sec:main-multi-layer}

We extend the two-layer protocol to the general finite tree setting formulated in Sec.~\ref{sec:main_tree_setup}. The relevant complexity parameter is the actual number $K$ of cut edges, equivalently the number of non-root clusters. 
This formulation naturally includes trees with non-uniform branching, branches terminating at different depths, and terminal observables attached to intermediate clusters.

\begin{theorem}[General tree-structured circuits; main-text statement]
\label{thm:tree_simulation_general}
Consider the finite-rooted tree setup of Sec.~\ref{sec:main_tree_setup}. 
Let $K:=|E|$ be the number of cut edges\footnote{
If only a selected subset \(E_{\rm cut}\) of edges in an underlying rooted-tree circuit is cut, the uncut connected components can be regarded as local components.
This yields a tree of the form considered here, and the theorem applies with \(K=|E_{\rm cut}|\).
}, 
and let $d$ be an upper bound on the bond dimensions of all cut edges. 
Then there exists a learning-based bottom-up protocol that estimates $\mu={\rm tr}[O\rho_{\rm tree}]$ within additive error $\epsilon\in(0,1]$ with probability at least $1-\delta$ using
\begin{equation}
    \mathcal O\!\left(
    \frac{d^3K^3}{\epsilon^2}
    \ln\!\left(
    \frac{Kd}{\delta}
    \right)
    \right)
\end{equation}
measurements.
\end{theorem}

\begin{figure}[t]
\centering
\begin{center}
 \includegraphics[width=85mm]{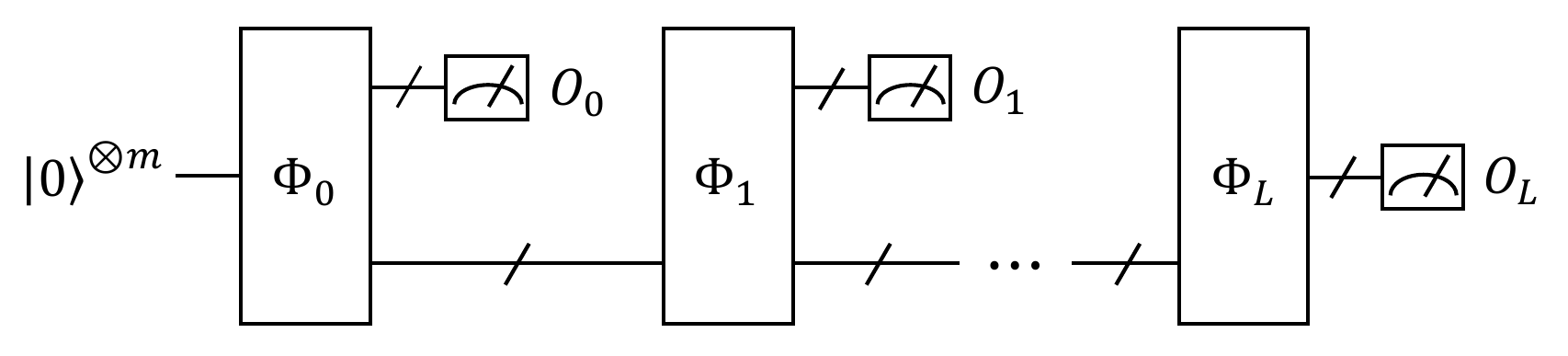}
\end{center}
\vspace{-10pt}
\caption{MPS-like tree. Each $\Phi_i$ is an unknown CPTP map, and adjacent maps $\Phi_i$ and $\Phi_{i+1}$ are connected by a wire of dimension at most $d$. Here, $m\in \mathbb{N}$, and $O_0,\ldots,O_L$ are Hermitian operators satisfying $\|O_i\|_{\infty} \leq 1$.\label{fig:mps}}
\end{figure}

More generally, if the bond dimensions are non-uniform and denoted by
$\{d_v\}_{v\in V^\circ}$, the measurement cost is
\begin{equation}
    \tilde{\mathcal O}\!\left(
    \frac{K^2}{\epsilon^2}
    \sum_{v\in V^\circ}d_v^3
    \right).
\end{equation}
This more general statement is given and proved in Appendix~\ref{sec:proof_multi_layer}.
Thus, the measurement cost is polynomial in the actual number of cuts.\footnote{For a complete $R$-ary tree of tree depth $L$ with $R\ge2$, $K=\sum_{l=1}^L R^l=\Theta(R^L)$, so Theorem~\ref{thm:tree_simulation_general} gives $\tilde{\mathcal O}(d^3K^3/\epsilon^2)$ measurements for regular multi-layer trees as a special case. For a chain, or MPS-like, tree with
$L$ cut edges (Fig.~\ref{fig:mps}), $K=L$, and the theorem gives
$\mathcal O(L^3d^3\epsilon^{-2}\ln(Ld/\delta))$ measurements.}
This should be compared with conventional optimal wire cutting with classical communication, for which the standard sufficient measurement bound is
\begin{equation}
    \mathcal O\!\left(
    \frac{(2d-1)^{2K}}{\epsilon^2}
    \right),
\end{equation}
due to the multiplicative accumulation of the rescaling factor $2d-1$ across cuts.
In the two-layer case, we prove in the next section that an exponential number of measurements in the number of cuts is also necessary for conventional learning-free wire-cutting methods. 

We now sketch the main idea for proving Theorem~\ref{thm:tree_simulation_general}.

\begin{proof}[Proof sketch of Theorem~\ref{thm:tree_simulation_general}]

The protocol proceeds bottom-up along the tree.
For a non-root node $v$, suppose that the learned effective observables
$\{\tilde M_u\}_{u\in{\rm Ch}(v)}$ of its children have already been constructed. 
We form the downstream observable
\begin{equation}
    \tilde B_v
    :=
    O_v\otimes
    \bigotimes_{u\in{\rm Ch}(v)}\tilde M_u,
\end{equation}
and use the learning protocol of Theorem~\ref{thm:error_bounds_for_each_x} to estimate
\begin{equation}
    \Phi_v^\dagger(\tilde B_v).
\end{equation}
The resulting estimator is stored as $\tilde M_v$. 
After all non-root nodes have been processed, we form
\begin{equation}
    \tilde B_r
    :=
    O_r\otimes
    \bigotimes_{u\in{\rm Ch}(r)}\tilde M_u
\end{equation}
and estimate ${\rm tr}[\rho\tilde B_r]$ by measuring the root state $\rho$ in an eigenbasis of $\tilde B_r$, with outcomes weighted by the corresponding eigenvalues.

We assign a local accuracy $\eta_v$ to each non-root cluster $v\in V^\circ$.
Let $T_v$ be the subtree rooted at $v$, including $v$ itself, and set
\begin{equation}
    A_v:=\sum_{w\in T_v}\eta_w,
    \qquad
    A_{\rm tot}:=\sum_{v\in V^\circ}\eta_v .
\end{equation}
On the event that all local learning steps succeed, the following relations hold:
\begin{equation}
    \|\tilde M_v-M_v\|_\infty
    \le
    e^{A_v}-1,
    \qquad
    \|\tilde M_v\|_\infty
    \le
    e^{A_v}.
\end{equation}
This is proven by induction from the leaves upward. The detailed statement and proof can be found in Proposition~\ref{prop:subtree_error_general_tree}.
Then, the corresponding root bound is
\begin{equation}\label{eq:root_norm_bound}
    \|\tilde B_r-B_r\|_\infty
    \le
    e^{A_{\rm tot}}-1.
\end{equation}
Thus, choosing
\begin{equation}
    \eta_v
    =
    \frac{\ln(1+\epsilon/2)}{K}
    \qquad
    (v\in V^\circ)
\end{equation}
ensures that the deterministic bias is at most $\epsilon/2$.

The same bound also controls the norm of the intermediate observables.
If the learning steps below $v$ have succeeded, 
\begin{eqnarray}
    \|\tilde B_v\|_\infty
    &\leq&
    \|O_v\|_\infty
    \prod_{u\in{\rm Ch}(v)}\|\tilde M_u\|_\infty \\
    &\leq&
    \prod_{u\in{\rm Ch}(v)}e^{A_u}
    =
    e^{\sum_{u\in{\rm Ch}(v)}A_u}\\
    &\leq&
    e^{A_{\rm tot}}\\
    &\leq&
    1+\frac{\epsilon}{2}
    \le
    \frac32 .
\end{eqnarray}
Therefore, the observable used in each local learning step has operator norm bounded by a constant. 
Assigning failure probability $\delta_v=\delta/(2K)$ to each local learning step, Theorem~\ref{thm:error_bounds_for_each_x} gives
\begin{equation}
    N_v
    =
    \mathcal O\!\left(
    \frac{d_v^3}{\eta_v^2}
    \ln\!\left(
    \frac{d_v}{\delta_v}
    \right)
    \right)
    =
    \mathcal O\!\left(
    \frac{d_v^3K^2}{\epsilon^2}
    \ln\!\left(
    \frac{Kd_v}{\delta}
    \right)
    \right)
\end{equation}
measurements at node $v$. 
Summing over all non-root clusters gives
\begin{equation}
    \mathcal O\!\left(
    \frac{K^2}{\epsilon^2}
    \sum_{v\in V^\circ}d_v^3
    \ln\left(\frac{Kd_v}{\delta}\right)
    \right)
\end{equation}
measurements for the non-root learning stages. Here, using $d_v\leq d$, this is bounded by
\begin{equation}
    \mathcal O\!\left(
    \frac{K^3d^3}{\epsilon^2}
    \ln\left(\frac{Kd}{\delta}\right)
    \right).
\end{equation}

On the global success event, the same norm bound \eqref{eq:root_norm_bound} gives
\begin{equation}
    \|\tilde B_r\|_\infty\le 3/2.
\end{equation}
Thus, Hoeffding's inequality implies that
\begin{equation}
    \mathcal O\!\left(
    \frac{1}{\epsilon^2}
    \ln \left( \frac{1}{\delta} \right)
    \right)
\end{equation}
additional root measurements estimate ${\rm tr}[\rho\tilde B_r]$ within error $\epsilon/2$. 
Combining this final statistical error with the deterministic bias gives total error at most $\epsilon$.

The only remaining issue is that $\tilde B_v$ is random because it depends on the estimates below $v$. 
In the formal proof, we condition on the data obtained before processing $v$, apply the local learning guarantee on the event that the strict descendants of $v$ have succeeded, and then remove the conditioning by the tower property. 
This gives the stated success probability;
see Appendix~\ref{sec:proof_multi_layer} for details.
\end{proof}

We finally comment on the implementation of the local basis measurements in the general-tree protocol. 
At each non-root node $v$, the only new diagonalization is that of the learned effective observable $\tilde M_v$, which acts on the incoming cut-edge Hilbert space $\mathcal H_v$ of dimension $d_v$. 
The observable used when processing a node has the tensor-product form
\begin{equation*}
\tilde B_v = O_v \otimes \bigotimes_{u\in{\rm Ch}(v)} \tilde M_u ,
\end{equation*}
so the protocol does not require diagonalizing a full tensor-product operator over the entire processed subtree. 
The corresponding local basis measurements must be implemented on the chosen local devices.

\section{Exponential separation from learning-free wire cutting}\label{sec:exponential_separation}

The polynomial upper bound in Sec.~\ref{sec:poly_circuit_knitting_tree} shows that, for tree-structured circuits, the sampling overhead can be made polynomial in the number of cuts by constructing the cuts from learned downstream effective observables. 
This raises the natural question of whether the same polynomial-in-cuts scaling follows from the tree structure alone, or whether the learning step is also essential. 
We answer this question in the simple two-layer setting by comparing our learning-based protocol with a corresponding class of learning-free protocols for cutting the same inter-cluster wires.

To make this comparison precise, we first formulate the two-layer expectation-estimation task under the local-access constraint of circuit knitting. 
This formulation keeps the same root state $\rho$, downstream local channels $\{\Phi_k\}_{k=1}^{R}$, and product observable $O$ as in the two-layer analysis of Sec.~\ref{sec:poly_circuit_knitting_tree}, while making explicit that $\rho$ and $\{\Phi_k\}_{k=1}^{R}$ are accessed only through separate local devices, with only classical communication allowed between these devices.
\begin{task}[Two-layer tree expectation-value estimation with local access]
\label{task:two-layer-tree-estimation}
Let
\begin{equation}
    \rho_{\rm targ}
    :=
    (\Phi_1\otimes\cdots\otimes\Phi_R)(\rho),
\end{equation}
where $\rho$ is an arbitrary unknown state on $\bigotimes_{k=1}^R \mathbb{C}^d$, and each $\Phi_k:\mathsf{L}(\mathbb{C}^d)\to \mathsf{L}(\mathbb{C}^{d'})$ is an arbitrary unknown CPTP map
with $d\leq d'$. 
We are given a known product observable
\begin{equation}
    O=\bigotimes_{k=1}^R O_k,
    \qquad
    O_k \in \mathsf{H}(\mathbb{C}^{d'}),
    \qquad
    \|O_k\|_\infty\le 1 .
\end{equation}
The goal is to estimate $\mu := \mathrm{tr}[O\rho_{\rm targ}]$ within additive error $\epsilon$ with high probability, using only a root device that prepares the state $\rho$, local devices that implement the channels $\{\Phi_k\}_{k=1}^R$, classical communication between these devices, and classical post-processing of the collected data.
\end{task}
\noindent In this formulation, the wires to be cut are the $R$ $d$-dimensional systems that would coherently connect the root system carrying $\rho$ to the inputs of the downstream channels $\{\Phi_k\}_{k=1}^R$ in the uncut two-layer circuit. Both approaches considered below cut these same links.

First, our learning-based protocol solves Task~\ref{task:two-layer-tree-estimation} by estimating $M_k := \Phi_k^\dagger(O_k)$ for each $k=1,\ldots,R$, and then constructing observable-adaptive cuts from these learned effective observables. 
By Theorem~\ref{thm:2-layer}, for any target additive error $\epsilon\in(0,1]$ and any failure probability $\delta\in(0,1)$, this protocol estimates $\mu$ within additive error $\epsilon$ with probability at least $1-\delta$ using
\begin{equation}
    \mathcal{O}\!\left(
    \frac{d^3R^3}{\epsilon^2}
    \ln\!\left(\frac{Rd}{\delta}\right)
    \right)
\end{equation}
total measurements, including the local learning cost.

We next define the corresponding learning-free protocol class.
Since each inter-cluster link to be cut is represented by an identity channel, the relevant learning-free procedure is a wire-cutting protocol.
As reviewed in Sec.~\ref{sec:time_like_cut}, standard QPD-based wire cutting represents a $d$-dimensional identity channel $\mathrm{id}_d$ by a pre-specified linear combination of MP channels:
\begin{equation}
    \mathrm{id}_d=\sum_{\lambda} a_{\lambda} \mathcal{M}_{\lambda},
    \qquad 
    \mathcal{M}_\lambda(\bullet)
    =
    \sum_y \mathrm{tr}(E_{\lambda,y} \bullet)\sigma_{\lambda,y}.
\end{equation}
Here $\{E_{\lambda,y}\}_y$ is a POVM on the cut system and $\sigma_{\lambda,y}$ is the corresponding state prepared on the downstream side.
Operationally, one samples $\lambda$ with probability $|a_\lambda|/\gamma$, where $\gamma=\sum_\lambda |a_\lambda|$, performs the corresponding MP channel $\mathcal{M}_{\lambda}$, and includes the factor $\gamma\,\mathrm{sgn}(a_\lambda)$ in the classical estimator. 

\begin{algorithm}[bp]
\caption{Learning-free wire-cutting template for two-layer trees}
\label{alg:learning-free-wire-cut}

\KwIn{
Local access to devices implementing $\rho$ and
$\{\Phi_k\}_{k=1}^R$;
known observables $\{O_k\}_{k=1}^R$;
shot budget $N$;
fixed wire-cutting rule $\mathsf{WC}$.
}

\KwOut{
An estimate $\hat{\mu}$ of $\mu$.
}

\For{$i=1,\ldots,N$}{
    Sample a cutting label $\ell_i$ according to $p_\ell$\;
    
    Measure $\rho$ with the POVM
    \[
        \left\{
        E_{\ell_i,y_i}
        =
        \bigotimes_{k=1}^R
        E^{(k)}_{\ell_i,y_i^k}
        \right\}_{y_i},
    \]
    obtaining $y_i=(y_i^1,\ldots,y_i^R)$\;
    
    Prepare
    \[
        \bigotimes_{k=1}^R
        \sigma^{(k)}_{\ell_i,y_i^k}
    \]
    and evolve it under $\bigotimes_{k=1}^R\Phi_k$\;
    
    Measure each output system in an eigenbasis of $O_k$, obtaining $z_i=(z_i^1,\ldots,z_i^R)$\;
}

\Return{
\[
    \hat{\mu}
    =
    \mathsf{Post}\bigl((\ell_i,y_i,z_i)_{i=1}^N\bigr);
\]}
\end{algorithm}

For the present comparison, the important point is that the cutting rule -- namely, the sampling distribution over cutting labels $\lambda$, the associated measure-and-prepare operations, and the classical post-processing -- is fixed in advance, before any shots are taken.
We use ``learning-free'' in this pre-specified sense. The cutting rule may be randomized and may include predetermined feed-forward within each shot, such as preparing a state depending on the measurement outcome at the cut, but it is not updated using the measurement outcomes from earlier shots. For Task~\ref{task:two-layer-tree-estimation}, we formalize such a rule as
\begin{equation}
    \mathsf{WC}
    =
    \left(
    p_\ell,
    \{E^{(k)}_{\ell,y}\},
    \{\sigma^{(k)}_{\ell,y}\},
    \mathsf{Post}
    \right)
\end{equation}
Here $p_\ell$ is a distribution over possibly composite \mbox{cutting} labels, so that a single label $\ell$ can specify the pre-sampled cutting choices at all cut locations.
For each label $\ell$, $\{E_{\ell,y}^{(k)}\}_y$ is a POVM on the $k$-th cut system, $\sigma_{\ell,y}^{(k)}$ is the corresponding state prepared at the input of the $k$-th downstream channel, and $\mathsf{Post}$ is an arbitrary classical post-processing map applied to the classical transcript over all shots. 
The rule $\mathsf{WC}$ may depend on the known classical inputs of the task, including the observables $\{O_k\}_{k=1}^R$, but it is fixed before the shots are taken. This class includes standard QPD-based wire-cutting schemes~\cite{peng2021simulating,Lowe2023fastquantumcircuit,brenner2023optimal,harada2025doubly,pednault2023alternative,PRXQuantum.6.010316}, including optimal wire cutting schemes, but does not include the observable-adaptive learning step used in our protocol. 

Given a fixed rule $\mathsf{WC}$ of this form, applying it to Task~\ref{task:two-layer-tree-estimation} gives the learning-free wire-cutting procedure summarized in Algorithm~\ref{alg:learning-free-wire-cut}.
We compare our learning-based protocol with all procedures generated in this way.
The polynomial upper bound above and the lower bound for this learning-free class, stated in the theorem below, yield an information-theoretic exponential separation between learning-based and learning-free wire cutting; see Fig.~\ref{fig:introduction_figure}(d).

\begin{theorem}[Exponential separation from learning-free wire cutting]
\label{thm:learning-free-wire-cut-lower-bound}
Consider Task~\ref{task:two-layer-tree-estimation} with target additive error $0<\epsilon \leq 1/3$.

\smallskip
\noindent\textup{\textbf{Upper bound.}}
For any failure probability $\delta\in(0,1)$, the learning-based
protocol estimates $\mu$ within additive error $\epsilon$ with probability at least $1-\delta$ using
\begin{equation}
    \mathcal{O}\!\left(
    \frac{d^3R^3}{\epsilon^2}
    \ln\!\left(\frac{Rd}{\delta}\right)
    \right)
\end{equation}
measurements in total, including the measurements used in the local learning steps.

\smallskip
\noindent\textup{\textbf{Lower bound.}}
Any learning-free wire-cutting protocol described by Algorithm~\ref{alg:learning-free-wire-cut} that estimates $\mu$ within additive error $\epsilon$ with constant success probability strictly larger than $1/2$ for all instances requires
\begin{equation}
    \Omega\!\left(
    \frac{(d+1)^R}{\epsilon^2}
    \right)
\end{equation}
measurements.
\end{theorem}

The measurement count in Theorem~\ref{thm:learning-free-wire-cut-lower-bound} follows the standard shot-count convention used in circuit-cutting analyses. Namely, it counts independent local experimental samples, or shots, generated by the protocol. A single local sample may include one local circuit execution together with the measurements prescribed by the protocol, including any mid-circuit measurements, but elementary measurement operations inside that local experiment are not counted separately.

The same exponential separation holds under the more fine-grained local-query count, where each use of an available local device is counted. The learning-based protocol retains the same polynomial scaling in local-query count, since the samples used for local learning and final estimation are already included in the displayed bound up to lower-order local-query factors. For the learning-free procedure in Algorithm~\ref{alg:learning-free-wire-cut}, each round uses at least one local device query, and under per-device accounting it uses one root query and $R$ downstream queries. Hence the lower bound on the number of rounds immediately implies an
$\Omega((d+1)^R/\epsilon^2)$ lower bound on total local queries as well.
Under per-device accounting, each round contributes an additional factor
of $R+1$, which is only polynomial in the number of cuts.
Thus the exponential-in-cuts separation is unaffected by this stricter
resource accounting.

Theorem~\ref{thm:learning-free-wire-cut-lower-bound} shows that the polynomial-in-cuts scaling of the learning-based protocol is not a consequence of the two-layer tree structure alone. Even for the same local-access task and the same inter-cluster links, any learning-free wire-cutting protocol whose rule is fixed before the shots are taken still requires exponentially many measurements in the number of cuts. The exponential-to-polynomial improvement therefore comes from the learning step: the protocol first identifies the downstream effective observables $M_k=\Phi_k^\dagger(O_k)$ and then uses this learned information to construct observable-adaptive cuts.

Finally, we briefly outline the proof of Theorem~\ref{thm:learning-free-wire-cut-lower-bound}.
The upper bound follows directly from Theorem~\ref{thm:2-layer}, so we focus on the lower-bound argument.
The full proof is given in Appendix~\ref{sec:lower-bound-qcc}.

\begin{proof}[Proof sketch of Theorem~\ref{thm:learning-free-wire-cut-lower-bound}]
The lower-bound argument is based on a reduction from expectation-value estimation to a binary discrimination problem.
We construct a hard distribution over instances indexed by a hidden bit $x\in\{0,1\}$, such that estimating the target expectation value to accuracy $\epsilon$ allows one to identify $x$ with constant success probability. 
Concretely, let $d=2^n$, let $\bar Z^{(k)}:=Z^{\otimes n}$ on the $k$-th cut system, set $\bar Z:=\bigotimes_{k=1}^R \bar Z^{(k)}$, and define
\begin{equation}
    \tau_x
    :=
    \frac{1}{d^R}
    \left(
        I + (-1)^x 3\epsilon\,\bar Z
    \right),\qquad x=0,1.
\end{equation}
The two states $\tau_0$ and $\tau_1$ differ only in the sign of the full-rank Pauli observable $\bar Z$.
A hard instance is generated by drawing independent Haar-random local unitaries
$U_1,\ldots,U_R$, with $U_k\in\mathsf{U}(\mathbb{C}^d)$ for each $k$, and fixing them for all $N$ shots.
With $U:=\bigotimes_{k=1}^R U_k$, we choose
\begin{equation}
    \rho_x := U\tau_x U^\dagger,
    \qquad
    \Phi_k(\bullet):=U_k^\dagger \bullet U_k,
    \qquad
    O_k:=\bar Z^{(k)} .
\end{equation}
Here, the downstream channels undo the random local changes of basis, and the target expectation value becomes
\begin{equation}
    \mu_x
    =
    \operatorname{tr}[\bar Z\tau_x]
    =
    (-1)^x 3\epsilon.
\end{equation}
Thus, any protocol that estimates $\mu_x$ within additive error $\epsilon$ with constant success probability can distinguish the hidden
bit $x$ with constant success probability.

For a learning-free wire-cutting protocol, however, the measurements $\{E_{\ell,y}^{(k)}\}$ in cuts are chosen without using learned estimates of the downstream effective observables. 
Let $T_N := (\ell_i,y_i,z_i)_{i=1}^N$ denote the classical data collected before the final post-processing.
Averaging over the independent Haar-random local bases
$U_1,\ldots,U_R$ used in the hard distribution, while keeping the learning-free wire-cutting rule fixed, shows that
\begin{equation}
    I(x:T_N)
    \le
    \mathcal O\!\left(
        \frac{N\epsilon^2}{(d+1)^R}
    \right).
\end{equation}
Since the final estimate is obtained from $T_N$ by arbitrary classical post-processing, the data-processing inequality implies that no post-processing can increase the information about $x$.
On the other hand, since $\mu_x=(-1)^{x} 3\epsilon$, an estimate within additive error $\epsilon$ determines the sign of $\mu_x$, and therefore the hidden bit $x$, with constant success probability. 
By Fano's inequality, this implies that 
\begin{equation}
    \Omega(1)\leq I(x:T_N).
\end{equation}
Combining the two bounds gives
\begin{equation}
    N
    =
    \Omega\!\left(
    \frac{(d+1)^R}{\epsilon^2}
    \right).
\end{equation}
\end{proof}

We note that the restriction $0<\epsilon\le 1/3$ is only used to keep the states in the proof sketch positive semidefinite when the two possible expectation values are chosen as $\pm 3\epsilon$. 
For larger constant $\epsilon$, using a constant separation instead gives $\Omega((d+1)^R)$, which is equivalent to the displayed bound up to a universal constant in this regime.

\section{Numerical simulation}\label{sec:numerics}

\begin{figure*}[pt]
    \centering
    \includegraphics[width=\textwidth]{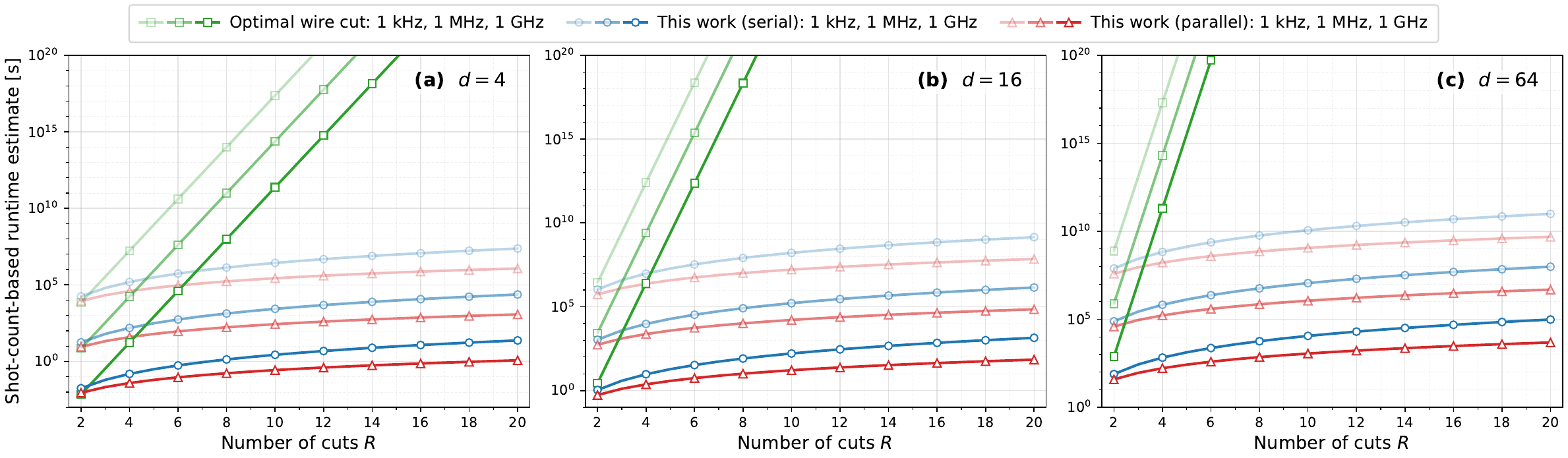}
    \caption{
    Shot-count-based runtime estimates for two-layer star circuits with cut dimensions $d=4,16,64$.
    We use $\epsilon=0.05$ and $\delta=0.05$, and convert a shot count $N$ to a runtime estimate by $T=N/f$, where $f$ is the effective sampling rate.
    Green curves correspond to optimal wire cutting, while blue and red curves correspond to our protocol with serial and parallel local learning, respectively.
    For each method, lighter to darker curves indicate $f=1\,\mathrm{kHz}$, $1\,\mathrm{MHz}$, and $1\,\mathrm{GHz}$.
    The vertical axis is truncated at $10^{20}$ seconds.
    }
    \label{fig:runtime-comparison}
\end{figure*}

\begin{figure*}[pt]
    \centering
    \includegraphics[width=\textwidth]{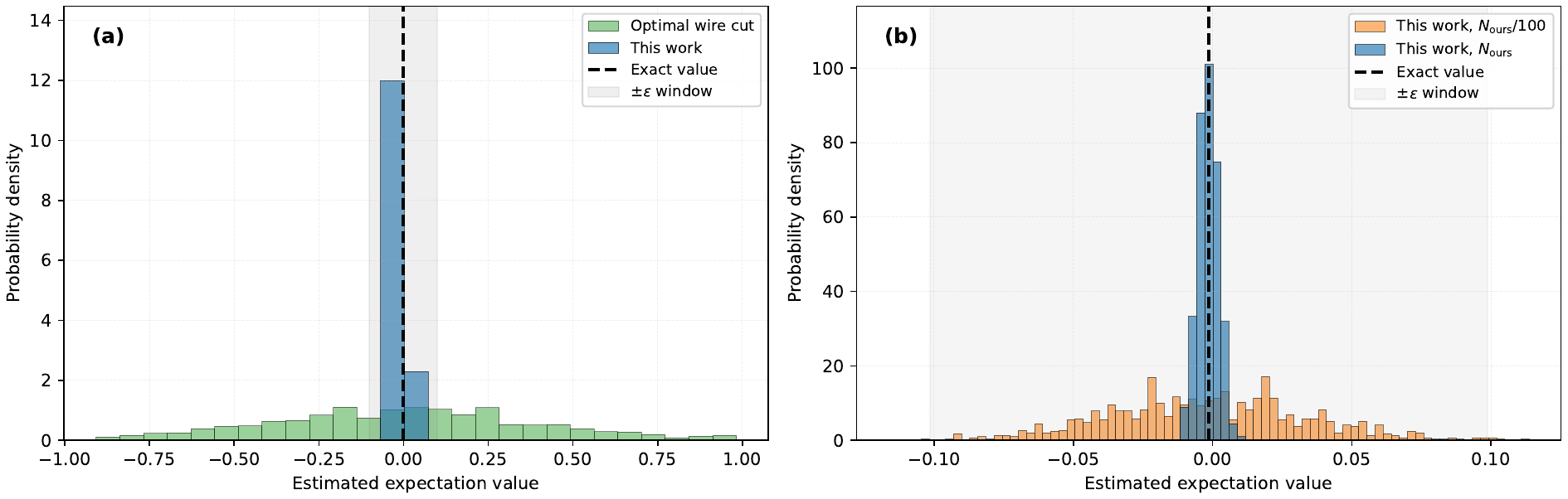}
    \caption{
    Finite-shot distributions for a fixed two-layer circuit instance generated from random local circuits, with $d=4$, $R=4$, $\epsilon=0.1$, and $\delta=0.1$.
    (a) Optimal wire cutting and our protocol are compared using the same total shot budget $N_{\rm ours}$, given by the proof-based allocation for our protocol.
    (b) For our protocol, the proof-based allocation $N_{\rm ours}$ is
    compared with a $1/100$-scaled allocation. Histograms are normalized as probability densities. The dashed vertical line marks the exact value, and the shaded region marks the
    $\pm\epsilon$ window.
    \label{fig:finite-shot-histogram}}
\end{figure*}

To illustrate the implications of the analytical bounds, we perform two numerical simulations. 
The first simulation evaluates the shot-count bounds, including their constant factors, for two-layer star circuits and compares the resulting runtime estimates with optimal wire cutting. 
This shows how the asymptotic separation appears for concrete choices of the cut dimension $d$, the number of cuts $R$, the target additive error $\epsilon$, and the failure probability $\delta$. The second simulation fixes an explicit randomly generated circuit instance and repeats the
estimation procedures many times under finite shot budgets. This allows us to compare the finite-shot distributions of optimal wire cutting and our
protocol under the same total shot budget, and to examine the effect of reducing
the proof-based allocation for our protocol.

For the first simulation, we consider two-layer circuits with $R$ cuts and cut dimension $d$.
For comparison, we use optimal wire cutting with rescaling factor $\gamma=2d-1$~\cite{brenner2023optimal,harada2025doubly,pednault2023alternative,PRXQuantum.6.010316}. 
The corresponding sufficient shot count is
\begin{equation}
    N_{\rm conv}
    =
    \left\lceil
    \frac{2(2d-1)^{2R}}{\epsilon^2}
    \ln\left(\frac{2}{\delta}\right)
    \right\rceil ,
\end{equation}
which guarantees additive error $\epsilon$ with failure probability at
most $\delta$ by Hoeffding's inequality.
For our learning-based protocol, we use the explicit constant-factor allocation of Protocol A in the proof of Theorem~\ref{thm:2-layer}. 
Namely, each local learning step uses the 2-design bound of Theorem~\ref{thm:error_bounds_for_each_x} with accuracy $\eta_a=9\epsilon/(20R)$ and failure probability $\delta/(2R)$, while the final estimation step uses Hoeffding's inequality with target error $\epsilon/10$ and failure probability $\delta/2$. This gives
\begin{equation}
    N_{\rm ours}
    =
    R N_{\mathrm{a},k} + N_{\mathrm{a},0},
\end{equation}
where
\begin{align}
    N_{\mathrm{a},k}
    &=
    \left\lceil
    \frac{800R^2}{81\epsilon^2}
    (d^2+1)
    \left(
    d+1+\frac{3\epsilon}{20R}
    \right)
    \ln\left(\frac{4Rd}{\delta}\right)
    \right\rceil, \\
    N_{\mathrm{a},0}
    &=
    \left\lceil
    \frac{200}{\epsilon^2}
    \ln\left(\frac{4}{\delta}\right)
    \right\rceil .
\end{align}
We also consider a fully parallel implementation in which the $R$ local learning tasks (Step~\textbf{(a-1)}) are run simultaneously, so that the corresponding wall-clock shot count is
\begin{equation}
    N_{\rm ours}^{\rm parallel}
    =
    N_{\mathrm{a},k}+N_{\mathrm{a},0}.
\end{equation}
To express these shot counts as runtime estimates, we use the simple proxy $T=N/f$, where $f$ is the effective sampling rate. 
In Fig.~\ref{fig:runtime-comparison}, we take $\epsilon=0.05$ and $\delta=0.05$, and compare $f=1~{\rm kHz}$, $1~{\rm MHz}$, and $1~{\rm GHz}$.

Fig.~\ref{fig:runtime-comparison} shows the resulting runtime estimates for $d=4,16,64$. The green curves, corresponding to optimal wire cutting, increase rapidly with the number of cuts $R$, and quickly exceed the
plotted range as either $R$ or $d$ grows. 
This behavior reflects the multiplicative quasiprobability rescaling factor $(2d-1)^{2R}$.
In contrast, the blue and red curves for our learning-based protocol grow much more slowly with $R$ over the same parameter range. 
The gap between the blue and red curves reflects the possible reduction in wall-clock runtime when the local learning tasks can be parallelized.

We next examine the finite-shot behavior of the estimators on a fixed explicit two-layer circuit instance. 
We take $d=4$ and $R=4$, so that each cut system consists of two qubits. 
The root circuit and the four local subsystem circuits are generated from fixed random seeds and kept fixed throughout the simulation. 
The local subsystem circuits act on the cut system together with additional ancilla qubits, and the terminal observables are chosen as random Pauli strings. 
For our protocol, the local effective observables are learned using a complete set of MUBs, viewed as a state 2-design. 
For $d=4$, this gives $d+1=5$ MUBs and $d(d+1)=20$ probe states. In the finite-shot simulation, we compute the exact Born probabilities for the fixed circuits and sample measurement outcomes from these probabilities. 
We repeat the full estimation procedure $1000$ times to obtain the histograms in Fig.~\ref{fig:finite-shot-histogram}.

Fig.~\ref{fig:finite-shot-histogram} shows the finite-shot distributions obtained from repeated independent runs on this fixed instance. 
In Fig.~\ref{fig:finite-shot-histogram}(a), optimal wire cutting and our
protocol are given the same total shot budget $N_{\rm ours}$. 
With this budget, the wire-cut estimator has a broad distribution due to quasiprobability rescaling, whereas the estimator of our protocol is concentrated near the exact value. 
Fig.~\ref{fig:finite-shot-histogram}(b) shows the effect of reducing the shot budget of our protocol. 
The distribution obtained with $N_{\rm ours}/100$ is broader than that obtained with $N_{\rm ours}$, as expected, but it remains centered near the exact value for this instance.
This comparison suggests that the constant-factor sufficient allocation used in the proof may be conservative for this fixed instance, and that sharper instance-dependent analyses of the learning step could lead to smaller shot estimates.

\section{Relation to prior work and applications}\label{sec:related_works}

We place our results in the context of several related approaches.
In Sec.~\ref{main:quantum_circuit_cutting}, we compare our learning-based approach with quantum circuit-cutting
methods, including QPD-based decompositions, structure-aware cut optimizations, and the cluster-simulation framework of Peng \textit{et al.}~\cite{peng2021simulating}. For the latter, we translate its finite-shot reconstruction bound into our tree notation and compare the scaling across several representative tree structures.
In Sec.~\ref{main:tree_dc_methods}, we relate our learning-based protocol to tree-based divide-and-conquer methods, such as hybrid tensor networks~\cite{PhysRevLett.127.040501}, Deep VQE~\cite{PRXQuantum.3.010346}, and entanglement forging~\cite{PRXQuantum.3.010309}, and explain how the same
learning viewpoint can be useful for these approaches by providing finite-shot guarantees for representative
tree-tensor contractions.
In Sec.~\ref{main:fundamental_limit}, we explain why our polynomial scaling is consistent with recent lower bounds on circuit knitting, especially for Refs.~\cite{PhysRevA.111.012433,marshall2023all}.
In Sec.~\ref{main:tensor_network}, we also compare our setting with classical tensor-network simulation in Ref.~\cite{doi:10.1137/050644756} from the viewpoints of access models, treewidth, and the resource being counted.

\subsection{Quantum circuit cutting}\label{main:quantum_circuit_cutting}

\subsubsection{QPD-based approaches}
We first discuss the relation between our learning-based cut and standard QPD-based wire-cutting methods. 
Although the general form of quasiprobability simulation was already reviewed in Sec.~\ref{sec:preliminaries}, we briefly recall it here because the difference between channel-level and expectation-value-level decompositions will also be important in the discussion below.
In the standard QPD setting, given a target operation $\mathcal{T}$ and a set $\mathcal{F}$ of implementable operations, one considers a channel-level decomposition
\begin{equation}\label{eq:channel-level_decomp}
    \mathcal{T} = \sum_i a_i \mathcal{E}_i,
\end{equation}
where $\mathcal{E}_i\in \mathcal{F}$ and $a_i \in \mathbb{R}$.
This decomposition must hold regardless of the input state and the later part of the circuit. 
As a result, it gives an unbiased estimator that works in arbitrary circuits, but the estimator carries the normalization $\gamma:=\sum_i|a_i|$, so the sampling overhead is increased by a factor of $\gamma^2$. 
Once $\mathcal{T}$ and $\mathcal{F}$ are fixed, the smallest possible $\gamma$ is determined by this channel-simulation task. 
For the wire cutting, where $\mathcal{T}=\mathrm{id}^n$ and $\mathcal{F}$ is the set of MP channels, the lower bound $\gamma\geq2d-1$ is known~\cite{yuan2021universal,brenner2023optimal}, and decompositions achieving this bound have been established~\cite{brenner2023optimal,harada2025doubly,pednault2023alternative,PRXQuantum.6.010316}.

Our result does not contradict the lower bound, because the object being decomposed is different. 
Standard QPD-based wire cutting requires an exact decomposition of the identity channel, independently of which observable is measured later.
By contrast, our learning-based cut only aims to reproduce the action of the identity channel for a fixed downstream circuit component.
For a fixed downstream channel $\Phi$ and a fixed observable $O$, we consider
\begin{equation}\label{eq:exp-level_decomp}
    \mathrm{tr}\left[O\,\Phi\circ \mathrm{id}(X)\right]
    =
    \sum_i a_i\,\mathrm{tr}\!\left[O\,(\Phi\circ \mathcal{E}_i)(X)\right],
    \quad
    \forall X .
\end{equation}
Theorem~\ref{thm:existence_of_no_cost_cut} shows that, for this fixed pair $(\Phi,O)$, there exists a single MP channel whose effective rescaling factor is $1$. However, constructing such an MP channel requires the exact classical description of $O_{\Phi}:=\Phi^{\dagger}(O)$, which is not available in practice. 
Theorem~\ref{thm:zero_cost_timelike_cut_with_an_approximated_unitary}, together with the learning protocols in Sec.~\ref{sec:learning_protocol}, gives a constructive version of this idea: 
after $N$ learning steps, the resulting observable-adaptive cut still has unit rescaling, while the error appears as a bias of $\tilde{\mathcal{O}}(\sqrt{d^3/N})$.
Therefore, compared with standard QPD-based wire cutting, our method changes the basic trade-off: instead of an unbiased but variance-amplified channel-level simulation, it gives an observable-dependent simulation with no variance amplification from the cut itself, at the cost of a controllable bias.

The difference between channel-level and expectation-value-level decompositions also helps interpret recent structure-aware methods~\cite{10025537,10196555,10313822}. 
For example, Ref.~\cite{10313822} seeks to reduce the practical cost of wire cutting by discarding Pauli-based cutting operations in the standard mid-circuit Pauli decomposition~\cite{peng2021simulating} that do not affect the target expectation-value calculation. 
In our framework, this can be reinterpreted as an expectation-value-level decomposition for a fixed circuit instance around the cut, rather than as a universal channel-level decomposition. 
The key difference from our method is that such approaches remain within a fixed family of Pauli-based cutting operations, whereas our approach allows a general MP channel and constructs the cut adaptively from the learned downstream effective observable. 
We refer the reader to Appendix~\ref{apsec:relation_to_wire_cut} for a more detailed discussion.

\subsubsection{Cluster-simulation framework by Peng \textit{et al.}}
We next discuss the relation to the cluster-simulation framework of Peng \textit{et al.}~\cite{peng2021simulating}. 
Their framework represents a quantum circuit as a tensor network and uses a Pauli-based edge-cutting decomposition to partition a large circuit into cluster circuits executable on smaller quantum devices.
Within this framework, they establish two types of cluster-simulation guarantees. 
First, for general clustered circuits, they obtain a simulator whose cost grows exponentially with the total number of inter-cluster qubit wires that are cut. 
This simulator is closely related to the quasiprobability-simulation viewpoint underlying QPD-based circuit cutting, where the sampling overhead is governed by the rescaling factors associated with the cut decompositions.
Second, they develop a graph-dependent simulation approach in which the cost can be controlled by structural properties of the tensor network induced by the clustering.
The latter result is more directly related to our work, because it addresses how circuit structure can reduce the simulation cost associated with circuit cutting. 
In this subsection, we therefore review this graph-dependent approach and compare it with our learning-based protocol.

To set up the comparison, we recall the tensor-network representation underlying the cluster-simulation method of Peng \textit{et al.}
In their framework, a quantum circuit $C$, together with the final classical post-processing $f$, is represented as a tensor network. If $y$ denotes the measurement outcome of $C$, the full contraction yields the desired expectation value $\mathbb{E}_yf(y)$. 
The tensor network is then partitioned into clusters, and the inter-cluster wires are cut using the Pauli-based edge-cutting decomposition:
\begin{equation}
    \mathrm{id}(\bullet)
    =
    \sum_{s=1}^{8}
    c_s\,\operatorname{tr}[O_s\,\bullet]\,\sigma_s ,
    \qquad
    c_s\in\left\{+\frac12,-\frac12\right\},
    \label{eq:peng_edge_cut_schematic}
\end{equation}
where $O_s\in\{I,X,Y,Z\}$ is a single-qubit Pauli observable, and $\sigma_s$ is one of the single-qubit states appearing in the decomposition. 
Equivalently, each label $s\in[8]$ specifies both the Pauli observable measured on the upstream side and the state prepared on the downstream side. 
Thus, cutting one qubit wire introduces an eight-valued cutting label.

After applying this decomposition to all inter-cluster wires, each cluster is converted into a tensor whose indices are the cutting labels on the wires incident to that cluster. 
Let $E_{\rm cut}(x)$ be the set of cut wires incident to cluster $x$, and define
\begin{equation}
    s_x := (s_e)_{e\in E_{\rm cut}(x)} .
\end{equation}
The corresponding tensor entry, denoted here by $a_x(s_x)$, is estimated by running the local cluster circuit with the measurement and preparation choices specified by $s_x$. 
Once all such local tensors are estimated, the original expectation value is reconstructed by the classical tensor contraction
\begin{equation}
    \mathbb{E}_{y}f(y)=\sum_{\{s_e\}}
    \prod_x a_x(s_x),
    \label{eq:peng_induced_tn}
\end{equation}
with the coefficients from Eq.~\eqref{eq:peng_edge_cut_schematic} absorbed into the local tensors. 
In practice, the entries $a_x(s_x)$ are estimated by finite-shot experiments, and the final contracted value is obtained by contracting the estimated local tensors. 
This is the sense in which the method of Ref.~\cite{peng2021simulating} has a tomography-like aspect: it estimates local tensor entries experimentally and controls the error of the final contracted value through entrywise accuracy guarantees.

To see this, let $\kappa$ be the number of clusters and let $n_{\rm max}$ be the maximum number of cut qubit-wires incident to a cluster.
Then the total number $D$ of local tensor entries to be estimated is bounded, schematically, by
\begin{equation}
    D 
    \leq
    \kappa 8^{n_{\rm max}}.
    \label{eq:peng_D}
\end{equation}
The entrywise perturbation analysis of \cite{peng2021simulating} shows that it is sufficient to estimate each tensor entry to accuracy
\begin{equation}
    \epsilon_{\rm entry}
    =
    \frac{\epsilon}{(e-1)D}
    \label{eq:peng_entry_accuracy}
\end{equation}
in order to guarantee that the contracted tensor-network value is accurate up to $\epsilon$. 
Since each entry is estimated as the mean of bounded random variables, Hoeffding's inequality gives the sufficient per-entry shot count
\begin{equation}
    N_{\rm entry}
    =
    \frac{2\ln(6D)}{\epsilon_{\rm entry}^2}
\end{equation}
Here, the success probability is taken to be $2/3$; equivalently, the total failure probability over all $D$ tensor entries is bounded by $1/3$ by a union bound.
Therefore, estimating all $D$ entries requires
\begin{equation}
    D\,N_{\rm entry}
    =\mathcal O\left( \frac{D^3 \ln(6D)}{\epsilon^2} \right)
    \label{eq:peng_quantum_executions1}
\end{equation}
local quantum executions/measurements.

At a high level, both approaches use finite-shot local experiments to estimate classical quantities whose errors are propagated through a tree contraction.
In this broad sense, the entrywise reconstruction of Ref.~\cite{peng2021simulating} can also be viewed as learning, in a fixed Pauli cutting basis, the cluster response for each combination of incident cut labels.
The difference lies in how boundary information is represented and propagated.
In the method of Ref.~\cite{peng2021simulating}, this response is represented by a tensor with one index for each incident cut wire, so a high-degree cluster gives rise to a high-order tensor in the fixed Pauli cutting basis.
By contrast, our protocol is observable-adaptive.
After the descendants of a node have been processed, all information relevant to the fixed product observable is compressed into a single Heisenberg-picture effective observable on the incoming edge.
More explicitly, for a non-root node $v$, the object passed upward is
\begin{equation}
    M_v
    =
    \Phi_v^\dagger
    \left(
        O_v
        \otimes
        \bigotimes_{u\in{\rm Ch}(v)} M_u
    \right),
    \label{eq:effective_observable_tree_peng}
\end{equation}
which acts only on the incoming cut space $\mathcal H_v$.
The local learning step estimates $M_v$ in operator norm and passes the estimate $\tilde M_v$ to the parent.
Thus, although the local observable used at $v$ contains the child messages $\{M_u\}_{u\in{\rm Ch}(v)}$, the message passed upward is a single observable-dependent operator, rather than a fixed-basis tensor indexed by all incident cut labels.

\begin{figure*}[pt]
    \centering
    \includegraphics[width=\textwidth]{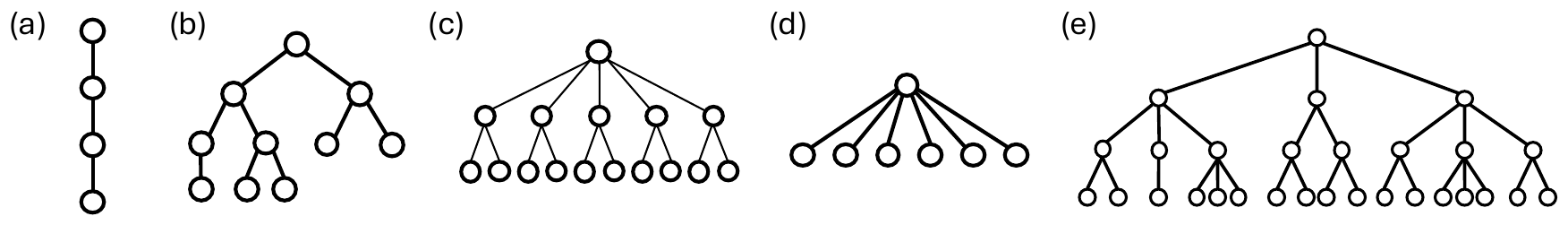}
    \vspace{-2em}
    \caption{
    Schematic illustrations of the tree families compared in Table~\ref{tab:peng_comparison}.
    (a)--(e) correspond to path, bounded-branching, bounded-dimension, two-layer star-shaped, and multiple-layer trees, respectively.
    Each vertex represents a local circuit component associated with a cluster, and each edge represents a cut system, or equivalently a cut wire, between neighboring clusters.}
    \label{fig:trees}
\end{figure*}

\begin{table*}[tp]
\centering
{
\scriptsize
\setlength{\tabcolsep}{2.4pt}
\renewcommand{\arraystretch}{2}
\resizebox{\textwidth}{!}{%
\begin{tabular}{|cc|c|ccc|ccc|}
\hline 
\multicolumn{2}{|c|}{Tree structure}
&
Num. of cuts
&
\multicolumn{3}{c|}{Entrywise reconstruction~\cite{peng2021simulating}}
&
\multicolumn{3}{c|}{This work
(Theorem~\ref{thm:tree_simulation_general})}
\\
\multicolumn{1}{|c}{Figure}
&
$(R,L,d)$
&
$K$
&
\multicolumn{1}{c}{Complexity}
&
Poly. in $K$?
&
Poly. in $d$?
&
\multicolumn{1}{c}{Complexity}
&
Poly. in $K$?
&
Poly. in $d$?
\\ 
\hline \hline

\multicolumn{2}{|c|}{Path tree (MPS-like)}
&
\multirow{2}{*}{$L$}
&
\multirow{2}{*}{$\displaystyle 
\mathcal{O}\!\left(
\frac{K^3 d^{18}}{\epsilon^2}\ln(Kd^6)
\right)$}
&
\multirow{2}{*}{Yes}
&
\multirow{2}{*}{Yes}
&
\multirow{2}{*}{$\displaystyle 
\mathcal{O}\!\left(
\frac{K^3 d^3}{\epsilon^2}\ln(Kd)
\right)$}
&
\multirow{2}{*}{Yes}
&
\multirow{2}{*}{Yes}
\\
Fig.~\ref{fig:trees}(a)
&
$(1,L,d)$
&&&&&&&
\\ 
\hline

\multicolumn{2}{|c|}{Bounded-branching tree}
&
\multirow{2}{*}{$\sum_{l=1}^{L}R^l$}
&
\multirow{2}{*}{%
$\displaystyle 
\mathcal{O}\!\left(
\frac{K^3 d^{\mathcal{O}(1)}}{\epsilon^2}
\ln(Kd^{\mathcal{O}(1)})
\right)$}
&
\multirow{2}{*}{Yes}
&
\multirow{2}{*}{Yes}
&
\multirow{2}{*}{$\displaystyle 
\mathcal{O}\!\left(
\frac{K^3d^3}{\epsilon^2}\ln(Kd)
\right)$}
&
\multirow{2}{*}{Yes}
&
\multirow{2}{*}{Yes}
\\
Fig.~\ref{fig:trees}(b)
&
$(\mathcal{O}(1),L,d)$
&&&&&&&
\\ 
\hline

\multicolumn{2}{|c|}{Bounded-dimension tree}
&
\multirow{2}{*}{$\sum_{l=1}^{L} R^l$}
&
\multirow{2}{*}{%
$\displaystyle 
\mathcal{O}\!\left(
\frac{K^3 2^{\mathcal{O}(K^{1/L})}}{\epsilon^2}\ln(K)
\right)$}
&
\multirow{2}{*}{%
\begin{tabular}[c]{@{}c@{}}
Yes if 
{\scriptsize $L \gtrsim \frac{\log K}{\log\log K}$}\\[-1pt]
No otherwise
\end{tabular}
}
&
\multirow{2}{*}{%
\begin{tabular}[c]{@{}c@{}}
Trivial\\[-1pt]
{\scriptsize $d=\mathcal O(1)$}
\end{tabular}
}
&
\multirow{2}{*}{$\displaystyle 
\mathcal{O}\!\left(
\frac{K^3}{\epsilon^2}\ln(K)
\right)$}
&
\multirow{2}{*}{Yes}
&
\multirow{2}{*}{%
\begin{tabular}[c]{@{}c@{}}
Trivial\\[-1pt]
{\scriptsize $d=\mathcal O(1)$}
\end{tabular}
}
\\
Fig.~\ref{fig:trees}(c)
&
$(R,L,\mathcal{O}(1))$
&&&&&&&
\\ 
\hline

\multicolumn{2}{|c|}{Two-layer tree (Star-shaped)}
&
\multirow{2}{*}{$R$}
&
\multirow{2}{*}{$\displaystyle 
\mathcal{O}\!\left(
\frac{K^3d^{9K}}{\epsilon^2}\ln(Kd^{3K})
\right)$}
&
\multirow{2}{*}{No}
&
\multirow{2}{*}{No}
&
\multirow{2}{*}{$\displaystyle 
\mathcal{O}\!\left(
\frac{K^3d^3}{\epsilon^2}\ln(Kd)
\right)$}
&
\multirow{2}{*}{Yes}
&
\multirow{2}{*}{Yes}
\\
Fig.~\ref{fig:trees}(d)
&
$(R,1,d)$
&&&&&&&
\\ 
\hline

\multicolumn{2}{|c|}{Multiple-layer tree}
&
\multirow{2}{*}{$\sum_{l=1}^{L} R^l$}
&
\multirow{2}{*}{%
$\displaystyle 
\mathcal{O}\!\left(
\frac{K^3d^{9(R+1)}}{\epsilon^2}\ln(Kd^{3(R+1)})
\right)$
}
&
\multirow{2}{*}{%
\begin{tabular}[c]{@{}c@{}}
No
\end{tabular}
}
&
\multirow{2}{*}{%
\begin{tabular}[c]{@{}c@{}}
No
\end{tabular}
}
&
\multirow{2}{*}{$\displaystyle 
\mathcal{O}\!\left(
\frac{K^3d^3}{\epsilon^2}\ln(Kd)
\right)$}
&
\multirow{2}{*}{Yes}
&
\multirow{2}{*}{Yes}
\\
Fig.~\ref{fig:trees}(e)
&
$(R,L,d)$
&&&&&&&
\\ 
\hline
\end{tabular}
}
}
\vspace{-0.5em}
\caption{
Comparison of the total number of local experimental samples, equivalently local quantum executions or measurement shots, required by the entrywise reconstruction of Peng et al.~\cite{peng2021simulating} and by our protocol. Here, $R$ denotes an upper bound on the branching number, $L$ denotes the tree depth, and $d$ denotes an upper bound on the cut-edge dimension.
``Poly. in $K$?'' asks whether the displayed cost is polynomial in the actual number of cut edges $K$, and ``Poly. in $d$?'' asks whether the degree of $d$ is bounded independently of the other growing tree parameters.
In the bounded-branching case, the $\mathcal O(1)$ exponent in $d^{\mathcal O(1)}$ is independent of $K$ but depends on the fixed branching bound; for branching bound $R\geq 2$, Eq.~\eqref{eq:peng_quantum_executions2} gives a factor $d^{9(R+1)}\geq d^{27}$.}
\label{tab:peng_comparison}
\end{table*}

This distinction leads to different scaling on high-degree trees. 
For the learning-based approach, Theorem~\ref{thm:tree_simulation_general} gives, for an arbitrary finite rooted tree with $K$ cut edges and incoming cut dimensions $\{d_v\}_{v\in V^\circ}$, the measurement bound
\begin{equation}
    N_{\rm learn}
    =
     \mathcal O\!\left(
        \frac{K^2}{\epsilon^2}
        \sum_{v\in V^\circ} d_v^3 \ln(Kd_{v})
    \right).
    \label{eq:ours_general_tree_peng}
\end{equation}
For the graph-dependent approach, in the tree setting we have $\kappa=K+1$, and the maximum-degree parameter $n_{\rm max}$ corresponds to
\begin{equation}
    n_{\rm max}
    =
    \max_v
    \log_2\!\left(
        d_v\prod_{u\in{\rm Ch}(v)}d_u
    \right),
\end{equation}
where the incoming factor $d_v$ is omitted for the root. Thus, we can translate the entrywise reconstruction bound~\eqref{eq:peng_quantum_executions1} into our tree notation as follows:
\begin{equation}
    \mathcal O\left(
    \max_{v}
    \frac{K^3
    \left(d_{v}\prod_{u\in{\rm Ch}(v)}d_u\right)^9}
    {\epsilon^2}
    \ln\!\left[K \biggl( d_v\prod_{u\in{\rm Ch}(v)}d_u \biggr)^3 \right]
    \right).
    \label{eq:peng_quantum_executions2}
\end{equation}
Based on this translated bound, Table~\ref{tab:peng_comparison} summarizes the measurement cost of the entrywise reconstruction in Ref.~\cite{peng2021simulating} and of our protocol for representative tree structures.
 
First, after translating the graph-dependent reconstruction bound of Ref.~\cite{peng2021simulating} into the tree parameters used here and explicitly tracking the finite-shot cost of estimating all tensor entries, we find polynomial dependence on $K$ for representative tree families, including path trees (Fig.~\ref{fig:trees}(a)) and bounded-branching trees (Fig.~\ref{fig:trees}(b)).
However, even in these polynomial regimes, there remains an important trade-off. 
The entrywise reconstruction of Ref.~\cite{peng2021simulating} is experimentally simple because it uses a fixed Pauli cutting basis and does not require diagonalizing an observable-adaptive basis. 
On the other hand, this fixed-basis reconstruction carries a much larger dependence on the cut dimension $d$. 
In particular, the $d^{\mathcal O(1)}$ dependence in the bounded-branching case corresponds to a fixed degree determined by the branching bound: for branching bound $R\ge2$, Eq.~\eqref{eq:peng_quantum_executions2} gives a factor $d^{9(R+1)}\ge d^{27}$.
In the path-tree case, the dimension factor is $d^{18}$ for the entrywise reconstruction, whereas it is $d^3$ for our learning-based protocol; for a qubit cut, this changes the factor from $2^{18}\approx 2.6 \times 10^5$ to $8$.

Second, when the branching number $R$ is allowed to grow, the entrywise reconstruction must estimate high-order boundary tensors indexed by all incident cutting labels, leading to factors such as $d^{9K}$ for the two-layer star (Fig.~\ref{fig:trees}(d)) and $d^{9(R+1)}$ for multi-layer trees with variable branching (Fig.~\ref{fig:trees}(e)). 
In contrast, our protocol compresses each processed subtree into a single observable-adaptive effective observable on the incoming edge, and therefore avoids estimating all entries of such fixed-basis boundary tensors.
As a consequence, Theorem~\ref{thm:tree_simulation_general} gives a polynomial finite-shot guarantee in $K$ and $d$ for arbitrary finite trees.

\subsection{Tree-based divide-and-conquer methods}\label{main:tree_dc_methods}

Tree-structured divide-and-conquer methods have also been developed beyond the standard circuit-cutting setting. 
A particularly broad formulation is provided by the hybrid tensor network (HTN) formalism~\cite{PhysRevLett.127.040501,Harada2025densitymatrix,PRXQuantum.6.010320}, where classical tensors and quantum tensors
\footnote{Here, a classical tensor refers to an ordinary tensor whose entries are explicitly stored and contracted by classical computation. 
A quantum tensor refers to a tensor represented by the amplitudes of quantum states prepared on a quantum device; for example, a family of states $\{\ket{\psi^i}=\sum_j \psi^i_j\ket{j}\}_i$ defines a tensor $\psi^i_j$.}
are combined within a single tensor-network ansatz.
This formalism provides a common language for several divide-and-conquer ansatzes, including Deep VQE~\cite{PRXQuantum.3.010346} and entanglement forging~\cite{PRXQuantum.3.010309}; see also Ref.~\cite{Harada2025densitymatrix}.
When the local tensors are known exactly, expectation-value estimation in these methods reduces to a tree contraction of local quantities.
In quantum-device implementations, however, these local quantities are obtained from finite-shot measurement data, and their statistical errors can propagate nontrivially through the contraction.
Below, we explain how the Heisenberg-picture learning primitive developed in this work gives a finite-shot stability guarantee for a representative hybrid tree tensor-network contraction.

To make this connection concrete, consider a two-layer hybrid tree tensor network (HTTN) state of the form
\begin{equation}\label{eq:httn_state}
    \ket{\psi_{\rm HT}}
    =
    \frac{1}{c}
    \sum_{i_1,\ldots,i_R}
    \phi_{i_1,\ldots,i_R}
    \ket{\psi_{1}^{i_1}} \otimes \cdots \otimes \ket{\psi_{R}^{i_R}},
\end{equation}
where $c$ is the normalization factor.
Here, $i_k\in[\chi_k]$ labels the bond index connecting the $k$-th leaf tensor to the parent tensor, and $\ket{\phi}:=\sum_{i_1,\ldots,i_R}\phi_{i_1,\ldots,i_R}\ket{i_1,\ldots,i_R}$ denotes the parent tensor written as a vector on the bond-index space.
The state $\ket{\psi_k^{i_k}}\in\mathcal H_k$ denotes the physical state associated with the $k$-th leaf tensor and bond index $i_k$, where $\mathcal H_k\simeq(\mathbb C^2)^{\otimes n_k}$ for an $n_k$-qubit leaf tensor.
For a product observable $O:= \bigotimes_{k=1}^{R} O_k$, with $\|O_k\|_\infty\leq 1$, the expectation value on $\ket{\psi_{\rm HT}}$ reduces to a contraction on the bond-index space:
\begin{equation}\label{eq:httn_exp_value}
    \braket{O}_{\psi_{\rm HT}} = \frac{\braket{\phi|\bigotimes_{k=1}^{R}M_k|\phi}}{\braket{\phi|\bigotimes_{k=1}^{R} S_k|\phi}},
\end{equation}
where the local operators are
\begin{equation}\label{eq:httn_local_tensor}
    (M_k)_{i_k,i'_k} := \braket{\psi_k^{i_k}|O_k|\psi_k^{i'_k}},
    \qquad
    (S_k)_{i_k,i'_k} := \braket{\psi_k^{i_k}|\psi_k^{i'_k}}.
\end{equation}
In the HTTN calculation, the target value $\braket{O}_{\psi_{\rm HT}}$ is evaluated in three steps. 
First, one computes Eq.~\eqref{eq:httn_local_tensor} on classical or quantum devices, obtaining the local operators $M_k$ and $S_k$. 
Second, one evaluates the expectation values of $\bigotimes_{k=1}^R M_k$ and $\bigotimes_{k=1}^R S_k$ on the parent tensor $\ket{\phi}$. 
Finally, one combines the two results according to Eq.~\eqref{eq:httn_exp_value}.

In practical HTTN calculations, however, the local operators $M_k$ and $S_k$ associated with quantum tensors are not available exactly, because they must be estimated from finite-shot measurement data.
The finite-shot issue is visible in the numerator of Eq.~\eqref{eq:httn_exp_value}. 
Suppose that the parent-level contraction uses estimates $\tilde M_k=M_k+\Delta M_k$ with 
$\|M_k\|_\infty\le1$, and define $\Delta_{\rm max}=\max_{k} \| \Delta M_k \|_{\infty}$. 
Then the observable contracted at the parent level can have norm
\begin{equation}
    \left\|
    \bigotimes_{k=1}^{R}\tilde M_k
    \right\|_\infty
    \leq \prod_{k}
    (1+ \| \Delta M_k \|_{\infty}) \leq (1+\Delta_{\rm max})^R.
\end{equation}
Thus, without operator-norm control and $R$-dependent accuracy allocation, local statistical fluctuations can be amplified by the product structure into an exponentially large single-shot range at the parent level.
Indeed, prior HTN works introduced Pauli-basis measurement/contraction rules for constructing operators such as $M_k$ and $S_k$~\cite{PhysRevLett.127.040501,Harada2025densitymatrix,PRXQuantum.6.010320}; some of these works also studied statistical noise numerically~\cite{PRXQuantum.6.010320} or briefly discussed statistical fluctuations in contracted observables~\cite{Harada2025densitymatrix}. 
However, an end-to-end finite-shot guarantee for tree contractions has not been established: such a guarantee would require controlling the local operator-norm errors $\|\Delta M_k\|_{\infty}$ induced by finite-shot data and allocating the local accuracies so that both the parent-level range and the accumulated bias remain bounded.

Our learning protocol provides such a guarantee in this representative setting: applied to each local contraction, it produces estimates $\tilde M_k$ with operator-norm accuracy $\eta=\mathcal O(\epsilon/R)$, keeping both the accumulated bias and the parent-level single-shot range bounded.
As a simple example, consider the type-(i) quantum tensors in the terminology of Ref.~\cite{Harada2025densitymatrix}. 
Suppose that the parent tensor is a quantum tensor represented by the normalized state $\ket{\phi}$, and that all $R$ leaf tensors are type-(i) quantum tensors. For simplicity, assume that each such leaf tensor has $n$ physical qubits and bond dimension $\chi=2^b$. Then it can be written as
\begin{equation}
    \ket{\psi_k^{i_k}}
    =
    Q_k
    \left(
        \ket{i_k}\otimes\ket{0}^{\otimes(n-b)}
    \right),
    \qquad
    i_k\in\{0,1\}^b,
\end{equation}
where $\ket{i_k}$ denotes the corresponding $b$-qubit computational-basis state and $Q_k$ is an $n$-qubit unitary.
For such type-(i) tensors, $S_k=I_\chi$, and the local contraction in Eq.~\eqref{eq:httn_local_tensor} can be written as the Heisenberg-evolved observable $M_k=\Phi_k^\dagger(O_k)$, where
\begin{equation}
    \Phi_k(\sigma)
    =
    Q_k
    \left(
        \sigma \otimes (\ket{0}\bra{0})^{\otimes(n-b)}
    \right)
    Q_k^\dagger .
\end{equation}
This is exactly the effective observable learned by our protocol.
Thus, in this type-(i) setting, the denominator in Eq.~\eqref{eq:httn_exp_value} reduces to $\braket{\phi|\phi}=1$.
Choosing the per-tensor learning accuracy as $\mathcal{O}(\epsilon/R)$ yields an $\mathcal O(\epsilon)$ bias contribution from the $R$ learned local effective observables.
After the learning stage, the remaining parent-tensor expectation value can be estimated by measuring $\bigotimes_{k=1}^R \tilde M_k$ on $\ket{\phi}$.
Since $\prod_{k=1}^R \|\tilde M_k\|_\infty \leq \left(1+\mathcal O(\epsilon/R)\right)^R = \mathcal{O}(1)$, the single-shot outcome is bounded by a constant, and $\mathcal{O}(1/\epsilon^2)$ final measurements are sufficient by Hoeffding's inequality.
Combining this with the observable-learning guarantee of Sec.~\ref{sec:learning_protocol}, the total number of measurements scales as
\begin{equation}
    \tilde{\mathcal O}
    \left(
        \frac{\chi^{3}R^3}{\epsilon^2}
    \right).
\end{equation}

More generally, the density-matrix formulation of HTNs in Ref.~\cite{Harada2025densitymatrix} provides a way to associate several types of quantum tensors with corresponding circuit or expansion-map representations. 
This suggests a possible route to extend our learning-based viewpoint beyond type-(i) tensors: once a quantum tensor is rewritten in such a representation, one may seek to learn the effective operator induced on the bond index using the protocol developed in Sec.~\ref{sec:learning_protocol}.
A complete sample-complexity analysis for these tensor types is left for future work, since it requires controlling normalization factors, possible postselection or success-probability overheads, and error propagation through the ratio form in Eq.~\eqref{eq:httn_exp_value}.

This perspective also clarifies how two consequences of Theorem~\ref{thm:existence_of_no_cost_cut} connect two apparently different divide-and-conquer viewpoints. 
Theorem~\ref{thm:zero_cost_timelike_cut_with_an_approximated_unitary} follows the circuit-cutting route: the learned downstream effective observable is used to choose an observable-adaptive measure-and-prepare replacement of the cut wire, in the spirit of expectation-value-level cut optimizations. 
Theorem~\ref{thm:zero_cost_proccessing}, in contrast, is closer to an HTTN-style contraction viewpoint. 
It uses the same learned effective observable through its eigenbasis and eigenvalues to define a measurement with classical post-processing, analogous to compressing a downstream quantum tensor into an effective operator on the bond index.

Thus, these two theorems can be viewed as two operational uses of the same Heisenberg-picture object, linking learning-based circuit cutting and HTTN-style contraction.

\subsection{Fundamental limit of circuit knitting}\label{main:fundamental_limit}

Several recent works have clarified fundamental limitations of circuit knitting and circuit cutting~\cite{PhysRevA.111.012433,marshall2023all}. In this subsection, we explain why our polynomial-scaling results do not contradict these lower bounds.

\subsubsection{Entanglement-cost lower bound}

We first discuss the entanglement-cost lower bounds of Ref.~\cite{PhysRevA.111.012433}. 
This work studies circuit knitting as a channel-simulation problem for a bipartite quantum channel $\mathcal{N}_{AB\to A'B'}$. Given a set $\mathcal{F}$ of allowed bipartite operations, such as LOCC, separable (SEP), or positive partial transpose (PPT) operations, the one-shot $\gamma$ factor is defined by the smallest $\ell_1$-norm of a channel-level quasiprobability decomposition:
\begin{equation}
\begin{split}
    \gamma_{\mathcal{F}}(\mathcal{N})
    :=
    \min
    \Bigg\{
        \sum_{j=1}^{m} |\alpha_j| :\
        &\mathcal{N}
        =
        \sum_{j=1}^{m} \alpha_j \mathcal{M}_j,\ 
        \mathcal{M}_j\in\mathcal{F},\\
        &\alpha_j\in \mathbb{R}, m\geq 1
    \Bigg\}.
\end{split}
\end{equation}
This is a universal channel-level quantity: 
the decomposition must reproduce the channel $\mathcal{N}$ itself, independently of the input state and of the observable measured later~\footnote{Here, ``channel-level'' means that the same decomposition is valid as an identity of linear maps, independently of the surrounding circuit, not that the QPD physically implements $\mathcal{N}$. Since standard QPDs generally involve negative coefficients, their operational use is also typically limited to expectation-value estimation.}.

Ref.~\cite{PhysRevA.111.012433} also considers an asymptotic variant of this task, called parallel cutting. 
Suppose that the same nonlocal channel $\mathcal{N}$ is used $n$ times. 
Instead of cutting the $n$ copies independently, one may regard $\mathcal{N}^{\otimes n}$ as a single large bipartite channel and optimize its QPD jointly.
Then, the effective per-use overhead is $\gamma_{\mathcal{F}}^{(n)}(\mathcal{N}):=\left[\gamma_{\mathcal{F}_n}(\mathcal{N}^{\otimes n})
\right]^{1/n}$, where $\mathcal{F}_n$ denotes the corresponding free operations acting on the $n$-copy systems. 
The regularized overhead, or regularized $\gamma$ factor, is the asymptotic value
\begin{equation}
    \gamma_{\mathcal{F}}^{\infty}(\mathcal{N})
    :=
    \lim_{n\to\infty}
    \gamma_{\mathcal{F}}^{(n)}(\mathcal{N}).
\end{equation}
Since cutting the $n$ copies independently is a special case of parallel cutting, one has $\gamma^{\infty}_{\mathcal F}(\mathcal N)\le \gamma_{\mathcal F}(\mathcal N)$.

The main result of Ref.~\cite{PhysRevA.111.012433} shows that this asymptotic per-use overhead is still lower bounded by entanglement-theoretic quantities. In particular, for the PPT- and SEP-assisted decompositions, one has
\begin{equation}
    \gamma^{\infty}_{\rm PPT}(\mathcal{N})
    \geq
    2^{E^{\rm PPT}_{C,0}(\mathcal{N})},
    \qquad
    \gamma^{\infty}_{\rm SEP}(\mathcal{N})
    \geq
    2^{E^{\rm SEP}_{C,0}(\mathcal{N})},
\end{equation}
where $E^{\mathcal F}_{C,0}(\mathcal{N})$ is the exact parallel entanglement cost of the channel under the corresponding free operation set. 
Since LOCC operations are contained in SEP operations, the SEP-assisted bound also implies a lower bound for LOCC-assisted circuit knitting. These results show that a bipartite channel with large entanglement cost cannot, in general, be simulated by local operations and classical post-processing without an exponentially large sampling overhead, even in the asymptotic parallel-cutting setting.

Our result does not contradict this limitation because it addresses a different simulation task. We do not claim that, in the channel-level sense considered above, either a general nonlocal channel or the identity channel admits a universal decomposition with low-$\gamma$. 
Instead, we relax the requirement to the expectation-value level. For a fixed downstream circuit component and target observable, our decomposition is only required to preserve the final expectation value of interest, and is adaptively built from information learned about the downstream effective observable; see the discussion around Eq.~\eqref{eq:channel-level_decomp}. 
When this effective observable is learned approximately, the resulting error appears as a controllable bias rather than a multiplicative rescaling factor. 
Thus, while Ref.~\cite{PhysRevA.111.012433} gives lower bounds on universal channel-level simulation cost, our result concerns a task-dependent, observable-adaptive guarantee for a fixed expectation-value task.

\subsubsection{Complexity-theoretic limitation}

A complementary limitation was studied in Ref.~\cite{marshall2023all} from a complexity-theoretic perspective. 
This work focuses on \textit{cut-local} circuit cutting schemes, namely schemes whose modifications are confined to the cut locations, such as partition-crossing gates.
It distinguishes two regimes. 
First, if the input state and measurement are fixed, then a small decomposition can formally exist, but an efficient classical algorithm that finds such decompositions for arbitrary circuits would imply $\mathrm{BPP}=\mathrm{BQP}$. 
Second, if the same decomposition is required to work for all inputs and measurements, then polynomially many cut-local terms are insufficient in general, even without computational assumptions.

Our protocol is different from the setting considered in this no-go result.
In the first regime, Ref.~\cite{marshall2023all} already allows the decomposition to depend on the fixed input state and measurement, but asks whether such decompositions can be found efficiently for arbitrary circuits by a classical algorithm. 
Our protocol does not provide such a classical procedure. 
It is restricted to tree-structured circuits, and the cut is obtained by running local quantum experiments to learn the relevant downstream effective observable, with this learning cost included in the sample complexity. 
The local subcircuits are still executed on quantum devices, so the protocol does not imply an efficient classical simulation of arbitrary quantum computation. 
For the second regime, our decomposition is also not universal. That is, it is tailored to a fixed downstream component and target observable, rather than being a single decomposition that works for all inputs and measurements.

Taken together, these lower bounds show that exponential barriers remain for universal channel-level circuit knitting and for cut-local decompositions of arbitrary circuits. 
Our result shows that, for tree-structured circuits and fixed expectation-value tasks, such barriers can be avoided by using local learning to construct observable-adaptive cuts.

\subsection{Classical tensor-network simulation}\label{main:tensor_network}

Our tree-structured setting is closely related to classical tensor-network simulation on low-treewidth graphs. 
A central result in this direction is due to Markov and Shi~\cite{doi:10.1137/050644756}. 
They showed that a quantum circuit with $T$ gates and underlying circuit graph $G_C$ can be simulated deterministically in time
\begin{equation}
    T^{\mathcal O(1)}
    \exp\!\left[
        \mathcal O\!\left(\operatorname{tw}(G_C)\right)
    \right],
\end{equation}
where $\operatorname{tw}(G_C)$ is the treewidth of the full gate-level circuit graph. 
Thus, when an explicit tensor-network description of the full circuit is available, circuits of logarithmic treewidth, and in particular circuits whose full gate-level graph is a tree, are classically tractable.

Our result concerns a different access model. 
Classical treewidth-based simulation contracts known tensors on a classical computer. 
By contrast, our protocol treats each local component as an experimentally accessible quantum channel and uses local quantum experiments to learn only the effective-observable information that must cross each cut. 
Thus, our polynomial scaling should not be interpreted as a classical hardness statement for tree-structured circuits. 
It is instead a measurement-complexity guarantee in a circuit-knitting setting where local quantum devices are available, full classical descriptions of the local components are not assumed, and global coherent execution is unavailable.

A useful way to see the distinction is to separate the coarse-grained interaction graph from the full gate-level graph. 
Let $H$ denote the coarse-grained graph whose vertices are the local components and whose edges are the cut wires, and let $G_{\rm full}$ denote the full gate-level circuit graph obtained by expanding every local component into its internal gates. 
In our setting, $H$ is a tree, so $\operatorname{tw}(H)=1$. 
However, this does not imply that $\operatorname{tw}(G_{\rm full})=1$. 
Indeed, if $G_i$ denotes the internal gate-level graph of the $i$-th local component, then $G_i$ is a subgraph of $G_{\rm full}$. 
By monotonicity of treewidth under taking subgraphs,
\begin{equation}
    \operatorname{tw}(G_{\rm full})
    \geq
    \max_i \operatorname{tw}(G_i).
\end{equation}
Therefore, a circuit can be tree-structured after coarse graining into local components while the full gate-level graph relevant to Markov-Shi simulation still has large treewidth. 
Equivalently, one may coarse-grain each component into a single tensor and obtain a tree tensor network, but then a classical contraction method requires an explicit description of each coarse-grained tensor.

Our protocol avoids this tensor-materialization step. 
It executes each local component on a quantum device and learns only the boundary information relevant to the target expectation value. 
In this sense, our method can be viewed as an experimental analogue of tree tensor contraction: classical contraction propagates explicit tensors, whereas our protocol propagates learned effective observables in the Heisenberg picture. 
The structural dependence is similar in spirit, but the resource being counted is different. 
Classical tensor-network simulation counts classical contraction time for known tensors, while our protocol counts the number of local quantum experiments needed to learn and use the relevant boundary information. 
The sample complexity of our protocol depends on the boundary dimension and the tree structure, rather than directly on the internal gate-level treewidth or non-Clifford complexity of each local component, although each shot still requires executing the corresponding local quantum circuit.

\section{Conclusion and discussions}\label{sec:coclusion}

In this work, we showed that the exponential-in-cuts measurement overhead of circuit cutting/knitting is not a universal consequence of divide-and-conquer.
For arbitrary finite tree-structured quantum circuits, we introduced an observable-adaptive cutting method for estimating expectation values of target observables.
The method constructs each cut from the Heisenberg-evolved observable relevant to the target measurement.
When this observable is learned approximately from local input-output experiments, the resulting cut has unit rescaling and incurs only a controlled additive bias.
By allocating the local learning accuracies across the tree, we obtained circuit-knitting protocols whose total measurement cost, including the local learning cost, is polynomial in the actual number of cut edges, the cut dimensions, and the inverse target precision.
This guarantee is uniform over the tree topology, rather than restricted to paths or bounded-branching trees.
We also proved an information-theoretic exponential separation from learning-free wire-cutting protocols in the two-layer setting.
This separation shows that the polynomial scaling does not follow from the tree structure alone, but from using learned downstream information to choose the cuts adaptively.

Conceptually, the result changes how the cost accumulates across cuts.
Conventional QPD-based cutting methods keep the estimator unbiased, but their rescaling factors multiply across cuts and amplify the variance.
Our construction removes the per-cut rescaling factor and replaces it with additive biases that can be reduced locally by improving the learning accuracy.
Thus, the central trade-off is shifted from multiplicative variance amplification to additive bias allocation.

An immediate question is how far this mechanism extends beyond trees.
One possible route is to combine standard gate cuts~\cite{10236453,Ufrecht2023cuttingmulticontrol,PRXQuantum.6.010316,schmitt2024cutting} or wire cuts~\cite{brenner2023optimal,harada2025doubly,PRXQuantum.6.010316,pednault2023alternative} with the present tree protocol.
For example, one may first decompose a general circuit into a coarse-grained cluster graph and then apply observable-adaptive learning along a tree decomposition of that graph.
In terms of its scaling behavior, such a hybrid approach would resemble classical tensor-network simulation of quantum circuits~\cite{doi:10.1137/050644756}: one expects polynomial dependence on the number of clusters, with a possible exponential dependence on an appropriate coarse-grained width parameter, such as the width of a tree decomposition.
Making this analogy precise in the local-experiment access model is an open problem.
Developing lower bounds would also be valuable for identifying which graph parameters, or which patterns of inter-cluster information flow, can force exponential overhead even in the presence of learning-assisted cuts.

Two aspects of the local learning step remain to be better understood.
First, our present operator-norm bounds are derived using matrix Bernstein-type inequalities and therefore inherit logarithmic factors in the cut dimension.
In related randomized tomography settings, such factors can sometimes be removed by using more refined proof techniques for sufficiently symmetric measurement ensembles~\cite{Guta_2020}.
It would be useful to determine whether the same is true for learning Heisenberg-evolved observables in our local input-output access model.
More generally, understanding whether the present $d^3$ dependence can be improved remains an important question.

Second, the current bound treats the downstream component as a general quantum channel.
In many relevant settings, however, the Heisenberg-evolved boundary observable may have a more compact description.
Examples include limited non-Clifford complexity or low-magic structure~\cite{g53f-z8cr}, bounded-gate-complexity dynamics~\cite{PRXQuantum.5.040306}, locality, and restricted noise models.
In such cases, one may not need to learn a general $d$-dimensional observable from scratch.
Identifying these structure-dependent regimes, while keeping track of both measurement cost and classical post-processing cost, would be an important step toward more practical observable-adaptive cutting protocols.

Another implementation question is how observable-adaptive cutting interacts with qubit-reuse compilation based on mid-circuit measurement and reset~\cite{PhysRevX.13.041057}.
Qubit reuse reduces device-size requirements in a way that is complementary to the measurement-overhead reduction studied in this work.
In tree-structured circuits, combining the two could trade larger local experiments and larger local learning tasks for fewer cut edges.
Quantifying this trade-off may be useful for applying observable-adaptive cutting to modular, distributed, or early fault-tolerant architectures.

Finally, our results point to a broader design principle for suppressing accumulated sampling overhead in virtual quantum-simulation techniques~\cite{PhysRevX.9.031013,PhysRevLett.132.050203,PhysRevA.109.022403,yao2024optimal,PhysRevLett.132.110203,yamamoto2024virtualentanglementpurificationnoisy}. In many virtual protocols, including circuit cutting, the main obstacle is not a single large overhead, but the multiplicative build-up of rescaling factors over many virtual operations or cut locations. In our setting, this accumulation is controlled by using local experimental resources to learn downstream effective observables, thereby replacing per-cut rescaling with a controllable additive bias. Other virtual-simulation protocols use different local resources, such as local coherent processing, to reduce the rescaling cost of individual virtual steps~\cite{wada2025state}. These examples suggest a common design principle: when additional local resources can be used to reduce or remove local rescaling factors, the global sampling cost may avoid multiplicative blow-up. Identifying the regimes in which this trade-off is possible may be useful for virtual simulation, quantum error mitigation~\cite{RevModPhys.95.045005,PhysRevLett.119.180509,PhysRevX.8.031027}, and hybrid quantum-classical simulation~\cite{endo2021hybrid}.

\begin{acknowledgements}
We thank Kosuke Mitarai and Bo Yang for helpful discussions.
This work was supported by JST [Moonshot R\&D] Grant No.~JPMJMS2061; MEXT Q-LEAP, Grant No.~JPMXS0120319794 and JPMXS0118067285; JST CREST Grant No.~JPMJCR23I4, and No.~JPMJCR25I4. H.H. was supported by JSPS KAKENHI Grant Number JP26KJ1951. K.W. was supported by JSPS KAKENHI Grant Number JP24KJ1963 and JST ASPIRE Grant Number JPMJAP2316.
\end{acknowledgements}

\bibliography{bib}

@article{PhysRevLett.119.180509,
  title = {Error Mitigation for Short-Depth Quantum Circuits},
  author = {Temme, Kristan and Bravyi, Sergey and Gambetta, Jay M.},
  journal = {Phys. Rev. Lett.},
  volume = {119},
  issue = {18},
  pages = {180509},
  numpages = {5},
  year = {2017},
  month = {Nov},
  publisher = {American Physical Society},
  doi = {10.1103/PhysRevLett.119.180509},
  url = {https://link.aps.org/doi/10.1103/PhysRevLett.119.180509}
}

@article{PhysRevX.8.031027,
  title = {Practical Quantum Error Mitigation for Near-Future Applications},
  author = {Endo, Suguru and Benjamin, Simon C. and Li, Ying},
  journal = {Phys. Rev. X},
  volume = {8},
  issue = {3},
  pages = {031027},
  numpages = {21},
  year = {2018},
  month = {Jul},
  publisher = {American Physical Society},
  doi = {10.1103/PhysRevX.8.031027},
  url = {https://link.aps.org/doi/10.1103/PhysRevX.8.031027}
}

@article{RevModPhys.95.045005,
  title = {Quantum error mitigation},
  author = {Cai, Zhenyu and Babbush, Ryan and Benjamin, Simon C. and Endo, Suguru and Huggins, William J. and Li, Ying and McClean, Jarrod R. and O'Brien, Thomas E.},
  journal = {Rev. Mod. Phys.},
  volume = {95},
  issue = {4},
  pages = {045005},
  numpages = {37},
  year = {2023},
  month = {Dec},
  publisher = {American Physical Society},
  doi = {10.1103/RevModPhys.95.045005},
  url = {https://link.aps.org/doi/10.1103/RevModPhys.95.045005}
}

@article{PhysRevX.6.021043,
  title = {Trading Classical and Quantum Computational Resources},
  author = {Bravyi, Sergey and Smith, Graeme and Smolin, John A.},
  journal = {Phys. Rev. X},
  volume = {6},
  issue = {2},
  pages = {021043},
  numpages = {14},
  year = {2016},
  month = {Jun},
  publisher = {American Physical Society},
  doi = {10.1103/PhysRevX.6.021043},
  url = {https://link.aps.org/doi/10.1103/PhysRevX.6.021043}
}

@article{peng2021simulating,
  title = {Simulating Large Quantum Circuits on a Small Quantum Computer},
  author = {Peng, Tianyi and Harrow, Aram W. and Ozols, Maris and Wu, Xiaodi},
  journal = {Phys. Rev. Lett.},
  volume = {125},
  issue = {15},
  pages = {150504},
  numpages = {6},
  year = {2020},
  month = {Oct},
  publisher = {American Physical Society},
  doi = {\doibase 10.1103/PhysRevLett.125.150504},
  url = {https://link.aps.org/doi/10.1103/PhysRevLett.125.150504}
}

@inproceedings{10.1145/3445814.3446758,
author = {Tang, Wei and Tomesh, Teague and Suchara, Martin and Larson, Jeffrey and Martonosi, Margaret},
title = {CutQC: Using Small Quantum Computers for Large Quantum Circuit Evaluations},
year = {2021},
isbn = {9781450383172},
publisher = {Association for Computing Machinery},
address = {New York, NY, USA},
url = {https://doi.org/10.1145/3445814.3446758},
doi = {\doibase 10.1145/3445814.3446758},
booktitle = {Proceedings of the 26th ACM International Conference on Architectural Support for Programming Languages and Operating Systems},
pages = {473–486},
numpages = {14},
location = {Virtual, USA},
series = {ASPLOS '21}
}

@article{Lowe2023fastquantumcircuit,
  doi = {10.22331/q-2023-03-02-934},
  url = {https://doi.org/10.22331/q-2023-03-02-934},
  title = {Fast quantum circuit cutting with randomized measurements},
  author = {Lowe, Angus and Medvidovi{\'{c}}, Matija and Hayes, Anthony and O'Riordan, Lee J. and Bromley, Thomas R. and Arrazola, Juan Miguel and Killoran, Nathan},
  journal = {{Quantum}},
  issn = {2521-327X},
  publisher = {{Verein zur F{\"{o}}rderung des Open Access Publizierens in den Quantenwissenschaften}},
  volume = {7},
  pages = {934},
  month = mar,
  year = {2023}
}

@ARTICLE{brenner2023optimal,
  author={Brenner, Lukas and Piveteau, Christophe and Sutter, David},
  journal={IEEE Transactions on Information Theory}, 
  title={Optimal Wire Cutting With Classical Communication}, 
  year={2025},
  volume={71},
  number={10},
  pages={7742-7752},
  doi={10.1109/TIT.2025.3601047}}

@article{harada2025doubly,
  title = {Doubly Optimal Parallel Wire Cutting without Ancilla Qubits},
  author = {Harada, Hiroyuki and Wada, Kaito and Yamamoto, Naoki},
  journal = {PRX Quantum},
  volume = {5},
  issue = {4},
  pages = {040308},
  numpages = {43},
  year = {2024},
  month = {Oct},
  publisher = {American Physical Society},
  doi = {10.1103/PRXQuantum.5.040308},
  url = {https://link.aps.org/doi/10.1103/PRXQuantum.5.040308}
}

@misc{pednault2023alternative,
      title={An alternative approach to optimal wire cutting without ancilla qubits}, 
      author={Edwin Pednault},
      year={2023},
      eprint={2303.08287},
      archivePrefix={arXiv},
      primaryClass={quant-ph}
}

@INPROCEEDINGS {10025537,
author = {G. Uchehara and T. M. Aamodt and O. Di Matteo},
booktitle = {2022 IEEE/ACM Third International Workshop on Quantum Computing Software (QCS)},
title = {Rotation-inspired circuit cut optimization},
year = {2022},
volume = {},
issn = {},
pages = {50-56},
doi = {10.1109/QCS56647.2022.00011},
url = {https://doi.ieeecomputersociety.org/10.1109/QCS56647.2022.00011},
publisher = {IEEE Computer Society},
address = {Los Alamitos, CA, USA},
month = {nov}
}

@INPROCEEDINGS{10313822,
  author={Chen, Daniel T. and Hansen, Ethan H. and Li, Xinpeng and Orenstein, Aaron and Kulkarni, Vinooth and Chaudhary, Vipin and Guan, Qiang and Liu, Ji and Zhang, Yang and Xu, Shuai},
  booktitle={2023 IEEE International Conference on Quantum Computing and Engineering (QCE)}, 
  title={Online Detection of Golden Circuit Cutting Points}, 
  year={2023},
  volume={01},
  number={},
  pages={26-31},
  doi={10.1109/QCE57702.2023.00012}}

@INPROCEEDINGS{10196555,
  author={Chen, Daniel T. and Hansen, Ethan H. and Li, Xinpeng and Kulkarni, Vinooth and Chaudhary, Vipin and Ren, Bin and Guan, Qiang and Kuppannagari, Sanmukh and Liu, Ji and Xu, Shuai},
  booktitle={2023 IEEE International Parallel and Distributed Processing Symposium Workshops (IPDPSW)}, 
  title={Efficient Quantum Circuit Cutting by Neglecting Basis Elements}, 
  year={2023},
  volume={},
  number={},
  pages={517-523},
  doi={10.1109/IPDPSW59300.2023.00091}}

@article{Mitarai_2021,
doi = {10.1088/1367-2630/abd7bc},
url = {https://dx.doi.org/10.1088/1367-2630/abd7bc},
year = {2021},
month = {feb},
publisher = {IOP Publishing},
volume = {23},
number = {2},
pages = {023021},
author = {Kosuke Mitarai and Keisuke Fujii},
title = {Constructing a virtual two-qubit gate by sampling single-qubit operations},
journal = {New Journal of Physics}
}

@article{Mitarai2021overheadsimulating,
  doi = {10.22331/q-2021-01-28-388},
  url = {https://doi.org/10.22331/q-2021-01-28-388},
  title = {Overhead for simulating a non-local channel with local channels by quasiprobability sampling},
  author = {Mitarai, Kosuke and Fujii, Keisuke},
  journal = {{Quantum}},
  issn = {2521-327X},
  publisher = {{Verein zur F{\"{o}}rderung des Open Access Publizierens in den Quantenwissenschaften}},
  volume = {5},
  pages = {388},
  month = jan,
  year = {2021}
}

@ARTICLE{10236453,
  author={Piveteau, Christophe and Sutter, David},
  journal={IEEE Transactions on Information Theory}, 
  title={Circuit Knitting With Classical Communication}, 
  year={2024},
  volume={70},
  number={4},
  pages={2734-2745},
  doi={10.1109/TIT.2023.3310797}}

@article{Ufrecht2023cuttingmulticontrol,
  doi = {\doibase 10.22331/q-2023-10-23-1147},
  url = {https://doi.org/10.22331/q-2023-10-23-1147},
  title = {Cutting multi-control quantum gates with {ZX} calculus},
  author = {Ufrecht, Christian and Periyasamy, Maniraman and Rietsch, Sebastian and Scherer, Daniel D. and Plinge, Axel and Mutschler, Christopher},
  journal = {{Quantum}},
  issn = {2521-327X},
  publisher = {{Verein zur F{\"{o}}rderung des Open Access Publizierens in den Quantenwissenschaften}},
  volume = {7},
  pages = {1147},
  month = oct,
  year = {2023}
}

@article{schmitt2024cutting,
  doi = {10.22331/q-2025-02-18-1634},
  url = {https://doi.org/10.22331/q-2025-02-18-1634},
  title = {Cutting circuits with multiple two-qubit unitaries},
  author = {Schmitt, Lukas and Piveteau, Christophe and Sutter, David},
  journal = {{Quantum}},
  issn = {2521-327X},
  publisher = {{Verein zur F{\"{o}}rderung des Open Access Publizierens in den Quantenwissenschaften}},
  volume = {9},
  pages = {1634},
  month = feb,
  year = {2025}
}

@ARTICLE{brandhofer2023optimal,
  author={Brandhofer, Sebastian and Polian, Ilia and Krsulich, Kevin},
  journal={IEEE Transactions on Quantum Engineering}, 
  title={Optimal Partitioning of Quantum Circuits Using Gate Cuts and Wire Cuts}, 
  year={2024},
  volume={5},
  number={},
  pages={1-10},
  doi={10.1109/TQE.2023.3347106}}

@article{PhysRevLett.127.040501,
  title = {Quantum Simulation with Hybrid Tensor Networks},
  author = {Yuan, Xiao and Sun, Jinzhao and Liu, Junyu and Zhao, Qi and Zhou, You},
  journal = {Phys. Rev. Lett.},
  volume = {127},
  issue = {4},
  pages = {040501},
  numpages = {6},
  year = {2021},
  month = {Jul},
  publisher = {American Physical Society},
  doi = {\doibase 10.1103/PhysRevLett.127.040501},
  url = {https://link.aps.org/doi/10.1103/PhysRevLett.127.040501}
}

@article{PRXQuantum.3.010346,
  title = {Deep Variational Quantum Eigensolver: A Divide-And-Conquer Method for Solving a Larger Problem with Smaller Size Quantum Computers},
  author = {Fujii, Keisuke and Mizuta, Kaoru and Ueda, Hiroshi and Mitarai, Kosuke and Mizukami, Wataru and Nakagawa, Yuya O.},
  journal = {PRX Quantum},
  volume = {3},
  issue = {1},
  pages = {010346},
  numpages = {12},
  year = {2022},
  month = {Mar},
  publisher = {American Physical Society},
  doi = {\doibase 10.1103/PRXQuantum.3.010346},
  url = {https://link.aps.org/doi/10.1103/PRXQuantum.3.010346}
}

@article{PRXQuantum.3.010309,
  title = {Doubling the Size of Quantum Simulators by Entanglement Forging},
  author = {Eddins, Andrew and Motta, Mario and Gujarati, Tanvi P. and Bravyi, Sergey and Mezzacapo, Antonio and Hadfield, Charles and Sheldon, Sarah},
  journal = {PRX Quantum},
  volume = {3},
  issue = {1},
  pages = {010309},
  numpages = {15},
  year = {2022},
  month = {Jan},
  publisher = {American Physical Society},
  doi = {\doibase 10.1103/PRXQuantum.3.010309},
  url = {https://link.aps.org/doi/10.1103/PRXQuantum.3.010309}
}

@article{Harada2025densitymatrix,
  doi = {10.22331/q-2025-08-07-1823},
  url = {https://doi.org/10.22331/q-2025-08-07-1823},
  title = {Density matrix representation of hybrid tensor networks for noisy quantum devices},
  author = {Harada, Hiroyuki and Suzuki, Yasunari and Yang, Bo and Tokunaga, Yuuki and Endo, Suguru},
  journal = {{Quantum}},
  issn = {2521-327X},
  publisher = {{Verein zur F{\"{o}}rderung des Open Access Publizierens in den Quantenwissenschaften}},
  volume = {9},
  pages = {1823},
  month = aug,
  year = {2025}
}

@article{PhysRevLett.115.070501,
  title = {Estimating Outcome Probabilities of Quantum Circuits Using Quasiprobabilities},
  author = {Pashayan, Hakop and Wallman, Joel J. and Bartlett, Stephen D.},
  journal = {Phys. Rev. Lett.},
  volume = {115},
  issue = {7},
  pages = {070501},
  numpages = {5},
  year = {2015},
  month = {Aug},
  publisher = {American Physical Society},
  doi = {10.1103/PhysRevLett.115.070501},
  url = {https://link.aps.org/doi/10.1103/PhysRevLett.115.070501}
}

@article{PhysRevLett.118.090501,
  title = {Application of a Resource Theory for Magic States to Fault-Tolerant Quantum Computing},
  author = {Howard, Mark and Campbell, Earl},
  journal = {Phys. Rev. Lett.},
  volume = {118},
  issue = {9},
  pages = {090501},
  numpages = {6},
  year = {2017},
  month = {Mar},
  publisher = {American Physical Society},
  doi = {10.1103/PhysRevLett.118.090501},
  url = {https://link.aps.org/doi/10.1103/PhysRevLett.118.090501}
}

@article{doi:10.1098/rspa.2019.0251,
author = {Seddon, James R.  and Campbell, Earl T. },
title = {Quantifying magic for multi-qubit operations},
journal = {Proceedings of the Royal Society A: Mathematical, Physical and Engineering Sciences},
volume = {475},
number = {2227},
pages = {20190251},
year = {2019},
doi = {10.1098/rspa.2019.0251},
URL = {https://royalsocietypublishing.org/doi/abs/10.1098/rspa.2019.0251}
}

@article{PRXQuantum.2.010345,
  title = {Quantifying Quantum Speedups: Improved Classical Simulation From Tighter Magic Monotones},
  author = {Seddon, James R. and Regula, Bartosz and Pashayan, Hakop and Ouyang, Yingkai and Campbell, Earl T.},
  journal = {PRX Quantum},
  volume = {2},
  issue = {1},
  pages = {010345},
  numpages = {42},
  year = {2021},
  month = {Mar},
  publisher = {American Physical Society},
  doi = {\doibase 10.1103/PRXQuantum.2.010345},
  url = {https://link.aps.org/doi/10.1103/PRXQuantum.2.010345}
}

@article{Guta_2020,
doi = {10.1088/1751-8121/ab8111},
url = {https://dx.doi.org/10.1088/1751-8121/ab8111},
year = {2020},
month = {apr},
publisher = {IOP Publishing},
volume = {53},
number = {20},
pages = {204001},
author = {Guţă, M and Kahn, J and Kueng, R and Tropp, J A},
title = {Fast state tomography with optimal error bounds},
journal = {Journal of Physics A: Mathematical and Theoretical}
}

@misc{zambrano2025fast,
      title={Fast quantum measurement tomography with dimension-optimal error bounds}, 
      author={Leonardo Zambrano and Sergi Ramos-Calderer and Richard Kueng},
      year={2025},
      eprint={2507.04500},
      archivePrefix={arXiv},
      primaryClass={quant-ph}, 
}

@article{webb2016clifford3design,
author = {Webb, Zak},
title = {The Clifford group forms a unitary 3-design},
year = {2016},
issue_date = {November 2016},
publisher = {Rinton Press, Incorporated},
address = {Paramus, NJ},
volume = {16},
number = {15–16},
issn = {1533-7146},
journal = {Quantum Info. Comput.},
month = nov,
pages = {1379–1400},
numpages = {22},
url = {https://dl.acm.org/doi/abs/10.5555/3179439.3179447}
}

@article{huang2020predicting,
  title={Predicting many properties of a quantum system from very few measurements},
  author={Huang, Hsin-Yuan and Kueng, Richard and Preskill, John},
  journal={Nature Physics},
  volume={16},
  number={10},
  pages={1050--1057},
  year={2020},
  publisher={Nature Publishing Group UK London},
  doi = {https://doi.org/10.1038/s41567-020-0932-7}
}

@misc{dankert2005efficient,
      title={Efficient Simulation of Random Quantum States and Operators}, 
      author={Christoph Dankert},
      year={2005},
      eprint={quant-ph/0512217},
      archivePrefix={arXiv},
      primaryClass={quant-ph}, 
}

@article{PhysRevA.77.012307,
  title = {Randomized benchmarking of quantum gates},
  author = {Knill, E. and Leibfried, D. and Reichle, R. and Britton, J. and Blakestad, R. B. and Jost, J. D. and Langer, C. and Ozeri, R. and Seidelin, S. and Wineland, D. J.},
  journal = {Phys. Rev. A},
  volume = {77},
  issue = {1},
  pages = {012307},
  numpages = {7},
  year = {2008},
  month = {Jan},
  publisher = {American Physical Society},
  doi = {10.1103/PhysRevA.77.012307},
  url = {https://link.aps.org/doi/10.1103/PhysRevA.77.012307}
}

@article{emerson2005scalable,
  title={Scalable noise estimation with random unitary operators},
  author={Emerson, Joseph and Alicki, Robert and {\.Z}yczkowski, Karol},
  journal={Journal of Optics B: Quantum and Semiclassical Optics},
  volume={7},
  number={10},
  pages={S347},
  year={2005},
  publisher={IOP Publishing},
  doi = {DOI 10.1088/1464-4266/7/10/021}
}

@inproceedings{ji2018pseudorandom,
author = {Ji, Zhengfeng and Liu, Yi-Kai and Song, Fang},
title = {Pseudorandom Quantum States},
year = {2018},
isbn = {978-3-319-96877-3},
publisher = {Springer-Verlag},
address = {Berlin, Heidelberg},
doi = {10.1007/978-3-319-96878-0_5},
booktitle = {Advances in Cryptology – CRYPTO 2018: 38th Annual International Cryptology Conference, Santa Barbara, CA, USA, August 19–23, 2018, Proceedings, Part III},
pages = {126–152},
numpages = {27},
location = {Santa Barbara, CA, USA}
}

@inproceedings{ananth2022cryptography,
  title={Cryptography from pseudorandom quantum states},
  author={Ananth, Prabhanjan and Qian, Luowen and Yuen, Henry},
  booktitle={Annual International Cryptology Conference},
  pages={208--236},
  year={2022},
  organization={Springer},
  doi = {https://doi.org/10.1007/978-3-031-15802-5_8}
}

@article{PhysRevA.80.012304,
  title = {Exact and approximate unitary 2-designs and their application to fidelity estimation},
  author = {Dankert, Christoph and Cleve, Richard and Emerson, Joseph and Livine, Etera},
  journal = {Phys. Rev. A},
  volume = {80},
  issue = {1},
  pages = {012304},
  numpages = {6},
  year = {2009},
  month = {Jul},
  publisher = {American Physical Society},
  doi = {10.1103/PhysRevA.80.012304},
  url = {https://link.aps.org/doi/10.1103/PhysRevA.80.012304}
}

@inproceedings{ma2025construct,
  title={How to construct random unitaries},
  author={Ma, Fermi and Huang, Hsin-Yuan},
  booktitle={Proceedings of the 57th Annual ACM Symposium on Theory of Computing},
  pages={806--809},
  year={2025},
  doi = {https://doi.org/10.1145/3717823.37182}
}

@book{Nielsen_Chuang_2010,
place={Cambridge},
title={Quantum Computation and Quantum Information: 10th Anniversary Edition},
publisher={Cambridge University Press},
author={Nielsen, Michael A. and Chuang, Isaac L.}, year={2010}}

@article{PRXQuantum.2.030316,
  title = {Models of Quantum Complexity Growth},
  author = {Brand\~ao, Fernando G.S.L. and Chemissany, Wissam and Hunter-Jones, Nicholas and Kueng, Richard and Preskill, John},
  journal = {PRX Quantum},
  volume = {2},
  issue = {3},
  pages = {030316},
  numpages = {40},
  year = {2021},
  month = {Jul},
  publisher = {American Physical Society},
  doi = {10.1103/PRXQuantum.2.030316},
  url = {https://link.aps.org/doi/10.1103/PRXQuantum.2.030316}
}

@article{PhysRevLett.125.200501,
  title = {Mixed-State Entanglement from Local Randomized Measurements},
  author = {Elben, Andreas and Kueng, Richard and Huang, Hsin-Yuan (Robert) and van Bijnen, Rick and Kokail, Christian and Dalmonte, Marcello and Calabrese, Pasquale and Kraus, Barbara and Preskill, John and Zoller, Peter and Vermersch, Beno\^{\i}t},
  journal = {Phys. Rev. Lett.},
  volume = {125},
  issue = {20},
  pages = {200501},
  numpages = {6},
  year = {2020},
  month = {Nov},
  publisher = {American Physical Society},
  doi = {10.1103/PhysRevLett.125.200501},
  url = {https://link.aps.org/doi/10.1103/PhysRevLett.125.200501}
}

@article{gross2007evenly,
  title={Evenly distributed unitaries: On the structure of unitary designs},
  author={Gross, David and Audenaert, Koenraad and Eisert, Jens},
  journal={Journal of Mathematical Physics},
  volume={48},
  pages   = {052104},
  number={5},
  year={2007},
  publisher={AIP Publishing},
  url = {https://doi.org/10.1063/1.2716992}
}

@INPROCEEDINGS{4262758,
  author={Ambainis, Andris and Emerson, Joseph},
  booktitle={Twenty-Second Annual IEEE Conference on Computational Complexity (CCC'07)}, 
  title={Quantum t-designs: t-wise Independence in the Quantum World}, 
  year={2007},
  volume={},
  number={},
  pages={129-140},
  doi={10.1109/CCC.2007.26}}

@incollection{delsarte1991spherical,
title = {SPHERICAL CODES AND DESIGNS},
editor = {D.G. Corneil and R. Mathon},
booktitle = {Geometry and Combinatorics},
publisher = {Academic Press},
pages = {68-93},
year = {1991},
isbn = {978-0-12-189420-7},
doi = {https://doi.org/10.1016/B978-0-12-189420-7.50013-X},
author = {Delsarte, Philippe and Goethals, Jean-Marie and Seidel, Johan Jacob}
}

@article{HOGGAR1982233,
title = {t-Designs in Projective Spaces},
journal = {European Journal of Combinatorics},
volume = {3},
number = {3},
pages = {233-254},
year = {1982},
issn = {0195-6698},
doi = {https://doi.org/10.1016/S0195-6698(82)80035-8},
url = {https://www.sciencedirect.com/science/article/pii/S0195669882800358},
author = {S.G. Hoggar}
}

@article{PhysRevA.69.062321,
  title = {Minimal universal two-qubit controlled-NOT-based circuits},
  author = {Shende, Vivek V. and Markov, Igor L. and Bullock, Stephen S.},
  journal = {Phys. Rev. A},
  volume = {69},
  issue = {6},
  pages = {062321},
  numpages = {8},
  year = {2004},
  month = {Jun},
  publisher = {American Physical Society},
  doi = {10.1103/PhysRevA.69.062321},
  url = {https://link.aps.org/doi/10.1103/PhysRevA.69.062321}
}

@article{PhysRevLett.106.180504,
  title = {Scalable and Robust Randomized Benchmarking of Quantum Processes},
  author = {Magesan, Easwar and Gambetta, J. M. and Emerson, Joseph},
  journal = {Phys. Rev. Lett.},
  volume = {106},
  issue = {18},
  pages = {180504},
  numpages = {4},
  year = {2011},
  month = {May},
  publisher = {American Physical Society},
  doi = {10.1103/PhysRevLett.106.180504},
  url = {https://link.aps.org/doi/10.1103/PhysRevLett.106.180504}
}

@book{gottesman1997stabilizer,
  title={Stabilizer codes and quantum error correction},
  author={Gottesman, Daniel},
  year={1997},
  publisher={California Institute of Technology}
}

@misc{kueng2015qubit,
      title={Qubit stabilizer states are complex projective 3-designs}, 
      author={Richard Kueng and David Gross},
      year={2015},
      eprint={1510.02767},
      archivePrefix={arXiv},
      primaryClass={quant-ph}, 
}

@article{webb2015clifford,
author = {Webb, Zak},
title = {The Clifford group forms a unitary 3-design},
year = {2016},
issue_date = {November 2016},
publisher = {Rinton Press, Incorporated},
address = {Paramus, NJ},
volume = {16},
number = {15–16},
issn = {1533-7146},
journal = {Quantum Info. Comput.},
month = nov,
pages = {1379–1400},
numpages = {22},
url = {https://dl.acm.org/doi/abs/10.5555/3179439.3179447}
}

@article{zhu2017multiqubit,
  title = {Multiqubit Clifford groups are unitary 3-designs},
  author = {Zhu, Huangjun},
  journal = {Phys. Rev. A},
  volume = {96},
  issue = {6},
  pages = {062336},
  numpages = {7},
  year = {2017},
  month = {Dec},
  publisher = {American Physical Society},
  doi = {10.1103/PhysRevA.96.062336},
  url = {https://link.aps.org/doi/10.1103/PhysRevA.96.062336}
}

@misc{zhu2016clifford,
      title={The Clifford group fails gracefully to be a unitary 4-design}, 
      author={Huangjun Zhu and Richard Kueng and Markus Grassl and David Gross},
      year={2016},
      eprint={1609.08172},
      archivePrefix={arXiv},
      primaryClass={quant-ph}, 
}

@INPROCEEDINGS{9605330,
  author={Van Den Berg, Ewout},
  booktitle={2021 IEEE International Conference on Quantum Computing and Engineering (QCE)}, 
  title={A simple method for sampling random Clifford operators}, 
  year={2021},
  volume={},
  number={},
  pages={54-59},
  doi={10.1109/QCE52317.2021.00021}}

@article{10.1063/1.4903507,
    author = {Koenig, Robert and Smolin, John A.},
    title = {How to efficiently select an arbitrary Clifford group element},
    journal = {Journal of Mathematical Physics},
    volume = {55},
    number = {12},
    pages = {122202},
    year = {2014},
    month = {12},
    issn = {0022-2488},
    doi = {10.1063/1.4903507}
}

@article{PhysRevA.70.052328,
  title = {Improved simulation of stabilizer circuits},
  author = {Aaronson, Scott and Gottesman, Daniel},
  journal = {Phys. Rev. A},
  volume = {70},
  issue = {5},
  pages = {052328},
  numpages = {14},
  year = {2004},
  month = {Nov},
  publisher = {American Physical Society},
  doi = {10.1103/PhysRevA.70.052328},
  url = {https://link.aps.org/doi/10.1103/PhysRevA.70.052328}
}

@article{wootters1989optimal,
title = {Optimal state-determination by mutually unbiased measurements},
journal = {Annals of Physics},
volume = {191},
number = {2},
pages = {363-381},
year = {1989},
issn = {0003-4916},
doi = {https://doi.org/10.1016/0003-4916(89)90322-9},
url = {https://www.sciencedirect.com/science/article/pii/0003491689903229},
author = {William K Wootters and Brian D Fields}
}

@article{PhysRevA.65.032320,
  title = {Mutually unbiased binary observable sets on N qubits},
  author = {Lawrence, Jay and Brukner, {\v{C}}aslav and Zeilinger, Anton},
  journal = {Phys. Rev. A},
  volume = {65},
  issue = {3},
  pages = {032320},
  numpages = {5},
  year = {2002},
  month = {Feb},
  publisher = {American Physical Society},
  doi = {10.1103/PhysRevA.65.032320},
  url = {https://link.aps.org/doi/10.1103/PhysRevA.65.032320}
}

@misc{seyfarth2019cyclic,
      title={Cyclic Mutually Unbiased Bases and Quantum Public-Key Encryption}, 
      author={Ulrich Seyfarth},
      year={2019},
      eprint={1907.02726},
      archivePrefix={arXiv},
      primaryClass={quant-ph}
}

@ARTICLE{9248636,
  author={Gokhale, Pranav and Angiuli, Olivia and Ding, Yongshan and Gui, Kaiwen and Tomesh, Teague and Suchara, Martin and Martonosi, Margaret and Chong, Frederic T.},
  journal={IEEE Transactions on Quantum Engineering}, 
  title={$O(N^3)$ Measurement Cost for Variational Quantum Eigensolver on Molecular Hamiltonians}, 
  year={2020},
  volume={1},
  number={},
  pages={1-24},
  doi={10.1109/TQE.2020.3035814}}

@article{renes2004symmetric,
  title={Symmetric informationally complete quantum measurements},
  author={Renes, Joseph M and Blume-Kohout, Robin and Scott, Andrew J and Caves, Carlton M},
  journal={Journal of Mathematical Physics},
  volume={45},
  number={6},
  pages={2171--2180},
  year={2004},
  publisher={American Institute of Physics},
  doi = {https://doi.org/10.1063/1.1737053
}
}

@article{ivonovic1981geometrical,
  title={Geometrical description of quantal state determination},
  author={Ivonovic, Igor D},
  journal={Journal of Physics A: Mathematical and General},
  volume={14},
  number={12},
  pages={3241},
  year={1981},
  publisher={IOP Publishing},
  doi = {10.1088/0305-4470/14/12/019}
}

@INPROCEEDINGS{1523643,
  author={Klappenecker, A. and Rotteler, M.},
  booktitle={Proceedings. International Symposium on Information Theory, 2005. ISIT 2005.}, 
  title={Mutually unbiased bases are complex projective 2-designs}, 
  year={2005},
  volume={},
  number={},
  pages={1740-1744},
  doi={10.1109/ISIT.2005.1523643}}

@article{Crawford2021efficientquantum,
  doi = {\doibase 10.22331/q-2021-01-20-385},
  url = {https://doi.org/10.22331/q-2021-01-20-385},
  title = {Efficient quantum measurement of {P}auli operators in the presence of finite sampling error},
  author = {Crawford, Ophelia and Straaten, Barnaby van and Wang, Daochen and Parks, Thomas and Campbell, Earl and Brierley, Stephen},
  journal = {{Quantum}},
  issn = {2521-327X},
  publisher = {{Verein zur F{\"{o}}rderung des Open Access Publizierens in den Quantenwissenschaften}},
  volume = {5},
  pages = {385},
  month = jan,
  year = {2021}
}

@article{PhysRevApplied.21.064001,
  title = {Minimal-Clifford shadow estimation by mutually unbiased bases},
  author = {Zhang, Qingyue and Liu, Qing and Zhou, You},
  journal = {Phys. Rev. Appl.},
  volume = {21},
  issue = {6},
  pages = {064001},
  numpages = {24},
  year = {2024},
  month = {Jun},
  publisher = {American Physical Society},
  doi = {10.1103/PhysRevApplied.21.064001},
  url = {https://link.aps.org/doi/10.1103/PhysRevApplied.21.064001}
}

@article{PhysRevA.109.062406,
  title = {Classical shadow tomography with mutually unbiased bases},
  author = {Wang, Yu and Cui, Wei},
  journal = {Phys. Rev. A},
  volume = {109},
  issue = {6},
  pages = {062406},
  numpages = {12},
  year = {2024},
  month = {Jun},
  publisher = {American Physical Society},
  doi = {10.1103/PhysRevA.109.062406},
  url = {https://link.aps.org/doi/10.1103/PhysRevA.109.062406}
}

@article{Mele2024introductiontohaar,
  doi = {10.22331/q-2024-05-08-1340},
  url = {https://doi.org/10.22331/q-2024-05-08-1340},
  title = {Introduction to {H}aar {M}easure {T}ools in {Q}uantum {I}nformation: {A} {B}eginner's {T}utorial},
  author = {Mele, Antonio Anna},
  journal = {{Quantum}},
  issn = {2521-327X},
  publisher = {{Verein zur F{\"{o}}rderung des Open Access Publizierens in den Quantenwissenschaften}},
  volume = {8},
  pages = {1340},
  month = may,
  year = {2024}
}

@misc{harrow2013church,
      title={The Church of the Symmetric Subspace}, 
      author={Aram W. Harrow},
      year={2013},
      eprint={1308.6595},
      archivePrefix={arXiv},
      primaryClass={quant-ph} 
}

@misc{oliveira2009concentration,
      title={Concentration of the adjacency matrix and of the Laplacian in random graphs with independent edges}, 
      author={Roberto Imbuzeiro Oliveira},
      year={2010},
      eprint={0911.0600},
      archivePrefix={arXiv},
      primaryClass={math.CO}, 
}

@article{tropp2012user,
  title={User-friendly tail bounds for sums of random matrices},
  author={Tropp, Joel A},
  journal={Foundations of computational mathematics},
  volume={12},
  number={4},
  pages={389--434},
  year={2012},
  publisher={Springer},
  doi = {https://doi.org/10.1007/s10208-011-9099-z}
}

@article{tropp2015introduction,
url = {http://dx.doi.org/10.1561/2200000048},
year = {2015},
volume = {8},
journal = {Foundations and Trends® in Machine Learning},
title = {An Introduction to Matrix Concentration Inequalities},
doi = {10.1561/2200000048},
issn = {1935-8237},
number = {1-2},
pages = {1-230},
author = {Joel A. Tropp}
}

@article{PRXQuantum.6.010316,
  title = {Optimal Quantum Circuit Cuts with Application to Clustered Hamiltonian Simulation},
  author = {Harrow, Aram W. and Lowe, Angus},
  journal = {PRX Quantum},
  volume = {6},
  issue = {1},
  pages = {010316},
  numpages = {23},
  year = {2025},
  month = {Jan},
  publisher = {American Physical Society},
  doi = {10.1103/PRXQuantum.6.010316},
  url = {https://link.aps.org/doi/10.1103/PRXQuantum.6.010316}
}

@inproceedings{10.1145/237814.237866,
author = {Grover, Lov K.},
title = {A Fast Quantum Mechanical Algorithm for Database Search},
year = {1996},
isbn = {0897917855},
publisher = {Association for Computing Machinery},
address = {New York, NY, USA},
url = {https://doi.org/10.1145/237814.237866},
doi = {10.1145/237814.237866},
booktitle = {Proceedings of the Twenty-Eighth Annual ACM Symposium on Theory of Computing},
pages = {212–219},
numpages = {8},
location = {Philadelphia, Pennsylvania, USA},
series = {STOC '96}
}

@article{doi:10.1137/S0097539795293172,
author = {Shor, Peter W.},
title = {Polynomial-Time Algorithms for Prime Factorization and Discrete Logarithms on a Quantum Computer},
journal = {SIAM Journal on Computing},
volume = {26},
number = {5},
pages = {1484-1509},
year = {1997},
doi = {10.1137/S0097539795293172}
}

@ARTICLE{fuller2021approximate,
  author={Fuller, Bryce and Hadfield, Charles and Glick, Jennifer R. and Imamichi, Takashi and Itoko, Toshinari and Thompson, Richard J. and Jiao, Yang and Kagele, Marna M. and Blom-Schieber, Adriana W. and Raymond, Rudy and Mezzacapo, Antonio},
  journal={IEEE Transactions on Quantum Engineering}, 
  title={Approximate Solutions of Combinatorial Problems via Quantum Relaxations}, 
  year={2024},
  volume={5},
  number={},
  pages={1-15},
  doi={10.1109/TQE.2024.3421294}
}

@Article{Amaro2022,
author={Amaro, David
and Rosenkranz, Matthias
and Fitzpatrick, Nathan
and Hirano, Koji
and Fiorentini, Mattia},
title={A case study of variational quantum algorithms for a job shop scheduling problem},
journal={EPJ Quantum Technology},
year={2022},
month={Feb},
day={10},
volume={9},
number={1},
pages={5},
issn={2196-0763},
doi={\doibase 10.1140/epjqt/s40507-022-00123-4},
url={https://doi.org/10.1140/epjqt/s40507-022-00123-4}
}

@article{Moll_2018,
doi = {10.1088/2058-9565/aab822},
url = {https://dx.doi.org/10.1088/2058-9565/aab822},
year = {2018},
month = {jun},
publisher = {IOP Publishing},
volume = {3},
number = {3},
pages = {030503},
author = {Nikolaj Moll and Panagiotis Barkoutsos and Lev S Bishop and Jerry M Chow and Andrew Cross and Daniel J Egger and Stefan Filipp and Andreas Fuhrer and Jay M Gambetta and Marc Ganzhorn and Abhinav Kandala and Antonio Mezzacapo and Peter Müller and Walter Riess and Gian Salis and John Smolin and Ivano Tavernelli and Kristan Temme},
title = {Quantum optimization using variational algorithms on near-term quantum devices},
journal = {Quantum Science and Technology}
}

@Article{Havlíček2019,
author={Havl{\'i}{\v{c}}ek, Vojt{\v{e}}ch
and C{\'o}rcoles, Antonio D.
and Temme, Kristan
and Harrow, Aram W.
and Kandala, Abhinav
and Chow, Jerry M.
and Gambetta, Jay M.},
title={Supervised learning with quantum-enhanced feature spaces},
journal={Nature},
year={2019},
month={Mar},
day={01},
volume={567},
number={7747},
pages={209-212},
issn={1476-4687},
doi={\doibase 10.1038/s41586-019-0980-2},
url={https://doi.org/10.1038/s41586-019-0980-2}
}

@article{peruzzo2014variational,
  title={A variational eigenvalue solver on a photonic quantum processor},
  author={Peruzzo, Alberto and McClean, Jarrod and Shadbolt, Peter and Yung, Man-Hong and Zhou, Xiao-Qi and Love, Peter J and Aspuru-Guzik, Al{\'a}n and O’brien, Jeremy L},
  journal={Nature Communications},
  volume={5},
  number={1},
  pages={4213},
  year={2014},
  publisher={Nature Publishing Group UK London},
  doi = {\doibase 10.1038/ncomms5213},
  URL = 
  {https://doi.org/10.1038/ncomms5213}
}

@article{kandala2017hardware,
  title={Hardware-efficient variational quantum eigensolver for small molecules and quantum magnets},
  author={Kandala, Abhinav and Mezzacapo, Antonio and Temme, Kristan and Takita, Maika and Brink, Markus and Chow, Jerry M and Gambetta, Jay M},
  journal={Nature},
  volume={549},
  number={7671},
  pages={242--246},
  year={2017},
  publisher={Nature Publishing Group UK London},
  doi = {10.1038/nature23879},
  URL = {https://doi.org/10.1038/nature23879}
}

@article{orus2019quantum,
  title={Quantum computing for finance: Overview and prospects},
  author={Or{\'u}s, Rom{\'a}n and Mugel, Samuel and Lizaso, Enrique},
  journal={Reviews in Physics},
  volume={4},
  pages={100028},
  year={2019},
  publisher={Elsevier},
  doi ={https://doi.org/10.1016/j.revip.2019.100028},
}

@misc{herman2022survey,
      title={A Survey of Quantum Computing for Finance}, 
      author={Dylan Herman and Cody Googin and Xiaoyuan Liu and Alexey Galda and Ilya Safro and Yue Sun and Marco Pistoia and Yuri Alexeev},
      year={2022},
      eprint={2201.02773},
      archivePrefix={arXiv},
      primaryClass={quant-ph} 
}

@article{bauer2020quantum,
author = {Bauer, Bela and Bravyi, Sergey and Motta, Mario and Chan, Garnet Kin-Lic},
title = {Quantum Algorithms for Quantum Chemistry and Quantum Materials Science},
journal = {Chemical Reviews},
volume = {120},
number = {22},
pages = {12685-12717},
year = {2020},
doi = {10.1021/acs.chemrev.9b00829},
URL = {https://doi.org/10.1021/acs.chemrev.9b00829}
}

@article{ma2020quantum,
  title={Quantum simulations of materials on near-term quantum computers},
  author={Ma, He and Govoni, Marco and Galli, Giulia},
  journal={npj Computational Materials},
  volume={6},
  number={1},
  pages={85},
  year={2020},
  publisher={Nature Publishing Group UK London},
  doi = {https://doi.org/10.1038/s41524-020-00353-z}
}

@article{bechtold2023investigating,
  title={Investigating the effect of circuit cutting in QAOA for the MaxCut problem on NISQ devices},
  author={Bechtold, Marvin and Barzen, Johanna and Leymann, Frank and Mandl, Alexander and Obst, Julian and Truger, Felix and Weder, Benjamin},
  journal={Quantum Science and Technology},
  volume={8},
  number={4},
  pages={045022},
  year={2023},
  publisher={IOP Publishing},
  doi = {10.1088/2058-9565/acf59c}
}

@article{yuan2021universal,
  title={Universal and operational benchmarking of quantum memories},
  author={Yuan, Xiao and Liu, Yunchao and Zhao, Qi and Regula, Bartosz and Thompson, Jayne and Gu, Mile},
  journal={npj Quantum Information},
  volume={7},
  number={1},
  pages={108},
  year={2021},
  publisher={Nature Publishing Group UK London},
  doi = {https://doi.org/10.1038/s41534-021-00444-9}
}

@article{doi:10.1142/S0129055X03001709,
author = {Horodecki, Michael and Shor, Peter W. and Ruskai, Mary Beth},
title = {Entanglement Breaking Channels},
journal = {Reviews in Mathematical Physics},
volume = {15},
number = {06},
pages = {629-641},
year = {2003},
doi = {10.1142/S0129055X03001709},
URL = {https://doi.org/10.1142/S0129055X03001709}
}

@book{rozanov2013probability,
  title={Probability theory: a concise course},
  author={Rozanov, Yurii A},
  year={2013},
  publisher={Courier Corporation}
}

@article{PRXQuantum.6.010320,
  title = {Hybrid Tree Tensor Networks for Quantum Simulation},
  author = {Schuhmacher, Julian and Ballarin, Marco and Baiardi, Alberto and Magnifico, Giuseppe and Tacchino, Francesco and Montangero, Simone and Tavernelli, Ivano},
  journal = {PRX Quantum},
  volume = {6},
  issue = {1},
  pages = {010320},
  numpages = {23},
  year = {2025},
  month = {Jan},
  publisher = {American Physical Society},
  doi = {10.1103/PRXQuantum.6.010320},
  url = {https://link.aps.org/doi/10.1103/PRXQuantum.6.010320}
}

@article{hoeffding1963probability,
 ISSN = {01621459, 1537274X},
 URL = {http://www.jstor.org/stable/2282952},
 author = {Wassily Hoeffding},
 journal = {Journal of the American Statistical Association},
 number = {301},
 pages = {13--30},
 publisher = {[American Statistical Association, Taylor & Francis, Ltd.]},
 title = {Probability Inequalities for Sums of Bounded Random Variables},
 urldate = {2025-12-20},
 volume = {58},
 year = {1963}
}

@article{PhysRevA.108.022615,
  title = {Quantum channel decomposition with preselection and postselection},
  author = {Nagai, Ryo and Kanno, Shu and Sato, Yuki and Yamamoto, Naoki},
  journal = {Phys. Rev. A},
  volume = {108},
  issue = {2},
  pages = {022615},
  numpages = {10},
  year = {2023},
  month = {Aug},
  publisher = {American Physical Society},
  doi = {10.1103/PhysRevA.108.022615},
  url = {https://link.aps.org/doi/10.1103/PhysRevA.108.022615}
}

@article{PhysRevLett.92.177902,
  title = {Efficient Decomposition of Quantum Gates},
  author = {Vartiainen, Juha J. and M\"ott\"onen, Mikko and Salomaa, Martti M.},
  journal = {Phys. Rev. Lett.},
  volume = {92},
  issue = {17},
  pages = {177902},
  numpages = {4},
  year = {2004},
  month = {Apr},
  publisher = {American Physical Society},
  doi = {10.1103/PhysRevLett.92.177902},
  url = {https://link.aps.org/doi/10.1103/PhysRevLett.92.177902}
}

@Article{app12020759,
AUTHOR = {Krol, Anna M. and Sarkar, Aritra and Ashraf, Imran and Al-Ars, Zaid and Bertels, Koen},
TITLE = {Efficient Decomposition of Unitary Matrices in Quantum Circuit Compilers},
JOURNAL = {Applied Sciences},
VOLUME = {12},
YEAR = {2022},
NUMBER = {2},
ARTICLE-NUMBER = {759},
URL = {https://www.mdpi.com/2076-3417/12/2/759},
ISSN = {2076-3417},
DOI = {10.3390/app12020759}
}

@misc{wada2025state,
      title={State-to-Hamiltonian conversion with a few copies}, 
      author={Kaito Wada and Jumpei Kato and Hiroyuki Harada and Naoki Yamamoto},
      year={2025},
      eprint={2509.14791},
      archivePrefix={arXiv},
      primaryClass={quant-ph}, 
}

@INPROCEEDINGS {tomesh2023divide,
author = { Tomesh, Teague and Saleem, Zain H. and Perlin, Michael A. and Gokhale, Pranav and Suchara, Martin and Martonosi, Margaret },
booktitle = { 2023 IEEE International Conference on Quantum Computing and Engineering (QCE) },
title = {{ Divide and Conquer for Combinatorial Optimization and Distributed Quantum Computation }},
year = {2023},
volume = {},
ISSN = {},
pages = {1-12},
doi = {10.1109/QCE57702.2023.00009},
url = {https://doi.ieeecomputersociety.org/10.1109/QCE57702.2023.00009},
publisher = {IEEE Computer Society},
address = {Los Alamitos, CA, USA},
month =sep}

@article{huang2025estimating,
  title = {Estimating local observables via cluster-level light-cone decomposition},
  author = {Huang, Junxiang and Tang, Yunxin and Yuan, Xiao},
  journal = {Phys. Rev. A},
  volume = {113},
  issue = {3},
  pages = {032414},
  numpages = {16},
  year = {2026},
  month = {Mar},
  publisher = {American Physical Society},
  doi = {10.1103/rzhg-szgj},
  url = {https://link.aps.org/doi/10.1103/rzhg-szgj}
}

@article{endo2021hybrid,
author = {Endo ,Suguru and Cai ,Zhenyu and Benjamin ,Simon C. and Yuan ,Xiao},
title = {Hybrid Quantum-Classical Algorithms and Quantum Error Mitigation},
journal = {Journal of the Physical Society of Japan},
volume = {90},
number = {3},
pages = {032001},
year = {2021},
doi = {10.7566/JPSJ.90.032001},

URL = {https://doi.org/10.7566/JPSJ.90.032001
}
}

@misc{scholten2024assessing,
      title={Assessing the Benefits and Risks of Quantum Computers}, 
      author={Travis L. Scholten and Carl J. Williams and Dustin Moody and Michele Mosca and William Hurley and William J. Zeng and Matthias Troyer and Jay M. Gambetta},
      year={2024},
      eprint={2401.16317},
      archivePrefix={arXiv},
      primaryClass={quant-ph}
}

@article{carrera2024combining,
  title={Combining quantum processors with real-time classical communication},
  author={Carrera Vazquez, Almudena and Tornow, Caroline and Riste, Diego and Woerner, Stefan and Takita, Maika and Egger, Daniel J},
  journal={Nature},
  volume={636},
  number={8041},
  pages={75--79},
  year={2024},
  publisher={Nature Publishing Group UK London},
  doi = {https://doi.org/10.1038/s41586-024-08178-2}
}

@article{PhysRevLett.130.110601,
  title = {Experimental Simulation of Larger Quantum Circuits with Fewer Superconducting Qubits},
  author = {Ying, Chong and Cheng, Bin and Zhao, Youwei and Huang, He-Liang and Zhang, Yu-Ning and Gong, Ming and Wu, Yulin and Wang, Shiyu and Liang, Futian and Lin, Jin and Xu, Yu and Deng, Hui and Rong, Hao and Peng, Cheng-Zhi and Yung, Man-Hong and Zhu, Xiaobo and Pan, Jian-Wei},
  journal = {Phys. Rev. Lett.},
  volume = {130},
  issue = {11},
  pages = {110601},
  numpages = {7},
  year = {2023},
  month = {Mar},
  publisher = {American Physical Society},
  doi = {10.1103/PhysRevLett.130.110601},
  url = {https://link.aps.org/doi/10.1103/PhysRevLett.130.110601}
}

@misc{marshall2023all,
      title={All this for one qubit? Bounds on local circuit cutting schemes}, 
      author={Simon C. Marshall and Jordi Tura and Vedran Dunjko},
      year={2023},
      eprint={2303.13422},
      archivePrefix={arXiv},
      primaryClass={quant-ph}, 
}

@article{PhysRevX.13.041057,
  title = {Qubit-Reuse Compilation with Mid-Circuit Measurement and Reset},
  author = {DeCross, Matthew and Chertkov, Eli and Kohagen, Megan and Foss-Feig, Michael},
  journal = {Phys. Rev. X},
  volume = {13},
  issue = {4},
  pages = {041057},
  numpages = {22},
  year = {2023},
  month = {Dec},
  publisher = {American Physical Society},
  doi = {10.1103/PhysRevX.13.041057},
  url = {https://link.aps.org/doi/10.1103/PhysRevX.13.041057}
}

@article{PhysRevX.9.031013,
  title = {Quantum Virtual Cooling},
  author = {Cotler, Jordan and Choi, Soonwon and Lukin, Alexander and Gharibyan, Hrant and Grover, Tarun and Tai, M. Eric and Rispoli, Matthew and Schittko, Robert and Preiss, Philipp M. and Kaufman, Adam M. and Greiner, Markus and Pichler, Hannes and Hayden, Patrick},
  journal = {Phys. Rev. X},
  volume = {9},
  issue = {3},
  pages = {031013},
  numpages = {11},
  year = {2019},
  month = {Jul},
  publisher = {American Physical Society},
  doi = {10.1103/PhysRevX.9.031013},
  url = {https://link.aps.org/doi/10.1103/PhysRevX.9.031013}
}

@article{PhysRevLett.132.050203,
  title = {Virtual Quantum Resource Distillation},
  author = {Yuan, Xiao and Regula, Bartosz and Takagi, Ryuji and Gu, Mile},
  journal = {Phys. Rev. Lett.},
  volume = {132},
  issue = {5},
  pages = {050203},
  numpages = {6},
  year = {2024},
  month = {Feb},
  publisher = {American Physical Society},
  doi = {10.1103/PhysRevLett.132.050203},
  url = {https://link.aps.org/doi/10.1103/PhysRevLett.132.050203}
}

@article{PhysRevA.109.022403,
  title = {Virtual quantum resource distillation: General framework and applications},
  author = {Takagi, Ryuji and Yuan, Xiao and Regula, Bartosz and Gu, Mile},
  journal = {Phys. Rev. A},
  volume = {109},
  issue = {2},
  pages = {022403},
  numpages = {23},
  year = {2024},
  month = {Feb},
  publisher = {American Physical Society},
  doi = {10.1103/PhysRevA.109.022403},
  url = {https://link.aps.org/doi/10.1103/PhysRevA.109.022403}
}

@article{yao2024optimal,
  title = {Optimal unilocal virtual quantum broadcasting},
  author = {Yao, Hongshun and Liu, Xia and Zhu, Chengkai and Wang, Xin},
  journal = {Phys. Rev. A},
  volume = {110},
  issue = {1},
  pages = {012458},
  numpages = {8},
  year = {2024},
  month = {Jul},
  publisher = {American Physical Society},
  doi = {10.1103/PhysRevA.110.012458},
  url = {https://link.aps.org/doi/10.1103/PhysRevA.110.012458}
}

@article{PhysRevLett.132.110203,
  title = {Virtual Quantum Broadcasting},
  author = {Parzygnat, Arthur J. and Fullwood, James and Buscemi, Francesco and Chiribella, Giulio},
  journal = {Phys. Rev. Lett.},
  volume = {132},
  issue = {11},
  pages = {110203},
  numpages = {8},
  year = {2024},
  month = {Mar},
  publisher = {American Physical Society},
  doi = {10.1103/PhysRevLett.132.110203},
  url = {https://link.aps.org/doi/10.1103/PhysRevLett.132.110203}
}

@article{jnane2022multicore,
  title={Multicore quantum computing},
  author={Jnane, Hamza and Undseth, Brennan and Cai, Zhenyu and Benjamin, Simon C and Koczor, B{\'a}lint},
  journal={Phys. Rev. Applied},
  volume={18},
  number={4},
  pages={044064},
  year={2022},
  publisher={APS},
  doi = { https://doi.org/10.1103/PhysRevApplied.18.044064}
}

@article{BARRAL2025100747,
title = {Review of Distributed Quantum Computing: From single QPU to High Performance Quantum Computing},
journal = {Computer Science Review},
volume = {57},
pages = {100747},
year = {2025},
issn = {1574-0137},
doi = {https://doi.org/10.1016/j.cosrev.2025.100747},
url = {https://www.sciencedirect.com/science/article/pii/S1574013725000231},
author = {David Barral and F. Javier Cardama and Guillermo Díaz-Camacho and Daniel Faílde and Iago F. Llovo and Mariamo Mussa-Juane and Jorge Vázquez-Pérez and Juan Villasuso and César Piñeiro and Natalia Costas and Juan C. Pichel and Tomás F. Pena and Andrés Gómez}
}

@misc{yamamoto2024virtualentanglementpurificationnoisy,
      title={Virtual entanglement purification via noisy entanglement}, 
      author={Kaoru Yamamoto and Yuichiro Matsuzaki and Yasunari Suzuki and Yuuki Tokunaga and Suguru Endo},
      year={2024},
      eprint={2411.10024},
      archivePrefix={arXiv},
      primaryClass={quant-ph}, 
}

@article{kanno2024quantum,
  title={Quantum computing quantum Monte Carlo with hybrid tensor network for electronic structure calculations},
  author={Kanno, Shu and Nakamura, Hajime and Kobayashi, Takao and Gocho, Shigeki and Hatanaka, Miho and Yamamoto, Naoki and Gao, Qi},
  journal={npj Quantum Information},
  volume={10},
  number={1},
  pages={56},
  year={2024},
  publisher={Nature Publishing Group UK London},
  doi = {https://doi.org/10.1038/s41534-024-00851-8}
}

@article{PhysRevResearch.3.043121,
  title = {Deep variational quantum eigensolver for excited states and its application to quantum chemistry calculation of periodic materials},
  author = {Mizuta, Kaoru and Fujii, Mikiya and Fujii, Shigeki and Ichikawa, Kazuhide and Imamura, Yutaka and Okuno, Yukihiro and Nakagawa, Yuya O.},
  journal = {Phys. Rev. Res.},
  volume = {3},
  issue = {4},
  pages = {043121},
  numpages = {17},
  year = {2021},
  month = {Nov},
  publisher = {American Physical Society},
  doi = {10.1103/PhysRevResearch.3.043121},
  url = {https://link.aps.org/doi/10.1103/PhysRevResearch.3.043121}
}

@misc{huembeli2022entanglement,
      title={Entanglement Forging with generative neural network models}, 
      author={Patrick Huembeli and Giuseppe Carleo and Antonio Mezzacapo},
      year={2022},
      eprint={2205.00933},
      archivePrefix={arXiv},
      primaryClass={quant-ph},
}

@article{
doi:10.1126/science.aam9288,
author = {Stephanie Wehner  and David Elkouss  and Ronald Hanson },
title = {Quantum internet: A vision for the road ahead},
journal = {Science},
volume = {362},
number = {6412},
pages = {eaam9288},
year = {2018},
doi = {10.1126/science.aam9288}}

@article{RevModPhys.87.1379,
  title = {Cavity-based quantum networks with single atoms and optical photons},
  author = {Reiserer, Andreas and Rempe, Gerhard},
  journal = {Rev. Mod. Phys.},
  volume = {87},
  issue = {4},
  pages = {1379--1418},
  numpages = {40},
  year = {2015},
  month = {Dec},
  publisher = {American Physical Society},
  doi = {10.1103/RevModPhys.87.1379},
  url = {https://link.aps.org/doi/10.1103/RevModPhys.87.1379}
}

@article{knorzer2025distributed,
	author={Knörzer, Johannes and Liu, Xiaoyu and Schiffer, Benjamin and Tura Brugués, Jordi},
	title={Distributed quantum information processing: a review of recent progress},
	journal={Reports on Progress in Physics},
	url={http://iopscience.iop.org/article/10.1088/1361-6633/ae74e0},
	year={2026}
}

@article{PhysRevA.111.012433,
  title = {Circuit knitting facing exponential sampling-overhead scaling bounded by entanglement cost},
  author = {Jing, Mingrui and Zhu, Chengkai and Wang, Xin},
  journal = {Phys. Rev. A},
  volume = {111},
  issue = {1},
  pages = {012433},
  numpages = {18},
  year = {2025},
  month = {Jan},
  publisher = {American Physical Society},
  doi = {10.1103/PhysRevA.111.012433},
  url = {https://link.aps.org/doi/10.1103/PhysRevA.111.012433}
}

@article{doi:10.1137/050644756,
author = {Markov, Igor L. and Shi, Yaoyun},
title = {Simulating Quantum Computation by Contracting Tensor Networks},
journal = {SIAM Journal on Computing},
volume = {38},
number = {3},
pages = {963-981},
year = {2008},
doi = {10.1137/050644756}
}

@misc{oshio2026nearheisenberglimitedparallelamplitudeestimation,
      title={Near-Heisenberg-limited parallel amplitude estimation with logarithmic depth circuit}, 
      author={Kohei Oshio and Kaito Wada and Naoki Yamamoto},
      year={2026},
      eprint={2508.06121},
      archivePrefix={arXiv},
      primaryClass={quant-ph}, 
}

@article{Martyn2025parallelquantum,
  doi = {10.22331/q-2025-08-27-1834},
  url = {https://doi.org/10.22331/q-2025-08-27-1834},
  title = {Parallel {Q}uantum {S}ignal {P}rocessing {V}ia {P}olynomial {F}actorization},
  author = {Martyn, John M. and Rossi, Zane M. and Cheng, Kevin Z. and Liu, Yuan and Chuang, Isaac L.},
  journal = {{Quantum}},
  issn = {2521-327X},
  publisher = {{Verein zur F{\"{o}}rderung des Open Access Publizierens in den Quantenwissenschaften}},
  volume = {9},
  pages = {1834},
  month = aug,
  year = {2025}
}

@article{PRXQuantum.5.040306,
  title = {Learning Quantum States and Unitaries of Bounded Gate Complexity},
  author = {Zhao, Haimeng and Lewis, Laura and Kannan, Ishaan and Quek, Yihui and Huang, Hsin-Yuan and Caro, Matthias C.},
  journal = {PRX Quantum},
  volume = {5},
  issue = {4},
  pages = {040306},
  numpages = {63},
  year = {2024},
  month = {Oct},
  publisher = {American Physical Society},
  doi = {10.1103/PRXQuantum.5.040306},
  url = {https://link.aps.org/doi/10.1103/PRXQuantum.5.040306}
}

@article{g53f-z8cr,
  title = {Efficient Distributed Inner-Product Estimation via Pauli Sampling},
  author = {Hinsche, M. and Ioannou, M. and Jerbi, S. and Leone, L. and Eisert, J. and Carrasco, J.},
  journal = {PRX Quantum},
  volume = {6},
  issue = {3},
  pages = {030354},
  numpages = {41},
  year = {2025},
  month = {Sep},
  publisher = {American Physical Society},
  doi = {10.1103/g53f-z8cr},
  url = {https://link.aps.org/doi/10.1103/g53f-z8cr}
}

@article{chuang1997prescription,
  title={Prescription for experimental determination of the dynamics of a quantum black box},
  author={Chuang, Isaac L and Nielsen, Michael A},
  journal={Journal of Modern Optics},
  volume={44},
  number={11-12},
  pages={2455--2467},
  year={1997},
  publisher={Taylor \& Francis}
}

@article{PhysRevLett.78.390,
  title = {Complete Characterization of a Quantum Process: The Two-Bit Quantum Gate},
  author = {Poyatos, J. F. and Cirac, J. I. and Zoller, P.},
  journal = {Phys. Rev. Lett.},
  volume = {78},
  issue = {2},
  pages = {390--393},
  numpages = {0},
  year = {1997},
  month = {Jan},
  publisher = {American Physical Society},
  doi = {10.1103/PhysRevLett.78.390},
  url = {https://link.aps.org/doi/10.1103/PhysRevLett.78.390}
}

@article{PhysRevLett.93.250407,
  title = {Quantum Calibration of Measurement Instrumentation},
  author = {D'Ariano, Giacomo Mauro and Maccone, Lorenzo and Presti, Paoloplacido Lo},
  journal = {Phys. Rev. Lett.},
  volume = {93},
  issue = {25},
  pages = {250407},
  numpages = {4},
  year = {2004},
  month = {Dec},
  publisher = {American Physical Society},
  doi = {10.1103/PhysRevLett.93.250407},
  url = {https://link.aps.org/doi/10.1103/PhysRevLett.93.250407}
}

\clearpage
\appendix
\onecolumngrid
\beginsupplement

\TOCwriteon

\begin{center}
  \Large \textbf{Appendices}
\end{center}

\addtocontents{toc}{\protect\setcounter{tocdepth}{3}}

\tableofcontents

\vspace{4mm}\noindent
We outline the structure of the appendices.
In Appendix~\ref{sec:QPD}, we provide a brief review of the quasiprobability simulation framework underlying quantum circuit cutting methods.
In Appendix~\ref{apsec:rescaling-free_wirecut}, we present the proofs of Theorems~\ref{thm:existence_of_no_cost_cut}--\ref{thm:zero_cost_proccessing} on the rescaling-free wire cut introduced in the main text.
In Appendix~\ref{apsec:learning_observable}, we first introduce useful tools and concepts for deriving tomography methods for observables evolved under an unknown quantum channel. Using these tools, we formulate learning procedures and provide their rigorous performance guarantees.
In Appendix~\ref{apsec:cluster_simulation}, we analyze the sample complexity upper bound of simulating tree-structured quantum circuits on local devices using our learning-based protocols.
In Appendix~\ref{sec:lower-bound-qcc}, we focus on two-layer tree-structured quantum circuits and perform an information-theoretic analysis that proves an exponential separation in the number of cuts between our learning-based protocol and the learning-free wire-cutting class considered in Sec.~\ref{sec:exponential_separation}.
In Appendix~\ref{apsec:relation_to_wire_cut}, we discuss the relationship between Theorem~\ref{thm:existence_of_no_cost_cut} and conventional wire-cutting methods.

\section{Quasiprobability simulation}\label{sec:QPD}

Quasiprobability simulation provides a general framework for classically reproducing the outcome statistics of a quantum process that cannot be directly implemented, using only the available operations on a given device.
This simulation scheme has been extensively employed across a wide range of quantum information and computation tasks, including quantum error mitigation~\cite{RevModPhys.95.045005,PhysRevLett.119.180509,PhysRevX.8.031027}, classical simulation of quantum systems~\cite{PhysRevLett.115.070501, PhysRevLett.118.090501, doi:10.1098/rspa.2019.0251, PRXQuantum.2.010345}, and quantum circuit cutting. 
For instance, in the context of classical simulation of quantum systems, a non-Clifford gate is regarded as the desired operation to be replaced, and it is typically simulated by stabilizer operations.

The key technical idea of the quasiprobability simulation is to express the target operation as a linear combination of implementable ones, referred to as a quasiprobability decomposition (QPD). Suppose that, due to physical limitations or experimental costs, the available set of quantum operations is restricted to a family $F$, while the desired operation $\mathcal{T}$ lies outside this set. Then, we represent $\mathcal{T}$ as a linear combination of operations $\{\mathcal{E}_i\}_i\subseteq F$:
\begin{equation}\label{apeq:QPD}
    \mathcal{T} = \sum_i a_i \mathcal{E}_i,
\end{equation}
where the coefficients $a_i$ are real numbers that may take negative values. Building on this decomposition, quasiprobability simulation effectively reproduces the action of $\mathcal{T}$ by randomly sampling $\mathcal{E}_i$ with the probabilities $|a_i|/\gamma$ and weighting outcomes by $\gamma \mathrm{sgn}(a_i)$, where $\gamma:=\sum_i|a_i|$.

To illustrate this idea, let us consider the task of estimating the expectation value of an observable $O$ with $\|O\|_{\infty}\leq1$ for the state obtained by applying $\mathcal{T}$ to an initial state $\rho$. 
Using Eq.~\eqref{apeq:QPD}, the desired quantity can be rewritten as
\begin{equation}\label{eq:QPD2}
    \mathrm{tr}[O \mathcal{T}(\rho)] = \sum_i \frac{|a_i|}{\gamma} \, \gamma \mathrm{sgn}(a_i) \, \mathrm{tr}[O \mathcal{E}_i(\rho)].
\end{equation}
where $\gamma:=\sum_i |a_i|$ and $\mathrm{sgn}(\cdot)$ denotes the sign function. 
Since $|a_i|/\gamma$ defines a normalized probability over the implementable operations, Eq.~\eqref{eq:QPD2} provides a natural Monte-Carlo sampling scheme: for each shot, select an index $i$ according to the probability $|a_i|/\gamma$, apply the corresponding operation $\mathcal{E}_i$ to $\rho$, measure the observable $O$, and weight the measurement outcome by $\gamma \mathrm{sgn}(a_{i})$. By averaging these weighted outcomes over the total number of shots $N$, we obtain the desired expectation value $\mathrm{tr}[O \mathcal{T}(\rho)]$. If $\|O\|_\infty\le 1$, each weighted single-shot outcome is bounded in
absolute value by $\gamma$. Therefore, Hoeffding's inequality implies
that $\mathcal O(\gamma^2\epsilon^{-2}\log(1/\delta))$ samples are sufficient to
estimate the expectation value within additive error $\epsilon$ with failure probability at most $\delta$.

\section{Observable-adaptive rescaling-free wire cuts}\label{apsec:rescaling-free_wirecut}

\subsection{Proof of Theorem~\ref{thm:existence_of_no_cost_cut}}
\label{sec:proof_of_existence_of_no_cost_cut}

\newtheorem*{T1}{\bf{Theorem~\ref{thm:existence_of_no_cost_cut}} (Existence of a rescaling-free wire cut)}
\begin{T1}
\textit{
Let $\Phi:\Linear{\mathbb{C}^{d_{\rm in}}} \to \Linear{\mathbb{C}^{d_{\rm out}}}$ be a CPTP map, and let $O \in \Herm{\mathbb{C}^{d_{\rm out}}}$ be a Hermitian operator. 
Then there exists a MP channel $\mathcal{M}_{\rm id}:\Linear{\mathbb{C}^{d_{\rm in}}} \to \Linear{\mathbb{C}^{d_{\rm in}}}$ such that, for any $X\in \Linear{\mathbb{C}^{d_{\rm in}}}$,
\begin{equation}
    \mathrm{tr}\left[O \,\Phi(X)\right] = \mathrm{tr}\left[ O\, (\Phi \circ \mathcal{M}_{\rm id}) (X) \right]. 
\end{equation}
An explicit construction of $\mathcal{M}_{\rm id}$ is given by
\begin{equation}\label{apeq:gamma_optimal_timelike_cut}
    \mathcal{M}_{\rm id}(X) = \sum_{j} \mathrm{tr}\left[V \pure{j} V^{\dagger} X\right] V\pure{j} V^{\dagger},
\end{equation}
where $\{\ket{j}\}$ denotes the computational basis, and $V$ is a unitary that diagonalizes $O_{\Phi}:=\Phi^{\dagger}(O)$, i.e., \begin{equation}\label{apeq:eigenvalue_decomposition_of_O}
    O_{\Phi}:= \Phi^{\dagger}(O) = V D V^{\dagger}.
\end{equation}
}
\end{T1}

\begin{proof}[Proof of Theorem~\ref{thm:existence_of_no_cost_cut}]
Let $D:=\sum_{j} \lambda_j \pure{j}$ with the computational basis $\{\ket{j}\}$.
By the definition of $\mathcal{M}_{\rm id}$ in Eq.~\eqref{apeq:gamma_optimal_timelike_cut}, it is straightforward to see
\begin{eqnarray}
    \mathcal{M}_{\rm id}(O_{\Phi}) 
    &=& \sum_{j}\braket{j|V^{\dagger}O_{\Phi}V|j} V\pure{j} V^{\dagger}\\
    &=& \sum_j \lambda_j V\pure{j} V^{\dagger} \\
    &=& O_{\Phi}\label{apeq:fixed_point},
\end{eqnarray}
which implies $O_{\Phi}$ is a fixed point of $\mathcal{M}_{\rm id}$.
Thus, we have
\begin{eqnarray}
    \mathrm{tr}\left[ O\, (\Phi \circ \mathcal{M}_{\rm id})(X) \right] 
    &=& \mathrm{tr}\left[ \mathcal{M}^{\dagger}_{\rm id}(O_{\Phi}) X \right] \\
    &=& \mathrm{tr}\left[ \mathcal{M}_{\rm id}(O_{\Phi}) X \right] \\
    &=& \mathrm{tr}\left[ O_{\Phi} X \right] \\
    &=& \mathrm{tr}\left[ O\, \Phi(X) \right].
\end{eqnarray}
Here, the second equality uses the self-adjointness of $\mathcal{M}_{\rm id}$ with respect to the Hilbert-Schmidt inner product, the third equality uses Eq.~\eqref{apeq:fixed_point}. This proves Theorem~\ref{thm:existence_of_no_cost_cut}.
\end{proof}

\subsection{Proof of Theorem~\ref{thm:zero_cost_timelike_cut_with_an_approximated_unitary}}\label{sec:proof-zero_cost_timelike_cut_with_an_approximated_unitary}

\newtheorem*{T2}{\bf{Theorem~\ref{thm:zero_cost_timelike_cut_with_an_approximated_unitary}} (Approximate rescaling-free wire cut)}
\begin{T2}
\textit{ 
Let $\Phi:\Linear{\mathbb{C}^{d_{\rm in}}} \to \Linear{\mathbb{C}^{d_{\rm out}}}$ be a CPTP map, and let $O \in \Herm{\mathbb{C}^{d_{\rm out}}}$ be a Hermitian operator.
Denote by $O_{\Phi}:=\Phi^{\dagger}(O)$ the corresponding effective observable, and let $\tilde{O}_{\Phi}$ be its Hermitian approximation satisfying $\|\tilde{O}_{\Phi} - O_{\Phi} \|_{\infty} \leq \eta$. 
Define the MP channel $\mathcal{M}_{\rm apr}:\Linear{\mathbb{C}^{d_{\rm in}}} \to \Linear{\mathbb{C}^{d_{\rm in}}}$ given by 
\begin{equation} 
    \mathcal{M}_{\rm apr}(X) := \sum_{j} \mathrm{tr} \left[ \tilde{V} \pure{j} \tilde{V}^{\dagger} X \right] \tilde{V} \pure{j} \tilde{V}^{\dagger}, 
\end{equation} 
where $\{\ket{j}\}$ denotes the computational basis, and $\tilde{V}$ is a unitary that diagonalizes $\tilde{O}_{\Phi}$ 
(i.e., $\tilde{O}_{\Phi}=\tilde{V}\tilde{D} \tilde{V}^{\dagger}$). 
Then for any  $X\in \Linear{\mathbb{C}^{d_{\rm in}}}$,
\begin{equation}
    \Bigl| \mathrm{tr}[O (\Phi \circ \mathcal{M}_{\rm apr}) (X)] - \mathrm{tr}[O\Phi(X)] \Bigr| \leq 2\eta \|X \|_{1}.
\end{equation}
}
\end{T2}

\begin{proof}[Proof of Theorem~\ref{thm:zero_cost_timelike_cut_with_an_approximated_unitary}]
Using the definition of $O_{\Phi}:=\Phi^{\dagger}(O)$ and the self-adjointness of $\mathcal{M}_{\rm apr}$, the difference between the expectation value with and without cut can be rewritten as
\begin{eqnarray}
     \mathrm{tr}[O (\Phi \circ \mathcal{M}_{\rm apr}) (X)] - \mathrm{tr}[O\Phi(X)]
     = \mathrm{tr}\left[ \left( \mathcal{M}_{\rm apr}(O_{\Phi}) - O_{\Phi} \right) X \right].
\end{eqnarray}
Here, Hölder's inequality, $| \mathrm{tr}[AB] | \leq \left\|A\right\|_{\infty}\left\|B \right\|_1$, yields
\begin{eqnarray}
    \Bigl| \mathrm{tr}\left[ \left( \mathcal{M}_{\rm apr}(O_{\Phi}) - O_{\Phi} \right) X \right] \Bigr|\leq \left\| \mathcal{M}_{\rm apr}(O_{\Phi}) - O_{\Phi} \right\|_{\infty} \left\| X\right\|_{1}.\label{eq:holder}
\end{eqnarray}
Since $\tilde{O}_{\Phi}$ is a fixed point of $\mathcal{M}_{\rm apr}$, i.e., $\mathcal{M}_{\rm apr}(\tilde{O}_{\Phi})=\tilde{O}_{\Phi}$, we obtain
\begin{eqnarray}
    \mathcal{M}_{\rm apr}(O_{\Phi}) - O_{\Phi} =  \mathcal{M}_{\rm apr}(O_{\Phi} - \tilde{O}_{\Phi}) - (O_{\Phi} - \tilde{O}_{\Phi}).
\end{eqnarray}
With the triangle inequality, the operator norm of this difference can be bounded as
\begin{eqnarray}
    \| \mathcal{M}_{\rm apr}(O_{\Phi}) - O_{\Phi}\|_{\infty}
    &\leq&  \| \mathcal{M}_{\rm apr}(O_{\Phi} - \tilde{O}_{\Phi})  \|_{\infty} + \| O_{\Phi} - \tilde{O}_{\Phi} \|_{\infty} \\
    &\leq& 2\| O_{\Phi} - \tilde{O}_{\Phi} \|_{\infty} \\
    &\leq& 2\eta,
\end{eqnarray}
where the second inequality follows from the contractivity of the MP channel, i.e.,  $\|\mathcal{M}_{\rm apr}(A)\|_{\infty} \leq \|A\|_{\infty},\forall A$, which is guaranteed by its construction as a unital pinching map.
Substituting this bound into Eq.~\eqref{eq:holder}, we arrive at
\begin{equation}
    \Bigl|  \mathrm{tr}[O (\Phi \circ \mathcal{M}_{\rm apr}) (X)] - \mathrm{tr}[O\Phi(X)] \Bigr| \leq 2\eta \|X \|_{1},
\end{equation}
which completes the proof.
\end{proof}

\subsection{Proof of Theorem~\ref{thm:zero_cost_proccessing}}\label{sec:proof_of_zero_cost_processing}

\newtheorem*{T3}{\bf{Theorem~\ref{thm:zero_cost_proccessing} (Post-processing from an approximate effective observable)}}
\begin{T3}
\textit{
Let $\Phi:\Linear{\mathbb{C}^{d_{\rm in}}} \to \Linear{\mathbb{C}^{d_{\rm out}}}$ be a CPTP map, and let $O \in \Herm{\mathbb{C}^{d_{\rm out}}}$ be a Hermitian operator.
Denote by $O_{\Phi}:=\Phi^{\dagger}(O)$ the corresponding effective observable, and let $\tilde{O}_{\Phi}$ be its Hermitian approximation satisfying $\|\tilde{O}_{\Phi} - O_{\Phi} \|_{\infty} \leq \eta$. 
Define the classical post-processing function
\begin{equation}\label{apeq:post-processing_functional}
    \mathcal{C}_{\rm apr}(X) := \sum_{j} \mathrm{tr} \left[ \tilde{V} \pure{j} \tilde{V}^{\dagger} X \right] \tilde{\lambda}_{j},
\end{equation}
where $\{\ket{j}\}$ denotes the computational basis, $\tilde{V}$ is a unitary that diagonalizes $\tilde{O}_{\Phi}$ (i.e., $\tilde{O}_{\Phi}=\tilde{V}\tilde{D} \tilde{V}^{\dagger}$), and $\tilde{D}:=\sum_j \tilde{\lambda}_j \pure{j}$. 
Then, for any  $X\in \Linear{\mathbb{C}^{d_{\rm in}}}$,
\begin{eqnarray}
    \Bigl|  \mathcal{C}_{\rm apr}(X) - \mathrm{tr}[O\Phi(X)] \Bigr| 
    &\leq& \eta \|X \|_{1},\\
    \max_{j} |\tilde{\lambda}_j| \leq \|O \|_{\infty}+\eta.\label{apeq:maximal_norm} 
\end{eqnarray}
}
\end{T3}

\begin{proof}[Proof of Theorem~\ref{thm:zero_cost_proccessing}]
By the definition of $\tilde{D}$, it is straightforward to rewrite Eq.~\eqref{apeq:post-processing_functional} as
\begin{equation}
    \mathcal{C}_{\rm apr}(X) = \mathrm{tr}[ \tilde{O}_{\Phi} X].
\end{equation}
Hence, using Hölder's inequality, 
\begin{eqnarray}
    \left| \mathcal{C}_{\rm apr}(X) - \mathrm{tr}[O \, \Phi(X)] \right| &\leq& \| \tilde{O}_{\Phi} - O_{\Phi} \|_{\infty} \| X \|_1
    \leq \eta \| X \|_1.
\end{eqnarray}
For Eq.~\eqref{apeq:maximal_norm}, since $O$ is Hermitian we have $-\|O\|_\infty I_{d_{\rm out}}\leq O \leq \|O\|_\infty I_{d_{\rm out}}$. Because the adjoint map of the CPTP map is unital and positive, 
\begin{equation}
    -\|O\|_\infty I_{d_{\rm in}}
    \le
    O_\Phi
    \le
    \|O\|_\infty I_{d_{\rm in}},
\end{equation}
and hence $\|O_\Phi\|_\infty\le \|O\|_\infty$. Thus,
\begin{eqnarray}
  \max_j |\tilde\lambda_j|
  &=&\|\tilde{O}_{\Phi} \|_{\infty} \\
  &\leq&  \|O_{\Phi} \|_{\infty} + \|\tilde{O}_{\Phi} - O_{\Phi} \|_{\infty} \\
   &\leq& \|O \|_{\infty} + \eta, 
\end{eqnarray}
which completes the proof.
\end{proof}

\section{Learning Heisenberg-evolved observables under unknown quantum channels}\label{apsec:learning_observable}

\subsection{Background}
\subsubsection{Useful tricks with the SWAP operator}
In the following, we present formulas for the SWAP operator, which will be used, especially, to introduce a protocol for learning local observables; see Sec.~\ref{sec:learning_protocol}. 
We denote by $\mathrm{SWAP}_n$ the operator that swaps two copies of an $n$-qubit system, defined as
\begin{equation}\label{eq:def_of_n_swap}
    \mathrm{SWAP}_{n} := \sum_{i,j\in\{0,1\}^n} \ket{i}\bra{j} \otimes \ket{j}\bra{i}, 
\end{equation}
where $\{\ket{i}\}_{i\in\{0,1\}^n}$ denotes the computational basis of $\mathbb{C}^d$ with $d=2^n$.
First, we introduce 
a well-known decomposition of $\mathrm{SWAP}_{n}$ with the tensor products of two $n$-qubit Pauli observables.
\begin{lemma}\label{lemma:swap_decomp_1}
The $2n$-qubit SWAP operator can be expressed as a uniform sum of tensor products of identical $n$-qubit Pauli operators:
\begin{equation}\label{eq:swap_decomp_by_n_paulis}
    \mathrm{SWAP}_{n}=\frac{1}{2^n}\sum_{P\in\{I,X,Y,Z\}^{\otimes n}} P\otimes P.
\end{equation}
\end{lemma}
\begin{proof}[Proof of Lemma~\ref{lemma:swap_decomp_1}]
It is straightforward to verify that the two-qubit SWAP operator can be decomposed with the single-qubit Pauli basis as $\mathrm{SWAP}_1=\frac{1}{2}\sum_{P\in\{I,\rm X,\rm Y,\rm Z\}}P\otimes P$. Since $\mathrm{SWAP}_n=\mathrm{SWAP}_1^{\otimes n}$, Eq.~\eqref{eq:swap_decomp_by_n_paulis} follows immediately, which completes the proof.
\end{proof}

Next, we introduce a useful identity known as the (partial) SWAP-trick.

\begin{lemma}[Partial-swap-trick]\label{lemma:partial_swap_trick}
For all operators $A,B\in \Linear{\mathbb{C}^d}$ with $d:=2^n$, we have
\begin{equation}
    \mathrm{tr}_{1}\left[ \mathrm{SWAP}_{n} (A\otimes B) \right] = AB.
\end{equation}
\end{lemma}
\begin{proof}[Proof of Lemma~\ref{lemma:partial_swap_trick}]
By definition of $\mathrm{SWAP}_n$ in Eq.~\eqref{eq:def_of_n_swap}, it follows that
\begin{eqnarray}
    \mathrm{tr}_{1}\left[ \mathrm{SWAP}_{n} (A_1\otimes B_1) \right] 
    &=& \sum_{i,j} \mathrm{tr}_{1}\left[ \ket{i}\bra{j} \otimes \ket{j}\bra{i} (A\otimes B) \right]
    =\sum_{i,j} (\braket{j|A|i} \ket{j}\bra{i}) B = AB.
\end{eqnarray}
\end{proof}

\subsubsection{Unitary and state designs}\label{sec:quantum-t-design}

Haar random unitaries and pure states underlie many primitives in various quantum information processing, such as randomized benchmarking~\cite{emerson2005scalable,PhysRevA.77.012307,PhysRevLett.106.180504,PhysRevLett.125.200501}, quantum tomography~\cite{Guta_2020,huang2020predicting,zambrano2025fast}, and quantum cryptography~\cite{ji2018pseudorandom,ananth2022cryptography,ma2025construct}.
However, the exact implementation of Haar random objects is generally infeasible, since typical unitaries require circuit depth exponential in the number of qubits~\cite{Nielsen_Chuang_2010,PhysRevA.69.062321,PRXQuantum.2.030316}. Fortunately, many applications do not require full Haar randomness: reproducing only the first few statistical moments of the Haar measure often suffices. 
This motivates the notion of \textit{unitary $t$-designs}~\cite{dankert2005efficient,PhysRevA.80.012304,gross2007evenly} and \textit{state $t$-designs}~\cite{4262758} (a.k.a., \textit{complex projective $t$-design}~\cite{HOGGAR1982233} or \textit{spherical $t$-designs}~\cite{delsarte1991spherical}), which formalize the ensembles that approximate the Haar distribution up to the $t$-th moment. We briefly recall their definitions below.

\begin{dfn}[Unitary $t$-design]
A probability distribution $\{p_i, U_i\}$ over $\Unit{\mathbb{C}^d}$ is a \textit{unitary $t$-design} if and only if, for all $A\in \Linear{(\mathbb{C}^d)^{\otimes t}}$,
\begin{equation}
    \sum_{i} p_i U_i^{\otimes t} A (U_i^{\dagger})^{\otimes t}= \int_{\Unit{\mathbb{C}^d}} d\mu_{\rm H}  U_i^{\otimes t} A (U_i^{\dagger})^{\otimes t},
\end{equation}
where $d\mu_{\rm H}$ denotes the normalized Haar measure on $\Unit{\mathbb{C}^d}$. 
\end{dfn}

\begin{dfn}[State $t$-design]\label{def:quantum-t-design}

A probability distribution $\{p_i, \ket{\psi_i}\}$ over pure states in $\mathbb{C}^d$ is a \textit{state $t$-design} (or \textit{complex projective $t$-design}) if and only if
\begin{equation}
    \sum_{i} p_i (\ket{\psi_i}\bra{\psi_i})^{\otimes t} = \int_{\phi} d\phi (\ket{\phi}\bra{\phi})^{\otimes t},
\end{equation}
where the integral is taken with respect to the normalized Haar measure on the unit sphere of $\mathbb{C}^d$.
\end{dfn}

Our randomized protocol discussed in Sec.~\ref{sec:main_leaerning_protocol} requires ensembles that form a unitary or state $2$-design. A representative example is given by the uniform distribution over the $n$-qubit Clifford group $\mathrm{Cl}(2^n)$~\cite{gottesman1997stabilizer}, which is known to form a unitary $3$-design~\cite{kueng2015qubit,webb2015clifford,zhu2017multiqubit} but not a $4$-design~\cite{zhu2016clifford}. The $n$-qubit Clifford circuits have been extensively studied, and importantly, they are considered to be particularly suitable for implementation for the following two reasons: (i) there exist efficient algorithms for uniform sampling from $\mathrm{Cl}(2^n)$~\cite{9605330,10.1063/1.4903507} and (ii) Clifford circuits can be implemented with at most $\mathcal{O}(n^2/\log{n})$ elementary gates (CNOT, Hadamard, and phase gates)~\cite{PhysRevA.70.052328}. 

Other prominent examples of state ensembles forming a $2$-design are provided by mutually unbiased bases (MUBs)~\cite{wootters1989optimal,PhysRevA.65.032320,seyfarth2019cyclic,9248636} and symmetric informationally complete positive operator-valued measures (SIC-POVMs)~\cite{renes2004symmetric}. Here, we focus on MUBs in particular, which are defined as follows.

\begin{dfn}[Mutually unbiased bases]
Let $\{\mathcal{B}_k\}_{k=1}^m$ be orthonormal bases of $\mathbb{C}^d$, where
$\mathcal{B}_k=\{\ket{\phi^{(k)}_i}\}_{i=1}^d$.
The bases are said to be mutually unbiased if and only if, for all $k\neq k'$ and all $i,j\in\{1,\dots,d\}$,
\begin{equation}
    \bigl|\braket{\phi^{(k)}_{i}|\phi^{(k')}_{j}}\bigr|^2 = \frac{1}{d}.
\end{equation}
\end{dfn}
A set of MUBs is called \emph{maximal} if $m=d+1$. In prime-power dimensions, including $d=2^n$, the maximal sets of $d+1$ MUBs exist~\cite{ivonovic1981geometrical,wootters1989optimal}, and the uniform distribution over $\bigcup_{k=1}^{d+1}\mathcal{B}_k$ constitutes a state $2$-design~\cite{1523643}. For qubits ($d=2^n$), there are efficient constructions of circuits that prepare MUB states or implement the corresponding unitaries,
e.g., methods via the stabilizer formalism~\cite{9248636,Crawford2021efficientquantum}.
A practical advantage of MUBs is the compactness of the ensemble: in dimension $d=2^n$, a maximal set consists of $d+1=2^n+1$ bases, or $d(d+1)=2^n(2^n+1)$ states in total. This is much smaller than the number of elements in the $n$-qubit Clifford group, whose size is $2^{n^2+2n}\prod_{j=1}^{n}(4^{j}-1)$. 
This compactness helps reduce classical overhead, especially in scenarios where many random samples are required (e.g., when the number of required random circuits is on the order of $2^n$ or larger)~\cite{harada2025doubly,PhysRevApplied.21.064001,PhysRevA.109.062406}.
For example, if $n$ is relatively small, one can make the lookup tables for the associated state-preparation circuits and reuse them during the sampling stage~\cite{harada2025doubly}.

\subsubsection{Symmetric subspace}\label{sec:symmetric_subspace}

To derive our randomized protocols in Sec.~\ref{sec:main_leaerning_protocol}, we need to employ moment calculations with respect to Haar random states. To provide the necessary tools for such calculations (Appendices~\ref{apsec:constructino} and \ref{appsec:performance_guarantee}), we here present a brief review of the symmetric subspace. For further details, see Refs.~\cite{Mele2024introductiontohaar,harrow2013church}.

\begin{dfn}[Permutation operator]
    Let $S_{k}$ denote the symmetric group on $k$ elements. For any $\pi\in S_k$, define the permutation operator $P_d(\pi)$ on the $k$-fold tensor product $(\mathbb{C}^{d})^{\otimes k}$ by
    \begin{equation}
        P_{d}(\pi) \ket{\psi_{1}} \otimes \cdots\otimes\ket{\psi_k} = \ket{\psi_{\pi^{-1}(1)}}\otimes \cdots \otimes\ket{\psi_{\pi^{-1}(k)}},
    \end{equation}
    for all $\ket{\psi_1},\ldots,\ket{\psi_k} \in \mathbb{C}^d$. Equivalently, it can be expressed as
    \begin{equation}
        P_{d}(\pi) = \sum_{i_1,\ldots,i_k\in [d]} \ket{i_{\pi^{-1}(1)},...,i_{{\pi}^{-1}(k)}} \bra{i_1,\ldots,i_k}.
    \end{equation}
\end{dfn}

\begin{dfn}[Symmetric subspace]
    The symmetric subspace of $(\mathbb{C}^d)^{\otimes k}$, denoted by $\vee^{k}\mathbb{C}^d$, is defined as
    \begin{equation}
        \vee^{k}\mathbb{C}^d := \left\{ \ket{\psi}\in (\mathbb{C}^d)^{\otimes k} : P_{d}(\pi) \ket{\psi}=\ket{\psi} ~~\text{for all}~~\pi\in S_k \right\}.
    \end{equation}
\end{dfn}

We now present two fundamental facts about the symmetric subspace: the orthogonal projector and its dimension. We omit the proofs of these results here, but the detailed arguments can be found in Refs.~\cite{Mele2024introductiontohaar,harrow2013church}.

\begin{fact}[Orthogonal projector onto $\vee^{k}\mathbb{C}^d$]\label{fact:proj_of_sym}
Define
\begin{equation}
P_{\rm sym}^{d,k}:=\frac{1}{k!}\sum_{\pi\in S_{k}} P_{d}(\pi).
\end{equation}
Then, $P_{\rm sym}^{d,k}$ is the orthogonal projector onto $\vee^{k}\mathbb{C}^d$.
\end{fact}
\begin{fact}[Dimension of $\vee^{k}\mathbb{C}^d$]\label{fact:dim_of_sym}
\begin{equation}
\dim{(\vee^{k}\mathbb{C}^d)}=\mathrm{tr}(P_{\rm sym}^{d,k})=\binom{d+k-1}{k}.
\end{equation}
\end{fact}

We conclude this subsection by introducing a useful identity that will be employed in the subsequent analysis. Define the tensor product state 
\begin{equation}
    \rho_{\text{k-sym}} := \int_{\phi} d\phi(\ket{\phi}\bra{\phi})^{\otimes k},
\end{equation}
where the integral is taken over the pure states $\ket{\phi}\in \mathbb{C}^d$ from the Haar measure. By unitary invariance, it follows that $\rho_{\text{k-sym}}$ commutes with $U^{\otimes k}$ for any unitary $U$. Moreover, it is known that $\vee^{k}\mathbb{C}^d$ affords the irreducible representation arising from the action $U \mapsto U^{\otimes k}$. Therefore, by Schur’s lemma, $\rho_{\text{k-sym}}$ is proportional to the orthogonal projection onto the symmetric subspace. Combining this with Facts~\ref{fact:proj_of_sym} and \ref{fact:dim_of_sym}, we obtain the following characterization of tensor powers of Haar-random pure states.

\begin{lemma}
Let $d\phi$ denote the normalized Haar measure on pure states of $\mathbb{C}^d$. For any $k\in \mathbb{N}$,
\begin{eqnarray}
    \int d\phi (\pure{\phi})^{\otimes k} 
    &=& \frac{P_{\rm sym}^{d,k}}{\mathrm{tr}(P_{\rm sym}^{d,k})}\\
    &=&  \frac{1}{d(d+1)\cdots(d+k-1)} \sum_{\pi\in \mathcal{S}_{k}} P_{d}(\pi).
\end{eqnarray}
\end{lemma}

This lemma is useful for deriving closed-form expressions for randomized protocol schemes and for evaluating their efficiency. In our calculations in Appendices~\ref{apsec:constructino} and \ref{appsec:performance_guarantee}, we frequently use the cases $k=1$ and $k=2$, so we present the corresponding formulas explicitly below.

\begin{corollary}\label{cor:2_design}
Let $d\phi$ denote the normalized Haar measure on pure states of $\mathbb{C}^d$ with $d:=2^n$. Then,
\begin{eqnarray}
    \int d\phi (\pure{\phi}) &=& \frac{I_d}{d},\\
    \int d\phi(\pure{\phi})^{\otimes 2} &=& \frac{1}{d(d+1)}\left( I_{d^2} +  \mathrm{SWAP}_{n}\right).
\end{eqnarray}
\end{corollary}

\subsection{Construction of the unbiased estimators}\label{apsec:constructino}

In this section, we explain how we derive the unbiased estimators introduced in Sec.~\ref{sec:main_leaerning_protocol}. 
To obtain $\hat{O}_{\Phi}$ satisfying $\| \hat{O}_{\Phi} - O_{\Phi}\|_{\infty} \leq \eta$, we need to collect informative data from the measurements of the circuit composed of an unknown CPTP map $\Phi: \Linear{\mathbb{C}^{d_{\rm in}}} \to \Linear{\mathbb{C}^{d_{\rm out}}}$ and a known Hermitian operator $O\in \Herm{\mathbb{C}^{d_{\rm out}}}$, probed with an appropriately chosen ensemble of input states $\mathcal{S}:=\{(p_i,\ket{\psi_i})\}$. 
We thus consider the following randomized protocol, which performs measurements on random inputs and stores the corresponding data:
\begin{description}
    \item[\textbf{Step 1 (Random sampling)}] Sample a state $\pure{\psi_i}$ according to the probability distribution 
    $p_i$.
    \item[\textbf{Step 2 (Measurement)}] Prepare the corresponding quantum state $\pure{\psi_i}$, evolve it under the CPTP map $\Phi$, and measure with the observable $O$ (i.e., apply a unitary rotation $W^{\dagger}$ immediately before measurement in the computational basis $\{\ket{j}\}$).
    Denote the measurement outcome by $j$.
    \item[\textbf{Step 3 (Data storage)}] Store the weighted state $\nu_{j} \pure{\psi_i}$ in classical memory.
\end{description}
For the above process (\textbf{Step 1-3}), we define the mapping $\mathcal{M}: \Herm{\mathbb{C}^{d_{\rm in}}} \rightarrow \Herm{\mathbb{C}^{d_{\rm in}}}$ from $O_{\Phi}$ to its classical description $\mathbb{E}_{i,j}[\nu_{j} \pure{\psi_i}]$, where the expectation is taken over the choice of the input state and
the measurement outcome. Explicitly, 
\begin{eqnarray}\label{eq:randomized-map}
    \mathcal{M}(O_{\Phi}) 
    &=& \sum_i p_i \sum_j \mathrm{tr}\left[ \pure{j} W^{\dagger} \Phi(\pure{\psi_i}) W \right] \nu_j \pure{\psi_i}\notag\\
    &=&\sum_i p_i \mathrm{tr}\left[ O_{\Phi} \pure{\psi_i} \right] \, \pure{\psi_i}. 
\end{eqnarray}
By choosing the ensemble $\mathcal{S}$ appropriately, the map $\mathcal{M}$ can be made to carry substantial information about $O_{\Phi}$. 
Representative examples of such a choice of $\mathcal{S}$ are an ensemble forming a 2-design or the tensor product of single-qubit stabilizer states. 
In Appendices~\ref{apsec:const_2-design} and \ref{apsec:const_stab}, we provide detailed discussions of these two cases, respectively. Moreover, by introducing an additional weighting in Step 3, we can also formulate a protocol based on Pauli operators. 
This is discussed in detail in Appendix~\ref{apsec:const_pau}. 

\subsubsection{2-design ensemble}\label{apsec:const_2-design}

Let $\mathcal{S}_{\text{2-dgn}}=\{(p_i,\ket{\psi_i})\}$ be an ensemble forming a state 2-design. By definition, 
\begin{equation}\label{eq:2-design}
    \sum_{i} p_i (\pure{\psi_i})^{\otimes 2} = \int_{\phi} (\ket{\phi}\bra{\phi})^{\otimes 2} d\phi,
\end{equation}
where the integral is over pure states with respect to the Haar measure; see Definition~\ref{def:quantum-t-design}. 
A standard result states that the Haar integral in Eq.~\eqref{eq:2-design} yields an operator $P_{\rm sym}^{d_{\rm in},2}$ proportional to the projector onto the symmetric subspace (see Corollary~\ref{cor:2_design}), namely,
\begin{equation}\label{eq:2-design-2}
    \sum_{i} p_i (\ket{\psi_i}\bra{\psi_i})^{\otimes 2} =\frac{P_{\rm sym}^{d_{\rm in},2}}{\binom{d_{\rm in}+1}{2}}= \frac{1}{d_{\rm in}(d_{\rm in}+1)}\left( \mathrm{SWAP}_n + I_{d_{\rm in}^2} \right),
\end{equation}
where $I_{d_{\rm in}^2}$ is the identity on the two-copy space, and $\mathrm{SWAP}_n$ is the swap operator. Substituting Eq.~\eqref{eq:2-design-2} in Eq.~\eqref{eq:randomized-map}, the map $\mathcal{M}$ acting on an operator $O_{\Phi}$ can be rewritten as
\begin{eqnarray}
    \mathcal{M}(O_{\Phi}) 
    &=& \sum_i p_i \mathrm{tr}_{1}\left[ (O_{\Phi}\otimes I_{d_{\rm in}}) (\pure{\psi_i})^{\otimes 2} \right]\\
    &=& \frac{1}{d_{\rm in}(d_{\rm in}+1)} \mathrm{tr}_{1}\left[ (O_{\Phi}\otimes I_{d_{\rm in}}) (\mathrm{SWAP}_n+I_{d_{\rm in}^2}) \right]\\
    &=& \frac{1}{d_{\rm in}(d_{\rm in}+1)}\left( O_{\Phi} + \mathrm{tr}\left(O_{\Phi}\right) I_{d_{\rm in}} \right).\label{eq:randomized-map-2design}
\end{eqnarray}
This expression shows that the randomized map $\mathcal{M}$ carries substantial information about the target operator $O_{\Phi}$, in the form of both the trace and operator content. From Eq.~\eqref{eq:randomized-map-2design}, the information about $O_{\Phi}$ can be recovered by applying the classical post-processing to the output 
\begin{equation}\label{eq:modification_2_design}
    \mathcal{M}(O_{\Phi}) \rightarrow d_{\rm in}(d_{\rm in}+1) \mathcal{M}(O_{\Phi})-\mathrm{tr}(O_{\Phi}) I_{d_{\rm in}}.
\end{equation}
Here, noting that the following equality holds
\begin{equation}
    \mathrm{tr}(O_{\Phi}) I_{d_{\rm in}} = \sum_i p_i \sum_j \nu_j \mathrm{tr}\left[ \pure{j} W^{\dagger}\Phi(\pure{\psi_i}) W\right] \times d_{\rm in}I_{d_{\rm in}},
\end{equation}
the above equality together with Eq.~\eqref{eq:randomized-map} readily implies that
\begin{equation}
    O_{\Phi} = \sum_{i} p_i \sum_j \mathrm{tr}\left[ \pure{j} W^{\dagger}\Phi(\pure{\psi_i}) W\right]
    \nu_j \left( d_{\rm in}(d_{\rm in}+1) \pure{\psi_i} - d_{\rm in}I_{d_{\rm in}}  \right)\label{eq:2_design_unbias}.
\end{equation}
Thus, we can reconstruct the classical description of $O_{\Phi}$ in practice, by modifying \textbf{Step 3} as follows. 

For $k=1,...,N$, we sample an input $\ket{\psi_{i}}$ with the probability $p_i$, evolve it under $\Phi$, perform the measurement in the basis $\{W\ket{j}\}$, obtain the outcome $j$ (\textbf{Step 1} and \textbf{Step 2}), and store the single-shot estimator as 
\begin{equation}
    \hat{\omega}_{\text{2-dgn}}^{(k)}:=\nu_{j_k}( d_{\rm in}(d_{\rm in}+1)\pure{\psi_{i_k}} - d_{\rm in}I_{d_{\rm in}}).
\end{equation}
Then, by averaging the data as
\begin{eqnarray}
    \hat{O}_{\Phi,\text{2-dgn}} &:=& \frac{1}{N}\sum_{k=1}^{N} \hat{\omega}_{\text{2-dgn}}^{(k)} \\
    &=& \frac{1}{N}\sum_{k=1}^{N} \nu_{j_k}( d_{\rm in}(d_{\rm in}+1)\pure{\psi_{i_k}} - d_{\rm in}I_{d_{\rm in}}),
\end{eqnarray}
we obtain an unbiased estimator of $O_{\Phi}$, i.e., $\mathbb{E}[\hat{O}_{\Phi,\text{2-dgn}}]=O_{\Phi}$. 

Finally, we provide some supplementary explanations regarding the above protocol. 
First, in the above discussion, the estimator $\hat{O}_{\Phi}$ was obtained by modifying 
the map $\mathcal{M}$, as demonstrated in Eq.~\eqref{eq:modification_2_design}, 
so that only the terms of $O_{\Phi}$ appearing in the expansion 
\eqref{eq:randomized-map-2design} of $\mathcal{M}(O_{\Phi})$ are extracted. 
The resulting estimator corresponds to a linear inversion estimator for $O_{\Phi}$, 
where the linear inversion technique is sometimes employed in quantum state/measurement tomography~\cite{Guta_2020,zambrano2025fast} and classical shadow tomography~\cite{huang2020predicting}. 

\subsubsection{Tensor-product single qubit stabilizer states}\label{apsec:const_stab}

We focus on the case where the input ensemble $\mathcal{S}_{\rm stab}$ is chosen as the set of all $n$-fold tensor products of the single-qubit stabilizer states 
\begin{equation}
\mathcal{Q}_{\rm stab}:=\left\{\ket{0},\ket{1},\ket{+},\ket{-},\ket{+i},\ket{-i}\right\},
\end{equation}
sampled uniformly with probability $p_i=1/6^n$. Formally, 
\begin{equation}
    \mathcal{S}_{\text{stab}}:=\left\{\left(\frac{1}{6^n}, \bigotimes_{s=1}^{n}\ket{u^s}\right): \ket{u^s} \in \mathcal{Q}_{\rm stab}\right\}.
\end{equation}

Before turning to general $n$, it is convenient to verify the single-qubit case ($n=1$). Let 
\begin{equation}
    \mathcal{Q}_x=\{\ket{\pm}\},~~\mathcal{Q}_y=\{\ket{\pm i}\},~~\mathcal{Q}_z=\{\ket{0},\ket{1}\},
\end{equation}
the eigenstate sets of the Pauli operators $\rm X, Y, Z$, respectively, and define $\mathcal{Q}_{\text{stab}}:=\mathcal{Q}_x\cup \mathcal{Q}_y \cup \mathcal{Q}_z$. The three bases $\mathcal{Q}_x, \mathcal{Q}_y, \mathcal{Q}_z$ form a maximal set of mutually unbiased bases (MUBs) in dimension $d_{\rm in}=2$, and the uniform ensemble over $\mathcal{Q}_{\text{stab}}$ is shown to be a complex-projective 2-design~\cite{1523643}. Equivalently,
\begin{equation}\label{eq:1_mubs_forms_2_design}
    \frac{1}{6} \sum_{\ket{v}\in \mathcal{Q}_{\text{stab}}} (\pure{u})^{\otimes 2} 
    = \frac{1}{6} \left( \mathrm{SWAP}+I \right).
\end{equation}
Using Eq.~\eqref{eq:1_mubs_forms_2_design} (equivalently, inserting $d_{\rm in}=2$ into Eq.~\eqref{eq:2_design_unbias}), we obtain
\begin{equation}\label{eq:1_mubs_unbias}
    \sum_{\ket{u}\in \mathcal{Q}_{\text{stab}}} \frac{1}{6} \mathrm{tr}\left[ A\pure{u} \right] \left( 6\pure{u}-2I \right) = A, \quad \forall A\in \mathsf{H}(\mathbb{C}^2).
\end{equation}
We now extend the discussion above to the multiple-qubit case. Taking the $n$-fold tensor products of both sides of Eq.~\eqref{eq:1_mubs_unbias}, for any product operator $A=\bigotimes_{s=1}^{n} A^s$ with $A^s\in \Herm{\mathbb{C}^2}$, we have
\begin{eqnarray}
    \bigotimes_{s=1}^{n} \sum_{\substack{\ket{u^s}\in \mathcal{Q}_{\text{stab}}}} \frac{1}{6} \mathrm{tr}\left[ A^s\pure{u^s} \right] \left( 6\pure{u^s}-2I \right) 
    = \sum_{\substack{\ket{u^1},...,\ket{u^n}\\ \in \mathcal{Q}_{\text{stab}}}} \frac{1}{6^n}  \prod_{s=1}^{n} \mathrm{tr}\left[ A^s \pure{u^s} \right]  \bigotimes_{s=1}^{n} \left( 6\pure{u^s} - 2I \right).
\end{eqnarray}
Using $\prod_{s=1}^{n}\mathrm{tr}\left[ A^s \pure{u^s} \right]=\mathrm{tr}\left[ A\pure{u} \right]$ for $\ket{u}:=\ket{u^1u^2\cdots u^n}$, we can rewrite this as
\begin{equation}\label{eq:product_mubs_unbias_A}
    A=\sum_{\ket{u}\in \mathcal{Q}_{\text{stab}}^{\otimes n}} \frac{1}{6^n} \mathrm{tr}\left[ A \pure{u} \right] \bigotimes_{s=1}^{n} \left( 6\pure{u^s} - 2I \right).
\end{equation}
By linearity, Eq.~\eqref{eq:product_mubs_unbias_A} also holds for any $A\in \Herm{\mathbb{C}^{2^n}}$. Hence, for the uniform product-stabilizer ensemble $\mathcal{S}_{\text{stab}}$, we can express 
\begin{equation}
    O_{\Phi} = \sum_{\ket{u}\in \mathcal{Q}_{\text{stab}}^{\otimes n}} \frac{1}{6^n} \sum_{j} \mathrm{tr}\left[ \pure{j} W^{\dagger} \Phi(\pure{u}) W \right]
    \nu_j \bigotimes_{s=1}^{n} \left( 6\pure{u^s} - 2I \right).\label{eq:n_mubs_unbias}
\end{equation}
Thus, to obtain the classical description of $O_{\Phi}$ in this case, we repeat the following randomized protocol independently for $k=1,...,N$. 
At the $k$-th round, we sample an input $\ket{u_k}\in \mathcal{Q}_{\text{stab}}^{\otimes n}$ uniformly, apply $\Phi$, measure in the basis $\{W\ket{j}\}$, obtain outcome $j_k$, and record the single-shot estimator
\begin{equation}
    \hat{\omega}_{\text{stab}}^{(k)}:=\nu_{j_k}\bigotimes_{s=1}^{n} (6\pure{u_k^{s}}-2I).
\end{equation}
Finally, the empirical estimator is obtained as the sample average
\begin{eqnarray}
    \hat{O}_{\Phi,\text{stab}} &:=& \frac{1}{N}\sum_{k=1}^{N} \hat{\omega}_{\text{stab}}^{(k)}\\
    &=& \frac{1}{N} \sum_{k=1}^{N}\nu_{j_k}\bigotimes_{s=1}^{n} (6\pure{u_k^{s}}-2I).
\end{eqnarray}
By construction and Eq.~\eqref{eq:n_mubs_unbias}, this is an unbiased estimator of $O_{\Phi}$.

\subsubsection{$n$-qubit Pauli operator}\label{apsec:const_pau}

To make the discussion self-contained, we recall the definitions of the single-qubit Pauli operators and their eigenvalue decompositions in the main text:
\begin{equation}
    P_i = \sum_{e\in[2]} c_{i,e}\,\pure{v_{i,e}}, \qquad i \in\{1,2,3,4\}=[4],
\end{equation}
where $P_1$, $P_2$, $P_3$, and $P_4$ correspond to the identity and the Pauli operators $\rm X$, $\rm Y$, and $\rm Z$, respectively.
The eigenstates and their associated eigenvalues are summarized as
\begin{align}
\begin{alignedat}{10}
\ket{v_{1,1}} &= \ket{0}, &\quad c_{1,1} &= +1, &\quad \ket{v_{1,2}} &= \ket{1}, &\quad c_{1,2} &= +1, \\
\ket{v_{2,1}} &= \ket{+}, &\quad c_{2,1} &= +1, &\quad \ket{v_{2,2}} &= \ket{-}, &\quad c_{2,2} &= -1, \\
\ket{v_{3,1}} &= \ket{+i}, &\quad c_{3,1} &= +1, &\quad \ket{v_{3,2}} &= \ket{-i}, &\quad c_{3,2} &= -1, \\
\ket{v_{4,1}} &= \ket{0}, &\quad c_{4,1} &= +1, &\quad \ket{v_{4,2}} &= \ket{1}, &\quad c_{4,2} &= -1.
\end{alignedat}
\end{align}

Building upon these single-qubit definitions, we can naturally extend them to the $n$-qubit setting by defining the Pauli ensemble
\begin{equation}
        \mathcal{S}_{\rm Pauli} := \left\{ \left( \frac{1}{4^n},\, 
    P_{\bm{i}} := \bigotimes_{s=1}^{n} P_{i_s} \right) : \bm{i}\in [4]^n
    \right\},
\end{equation}
where $\bm{i} = (i_1,\dots,i_n)$.  
Each operator in this ensemble can be expressed via its eigenvalue decomposition as
\begin{equation}
\label{appeq:eigenvalue_decompositions_of_n_pauli_app}
    P_{\bm{i}}
    = \sum_{\bm{e}\in[2]^n} c_{\bm{i},\bm{e}}\,
    \pure{v_{\bm{i},\bm{e}}},~~~~~~~~~
    c_{\bm{i},\bm{e}}
    := \prod_{s=1}^{n} c_{i_s,e_s}, \qquad
    \ket{v_{\bm{i},\bm{e}}} := \bigotimes_{s=1}^{n} \ket{v_{i_s,e_s}}.
\end{equation}

With these definitions, we can describe the sampling process as follows:  
for each trial, a Pauli operator $P_{\bm{i}}$ is drawn uniformly at random from $\mathcal{S}_{\rm Pauli}$, and one of its eigenstates $\ket{v_{\bm i,\bm{e}}}$ is selected uniformly to serve as the probe state. 
\begin{enumerate}
    \item \textbf{Random sampling}: Sample a Pauli operator $P_{\bm{i}_k}$ uniformly at random, and sample one of its eigenstates $\ket{v_{\bm{i}_k,\bm{e}_k}}$ uniformly.
    \item \textbf{Measurement}: Prepare the corresponding quantum state $\ket{v_{\bm{i},\bm{e}}}$, evolve it under the map $\Phi$, and measure with the observable $O$. Denote the measurement outcome by $j$.
    \item \textbf{Data storage}: Store the weighted Pauli operator $\nu_{j} c_{\bm i,\bm{e}} P_{\bm{i}}$ in classical memory.
\end{enumerate}
As in the previous discussions, we show that the output of this procedure contains sufficient information of $O_{\Phi}$ in expectation. To this end, we define the mapping $\mathcal{M}:\Herm{\mathbb{C}^{d_{\rm in}}}\rightarrow\Herm{\mathbb{C}^{d_{\rm in}}}$, which maps the observable $O_{\Phi}$ to its classical description $\mathbb{E}_{i,\bm j,\bm e}[\nu_{j} c_{\bm{i},\bm{e}} P_{\bm{i}}]$, obtained by averaging over the labels $(\bm{i},\bm{e})$ and measurement outcomes $j$. Explicitly,
\begin{eqnarray}
    \mathcal{M}(O_{\Phi}) 
    &=& \sum_{\bm{i}\in[4]^n,\bm{e}\in[2]^n} \frac{1}{8^n} \sum_{j} \mathrm{tr}\left[ \pure{j}W^{\dagger}\Phi(\pure{v_{\bm{i},\bm{e}}})W \right]\, \nu_{j} c_{\bm{i},\bm{e}} P_{\bm{i}}\\
    &=& \sum_{\bm{i}\in[4]^n} \frac{1}{8^n} \mathrm{tr}\left[ O_{\Phi} P_{\bm{i}} \right]\, P_{\bm{i}} ,
\end{eqnarray}
where we have used the eigenvalue decomposition of $O_{\Phi}$ and $P_{\bm{i}}$ in Eqs.~\eqref{eq:svd_of_O} and \eqref{appeq:eigenvalue_decompositions_of_n_pauli_app}, respectively. Here, the decomposition formula of swap operators in Lemma~\ref{lemma:swap_decomp_1}
enables a further simplification of the above equality as
\begin{eqnarray}
    \mathcal{M}(O_{\Phi}) 
    = \frac{1}{4^n} O_{\Phi} \left(= \frac{1}{4^n}\mathrm{tr}_{1}\left[ (O_{\Phi}\otimes I) \left(\frac{1}{2^n}\sum_{\bm{i}\in[4]^n} P_{\bm{i}}\otimes P_{\bm i} \right) \right] \right) .\label{eq:pauli_unbias}
\end{eqnarray}
Hence, to obtain the classical description of $O_{\Phi}$, we repeat the following randomized protocol independently for $k=1,...,N$. At the $k$-th round, we uniformly sample the labels of a Pauli operator and its eigenstate $(\bm{i}_{k},\bm{e}_{k})$, prepare the input state $\ket{v_{\bm{i}_k,\bm{e}_k}}$, apply $\Phi$, rotate the state by $W$, measure in the computational basis $\{\ket{j}\}$, obtain outcome $j_k$, and record the single-shot estimator as
\begin{equation}
\hat{\omega}_{\text{pauli}}^{(k)}:=4^n \nu_{j_k} c_{\bm{i}_k,\bm{e}_k} P_{\bm{i}_k}
\end{equation}
in classical memory. Finally, the empirical estimator is obtained as the sample average
\begin{eqnarray}
    \hat{O}_{\Phi,\text{pauli}}
    &:=& \frac{1}{N}\sum_{k=1}^{N} 
    \hat{\omega}_{\text{pauli}}^{(k)} \\
    &=& \frac{1}{N} \sum_{k=1}^{N} 4^n \nu_{j_k} c_{\bm{i}_k,\bm{e}_k} P_{\bm{i}_k}.
\end{eqnarray}
By construction and Eq.~\eqref{eq:pauli_unbias}, this is an unbiased estimator of $O_{\Phi}$.

\subsection{Rigorous performance guarantees (Proof of Theorem~\ref{thm:error_bounds_for_each_x})}\label{appsec:performance_guarantee}

In the main text, we provide the following performance guarantee for the estimator $\hat{O}_{\Phi,x}$ with each $x\in\{\text{2-dgn, stab, pauli}\}$.

\newtheorem*{T4}{\bf{Theorem~\ref{thm:error_bounds_for_each_x} (Performance guarantee for $\hat{O}_{\Phi,x}$)}}
\begin{T4}
\textit{
Let $\Phi:\Linear{\mathbb{C}^{d_{\rm in}}} \to \Linear{\mathbb{C}^{d_{\rm out}}}$ be an unknown CPTP map, and let $O \in \Herm{\mathbb{C}^{d_{\rm out}}}$ be a known Hermitian operator. Define $O_\Phi:=\Phi^\dagger(O)$.
Let $\hat{O}_{\Phi,x}$ be the estimator defined in
Sec.~\ref{sec:main_leaerning_protocol} for
$x\in\{\text{2-dgn},\text{stab},\text{pauli}\}$.
Then, for each $x$, it suffices to take $N$ at least as follows to guarantee $\|\hat{O}_{\Phi,x}-O_{\Phi}\|_{\infty}\le \eta$ with probability at least $1-\delta$.
\begin{eqnarray}
    &\text{[2-dgn]}~&\frac{2 \max\{1,\|O\|^2_{\infty}\}(d_{\rm in}^2+1) (d_{\rm in}+1+\eta/3) } {\eta^2} \ln\left(\frac{2d_{\rm in}}{\delta}\right),\notag\\
    &\text{[stab]}~& \frac{2 \max\{1,\|O\|_{\infty}^2\}(d_{\rm in}^{3.33}+1+(\eta/3)(d_{\rm in}^2+1)) } {\eta^2} \ln{\left(\frac{2d_{\rm in}}{\delta}\right)},\notag\\
    &\text{[pauli]}~& \frac{2 \max\{1,\|O\|_{\infty}^2\}(d_{\rm in}^4+1+(\eta/3)(d_{\rm in}^2+1)) } {\eta^2} \ln{\left(\frac{2d_{\rm in}}{\delta}\right)}.\notag
\end{eqnarray}}
\end{T4}

\vspace{4pt}

In the following subsections, we provide the proof of Theorem~\ref{thm:error_bounds_for_each_x} for each $x\in\{\text{2-dgn, stab, pauli}\}$. Now, we outline the basic proof strategy below. 
For each choice of $x$, the estimator $\hat{O}_{\Phi,x}$ of $O_{\Phi}$, generated by the algorithm described in Sec.~\ref{sec:learning_protocol} is obtained as the average of independently sampled random operators $\hat{\omega}_x^{(k)}$; see Sec.~\ref{sec:main_leaerning_protocol} or Appendix~\ref{apsec:constructino}. Since $\hat{\omega}_x^{(k)}$ is an unbiased estimator of $O_{\Phi}$, for each $k\in[N]$,
\begin{equation}
     X^{(k)} = \frac{1}{N}\bigl(\hat{\omega}_x^{(k)} - O_{\Phi}\bigr),
\end{equation}
is a centered random matrix satisfying $\mathbb{E}[X^{(k)}]=0$. Moreover, each $X^{(k)}$ is independent, random, self-adjoint, and bounded as
\begin{eqnarray}
    \left\|X^{(k)}\right\|_{\infty} 
    &=& \frac{1}{N}\,\left\|\hat{\omega}_x^{(k)}-O_{\Phi}\right\|_{\infty} \\
    &\leq& \frac{1}{N}\left\{\|\hat{\omega}_x^{(k)}\|_{\infty}+\,\|O_{\Phi}\|_{\infty}\right\}.
\end{eqnarray}
Then, for such independent random matrices, to analyze the deviation $\hat{O}_{\Phi,x}-O_{\Phi}=\sum_{k=1}^N X^{(k)}$, we employ the matrix extension of Bernstein's inequality developed independently in Refs.~\cite{oliveira2009concentration,tropp2012user}. 
\begin{lemma}[Matrix Bernstein inequality: Hermitian case~\cite{tropp2015introduction}]\label{lemma:matrix_bernstein_inequality}
Consider a finite sequence $X^{(1)},\ldots,X^{(N)}$ of independent, random, Hermitian matrices with dimension $d$. Assume that each random matrix satisfies $\mathbb{E}[X^{(k)}] = 0$ and $\|X^{(k)}\|_{\infty}\leq \alpha$ almost surely for all $k$. Then, for all $t\geq0$, 
\begin{equation}\label{eq:matrix_benstein}
    \mathrm{Pr}\left[ \left\| \sum_{k=1}^{N} X^{(k)} 
     \right\|_{\infty} \geq t \right] \leq 2d \exp\left( \frac{-t^2/2}{\sigma^2+\alpha t/3} \right)
\end{equation}
where $\sigma^2=\| \sum_{k=1}^{N} \mathbb{E}[(X^{(k)})^2] \|_{\infty}$.
\end{lemma}

Hence, we complete the proof of Theorem~\ref{thm:error_bounds_for_each_x} by evaluating $\alpha$ and $\sigma^2$ for each choice of $x$.

\subsubsection{2-design ensemble}\label{sec:prt_global_2-design_complexity}

\begin{proof}
The algorithm described in Sec.~\ref{sec:main_leaerning_protocol} produces an estimator of observable $O_{\Phi}$ as
\begin{eqnarray}
    \hat{O}_{\Phi,\text{2-dgn}} 
    &:=& \frac{1}{N}\sum_{k=1}^{N} \hat{\omega}_{\text{2-dgn}}^{(k)}\\
    &=& \frac{1}{N}\sum_{k=1}^{N} \nu_{j_k} \left( d_{\rm in}(d_{\rm in}+1)\pure{\psi_{i_k}} - d_{\rm in}I_{d_{\rm in}} \right),
\end{eqnarray}
where each $\hat{\omega}_{\text{2-dgn}}^{(k)}$ is an i.i.d. copy of a random matrix $\omega_{\text{2-dgn}}\in \Herm{\mathbb{C}^{d_{\rm in}}}$ taking the value
\begin{equation}
   \nu_{j} \left( d_{\rm in}(d_{\rm in}+1)\pure{\psi_{i}} - d_{\rm in}I_{d_{\rm in}} \right)
\end{equation}
with probability $p_{i} \bra{j} W^{\dagger}\Phi(\pure{\psi_{i}})W\ket{j}$ for all $k\in\{1,...,N\}$.
Since the estimator is unbiased, i.e., $\mathbb{E}[\hat{O}_{\Phi}]=O_{\Phi}$, the deviation $\hat{O}_{\Phi,\text{2-dgn}}-O_{\Phi}$ can be written as a sum of $N$ independent, centered random matrices of the form $X_{\text{2-dgn}}^{(k)}:=( \omega_{\text{2-dgn}}^{(k)}  - O_{\Phi} )/N$.
These matrices obey
\begin{eqnarray}
    \left\| X_{\text{2-dgn}}^{(k)} \right\|_{\infty}
    &\leq& \frac{1}{N} \left\{  \left\| \omega_{\text{2-dgn}}^{(k)} \right\|_{\infty} + \left\| O_{\Phi} \right\|_{\infty} \right\}\\
    &\leq& \frac{1}{N} \left\{ \max_{\ket{\psi_i},\nu_j} |\nu_{j}| \left\|   d_{\rm in}(d_{\rm in}+1) \pure{\psi_{i}} -d_{\rm in}I_{d_{\rm in}}  \right\|_{\infty} + \left\| O_{\Phi} \right\|_{\infty}\right\}\\
    &\leq& \frac{ (d_{\rm in}^2+1) \|O \|_{\infty} }{N}:=\alpha\label{eq:2_design_norm}.
\end{eqnarray}
where the last inequality follows from $
|\nu_j|\le \|O\|_\infty$ and $\left\|d_{\rm in}(d_{\rm in}+1)|\psi_i\rangle\langle\psi_i|-d_{\rm in}I_{d_{\rm in}}\right\|_\infty\le d_{\rm in}^2$. 
We also use $\|O_\Phi\|_\infty=\|\Phi^\dagger(O)\|_\infty\le\|O\|_\infty$, which comes from the fact that $\Phi^\dagger$ is positive and unital.
Also, the random matrix $\omega_{\text{2-dgn}}$ obeys
\begin{eqnarray}
    \mathbb{E}[(\omega_{\text{2-dgn}}- \mathbb{E}[\omega_{\text{2-dgn}}])^2 ] &=& \mathbb{E}[\omega_{\text{2-dgn}}^2] - (\mathbb{E}[\omega_{\text{2-dgn}}])^2 \\
    &=& \mathbb{E}[\omega_{\text{2-dgn}}^2] - O_{\Phi}^2\label{eq:variance_decomp}
\end{eqnarray}
and we have
\begin{eqnarray}
    \mathbb{E}[\omega_{\text{2-dgn}}^2] 
    &=& \sum_{i} p_{i} \sum_j \mathrm{tr}\left[ \pure{j} W^{\dagger}\Phi(\pure{\psi_{i}}) W\right] \nu_j^2 \left( d_{\rm in}(d_{\rm in}+1) \pure{\psi_{i}} - d_{\rm in}I_{d_{\rm in}}  \right)^2\\
    &=& \sum_{i} p_{i} \mathrm{tr}\left[ \Phi^{\dagger}(O^2) \pure{\psi_{i}}\right] \left( d_{\rm in}^2(d_{\rm in}^2-1) \pure{\psi_{i}} + d_{\rm in}^2I_{d_{\rm in}} \right)\\
    &=& d_{\rm in}^2(d_{\rm in}^2-1) \sum_{i} p_{i} \mathrm{tr}\left[ \Phi^{\dagger}(O^2) \pure{\psi_{i}}\right]   \pure{\psi_{i}} + d_{\rm in}^2 \sum_{i}p_{i} \, \mathrm{tr}\left[ \Phi^{\dagger}(O^2) \pure{\psi_{i}} \right] I_{d_{\rm in}}\\
    &=& d_{\rm in}^2(d_{\rm in}^2-1) \mathrm{tr}_{1} \left[ \left(\Phi^{\dagger}(O^2) \otimes I\right) \frac{\mathrm{SWAP}_{n}+I_{d_{\rm in}^2}}{d_{\rm in}(d_{\rm in}+1)} \right] + d_{\rm in}\,\mathrm{tr}\left[ \Phi^{\dagger}(O^2) \right] I_{d_{\rm in}} \\
    &=& d_{\rm in}(d_{\rm in}-1) \Phi^{\dagger}(O^2) + d_{\rm in}^2 \,\mathrm{tr}\left[ \Phi^{\dagger}(O^2) \right] I_{d_{\rm in}},
\end{eqnarray}
where in the second equality we use the eigenvalue decomposition of $O$ in Eq.~\eqref{eq:svd_of_O} together with the adjoint map $\Phi^{\dagger}$ of $\Phi$; in the fourth equality we apply the Corollary~\ref{cor:2_design}, which follows from the fact that the probability distribution $\mathcal{S}_{\text{2-dgn}}$ forms a state 2-design. 
Finally, in the last line, we employ the partial-swap-trick from Lemma~\eqref{lemma:partial_swap_trick}.

Hence, the variance $\sigma^2$ can be upper bounded by 
\begin{eqnarray}
    \left\| \sum_{k=1}^{N} \mathbb{E}\left[ \left(\frac{\omega^{(k)}_{\text{2-dgn}}-O_{\Phi}}{N}\right)^2 \right] \right\|_{\infty} 
    &=& \frac{1}{N} \left\| d_{\rm in}(d_{\rm in}-1) \Phi^{\dagger}(O^2) + d_{\rm in}^2\,\mathrm{tr}\left[ \Phi^{\dagger}(O^2) \right] I_{d_{\rm in}} - O_{\Phi}^2 \right\|_{\infty} \\
    &\leq& \frac{1}{N} \left\{ d_{\rm in}(d_{\rm in}-1) \| \Phi^{\dagger}(O^2) \|_{\infty}
    + d_{\rm in}^2 \mathrm{tr}\left[ \Phi^{\dagger}(O^2) \right]  + \|O_{\Phi}^2\|_{\infty} \right\}\\
    &\leq& \frac{1}{N} \left\{ d_{\rm in}(d_{\rm in}-1) \| O \|_{\infty}^2
    + d_{\rm in}^3 \| O \|_{\infty}^2  + \|O_{\Phi}\|_{\infty}^2 \right\}\\
    &\leq& \frac{\| O \|_{\infty}^2(d_{\rm in}^2+1)(d_{\rm in}+1)}{N}:=\sigma^2.\label{eq:2_design_variance}
\end{eqnarray}
In the third inequality, we used $0\leq O^2\le \|O\|_\infty^2 I_{d_{\rm out}}$. 
Since $\Phi^\dagger$ is positive and unital, this implies $0\le \Phi^\dagger(O^2) \leq \|O\|_\infty^2 I_{d_{\rm in}}$ and hence $\|\Phi^\dagger(O^2)\|_\infty\le \|O\|_\infty^2$ and $\operatorname{tr}[\Phi^\dagger(O^2)]\leq d_{\rm in}\|O\|_\infty^2$. We also used $\|O_\Phi^2\|_\infty=\|O_\Phi\|_\infty^2\leq\|O\|_\infty^2$.
Substituting Eqs.~\eqref{eq:2_design_norm} and \eqref{eq:2_design_variance} into Eq.~\eqref{eq:matrix_benstein} yields
\begin{eqnarray}
    \mathrm{Pr}\left[ \left\| \hat{O}_{\Phi,\text{2-dgn}} - O_{\Phi}
     \right\|_{\infty} \geq \eta \right] \leq 2d_{\rm in} \exp\left( \frac{-N\eta^2/2}{ \|O\|_{\infty}^2 (d_{\rm in}^2+1)(d_{\rm in}+1) + \|O\|_{\infty}(d_{\rm in}^2+1)\eta/3} \right).
\end{eqnarray}
Thus, it can be concluded that
\begin{equation}
    N \geq \frac{2 \max\{1,\|O\|_{\infty}^2\}(d_{\rm in}^2+1) (d_{\rm in}+1+\eta/3) } {\eta^2} \ln{\left(\frac{2d_{\rm in}}{\delta}\right)}
\end{equation}
suffices to guarantee that $\|\hat{O}_{\Phi,\text{2-dgn}}-O_{\Phi} \|_{\infty}\leq \eta$ with probability at least $1-\delta$.
\end{proof}

\subsubsection{Tensor-product single-qubit stabilizer states}\label{sec:prt_local_2_design_complexity}
\begin{proof}
Similar to the previous discussion, since the estimator $\hat{O}_{\Phi,\text{stab}}$ has the explicit representation $\hat{O}_{\Phi,\text{stab}} = \frac{1}{N} \sum_{i=1}^{N} \hat{\omega}_{\text{stab}}^{(k)}$ in Eq.~\eqref{eq:estimator_stab}, it follows that
\begin{equation}
    \hat{O}_{\Phi,\text{stab}} - O_{\Phi} =  \sum_{k=1}^{N} \frac{1}{N} \left\{\hat{\omega}_{\text{stab}}^{(k)}-O_{\Phi}\right\}.
\end{equation}
Here, each $\hat{\omega}_{\text{stab}}^{(k)}$ is an i.i.d. copy of a random matrix $\omega_{\text{stab}}=\nu_{j}\bigotimes_{s=1}^{n} (6\pure{u^s}-2I)$, which occurs with probability $(1/6^n)\bra{j} W^{\dagger} \Phi(\pure{u}) W\ket{j}$, where $\ket{u}=\bigotimes_{s=1}^{n} \ket{u^s}$.
Combining this with the unbiasedness of $\hat{O}_{\Phi,\text{stab}}$ shown in Appendix~\ref{apsec:const_stab} implies that $\hat{O}_{\Phi,\text{stab}} - O_{\Phi}$ is a sum of centered random matrices $X_{\text{stab}}^{(k)}:=(\hat{\omega}_{\text{stab}}^{(k)}-O_{\Phi})/N$. 
The operator norm of each random matrix $X_{\text{stab}}^{(k)}$ can be upper bounded as
\begin{eqnarray}
    \| X_{\text{stab}}^{(k)} \|_{\infty} &\leq& \frac{1}{N} \left\{ \| \omega_{\text{stab}}^{(k)} \|_{\infty} + \| O_{\Phi} \|_{\infty} \right\}\\
    &\leq& \frac{1}{N} \left\{  \max_{\ket{u},\nu_j} |\nu_j | \prod_{s=1}^{n} \| 6\pure{u^s}-2I \|_{\infty}  + \| O_{\Phi} \|_{\infty} \right\}\\
    &\leq& \frac{ (4^n+1) \|O \|_{\infty} }{N} := \alpha.
\end{eqnarray}
Since $\mathbb{E}[(X^{(k)}_{\text{stab}})^2 ]= (\mathbb{E}[\omega_{\text{stab}}^2]-O_{\Phi}^2)/N^2$, it suffices to calculate $\mathbb{E}[ \omega_{\text{stab}}^2 ]$ to evaluate the variance parameter:
\begin{eqnarray}
    \mathbb{E}[\omega_{\text{stab}}^2] 
    &=& \mathbb{E}_{u} \sum_j \mathrm{tr}\left[ \pure{j} W^{\dagger}\Phi\left(\pure{u} \right) W\right] \nu_j^2 \left( \bigotimes_{s=1}^{n} (6\pure{u^s}-2I)  \right)^2\\
    &=& \frac{1}{6^n} \sum_{\ket{u}\in \mathcal{Q}_{\rm stab}^{\otimes n}} \mathrm{tr}\left[ \Phi^{\dagger}(O^2) \bigotimes_{s=1}^{n} \pure{u^{s}} \right] \left( 4^n \bigotimes_{s=1}^{n} (3\pure{u^{s}}+I) \right).
\end{eqnarray}
Here, to proceed with the above calculation, we first examine the case where the input is a product operator $A:=\bigotimes_{s=1}^{n}A^s$ with $A^{s}\in \Herm{\mathbb{C}^2}$, instead of $\Phi^{\dagger}(O^2)$.
\begin{eqnarray}
    \sum_{\ket{u}\in \mathcal{Q}_{\rm stab}^{\otimes n}} \mathrm{tr} \left[ A \bigotimes_{s=1}^{n} \pure{u^{s}} \right] \left( \bigotimes_{s=1}^{n} \left( 3\pure{u^{s}}+I \right) \right)
    &=& \bigotimes_{s=1}^{n} \left( \sum_{\ket{u^s} \in \mathcal{Q}_{\rm stab}} \mathrm{tr}\left[ A^s \pure{u^{s}}\right] (3 \pure{u^{s}}+I)  \right) \\
    &=& 3^n \bigotimes_{s=1}^{n} \left( 2\mathrm{tr}(A^s) I + A^s \right)\\
    &=& 3^n \sum_{\alpha \subseteq [n]} 2^{|\alpha|} \left(\mathrm{tr}_{\alpha} \left[A \right] \otimes \bigotimes_{s\in \alpha} I_s \right),
\end{eqnarray}
where the second equality uses $\sum_{\ket{u^s}\in \mathcal{Q}_{\rm stab}} \mathrm{tr}\left[ A^s \pure{u^{s}}\right] (3 \pure{u^{s}}+I)  = 3 \left( 2 \mathrm{tr}(A^s)I+A^s \right)$ for all $s\in[n]$, which follows from the fact that uniform distribution over the set $\mathcal{Q}_{\rm stab}$ forms a state 2-design. 
By linear extension of such product operators $A$, the result can be applied to $\Phi^{\dagger}(O^2)$. Therefore, the operator norm of $\mathbb{E}[\omega_{\rm stab}^2]$ can be upper bounded as
\begin{eqnarray}
    \left\| \mathbb{E}[\omega_{\text{stab}}^2]\right\|_{\infty}
    &=& \left\| 2^n \sum_{\alpha \subseteq [n]} 2^{|\alpha|} \left(\mathrm{tr}_{\alpha} \left[\Phi^{\dagger}(O^2) \right] \otimes \bigotimes_{s\in \alpha} I_s \right)\right\|_{\infty}\\
    &\leq& 2^n \sum_{\alpha \subseteq [n]} 2^{|\alpha|}  \left\|\mathrm{tr}_{\alpha} \left[\Phi^{\dagger}(O^2) \right] \right\|_{\infty} \left\| \bigotimes_{s\in \alpha} I_s \right\|_{\infty} \\
    &\leq& 2^n \left\| O \right\|^2_{\infty} \sum_{\alpha \subseteq [n]} 2^{|\alpha|} \cdot 2^{|\alpha|} \\
    &\leq& 2^n \left\| O \right\|^2_{\infty} \sum_{j=0}^{n} \binom{n}{j} 4^{j} \\
    &\leq& 2^n \left\| O \right\|^2_{\infty} (4+1)^n\leq 10^n \|O \|_{\infty}^2,
\end{eqnarray}
where we used $0\leq O^2\le \|O\|_\infty^2 I_{d_{\rm out}}$ and the positivity and unitality of $\Phi^\dagger$, which imply $0\le \Phi^\dagger(O^2)\leq\|O\|_\infty^2 I_{d_{\rm in}}$.
Consequently, for each subset $\alpha$, $\left\|\operatorname{tr}_{\alpha}\!\left[\Phi^\dagger(O^2)\right]\right\|_\infty\leq2^{|\alpha|}\|O\|_\infty^2$, with the convention for $\operatorname{tr}_{\alpha}$ used above.
Hence, the variance $\sigma^2$ can be upper bounded by 
\begin{eqnarray}
    \left\| \sum_{i=1}^{N} \mathbb{E}\left[ \left(\frac{\omega^{(k)}_{\text{stab}} -O_{\Phi}}{N}\right)^2 \right] \right\|_{\infty} 
    &\leq& \frac{1}{N} \left\{ \left\| \mathbb{E}[(\omega_{\text{stab}}^{(k)})^2]\right\|_{\infty} + \| O^2_{\Phi} \|_{\infty} \right\} \\
    &\leq& \frac{\| O \|_{\infty}^2(10^n+1)}{N}:= \sigma^2.
\end{eqnarray}
Thus, Matrix Bernstein inequality in Lemma~\ref{lemma:matrix_bernstein_inequality} implies that
\begin{equation}
    N \geq \frac{2 \max\{1,\|O\|_{\infty}^2\}(10^n+1+(\eta/3)(4^n+1)) } {\eta^2} \ln{\left(\frac{2^{n+1}}{\delta}\right)}
\end{equation}
is sufficient to ensure that $\| \hat{O}_{\Phi,\text{stab}}-O_{\Phi} \|_{\infty} \leq \eta$ with probability at least $1-\delta$. Since $d_{\rm in}=2^n$ and
$10^n=d_{\rm in}^{\log_2 10}\le d_{\rm in}^{3.33}$,
this implies the stabilizer bound stated in Theorem~\ref{thm:error_bounds_for_each_x}.
\end{proof}

\subsubsection{$n$-qubit Pauli operator}
\begin{proof}
From the unbiasedness of $\hat{O}_{\Phi,\text{pauli}}$ together with Eq.~\eqref{eq:estimator_pauli}, we have
\begin{equation}
    \hat{O}_{\Phi, \text{pauli}} - O_{\Phi} = \sum_{k=1}^{N} \frac{1}{N} \left\{ \omega_{\text{pauli}}^{(k)}-O_{\Phi} \right\}:=\sum_{k=1}^{N} X_{\rm pauli}^{(k)},
\end{equation}
which is a sum of independent centered random matrices.
For each random matrix $X_{\rm pauli}^{(k)}$, we have the upper bound as
\begin{eqnarray}
    \| X_{\rm pauli}^{(k)} \|_{\infty} 
    &\leq& \frac{1}{N} \left\{ \left\| \omega_{\rm pauli}^{(k)} \right\|_{\infty} + \| O_{\Phi} \|_{\infty} \right\}\\
    &\leq& \frac{1}{N} \left\{ \max_{j ,\bm i, \bm e} \left\| 4^n \nu_{j} c_{\bm{i},\bm{e}} P_{\bm{i}} \right\|_{\infty} + \| O_{\Phi} \|_{\infty} \right\}\\
    &=& \frac{ (4^n+1)\|O \|_{\infty} }{N} := \alpha.
\end{eqnarray}
Here, $\mathbb{E}[\omega^2_{\rm pauli}]$ can be evaluated as
\begin{eqnarray}
    \mathbb{E}[\omega^2_{\rm pauli}]
    &=& \sum_{\bm i, \bm e} \frac{1}{8^n} \sum_{j} \mathrm{tr}\left[ \pure{j} W^{\dagger} \Phi(\pure{v_{\bm i, \bm e}}) W \right] \left( \nu_{j} 4^n c_{\bm i,\bm e}P_{\bm{i}} \right)^2\\
    &=& \frac{1}{8^n} \sum_{\bm i, \bm e} \mathrm{tr}\left[ \Phi^{\dagger}(O^2) \pure{v_{\bm i, \bm e}} \right] 16^n I^{\otimes n}\\
    &=& 2^n \sum_{\bm i\in[4]^n} \mathrm{tr}\left[ \Phi^{\dagger}(O^2) \right] I^{\otimes n}\\
    &=& 8^n \mathrm{tr}\left[ \Phi^{\dagger}(O^2) \right] I^{\otimes n},
\end{eqnarray}
where in the second-to-third line, we employ the fact that $\sum_{\bm e} \pure{v_{\bm i, \bm e}}=I^{\otimes n}$ for all $\bm i \in[4]^n$. 
Hence, we arrive at the following upper bound of the variance parameter:
\begin{eqnarray}
    \left\| \sum_{i=1}^{N} \mathbb{E}\left[ \left(\frac{\omega^{(k)}_{\text{pauli}}-O_{\Phi}}{N}\right)^2 \right] \right\|_{\infty}
    &\leq& \frac{1}{N} \left\{ \left\| \mathbb{E}[(\omega_{\text{pauli}}^{(k)})^2]\right\|_{\infty} + \| O^2_{\Phi} \|_{\infty} \right\}\\
    &\leq& \frac{\| O \|_{\infty}^2(16^n+1)}{N} := \sigma^2,
\end{eqnarray}
where we used $0\le \Phi^\dagger(O^2)\le\|O\|_\infty^2 I_{d_{\rm in}}$ which follows from $O^2\le \|O\|_\infty^2 I_{d_{\rm out}}$ and the
positivity and unitality of $\Phi^\dagger$. Therefore, $\operatorname{tr}[\Phi^\dagger(O^2)]\leq d_{\rm in}\|O\|_\infty^2
=2^n\|O\|_\infty^2$.

Like the preceding discussion, the Matrix Bernstein inequality shows that the number of iterations
\begin{equation}
    N \geq \frac{2 \max\{1,\|O\|_{\infty}^2\}(d_{\rm in}^4+1+(\eta/3)(d_{\rm in}^2+1)) } {\eta^2} \ln{\left(\frac{2d_{\rm in}}{\delta}\right)}
\end{equation}
suffices to ensure that $\| \hat{O}_{\Phi,\text{pauli}}-O_{\Phi} \|_{\infty} \leq \eta$ with probability at least $1-\delta$.
\end{proof}

\section{Cluster simulation of tree-structured quantum circuits}\label{apsec:cluster_simulation}

\subsection{Technical lemmas}\label{sec:tecchnical_lemmas}
In this subsection, we present several technical lemmas that will be used for the complexity analysis in the following subsections.
First, the following two lemmas are useful for bounding the norms of tensor products of local operators generated through statistical estimation procedures, and their deviation from the true value.

\begin{lemma}[See, e.g., Ref.~\cite{peng2021simulating}]\label{lemma:bound_exponential}
Let $r\in \mathbb{N}$ and $x$ be a real number satisfying $0\leq x \leq 1/r$. Then, the following inequality holds:
\begin{equation}
(1+x)^r \leq 1+(e-1)rx.
\end{equation}
\end{lemma}
\begin{proof}[Proof of Lemma~\ref{lemma:bound_exponential}]
Since $\ln(1+x)\leq x$ for all $x>-1$, we obtain $(1+x)^{r}=\exp{(r \ln(1+x))}\leq e^{rx}$. Now we consider the exponential function $f(t)=e^{t}$. Since $f(t)$ is convex, for $t\in[0,1]$ its graph lies below the line segment connecting the points $(0,1)$ and $(1,e)$, i.e., $e^{t} \leq (e-1)t+1$. Applying this with $t=rx$ and noting that $0\leq rx \leq 1$ by assumption, we have
\begin{equation}
    (1+x)^{r} \leq e^{rx} \leq 1+(e-1)rx,
\end{equation}
which concludes the proof.
\end{proof}

\begin{lemma}\label{lemma:bound_the_op_norm_of_products}
Let $\{A_i\}_{i=1}^r$ and $\{B_i\}_{i=1}^r$ be sets of matrices such that $\|A_i - B_i \|_{\infty} \leq \eta_i$ for all $i=1,\dots,r$ where $\eta_i \in (0,1]$. Then the following inequality holds:
\begin{equation}
    \left\| \bigotimes_{i=1}^{r} A_i - \bigotimes_{i=1}^{r} B_i \right\|_{\infty} 
    \leq \sum_{i=1}^{r} \eta_i \prod_{j \neq i} \max\{ \|A_j\|_{\infty}, \, \|B_j\|_{\infty}\}.
\end{equation}
\end{lemma}
\begin{proof}[Proof of Lemma~\ref{lemma:bound_the_op_norm_of_products}]
By the telescoping identity, we have
\begin{equation}
    \bigotimes_{i=1}^{r} A_{i} - \bigotimes_{i=1}^{r} B_i = \sum_{i=1}^{r} \Bigl( \bigotimes_{j<i} B_j \Bigr) \otimes (A_i - B_i) \otimes \Bigl( \bigotimes_{j>i} A_j \Bigr).
\end{equation}
Since $\| A\otimes B \|_{\infty}=\|A \|_{\infty}\|B\|_{\infty}$ and triangle inequality, it follows that
\begin{eqnarray}
    \left\| \bigotimes_{i=1}^{r} A_{i} - \bigotimes_{i=1}^{r} B_i \right\|_{\infty} 
    &\leq& \sum_{i=1}^{r} \|A_i - B_i\|_{\infty} \prod_{j<i}\|B_j\|_{\infty} \prod_{j>i}\|A_j\|_{\infty}\\
    &\leq& \sum_{i=1}^{r} \eta_i \prod_{j \neq i} \max\{ \|A_j\|_{\infty}, \, \|B_j\|_{\infty}\},
\end{eqnarray}
where the last inequality uses the assumption $\|A_i - B_i\|_{\infty} \leq \eta_i$.
\end{proof}

\begin{lemma}[Union bound; see, e.g., Ref.~\cite{rozanov2013probability}]\label{lemma:union-bound}
Given a finite set $\{E_i\}_{i=1}^{L}$ of events, the following inequality holds:
\begin{equation}
    \mathrm{Pr}\left[ \bigcup_{i=1}^{L} E_i \right] \leq \sum_{i=1}^{L} \mathrm{Pr}(E_{i}).
\end{equation}
\end{lemma}
In our analysis, we often use the following equivalent inequality obtained by applying De~Morgan’s law:
\begin{equation}
    1 - \sum_{i=1}^{L} \mathrm{Pr}\left[ E_i^{\rm c} \right] \leq \mathrm{Pr}\left[ \bigcap_{i=1}^{L} E_i \right].
\end{equation}
Therefore, to ensure that all events $E_{1},...,E_{L}$ occur simultaneously with probability at least $1-\delta$, it suffices to require that
\begin{equation}
    \sum_{i=1}^{L} \mathrm{Pr}[E_{i}^{c}] \leq \delta.
\end{equation}
In other words, if the individual failure probabilities $\mathrm{Pr}[E_{i}^{c}]$ are chosen such that their sum does not exceed $\delta$, then the probability that all events occur is at least $1-\delta$. As a simple example, it suffices to assign equal bounds $\mathrm{Pr}[E_{i}^c] \leq \delta/L$ for all $i$ for the desired overall confidence.

\subsection{Proof of Theorem~\ref{thm:2-layer}}\label{sec:proof_2_layer}

\newtheorem*{T7}{\bf{Theorem~\ref{thm:2-layer}} (Two-layer tree circuits)}
\begin{T7}
\textit{
Protocols A and B estimate $\mu={\rm tr}[O\rho_{\rm tree}]$ up to additive error $\epsilon\in(0,1]$ with probability at least $1-\delta$, provided that the
numbers of measurements satisfy, for $k=1,\ldots,R$,
\begin{equation}
    N_{{\rm a},k}
    =
    N_{{\rm b},k}
    =
    \mathcal O\!\left(
    \frac{d^3R^2}{\epsilon^2}
    \ln\!\left(
    \frac{Rd}{\delta}
    \right)
    \right),\qquad
    N_{{\rm a},0}
    =
    N_{{\rm b},0}
    =
    \mathcal O\!\left(
    \frac{1}{\epsilon^2}
    \ln\!\left(
    \frac{1}{\delta}
    \right)
    \right).
\end{equation}
Consequently,
\begin{equation}
    N_{{\rm a,tot}}
    =
    N_{{\rm b,tot}}
    =
    \mathcal O\!\left(
    \frac{d^3R^3}{\epsilon^2}
    \ln\!\left(
    \frac{Rd}{\delta}
    \right)
    \right).
\end{equation}
}
\end{T7}
Here, Protocol A and Protocol B refer to the measure-and-prepare and post-processing implementations introduced in Sec.~\ref{sec:2-layer-tree}, respectively.

\begin{proof}[Proof of Theorem~\ref{thm:2-layer}]
We first analyze Protocol A, and then B.

\vspace{0.5em}
\noindent \textbf{Analysis of Protocol A.}
Assume that the operator-norm error of each estimated effective observable $\tilde{M}_k$ in Step~\textbf{(a-1)} is upper bounded by $\eta_{\rm a}$, i.e., 
\begin{equation}
    \left\| \tilde{M}_k - M_k \right\|_{\infty}\leq \eta_{\rm a}, \quad k=1,...,R.
\end{equation}
Define the approximate target mean of Step~\textbf{(a-3)} as 
\begin{equation}
    \mu_{\rm a} := \mathrm{tr}\left[ O \left( \bigotimes_{k=1}^{R} \Phi_k \circ \mathcal{M}_{\rm apr}^{(k)} \right) (\rho) \right].\label{apeq:mu_a}
\end{equation}
Under this assumption, our first goal is to relate the individual errors $\eta_{\rm a}$ to the total bias in Step~\textbf{(a-3)}. To this end, we express the bias,
\begin{equation}\label{apeq:bias}
    \mathrm{Bias}_{\rm a} := \left| \mu_{\rm a} - \mu \right|,
\end{equation}
solely in terms of the quantity $\eta_{\rm a}$.
Using Eq.~\eqref{eq:mu_another}, this can be rewritten and bounded as
\begin{eqnarray}
    \mathrm{Bias}_{\rm a}
    = \left| \mathrm{tr}\left[ \left(\bigotimes_{k=1}^{R} \mathcal{M}_{\rm apr}^{(k)}(M_k) \right) \rho \right] - \mathrm{tr}\left[ \left(\bigotimes_{k=1}^{R} M_k \right) \rho \right] \right|
    \leq \left\| \bigotimes_{k=1}^{R} \mathcal{M}_{\rm apr}^{(k)}(M_k) - \bigotimes_{k=1}^{R} M_k \right\|_{\infty},\label{apeq:bias2}
\end{eqnarray}
where the inequality follows from Hölder's inequality, $|\mathrm{tr}[AB]|\leq \|A \|_{\infty} \| B\|_{1}$.
Applying the telescoping identity (see Lemma~\ref{lemma:bound_the_op_norm_of_products}), we obtain
\begin{eqnarray}
    \mathrm{Bias}_{\rm a}
    &\leq& \sum_{k=1}^{R} \left\| \mathcal{M}_{\rm apr}^{(k)}(M_k) - M_k \right\|_{\infty} \prod_{l \neq k} \max \left\{ \left\|\mathcal{M}_{\rm apr}^{(l)}(M_l) \right\|_{\infty}, \left\| M_l \right\|_{\infty} \right\}.\label{apeq:bias3-0}
\end{eqnarray}
Since each MP channel $\mathcal{M}_{\rm apr}^{(k)}$ is a unital pinching map, it is contractive in operator norm, $\| \mathcal{M}_{\rm apr}^{(k)}(A) \|_{\infty}\leq \|A \|_{\infty}$ for any Hermitian operator $A$. Moreover, $\mathcal{M}_{\rm apr}^{(k)}(\tilde{M}_k)=\tilde{M}_k$ by construction; see Appendix~\ref{sec:proof_of_existence_of_no_cost_cut}. Hence, 
\begin{eqnarray}
    \left\| \mathcal{M}_{\rm apr}^{(k)}(M_k) - M_k \right\|_{\infty} &\leq& \left\| \mathcal{M}_{\rm apr}^{(k)}(M_k - \tilde{M}_k) \right\|_{\infty} + \left\| \tilde{M}_k - M_k \right\|_{\infty}\\
    &\leq& 2 \left\| \tilde{M}_k - M_k \right\|_{\infty} \leq 2\eta_{\rm a}.\label{apeq:mp_norm-1}
\end{eqnarray}
Similarly, since $\Phi_l$ is CPTP, $\Phi_l^\dagger$ is completely positive and unital. 
As positive unital maps are contractive in the operator norm, we have $\|M_l\|_\infty=\|\Phi_l^\dagger(O_l)\|_\infty\leq \|O_l\|_\infty$.
Using $\|\mathcal{M}_{\rm apr}^{(l)}(M_l)\|_\infty \le \|M_l\|_\infty$ and Eq.~\eqref{apeq:mp_norm-1} in Eq.~\eqref{apeq:bias3-0} yields
\begin{eqnarray}
    \mathrm{Bias}_{\rm a} \leq \sum_{k=1}^{R} 2\eta_{\rm a} \prod_{l \neq k} \left\|M_l \right\|_{\infty}
    \leq 2R\eta_{\rm a}.\label{apeq:bias3}
\end{eqnarray}
Hence, by choosing
\begin{equation}
    \eta_{\rm a} = \frac{9\epsilon}{20R},\label{apeq:eta_condition}
\end{equation}
the $\mathrm{Bias}_{\rm a}$ is upper bounded by $\frac{9\epsilon}{10}$.

Let $\hat\mu_{\rm a}$ be the empirical mean from $N_{\rm a,0}$ shots in Step~\textbf{(a-3)}.
Define
\begin{equation}
E^{\rm a}_0 :=
\left\{
|\hat\mu_{\rm a}-\mu_{\rm a}|\le \frac{\epsilon}{10}
\right\}.
\end{equation}
Conditioned on the learned observables $\{\tilde M_k\}_{k=1}^R$, each
single-shot outcome in Step~\textbf{(a-3)} is bounded in $[-1,1]$. Hence, by
Hoeffding's inequality,
\begin{equation}
\Pr\!\left[(E^{\rm a}_0)^c
\,\middle|\,
\tilde M_1,\ldots,\tilde M_R
\right]
\le
2\exp\!\left(
-\frac{N_{\mathrm{a},0}\epsilon^2}{200}
\right).
\end{equation}
Since the right-hand side is uniform in the learning outcomes, the tower
property gives
\begin{equation}
\Pr[(E^{\rm a}_0)^c]
\le
2\exp\!\left(
-\frac{N_{\mathrm{a},0}\epsilon^2}{200}
\right).
\end{equation}
Thus it suffices to choose
\begin{equation}
N_{\mathrm{a},0}
=
O\!\left(
\frac{1}{\epsilon^2}\ln\frac{1}{\delta}
\right)
\end{equation}
so that $\Pr[(E^{\rm a}_0)^c]\le \delta/2$.

For each local learning step, define
\begin{equation}
    E^{\rm a}_k :=\left\{
    \|\tilde M_k-M_k\|_\infty \le \eta_{\rm a}
    \right\},
    \qquad k=1,\ldots,R .
\end{equation}
Applying Theorem~\ref{thm:error_bounds_for_each_x} with target accuracy
$\eta=\eta_{\rm a}=9\epsilon/(20R)$ and local failure probability $\delta/(2R)$, and using $\|O_k\|_\infty\le 1$ and $\dim \mathcal H_k\le d$, it suffices to take
\begin{equation}
    N_{a,k}
    =
    \mathcal O\!\left(
    \frac{d^3}{\eta_{\rm a}^2}
    \ln\left(\frac{2d}{\delta/(2R)}\right)
    \right)
    =
    \mathcal O\!\left(
    \frac{d^3R^2}{\epsilon^2}
    \ln \left(\frac{Rd}{\delta}\right)
    \right)
\end{equation}
to ensure $\Pr[(E^a_k)^c]\leq \delta/(2R)$.

On the event $\bigcap_{k=1}^R E^{\rm a}_k$, we have
${\rm Bias}_{\rm a}\le 9\epsilon/10$. On $E^{\rm a}_0$, we have
$|\hat\mu_{\rm a}-\mu_{\rm a}|\le \epsilon/10$. Therefore, on
$\bigcap_{k=0}^R E^{\rm a}_k$,
\begin{equation}
|\hat\mu_{\rm a}-\mu|
\le
|\hat\mu_{\rm a}-\mu_{\rm a}|+|\mu_{\rm a}-\mu|
\le
\frac{\epsilon}{10}+\frac{9\epsilon}{10}
=
\epsilon .
\end{equation}
Moreover,
\begin{equation}
\Pr\!\left[
\left(\bigcap_{k=0}^R E^{\rm a}_k\right)^c
\right]
\le
\Pr[(E^{\rm a}_0)^c]
+
\sum_{k=1}^R \Pr[(E^{\rm a}_k)^c]
\le
\frac{\delta}{2}
+
R\frac{\delta}{2R}
=
\delta .
\end{equation}
This proves the claim for Protocol A.

\vspace{0.5em}
\noindent \textbf{Analysis of Protocol B.}
Assume that each estimated effective observable $\tilde{M}_k$ in Step~\textbf{(b-1)} satisfies
\begin{equation}\label{apeq:assumption_b}
    \left\| \tilde{M}_k - M_k \right\|_{\infty}\leq \eta_{\rm b}, \quad k=1,...,R.
\end{equation}
Define the approximate target mean of Step~\textbf{(b-3)} as 
\begin{equation}
    \mu_{\rm b} :=  \mathrm{tr}\left[\left(\bigotimes_{k=1}^{R}\tilde{M}_k \right) \rho\right],\label{apeq:mu_b}
\end{equation}
where the equality follows from the definition of $\mathcal{C}_{\rm apr}^{(k)}$. 

Under this assumption of Eq.~\eqref{apeq:assumption_b}, we first analyze the relationship between the individual errors $\eta_{\rm b}$ and the total bias
\begin{equation}\label{apeq:bias_thm3}
    \mathrm{Bias}_{\rm b} := \left| \mu_{\rm b}  - \mu \right|
\end{equation}
for Step~\textbf{(b-3)}. Now, the total bias can be bounded as
\begin{eqnarray}\label{apeq:bias2_thm3}
    \mathrm{Bias}_{\rm b} = \left| \mathrm{tr}\left[ \left(\bigotimes_{k=1}^{R}\tilde{M}_k-\bigotimes_{k=1}^{R}M_k\right)\rho \right] \right|
    \leq \left\| \bigotimes_{k=1}^{R}\tilde{M}_k - \bigotimes_{k=1}^{R}M_k \right\|_{\infty},
\end{eqnarray}
where we used Hölder's inequality, $|\mathrm{tr}[AB]| \leq \| A\|_{\infty} \| B\|_{1}$. 
Using the telescoping identity in Lemma~\ref{lemma:bound_the_op_norm_of_products}, we obtain
\begin{eqnarray}
    \mathrm{Bias}_{\rm b} 
    &\leq&  \sum_{k=1}^{R} \left\| \tilde{M}_k-M_k\right\|_{\infty} \prod_{l \neq k} \max \left\{ \left\|M_l\right\|_{\infty}, \| \tilde{M}_l \|_{\infty} \right\} \\
    &\leq& \sum_{k=1}^{R} \left\| \tilde{M}_k-M_k\right\|_{\infty} \prod_{l \neq k} \max \left\{ \left\|M_l\right\|_{\infty}, \left\| M_l \right\|_{\infty} + \left\| \tilde{M}_l-M_l\right\|_{\infty} \right\} \\
    &\leq& R \eta_{\rm b} (1+\eta_{\rm b})^{R-1},\label{apeq:bias3_thm3}
\end{eqnarray}
where in the last inequality we used $\left\|M_l\right\|_{\infty}\le 1$ and
$\|\tilde{M}_l - M_l\|_{\infty}\leq \eta_{\rm b}$ for all $l$.

We now derive an explicit upper bound on the total bias $\mathrm{Bias}_{\rm b}$ in terms of the estimation errors $\eta_{\rm b}$. We further simplify the derived expression using the following fact (see Lemma~\ref{lemma:bound_exponential}):
\begin{equation}
    (1+x)^r \leq 1+(e-1) rx,\quad (0\leq x \leq 1/r).
\end{equation}
In particular, we set 
\begin{equation}
    \eta_{\rm b}= \frac{\epsilon}{2R(e-1)}\label{apeq:xi_set}
\end{equation}
for $k=1,...,R$.
Substituting this into Eq.~\eqref{apeq:bias3_thm3} yields
\begin{eqnarray}
    \mathrm{Bias}_{\rm b} 
    &\leq& R \times \frac{\epsilon}{2R(e-1)} \times \left( 1 + \frac{\epsilon}{2R(e-1)} \right)^{R-1}\\
    &\leq& R \times \frac{\epsilon}{2R(e-1)}\times \left( 1+ \frac{\epsilon}{2} \right)\\
    &\leq& \frac{\epsilon}{2} \times \frac{1.5}{e-1} \leq \frac{\epsilon}{2},\label{apeq:bias_bound}
\end{eqnarray}
where in the last two inequalities we use $\epsilon\leq 1$ and the fact that $1.5/(e-1)\leq 1$.
In Step~\textbf{(b-3)}, each run outputs a classical random variable $Y$ such that $\mathbb{E}[Y|\tilde{M}_1,...,\tilde{M}_R]=\mu_{\rm b}$. Moreover, conditioned on fixed $\{\tilde{M}_k\}$, the output is upper bounded by
\begin{equation}
    |Y| \leq \gamma(\tilde{M}) := \left\| \bigotimes_{k=1}^{R} \tilde{M}_{k} \right\|_{\infty} \leq \prod_{k=1}^{R} \left\| \tilde{M}_k \right\|_{\infty}.
\end{equation}
On the good-tomography event in Eq.~\eqref{apeq:tomography_good_event_b} below, we will show that $\gamma(\tilde{M})\leq 3/2$.

Next, we consider the statistical error in Step~\textbf{(b-3)}. Let $\hat{\mu}
_{\rm b}$ be the empirical mean from $N_{\mathrm{b},0}$ shots.
Define the statistical event
\begin{eqnarray}
    E_{0}^{\rm b} := \left\{ \left| \hat{\mu}_{\rm b} - \mu_{\rm b} \right| \leq \frac{\epsilon}{2} \right\},\label{apeq:protocol_b_3}
\end{eqnarray}
where $\mu_{\rm b}$ is defined in Eq.~\eqref{apeq:mu_b}.
Conditioned on fixed $\{\tilde{M}_{k}\}$, the $N_{\rm b,0}$ samples are i.i.d. and bounded in $[-\gamma(\tilde{M}),\gamma(\tilde{M})]$. Thus, Hoeffding's inequality gives
\begin{equation}
    \mathrm{Pr}\left[ (E_{0}^{\rm b})^{c} \mid \tilde{M}_1,...,\tilde{M}_{R} \right] \leq 2 \exp\left( -\frac{N_{\rm b,0} \epsilon^2}{8 \gamma(\tilde{M})^2 } \right).
\end{equation}
Here, let $T$ be the good-tomography event where the tomography errors are within the desired value as in Eq.~\eqref{apeq:xi_set}
\begin{eqnarray}
    T:= \bigcap_{k=1}^{R} \left\{ \left\| \tilde{M}_k - M_k \right\|_{\infty} \leq \eta_{\rm b} \right\}, \qquad
    \eta_{\rm b}=\frac{\epsilon}{2R(e-1)}.\label{apeq:tomography_good_event_b}
\end{eqnarray} 
On the event $T$, we have
\begin{eqnarray}
    \gamma(\tilde{M}) = \left\| \bigotimes_{k=1}^{R} \tilde{M}_{k} \right\|_{\infty} \leq \prod_{k=1}^{R} \left\| \tilde{M}_k \right\|_{\infty} \leq (1+\eta_{\rm b})^R,
\end{eqnarray}
where we used the triangle inequality, $\|\tilde{M}_k \|_{\infty} \leq \| M_k \|_{\infty} + \| M_k - \tilde{M}_k \|_{\infty}\leq 1+\eta_{\rm b}$.
With $\eta_{\rm b}=\epsilon/[2R(e-1)]\leq 1/R$, Lemma~\ref{lemma:bound_exponential} yields
\begin{eqnarray}
    (1+\eta_{\rm b})^R \leq 1+\frac{\epsilon}{2} \leq \frac{3}{2},
\end{eqnarray}
so $\gamma(\tilde{M})\leq3/2$ holds on the good-tomography event $T$. Consequently, for any realization with $T$,
\begin{equation}
    \mathrm{Pr}\left[ (E_{0}^{\rm b})^{c} \mid \tilde{M}_1,...,\tilde{M}_{R} \right] \leq 2 \exp\left( -\frac{N_{\rm b,0} \epsilon^2}{18} \right).
\end{equation}
By the tower property and the fact that the above conditional bound holds whenever $T$ occurs, we have
\begin{eqnarray}
    \mathrm{Pr}\left[ (E_{0}^{\rm b})^{\rm c} \cap T \right] = \mathbb{E}\left[ \mathbf{1}_T \mathrm{Pr}\left[ (E_{0}^{\rm b})^{\rm c} \mid \tilde{M}_1,...,\tilde{M}_R \right] \right] \leq 2\exp\left(-\frac{N_{\rm b,0} \epsilon^2}{18}\right),
\end{eqnarray}
where $\mathbf{1}_{T}$ is the indicator function about $T$.
Thus, it is sufficient to choose
\begin{equation}
    N_{\rm b,0} = \mathcal{O}\left( \frac{1}{\epsilon^2} \ln\left( \frac{1}{\delta} \right)\right),\label{apeq:N_b0}
\end{equation}
so that $\mathrm{Pr}[(E_{0}^{\rm b})^{\rm c} \cap T]\leq \delta/2$.
Then, on $T \cap E_{0}^{\rm b}$, we have $\mathrm{Bias}_{\rm b}\leq\epsilon/2$ and $| \hat{\mu}_{\rm b} - \mu_{\rm b} |\leq \epsilon/2$. The triangle inequality yields 
\begin{equation}
    \left| \hat{\mu}_{\rm b} - \mu \right| \leq \mathrm{Bias}_{\rm b} + | \hat{\mu}_{\rm b} - \mu_{\rm b} | \leq \epsilon.
\end{equation}

It remains to choose the number of samples in each local learning step. For $k=1,\ldots,R$, define
\begin{equation}
    E^{\rm b}_k :=
    \left\{
    \|\tilde M_k-M_k\|_\infty \le \eta_{\rm b}
    \right\},
    \qquad
    k=1,\ldots,R .
\end{equation}
Applying Theorem~\ref{thm:error_bounds_for_each_x} with target accuracy $\eta=\eta_b=\epsilon/[2R(e-1)]$ and local failure probability $\delta/(2R)$, and using $\|O_k\|_{\infty}\leq 1$ and $\dim\mathcal H_k\le d$, it suffices to take
\begin{equation}
    N_{b,k}=
    \mathcal O\!\left(
    \frac{d^3}{\eta_b^2}
    \ln \left(\frac{2d}{\delta/(2R)} \right)
    \right) =
    \mathcal O\!\left(
    \frac{d^3R^2}{\epsilon^2}
    \ln \left( \frac{Rd}{\delta} \right)
    \right)
\end{equation}
so that $\Pr[(E^{\rm b}_k)^c]\leq \delta/(2R)$.
With this choice, the good-tomography event $T$ defined in Eq.~\eqref{apeq:tomography_good_event_b} satisfies
\begin{equation}
\Pr[T^c]\le \sum_{k=1}^R \Pr[(E^{\rm b}_k)^c]
\le R\frac{\delta}{2R}
=
\frac{\delta}{2}.
\end{equation}
Moreover, the choice of $N_{\mathrm{b},0}$ above (Eq.~\eqref{apeq:N_b0}) ensures that
\begin{equation}
    \Pr[(E^{\rm b}_0)^c\cap T]\le \frac{\delta}{2}.
\end{equation}
Combining these bounds with the union bound, we obtain
\begin{equation}
    \Pr[(T\cap E^{\rm b}_0)^c]
    \leq
    \sum_{k=1}^R \Pr[(E^{\rm b}_k)^c]
    +
    \Pr[(E^{\rm b}_0)^c\cap T]
    \leq
    R\frac{\delta}{2R}+\frac{\delta}{2}
    =
    \delta.
\end{equation}
This proves the theorem statement for Protocol B.
\end{proof}

\subsection{Proof of Theorem~\ref{thm:tree_simulation_general}}\label{sec:proof_multi_layer}

In this subsection, we prove the formal version of Theorem~\ref{thm:tree_simulation_general}. Note that we assume $K\ge1$ and the case $K=0$ is handled by directly measuring the root system.

\newtheorem*{T5}{\bf{Theorem~\ref{thm:tree_simulation_general} (General tree-structured circuits, formal)}}
\begin{T5}
\textit{
Consider the finite rooted tree setup of Sec.~\ref{sec:main_tree_setup}. 
Let $K$ be
the number of cut edges, and let $d_v$ be the bond dimension of the cut
edge entering node $v\in V^\circ$. 
Then, Algorithm~\ref{alg:general-tree-appendix} estimates the expectation value
\begin{equation}
    \mu={\rm tr}[O\rho_{\rm tree}]
\end{equation}
within additive error $\epsilon\in(0,1]$ with probability
at least $1-\delta$ using
\begin{equation}
    N_{\rm tot}
    =
    \mathcal O\!\left(
    \frac{K^2}{\epsilon^2}
    \sum_{v\in V^\circ}d_v^3
    \ln\!\left(
    \frac{Kd_v}{\delta}
    \right)
    +
    \frac{1}{\epsilon^2}
    \ln\!\left(
    \frac{1}{\delta}
    \right)
    \right)
\end{equation}
measurements. In particular, if $d_v\le d$ for all $v\in V^\circ$, then
\begin{equation}
    N_{\rm tot}
    =
    \mathcal O\!\left(
    \frac{d^3K^3}{\epsilon^2}
    \ln\!\left(
    \frac{Kd}{\delta}
    \right)
    \right).
\end{equation}
}
\end{T5}

\subsubsection{General algorithm overview}

\begin{algorithm}[b]
\caption{Bottom-up learning protocol on a finite rooted tree}
\label{alg:general-tree-appendix}
\DontPrintSemicolon

\KwIn{
A finite rooted tree $T=(V,E)$ with root $r$;
local maps $\{\Phi_v\}_{v\in V^\circ}$;
root state $\rho$;
terminal observables $\{O_v\}_{v\in V}$;
shot budgets $\{N_v\}_{v\in V^\circ}$ and $N_0$.
}

\KwOut{
An estimate $\hat\mu$ of $\mu={\rm tr}[O\rho_{\rm tree}]$.
}

Choose a bottom-up ordering of $V^\circ$, so that every child is processed before its parent.\;

\ForEach{$v\in V^\circ$ in this order}{
    Form
    \[
        \tilde B_v
        \leftarrow
        O_v\otimes
        \bigotimes_{u\in{\rm Ch}(v)}\tilde M_u .
    \]
    Using $N_v$ fresh samples, apply the 2-design learning protocol of
    Theorem~\ref{thm:error_bounds_for_each_x} to estimate
    \[
        \Phi_v^\dagger(\tilde B_v).
    \]
    Store the resulting estimator as $\tilde M_v$.\;
}

Form
\[
    \tilde B_r
    \leftarrow
    O_r\otimes
    \bigotimes_{u\in{\rm Ch}(r)}\tilde M_u .
\]

Measure the root state $\rho$ in an eigenbasis of $\tilde B_r$ for $N_0$ shots, weight each outcome by the corresponding eigenvalue, and return
the empirical mean $\hat\mu$.\;
\Return{$\hat\mu$}
\end{algorithm}

We now describe the protocol and the shot allocation used in Theorem~\ref{thm:tree_simulation_general}.
The protocol processes the vertices in a bottom-up order, so that every child is processed before its parent. 
Specifically, suppose that, when processing a node $v$, the estimates $\{\tilde M_u\}_{u\in{\rm Ch}(v)}$ for all its children have already been constructed. 
We then form the downstream observable
\begin{equation}
    \tilde B_v
    :=
    O_v\otimes
    \bigotimes_{u\in{\rm Ch}(v)}\tilde M_u .
\end{equation}
Using $N_v$ fresh samples from the local device implementing $\Phi_v$, we apply the 2-design learning protocol of
Theorem~\ref{thm:error_bounds_for_each_x} to estimate
\begin{equation}
    \Phi_v^\dagger(\tilde B_v),
\end{equation}
and store the resulting estimator as $\tilde M_v$. Here, the target local accuracy $\eta_v$ in operator norm and failure probability $\delta_v$ are chosen as
\begin{equation}\label{eq:error_choice}
    \eta_v=\frac{\ln(1+\epsilon/2)}{K},
    \qquad
    \delta_v=\frac{\delta}{2K}.
\end{equation}
Thus, the number of samples at node $v$ is chosen as
\begin{equation}
\label{eq:general_tree_local_shots}
    N_v
    =
    \mathcal O\!\left(
    \frac{d_v^3}{\eta_v^2}
    \ln\frac{d_v}{\delta_v}
    \right)
    =
    \mathcal O\!\left(
    \frac{d_v^3K^2}{\epsilon^2}
    \ln\frac{Kd_v}{\delta}
    \right).
\end{equation}

After all non-root vertices have been processed, we form
\begin{equation}
    \tilde B_r
    :=
    O_r\otimes
    \bigotimes_{u\in{\rm Ch}(r)}\tilde M_u .
\end{equation}
We then measure the root state $\rho$ in an eigenbasis of $\tilde B_r$,
weight each outcome by the corresponding eigenvalue, and return the empirical
mean $\hat\mu$. We also reserve the remaining failure probability for the final root
measurement and set $\delta_{0}:=\delta/2$. Thus, the number of root measurements is
\begin{equation}
\label{eq:general_tree_root_shots}
    N_0
    =\mathcal O\!\left(
    \frac{1}{\epsilon^2}
    \ln\frac{1}{\delta_0}
    \right)
    =
    \mathcal O\!\left(
    \frac{1}{\epsilon^2}
    \ln\frac{1}{\delta}
    \right).
\end{equation}
The same procedure is summarized in Algorithm~\ref{alg:general-tree-appendix}.

\vspace{1em}
\subsubsection{Success events}
For the auxiliary statements below, fix arbitrary local accuracy parameters $\boldsymbol{\eta}:=\{\eta_v\}_{v\in V^\circ}$ and $\eta_v>0$.
For each $v\in V^\circ$, define the local success event
\begin{equation}
    E_v(\eta_v)
    :=
    \left\{
    \left\|
    \tilde M_v-\Phi_v^\dagger(\tilde B_v)
    \right\|_\infty
    \le \eta_v
    \right\}.
\end{equation}
Thus, $E_v(\eta_v)$ is the event that the local learning step at node $v$ succeeds with accuracy $\eta_v$. 

For each $v\in V^\circ$, let $T_v$ denote the subtree rooted at $v$,
including $v$ itself. We define the subtree success event by
\begin{equation}
    G_v(\boldsymbol{\eta})
    :=
    \bigcap_{w\in T_v}E_w(\eta_w),
\end{equation}
and the strict-descendant success event by
\begin{equation}
    D_v(\boldsymbol{\eta})
    :=
    \bigcap_{w\in T_v\setminus\{v\}}E_w(\eta_w).
\end{equation}
Here, $G_v$ means that all local learning steps in the subtree rooted at $v$, including the step at $v$ itself, succeed, while $D_v$ means that all strict descendants of $v$ succeed. We also define the global success event
\begin{equation}
    G(\boldsymbol{\eta})
    :=
    \bigcap_{v\in V^\circ}E_v(\eta_v).
\end{equation}
When the accuracy parameters are clear from context, we simply write
$E_v$, $G_v$, $D_v$, and $G$.

\subsubsection{Proof using the auxiliary statements}

Before presenting the full proof of Theorem~\ref{thm:tree_simulation_general}, we state the auxiliary results used to analyze
Algorithm~\ref{alg:general-tree-appendix}. The detailed proof for each statement can be found after the proof of Theorem~\ref{thm:tree_simulation_general}.
The first lemma bounds the norm of
the ideal effective observables.

\begin{lemma}[Norm bound for ideal effective observables]
\label{lem:ideal_norm_general_tree}
For every $v\in V^\circ$,
\begin{equation}
    \|M_v\|_\infty\le 1,\qquad \|B_r\|_\infty\le 1.
\end{equation}
\end{lemma}

The following proposition is the deterministic part of the proof. It shows that local learning errors accumulate over the actual descendant subtree of a node.

\begin{proposition}[Subtree-wise error propagation]
\label{prop:subtree_error_general_tree}
For each $v\in V^\circ$, define $A_v:=\sum_{w\in T_v}\eta_w$.
Then, on the event $G_v$,
\begin{equation}
    \|\tilde M_v-M_v\|_\infty\le e^{A_v}-1,
    \qquad
    \|\tilde M_v\|_\infty\le e^{A_v}.
\end{equation}
\end{proposition}

The next corollary records the root bias bound and the norm bounds used below.

\begin{corollary}[Bias and norm control]
\label{cor:root_and_intermediate_norm_general_tree}
Define $A_{\rm tot}:=\sum_{v\in V^\circ}\eta_v$. Then, on the event $G$,
\begin{equation}
    \|\tilde B_r-B_r\|_\infty\le e^{A_{\rm tot}}-1,
    \qquad
    \|\tilde B_r\|_\infty\le e^{A_{\rm tot}}.
\end{equation}
Moreover, for every $v\in V^\circ$, on the event $D_v$,
\begin{equation}
    \|\tilde B_v\|_\infty\le e^{A_{\rm tot}}.
\end{equation}
\end{corollary}

The following lemma controls the probability of the adaptive learning process.
The point is that $\tilde B_v$ is random because it is constructed from the
learning outcomes in the descendant subtrees. Fix the bottom-up ordering, and
let $\mathcal F_{<v}$ denote the information generated by all data and estimators obtained before processing $v$. Then $\tilde B_v$ and $D_v$ are determined by
$\mathcal F_{<v}$.

\begin{lemma}[Success probability of adaptive bottom-up learning]
\label{lem:adaptive_success_general_tree}
Fix a bottom-up processing order, and let $\mathcal F_{<v}$ denote
all information available before processing $v$, including all data and estimators obtained at previously processed nodes.
Suppose that, for every $v\in V^\circ$, the event $D_v$ is determined
by $\mathcal F_{<v}$, and that
\begin{equation}
    \Pr(E_v^c\mid\mathcal F_{<v})\le \delta_v
    \qquad
    \text{on the event }D_v .
\end{equation}
Then
\begin{equation}
    \Pr(G^c)\le \sum_{v\in V^\circ}\delta_v.
\end{equation}
\end{lemma}

Combining the above statements, we prove Theorem~\ref{thm:tree_simulation_general} as follows.

\begin{proof}[Proof of Theorem~\ref{thm:tree_simulation_general}]
We use the accuracy and failure-probability budgets in
Eq.~\eqref{eq:error_choice}. Thus,
\begin{equation}
    \eta_v=\frac{\ln(1+\epsilon/2)}{K},
    \qquad
    \delta_v=\frac{\delta}{2K}
    \qquad
    (v\in V^\circ),
\end{equation}
and we set $\delta_0:=\delta/2$.
With this choice,
\begin{equation}
    A_{\rm tot}
    :=
    \sum_{v\in V^\circ}\eta_v
    =
    \ln\left(1+\frac{\epsilon}{2}\right).
\end{equation}

First, assume that the global success event $G$ holds. We bound the deterministic bias caused by replacing the ideal root observable $B_r$ with
the learned observable $\tilde B_r$.
By Corollary~\ref{cor:root_and_intermediate_norm_general_tree}, on the global success event $G$,
\begin{equation}
    \|\tilde B_r-B_r\|_\infty
    \le
    e^{A_{\rm tot}}-1
    = \frac{\epsilon}{2}.
\end{equation}
Since $\mu={\rm tr}[\rho B_r]$, we obtain, on $G$,
\begin{equation}
    \left|
    {\rm tr}[\rho\tilde B_r]-\mu
    \right|
    \le
    \|\tilde B_r-B_r\|_\infty
    \le \frac{\epsilon}{2}.
\end{equation}

Next, we bound the probability of $G$. 
Fix a non-root node $v\in V^\circ$. 
On the strict-descendant success event $D_v$, Corollary~\ref{cor:root_and_intermediate_norm_general_tree} gives
\begin{equation}
    \|\tilde B_v\|_\infty
    \le
    e^{A_{\rm tot}}
    =
    1+\frac{\epsilon}{2}
    \le \frac{3}{2}.
\end{equation}
Thus, conditioned on the information $\mathcal F_{<v}$ obtained before processing $v$ and on the event $D_v$, the observable $\tilde B_v$ is fixed and has operator norm bounded by a constant.
Therefore, by the 2-design learning guarantee of
Theorem~\ref{thm:error_bounds_for_each_x}, the choice $N_v=\mathcal O\!\left(
    d_v^3\eta_v^{-2}
    \ln(d_v/\delta_v)
    \right)$
ensures
\begin{equation}
    \Pr(E_v^c\mid\mathcal F_{<v})\le \delta_v
    \qquad
    \text{on }D_v .
\end{equation}
Thus, for every non-root vertex $v$, the conditional success
condition required in Lemma~\ref{lem:adaptive_success_general_tree} is satisfied with failure probability
$\delta_v$. Applying Lemma~\ref{lem:adaptive_success_general_tree} then gives
\begin{equation}
    \Pr(G^c) \leq  \sum_{v\in V^\circ}\delta_v =
    \frac{\delta}{2}
\end{equation}
It remains to control the final measurement at the root. 
By Corollary~\ref{cor:root_and_intermediate_norm_general_tree}, on $G$,
\begin{equation}
    \|\tilde B_r\|_\infty
    \le
    e^{A_{\rm tot}}
    \leq \frac{3}{2}.
\end{equation}
Thus each single-shot outcome in the final root measurement is bounded in
absolute value by $3/2$. Here, define
\begin{equation}
    E_0:=\left\{\left|\hat\mu-{\rm tr}[\rho\tilde B_r]\right|\le\frac{\epsilon}{2}\right\}.
\end{equation}
Then, by Hoeffding's inequality, the choice $N_0=\mathcal O(\epsilon^{-2}\ln(1/\delta_0))$
ensures
\begin{equation}
    \Pr(E_0^c\cap G)\le \delta_0.
\end{equation}

On the event $G\cap E_0$, the total error satisfies
\begin{equation}
    |\hat\mu-\mu|
    \le
    \left|
    \hat\mu-{\rm tr}[\rho\tilde B_r]
    \right|
    +
    \left|
    {\rm tr}[\rho\tilde B_r]-\mu
    \right|
    \leq
    \frac{\epsilon}{2}+ \frac{\epsilon}{2}
    =
    \epsilon.
\end{equation}
Therefore, we have
\begin{equation}
    \Pr(|\hat\mu-\mu|>\epsilon)
    \le
    \Pr(G^c)+\Pr(E_0^c\cap G)
    \leq \frac{\delta}{2} + \frac{\delta}{2} =
    \delta.
\end{equation}

We finally count the number of measurements. Since
\begin{eqnarray}
    \delta_v&=&\frac{\delta}{2K},
    \qquad
    \eta_v
    =
    \frac{\ln(1+\epsilon/2)}{K}
    =
    \Theta\!\left(\frac{\epsilon}{K}\right)\quad \text{for $\epsilon\in(0,1]$},
\end{eqnarray}
we have
\begin{equation}
    N_v
    =
    \mathcal O\!\left(
    \frac{d_v^3K^2}{\epsilon^2}
    \ln\frac{Kd_v}{\delta}
    \right).
\end{equation}
Summing over all non-root nodes and the root-measurement cost gives
\begin{eqnarray}
    N_{\rm tot}
    &=&N_0+\sum_{v\in V^\circ}N_v
    = \mathcal{O}\left( \frac{1}{\epsilon^2} \ln\frac{1}{\delta} \right) + \sum_{v\in V^\circ} \mathcal O\!\left(
    \frac{d_v^3K^2}{\epsilon^2}
    \ln\left(\frac{Kd_v}{\delta}\right)
    \right)\\
    &=&
    \mathcal O\!\left(
    \frac{K^2}{\epsilon^2}
    \sum_{v\in V^\circ}
    d_v^3
    \ln\left(\frac{Kd_v}{\delta}\right)
    +
    \frac{1}{\epsilon^2}
    \ln\frac{1}{\delta}
    \right).
\end{eqnarray}
If $d_v\le d$ for every $v\in V^\circ$, then $\sum_{v\in V^\circ}d_v^3\le Kd^3$ and $\ln(Kd_v/\delta)\le \ln(Kd/\delta)$. Thus,
\begin{equation}
    N_{\rm tot}
    =
    \mathcal O\!\left(
    \frac{d^3K^3}{\epsilon^2}
    \ln\left(\frac{Kd}{\delta}\right)
    \right).
\end{equation}
This completes the proof.
\end{proof}

\subsubsection{Proofs of auxiliary results}

\begin{proof}[Proof of Lemma~\ref{lem:ideal_norm_general_tree}]
We first recall a basic property of the adjoint of a CPTP map. Since
$\Phi_v$ is CPTP, its adjoint $\Phi_v^\dagger$ is positive and unital.
Hence, for any Hermitian operator $A$,
\begin{equation}
    \|\Phi_v^\dagger(A)\|_\infty\le \|A\|_\infty .
\end{equation}
Indeed, for any Hermitian $A$, we have
\begin{equation}
    -\|A\|_\infty I\le A\le \|A\|_\infty I.
\end{equation}
Applying the positive and unital map $\Phi_v^\dagger$ gives
\begin{equation}
    -\|A\|_\infty I
    =
    \Phi_v^\dagger(-\|A\|_\infty I)
    \le
    \Phi_v^\dagger(A)
    \le
    \Phi_v^\dagger(\|A\|_\infty I)
    =
    \|A\|_\infty I.
\end{equation}
Therefore, $\|\Phi_v^\dagger(A)\|_\infty\le \|A\|_\infty$.

Next, we prove $\|M_v\|_\infty\le 1$ by induction from the leaves to the root.
First, let $v$ be a leaf, i.e., $B_v=O_v$. 
Since $\|O_v\|_\infty\le 1$, the contractivity shown above gives
\begin{equation}
    \|M_v\|_\infty
    =
    \|\Phi_v^\dagger(B_v)\|_\infty
    \le
    \|B_v\|_\infty
    \le
    1.
\end{equation}

Next, let $v$ be an internal non-root node, and assume that
\begin{equation}
    \|M_u\|_\infty\le 1
\end{equation}
holds for all children $u\in{\rm Ch}(v)$. 
Therefore,
\begin{equation}
    \|B_v\|_\infty
    =
    \left\|
    O_v\otimes
    \bigotimes_{u\in{\rm Ch}(v)}M_u
    \right\|_\infty
    \le
    \|O_v\|_\infty
    \prod_{u\in{\rm Ch}(v)}\|M_u\|_\infty
    \le
    1.
\end{equation}
Using the contractivity of $\Phi_v^\dagger$,
\begin{equation}
    \|M_v\|_\infty
    =
    \|\Phi_v^\dagger(B_v)\|_\infty
    \le
    \|B_v\|_\infty
    \le
    1.
\end{equation}
By induction, this proves $\|M_v\|_\infty\le 1$ for all $v\in V^\circ$.

Finally, we prove the bound at the root. 
Since $\|O_r\|_\infty\le 1$ and $\|M_u\|_\infty\le 1$ for all
$u\in{\rm Ch}(r)$, we obtain
\begin{equation}
    \|B_r\|_\infty
    =
    \left\|
    O_r\otimes
    \bigotimes_{u\in{\rm Ch}(r)}M_u
    \right\|_\infty
    \le
    \|O_r\|_\infty
    \prod_{u\in{\rm Ch}(r)}\|M_u\|_\infty
    \le
    1.
\end{equation}
This completes the proof.
\end{proof}

\begin{proof}[Proof of Proposition~\ref{prop:subtree_error_general_tree}]
We prove the statement by induction from the leaves upward.
First, suppose that $v$ is a leaf. Then, by definition,
\begin{equation}
    \tilde B_v=B_v=O_v.
\end{equation}
On the subtree success event $G_v$, the local success event $E_v$ holds.
Therefore,
\begin{equation}
    \|\tilde M_v-M_v\|_\infty
    =
    \left\|
    \tilde M_v-\Phi_v^\dagger(\tilde B_v)
    \right\|_\infty
    \le
    \eta_v.
\end{equation}
Since $A_v=\eta_v$ for a leaf and $\eta_v\le e^{\eta_v}-1$, we obtain
\begin{equation}
    \|\tilde M_v-M_v\|_\infty
    \le
    e^{A_v}-1.
\end{equation}
Using Lemma~\ref{lem:ideal_norm_general_tree} and the triangle inequality, we also get
\begin{equation}
    \|\tilde M_v\|_\infty
    \le
    \|M_v\|_\infty+\|\tilde M_v-M_v\|_\infty
    \le
    1+e^{A_v}-1
    =
    e^{A_v}.
\end{equation}

Now let $v$ be an internal non-root vertex, and assume that the claim holds
for all children $u\in{\rm Ch}(v)$. On $G_v$, the event $E_v$ holds, and the subtree success event $G_u$ holds for every child $u\in{\rm Ch}(v)$.
Hence,
\begin{eqnarray}
    \|\tilde M_v-M_v\|_\infty
    &=&
    \left\|
    \tilde M_v-\Phi_v^\dagger(\tilde B_v)
    +
    \Phi_v^\dagger(\tilde B_v-B_v)
    \right\|_\infty                                      \\
    &\le&
    \eta_v+\|\Phi_v^\dagger(\tilde B_v-B_v)\|_\infty     \\
    &\le&
    \eta_v+\|\tilde B_v-B_v\|_\infty .\label{apeq:m_b_bound}
\end{eqnarray}
In the last inequality, we used the operator-norm contractivity of
$\Phi_v^\dagger$ on Hermitian operators.

It remains to bound $\|\tilde B_v-B_v\|_\infty$. Define
\begin{equation}
    \Delta_u:=\tilde M_u-M_u
    \qquad
    (u\in{\rm Ch}(v)),
\end{equation}
and set
\begin{equation}
    C_v:=\sum_{u\in{\rm Ch}(v)}A_u.
\end{equation}
By the definitions of $B_v$ and $\tilde B_v$,
\begin{equation}
    \|\tilde B_v-B_v\|_\infty
    =
    \left\|
    O_v\otimes
    \left(
    \bigotimes_{u\in{\rm Ch}(v)}\tilde M_u
    -
    \bigotimes_{u\in{\rm Ch}(v)}M_u
    \right)
    \right\|_\infty                                      
    \le
    \left\|
    \bigotimes_{u\in{\rm Ch}(v)}\tilde M_u
    -
    \bigotimes_{u\in{\rm Ch}(v)}M_u
    \right\|_\infty ,
\end{equation}
where the inequality uses $\|O_v\|_\infty\le 1$.

We now expand the tensor-product difference. Since
$\tilde M_u=M_u+\Delta_u$, we have
\begin{equation}
    \bigotimes_{u\in{\rm Ch}(v)}\tilde M_u
    -
    \bigotimes_{u\in{\rm Ch}(v)}M_u
    =
    \sum_{\emptyset\neq I\subseteq{\rm Ch}(v)}
    \left(
    \bigotimes_{u\in I}\Delta_u
    \otimes
    \bigotimes_{u\notin I}M_u
    \right),
\end{equation}
where any fixed ordering of the children is used. Here, the subset
$I\subseteq{\rm Ch}(v)$ specifies which tensor factors contribute an error
term $\Delta_u$. The empty subset is excluded because it gives the term
$\bigotimes_{u\in{\rm Ch}(v)}M_u$, which is subtracted on the left-hand side.

Taking the operator norm and using the triangle inequality gives
\begin{equation}
    \|\tilde B_v-B_v\|_\infty
    \le
    \sum_{\emptyset\neq I\subseteq{\rm Ch}(v)}
    \left\|
    \bigotimes_{u\in I}\Delta_u
    \otimes
    \bigotimes_{u\notin I}M_u
    \right\|_\infty                                     
    =
    \sum_{\emptyset\neq I\subseteq{\rm Ch}(v)}
    \prod_{u\in I}\|\Delta_u\|_\infty
    \prod_{u\notin I}\|M_u\|_\infty .
\end{equation}
By the induction hypothesis and Lemma~\ref{lem:ideal_norm_general_tree},
\begin{equation}
    \|\Delta_u\|_\infty
    =
    \|\tilde M_u-M_u\|_\infty
    \le
    e^{A_u}-1,
    \qquad
    \|M_u\|_\infty\le 1.
\end{equation}
Therefore,
\begin{eqnarray}
    \|\tilde B_v-B_v\|_\infty
    &\le&
    \sum_{\emptyset\neq I\subseteq{\rm Ch}(v)}
    \prod_{u\in I}(e^{A_u}-1)                              \\
    &=&
    \prod_{u\in{\rm Ch}(v)}
    \left(1+(e^{A_u}-1)\right)-1                            \\
    &=&
    \prod_{u\in{\rm Ch}(v)}e^{A_u}-1                         \\
    &=&
    e^{C_v}-1.
\end{eqnarray}
Substituting this tensor-product bound $\|\tilde B_v-B_v\|_\infty \le e^{C_v}-1$ into the local-error decomposition $\|\tilde M_v-M_v\|_\infty \le \eta_v+\|\tilde B_v-B_v\|_\infty$ in Eq.~\eqref{apeq:m_b_bound} gives
\begin{equation}
    \|\tilde M_v-M_v\|_\infty
    \le \eta_v+e^{C_v}-1 .
\end{equation}
It remains to put this bound in the inductive form.
Since the child subtrees are disjoint, we have $A_v=\eta_v+C_v$. Moreover, $e^{\eta_v}-1\ge \eta_v$ and $e^{C_v}\ge 1$, and hence
\begin{equation}
    \eta_v+e^{C_v}-1
    \le
    e^{C_v}(e^{\eta_v}-1)+e^{C_v}-1
    =
    e^{A_v}-1 .
\end{equation}
Thus,
\begin{equation}
    \|\tilde M_v-M_v\|_\infty \le e^{A_v}-1,
\end{equation}
which closes the induction step for the error bound.
Finally, using Lemma~\ref{lem:ideal_norm_general_tree} and the triangle
inequality,
\begin{equation}
    \|\tilde M_v\|_\infty
    \le
    \|M_v\|_\infty+\|\tilde M_v-M_v\|_\infty
    \le
    1+e^{A_v}-1
    =
    e^{A_v}.
\end{equation}
\end{proof}

\begin{proof}[Proof of Corollary~\ref{cor:root_and_intermediate_norm_general_tree}]
Assume first that the global success event $G$ holds. 
Then Proposition~\ref{prop:subtree_error_general_tree} applies to every child of the root. Since the subtrees rooted at the children of $r$ are disjoint, we
have
\begin{equation}
    A_{\rm tot}
    =
    \sum_{u\in{\rm Ch}(r)}A_u .
\end{equation}
Expanding the tensor-product difference over nonempty subsets of
${\rm Ch}(r)$, and using $\|\tilde M_u-M_u\|_\infty\le e^{A_u}-1$ and $\|M_u\|_\infty\le 1$, we obtain
\begin{eqnarray}
    \|\tilde B_r-B_r\|_\infty
    \le
    \sum_{\emptyset\neq I\subseteq{\rm Ch}(r)}
    \prod_{u\in I}(e^{A_u}-1) 
    =
    \prod_{u\in{\rm Ch}(r)} e^{A_u}-1
    =
    e^{A_{\rm tot}}-1 .
\end{eqnarray}
Moreover,
\begin{equation}
    \|\tilde B_r\|_\infty
    \le
    \|O_r\|_\infty
    \prod_{u\in{\rm Ch}(r)}\|\tilde M_u\|_\infty
    \le
    \prod_{u\in{\rm Ch}(r)}e^{A_u}
    =
    e^{A_{\rm tot}}.
\end{equation}

Next, fix $v\in V^\circ$. On $D_v$, all strict descendants of
$v$ are good. Thus, for every child $u\in{\rm Ch}(v)$, the
subtree success event $G_u$ holds, and Proposition~1 gives
\begin{equation}
    \|\tilde M_u\|_\infty\le e^{A_u}.
\end{equation}
Therefore,
\begin{equation}
    \|\tilde B_v\|_\infty
    \le
    \|O_v\|_\infty
    \prod_{u\in{\rm Ch}(v)}\|\tilde M_u\|_\infty
    \le
    \prod_{u\in{\rm Ch}(v)}e^{A_u}
    \le
    e^{A_{\rm tot}} .
\end{equation}
The same bound also covers the case where $v$ is a leaf, by the empty tensor
product convention and $\|O_v\|_\infty\le1$.
\end{proof}

\begin{proof}[Proof of Lemma~\ref{lem:adaptive_success_general_tree}]
For each $v\in V^\circ$, the event $D_v$ is determined by
$\mathcal F_{<v}$. Hence, by the tower property,
\begin{eqnarray}
    \Pr(E_v^c\cap D_v)
    =
    \mathbb E\!\left[
    \mathbf 1_{D_v}\mathbf 1_{E_v^c}
    \right]                                             
    =
    \mathbb E\!\left[
    \mathbf 1_{D_v}
    \Pr(E_v^c\mid \mathcal F_{<v})
    \right]                                          
    \le
    \mathbb E\!\left[
    \mathbf 1_{D_v}\delta_v
    \right]                                            
    \le
    \delta_v.
\end{eqnarray}
Here we used the assumed conditional bound on the event $D_v$.

Now suppose that $G$ fails. Then at least one local event $E_v$ fails.
Choose a failed vertex of maximum depth. Since it has maximum depth among the failed vertices, all its strict descendants are good. Therefore,
\begin{equation}
    G^c
    \subseteq
    \bigcup_{v\in V^\circ}(E_v^c\cap D_v).
\end{equation}
By the union bound,
\begin{equation}
    \Pr(G^c)
    \le
    \sum_{v\in V^\circ}\Pr(E_v^c\cap D_v)
    \le
    \sum_{v\in V^\circ}\delta_v.
\end{equation}
\end{proof}

\section{Exponential separation from learning-free wire cutting (Proof of Theorem~\ref{thm:learning-free-wire-cut-lower-bound})}\label{sec:lower-bound-qcc}

\noindent\textbf{Theorem~\ref{thm:learning-free-wire-cut-lower-bound} (Exponential separation from learning-free wire cutting).}
{\itshape
Consider Task~\ref{task:two-layer-tree-estimation} with target additive error $0<\epsilon \leq 1/3$.

\smallskip
\noindent\textup{\textbf{Upper bound.}}
For any failure probability $\delta\in(0,1)$, the learning-based
protocol estimates $\mu$ within additive error $\epsilon$ with probability at least $1-\delta$ using
\begin{equation}
    \mathcal{O}\!\left(
    \frac{d^3R^3}{\epsilon^2}
    \ln\!\left(\frac{Rd}{\delta}\right)
    \right)
\end{equation}
measurements in total, including the measurements used in the local learning steps.

\noindent\textup{\textbf{Lower bound.}}
Any learning-free wire-cutting protocol described by Algorithm~\ref{alg:learning-free-wire-cut} that estimates $\mu$ within additive error $\epsilon$ with constant success probability strictly larger than $1/2$ for all instances requires
\begin{equation}
    \Omega\!\left(
    \frac{(d+1)^R}{\epsilon^2}
    \right)
\end{equation}
measurements.
}

\begin{proof}[Proof of Theorem~\ref{thm:learning-free-wire-cut-lower-bound}]
The upper bound follows directly from Theorem~\ref{thm:2-layer}. We prove
the lower bound below.
It is enough to construct a random family of hard instances. Indeed, if a
learning-free wire-cutting protocol estimates $\mu$ within additive error $\epsilon$ with constant success probability for all instances
using $N$ shots, then it must also succeed on instances drawn from this
family. 

We construct a hard distribution over instances indexed by a hidden bit $x\in\{0,1\}$, such that estimating the target expectation value to accuracy $\epsilon$ allows one to identify $x$ with constant success probability.
Let $d=2^n$. For the $k$-th cut system, define $\bar Z^{(k)} := Z^{\otimes n}$ and set
\begin{equation}
    \bar Z := \bigotimes_{k=1}^R \bar Z^{(k)} .
\end{equation}
For a hidden bit $x\in\{0,1\}$, define the $d^{R}$-dimensional states
\begin{equation}
    \tau_x
    :=
    \frac{1}{d^R}
    \left(
        I+(-1)^x 3\epsilon\,\bar Z
    \right),
    \qquad x=0,1.
    \label{eq:hard-tau-x}
\end{equation}
The assumption $0<\epsilon\le 1/3$ ensures that $\tau_x$ is positive semidefinite. 

Next, we hide this Pauli observable by random local bases. 
For each local system $k\in\{1,\ldots,R\}$, we independently sample $U_k$ from the Haar measure on $\mathrm{U}(\mathbb{C}^d)$, and fix $U_1,\ldots,U_R$ for all $N$ shots.
With $U := \bigotimes_{k=1}^R U_k$, we define the input state
\begin{equation}
    \rho_x := U\tau_x U^\dagger=\frac{1}{d^R}
    \left(
        I+(-1)^x 3\epsilon\, U \bar Z U^{\dagger}
    \right),\qquad x=0,1.
    \label{eq:hard-rho-x}
\end{equation}
We also choose the downstream channels and observables as
\begin{equation}
    \Phi_k(\bullet) := U_k^\dagger \bullet U_k,
    \qquad
    O_k := \bar Z^{(k)} .
    \label{eq:hard-channel-observable}
\end{equation}
Since the downstream channels undo the random local changes of basis, $(\Phi_1\otimes\cdots\otimes\Phi_R)(\rho_x)=\tau_x $, and $O=\bigotimes_{k=1}^R O_k=\bar Z$, the target expectation value is
\begin{equation}
    \mu_x
    =
    \operatorname{tr}[O(\Phi_1\otimes\cdots\otimes\Phi_R)(\rho_x)]
    =
    \operatorname{tr}[\bar Z\tau_x]
    =
    (-1)^x 3\epsilon .
    \label{eq:hard-mu-x}
\end{equation}
Hence, by estimating the expectation value $\mu_x$ in Algorithm~\ref{alg:learning-free-wire-cut} within additive error $\epsilon$, one can predict the hidden bit $x$ with constant success probability.

We now turn this discrimination reduction into a sample lower bound by comparing two bounds on the mutual information between the hidden bit $x$ and the classical data produced by the protocol. 
Let $T_N := \bigl((\ell_i,y_i,z_i)\bigr)_{i=1}^N$ denote the classical data collected by Algorithm~\ref{alg:learning-free-wire-cut} before the final post-processing. 
Let $\hat{\mu}=\mathsf{Post}(T_N)$ be the final estimate. Define an estimator of the hidden bit by
\begin{equation}
    \hat{x}(T_N)
    :=
    \begin{cases}
    0, & \hat{\mu}\ge 0,\\
    1, & \hat{\mu}<0.
    \end{cases}
\end{equation}
For concreteness, we first take the success probability to be at least $2/3$. 
If the protocol estimates $\mu_x$ within additive error $\epsilon$ with success probability at least $2/3$, then $P_{\rm err}:=\Pr[\hat{x}\neq x]\le \frac{1}{3}$. Here $x$ is uniformly distributed over $\{0,1\}$, so $H(x)=\log 2$ where $H(A)$ is the Shannon entropy of $A$. 
Since $\hat{x}$ is a function of $T_N$, we have
$H(x|T_N)\le H(x|\hat{x})$.
By Fano's inequality for binary hypothesis testing, $H(x|T_N)\le h_2(P_{\rm err})\le h_2(1/3)$, where $h_2(p):=-p\log p-(1-p)\log(1-p)$ is the binary entropy. 
Thus, we have
\begin{equation}
    I(x:T_N)
    =
    H(x)-H(x|T_N)
    \ge
    \log 2-h_2(1/3)
    =: c_0>0,
    \label{eq:fano-lower-bound}
\end{equation}
where $I(A:B)$ denotes the mutual information between random variables $A$ and $B$.
For any other fixed success
probability larger than $1/2$, the same argument gives another positive constant $c_0$.

We now upper bound the same mutual information for any fixed
learning-free wire-cutting rule $\mathsf{WC}$. 
Let $\bm{\ell}:=(\ell_1,...,\ell_N)$, $\bm{y}=(y_1,...,y_N)$, and $\bm{z}=(z_1,...,z_N)$.
Since $U$ and the
cutting labels $\bm{\ell}$ are independent of $x$, we have $I(x:U,\bm{\ell})=0$. Therefore,
\begin{align}
    I(x:T_N)
    &\le I(x:T_N,U) \notag\\
    &= I(x:U,\bm{\ell})
       + I(x:\bm{y},\bm{z}
       \mid U,\bm{\ell}) \notag\\
    &= I(x:\bm{y},\bm{z}
       \mid U,\bm{\ell}).
    \label{eq:MI-add-U}
\end{align}
Conditioned on $U$, $\bm{\ell}$, and $\bm{y}$, all downstream input states are fixed by the wire-cutting rule $\mathsf{
WC}$: in the
$i$-th shot, the prepared state is $\bigotimes_{k=1}^R \sigma^{(k)}_{\ell_i,y_i^k}$.
Moreover, the downstream channels are fixed once $U$ is fixed. Thus, the conditional distribution of the downstream outcomes $\bm z$ does
not depend on $x$ once $U,\bm{\ell},\bm y$ are specified. Equivalently,
\begin{equation}
    I(x:\bm{z}\mid \bm{y},U,\bm{\ell})=0.
\end{equation}
Using the chain rule for mutual information, we obtain
\begin{align}
    I(x:\bm{y},\bm{z}\mid U,\bm{\ell})
    &=
    I(x:\bm{y}\mid U,\bm{\ell})
    +
    I(x:\bm{z}\mid \bm{y},U,\bm{\ell}) \notag\\
    &=
    I(x:\bm{y}\mid U,\bm{\ell}).
    \label{eq:MI-remove-z}
\end{align}
Combining Eqs.~\eqref{eq:MI-add-U} and \eqref{eq:MI-remove-z}, we have
\begin{equation}
    I(x:T_N)
    \le
    I(x:\bm{y}\mid U,\bm{\ell}).
    \label{eq:MI-reduce-y}
\end{equation}
It remains to decompose the right-hand side into single-shot
contributions. Conditioned on $x$, $U$, and $\bm{\ell}$, the
outcomes $y_1,\ldots,y_N$ are independent, and the distribution of $y_i$ depends on the cutting labels only through $\ell_i$. Therefore,
\begin{align}
    I(x:\bm{y}\mid U,\bm{\ell})
    &=
    H(\bm{y}\mid U,\bm{\ell})
    -
    H(\bm{y}\mid x,U,\bm{\ell}) \notag\\
    &\le
    \sum_{i=1}^N H(y_i\mid U,\ell_i)
    -
    \sum_{i=1}^N H(y_i\mid x,U,\ell_i) \notag\\
    &=
    \sum_{i=1}^N I(x:y_i\mid U,\ell_i),
    \label{eq:MI-sum}
\end{align}
and
\begin{equation}
    I(x:T_N)
    \le
    \sum_{i=1}^N I(x:y_i\mid U,\ell_i).
\end{equation}

We now bound each single-shot contribution
$I(x:y_i\mid U,\ell_i)$. Fix $i$. Since the cutting label
$\ell_i$ is sampled independently of $x$ and $U$, we have
\begin{equation}
    I(x:y_i\mid U,\ell_i)
    =
    I(x:y_i,\ell_i\mid U).
\end{equation}
Moreover, by the chain rule for mutual information and $I(U:y_i,\ell_i)\ge 0$, we obtain
\begin{eqnarray}
    I(x:y_i,\ell_i\mid U) = I(x,U:y_i,\ell_i) - I(U:y_i,\ell_i)
    \leq
    I(x,U:y_i,\ell_i).
\end{eqnarray}
To bound the term $I(x,U:y_i,\ell_i)$, we use the standard fact that any finite POVM can be refined into a rank-one POVM, with the original outcome recovered by classical processing. 
In our case, for a fixed cutting label $\ell_i$, the product POVM $\{E_{\ell_i,y_i}=\bigotimes_{k=1}^R E^{(k)}_{\ell_i,y_i^k}\}_{y_i}$ can be refined into rank-one product elements of the form $\left\{\omega_{\ell_i,s}d^R|v_{\ell_i,s}\rangle\langle v_{\ell_i,s}|\right\}_s$, where $|v_{\ell_i,s}\rangle=|v^1_{\ell_i,s}\rangle\otimes\cdots\otimes|v^R_{\ell_i,s}\rangle$ and $\sum_s \omega_{\ell_i,s}=1$.
Since the original outcome $y_i$ is obtained from the refined outcome $s$ by classical coarse graining, the data-processing inequality gives
\begin{equation}
    I(x,U:y_i,\ell_i)\le I(x,U:s,\ell_i).
\end{equation}

Let us define $q_{\ell,s}^{x,U}:=p_{\ell}\omega_{\ell,s}d^R\langle v_{\ell,s}|U\tau_x U^\dagger|v_{\ell,s}\rangle$, which denotes the joint probability of the cutting label $\ell_i$ taking value $\ell$ and the refined outcome $s$, conditioned on $x$ and $U$. We also write $\bar{q}_{\ell,s}:=\mathbb{E}_{U,x}\!\left[q_{\ell,s}^{x,U}\right]$. Then, we have
\begin{eqnarray}
    I(x:y_i|U,\ell_i)
    &\leq& I(x,U:s,\ell_i)\\
    &=&
    \sum_{\ell,s}
    \left\{
    -\bar q_{\ell,s}\log \bar q_{\ell,s}
    +
    \mathbb{E}_{U,x}
    \left[
    q_{\ell,s}^{x,U}\log q_{\ell,s}^{x,U}
    \right]
    \right\} \\
    &\leq&
    \sum_{\ell,s}
    \left\{
    -\bar q_{\ell,s}\log \bar q_{\ell,s}
    +
    \mathbb{E}_{U,x}
    \left[
    q_{\ell,s}^{x,U}
    \left(
    \log \bar q_{\ell,s}
    +
    \frac{q_{\ell,s}^{x,U}-\bar q_{\ell,s}}{\bar q_{\ell,s}}
    \right)
    \right]
    \right\}\\
    &=&
    \sum_{\ell,s}
    \frac{
    \mathbb{E}_{U,x}\!\left[(q_{\ell,s}^{x,U})^2\right]
    -
    \bar q_{\ell,s}^{\,2}
    }{
    \bar q_{\ell,s}
    }.
    \label{eq:denom_nume}
\end{eqnarray}
Here, the inequality uses the bound
$\log t\leq \log a + (t-a)/a$ for $t,a>0$, with
$a=\bar q_{\ell,s}$. 
Equivalently, this follows from the concavity of
$\log t$. 
Hence, it remains to evaluate the first and second moments of
$q_{\ell,s}^{x,U}$.

Averaging over $x\in\{0,1\}$ yields
$\mathbb{E}_x[\tau_x]=I/d^R$, and we have
\begin{eqnarray}
    \bar q_{\ell,s}
    &=&
    p_{\ell}\omega_{\ell,s}d^R
    \mathbb{E}_{U,x}
    \left[
    \bra{v_{\ell,s}}U\tau_xU^\dagger
    \ket{v_{\ell,s}}
    \right]\\
    &=&
    p_{\ell}\omega_{\ell,s}.
    \label{eq:nume_1}
\end{eqnarray}

For the second moment, we have
\begin{eqnarray}
    \mathbb{E}_{U,x}\!\left[\left(q_{\ell,s}^{x,U}\right)^2\right]
    &=& p_\ell^2 \omega_{\ell,s}^2 d^{2R}
    \mathrm{tr}\left[
    \mathbb{E}_{x}[\tau_x^{\otimes 2}]
    \bigotimes_{k=1}^{R}
    \mathbb{E}_{U_k}
    \left(
    U_k^{\dagger} \pure{v^k_{\ell,s}} U_k
    \right)^{\otimes 2}
    \right]\\
    &=& p_\ell^2 \omega_{\ell,s}^2 d^{2R}
    \mathrm{tr}\left[
    \frac{1}{d^{2R}}
    \left(
    I \otimes I
    +
    9\epsilon^2 \bar{Z} \otimes \bar{Z}
    \right)
    \bigotimes_{k=1}^{R}
    \mathbb{E}_{U_k}
    \left(
    U_k^{\dagger} \pure{v^k_{\ell,s}} U_k
    \right)^{\otimes 2}
    \right]\\
    &=& p_\ell^2 \omega_{\ell,s}^2
    \mathrm{tr}\left[
    \left(
    I \otimes I
    +
    9\epsilon^2 \bar{Z} \otimes \bar{Z}
    \right)
    \bigotimes_{k=1}^{R}
    \frac{I\otimes I + \mathrm{SWAP}}{d(d+1)}
    \right]\\
    &=& p_\ell^2 \omega_{\ell,s}^2
    \left(
    \prod_{k=1}^{R}
    \mathrm{tr}\left[
    \frac{I\otimes I + \mathrm{SWAP}}{d(d+1)}
    \right]
    +
    9\epsilon^2
    \prod_{k=1}^{R}
    \mathrm{tr}\left[
    \bar{Z}^{(k)} \otimes \bar{Z}^{(k)}
    \cdot
    \frac{I\otimes I + \mathrm{SWAP}}{d(d+1)}
    \right]
    \right)\\
    &=& p_\ell^2 \omega_{\ell,s}^2
    \left(
    1 + \frac{9\epsilon^2}{(d+1)^R}
    \right).
    \label{eq:denom_1}
\end{eqnarray}
where in the second equality we used $\mathbb{E}_{x}[\tau_x^{\otimes 2}]=(I\otimes I+9\epsilon^2 \bar{Z}\otimes \bar{Z})/d^{2R}$ by the definition of $\tau_x$, and in the third equality we applied the
standard Haar twirling identity for rank-one projectors,
\begin{equation}
    \mathbb{E}_{U_k}
    \left(
    U_k^{\dagger} \pure{v^k_{\ell,s}} U_k
    \right)^{\otimes 2}
    =
    \frac{I \otimes I + \mathrm{SWAP}}{d(d+1)}
\end{equation}
valid for any fixed $\ket{v^k_{\ell,s}}$ in dimension $d$.
In the fourth equality, we use
\begin{equation}
    \mathrm{tr}\left[
    \frac{I\otimes I+\mathrm{SWAP}}{d(d+1)}
    \right]=1,\qquad 
    \mathrm{tr}\left[\bar{Z}^{(k)}\otimes \bar{Z}^{(k)}\frac{I\otimes I+\mathrm{SWAP}}{d(d+1)}
    \right]=\frac{1}{d+1}.
\end{equation}
Inserting the calculation results Eqs.~\eqref{eq:nume_1} and
\eqref{eq:denom_1} into Eq.~\eqref{eq:denom_nume}, we have
\begin{eqnarray}
    I(x:y_i|U,\ell_i)
    &\leq&
    \sum_{\ell,s}
    \frac{
    p_{\ell}^{2}\omega_{\ell,s}^{2}
    \left(
    1+\frac{9\epsilon^2}{(d+1)^R}
    \right)
    -
    p_{\ell}^{2}\omega_{\ell,s}^{2}
    }{
    p_{\ell}\omega_{\ell,s}
    }\\
    &=&
    \frac{9\epsilon^2}{(d+1)^R}
    \sum_{\ell,s}p_{\ell}\omega_{\ell,s}\\
    &=&
    \frac{9\epsilon^2}{(d+1)^R}.
    \label{eq:single_shot_mi_bound}
\end{eqnarray}
Here, in the last equality, we used
$\sum_{\ell}p_\ell=1$ and $\sum_s\omega_{\ell,s}=1$ for each
$\ell$.

The above result, together with Eqs.~\eqref{eq:MI-sum} and
\eqref{eq:fano-lower-bound}, yields
\begin{equation}
    c_0
    \leq
    I(x:T_N)
    \leq
    \frac{9N\epsilon^2}{(d+1)^R}.
\end{equation}
Thus, we obtain
\begin{equation}
    N
    \geq
    \frac{c_0}{9}
    \frac{(d+1)^R}{\epsilon^2}
    =
    \Omega\!\left(
    \frac{(d+1)^R}{\epsilon^2}
    \right).
\end{equation}
This proves the lower-bound part of Theorem~\ref{thm:learning-free-wire-cut-lower-bound}.
\end{proof}

\section{Relation to existing quantum circuit-cutting methods}\label{apsec:relation_to_wire_cut}

We review several broad directions in prior work aimed at reducing the sampling overhead in quantum circuit cutting, and discuss how our results relate to these approaches.
Quantum circuit cutting was introduced in Ref.~\cite{peng2021simulating} as a method for partitioning a clustered quantum system.
The cutting approach employed in this work belongs to the category of wire cutting rather than gate cutting, and is based on mid-circuit Pauli measurements. 
Specifically, the authors use the following decomposition of the $n$-qubit identity channel $\mathrm{id}^n$ to partition a quantum circuit into subcircuits:
\begin{equation}\label{eq:pauli_decomp}
    \mathrm{id}^n(\rho) = \frac{1}{2^n} \sum_{\bm{i}\in[4]^n} \mathrm{tr}\left[P_{\bm i} \rho\right] P_{\bm i} := \sum_{\bm{i}\in[4]^n} \mathcal{E}_{P_{\bm i}}(\rho),
\end{equation}
where $[4]=\{1,\dots,4\}$, $\bm{i}=(i_1,\dots,i_n)$, and $\{P_1,P_2,P_3,P_4\}=\{I,X, Y, Z\}$.
Each term can be viewed as a process that measures the input state $\rho$ in the Pauli observable $P_{\bm i}$ and then prepares $P_{\bm i}$ as a new input. 
Note that $\mathcal{E}_{P_{\bm i}}$ is a linear map and is not CPTP. Thus, in the standard wire-cutting protocol, it is implemented via randomized eigenstate preparation together with classical post-processing to form an unbiased estimator.
Concretely, one rewrites the expansion into an implementable form using randomized eigenstate preparation and classical post-processing:
\begin{equation}\label{apeq:pauli_decomp}
    \mathrm{id}^n = \sum_{\bm{i}\in[4]^n} \mathcal{E}_{P_{\bm i}},
    \quad 
    \mathcal{E}_{P_{\bm i}}(\rho) = \sum_{\bm e\in [2]^n} c_{\bm i,\bm e} \mathrm{tr}[ \pure{v_{\bm i,\bm e}} \rho] \sum_{\bm{e'}\in[2]^n} p_{\bm e'} c_{\bm i,\bm e'} \pure{v_{\bm i,\bm e'}},
\end{equation}
where $p_{\bm e'}:=1/2^n$ is the probability distribution, and $c_{\bm i,\bm e}\in\{\pm1\}$ and $\ket{v_{\bm i, \bm e}}$ denote the eigenvalue and eigenstate of $P_{\bm i}$, respectively, following the notation introduced in Appendix~\ref{apsec:const_pau}. 
Operationally, $\mathcal{E}_{P_{\bm{i}}}(\rho)=(1/2^{n}) \mathrm{tr}[P_{\bm i}\rho]P_{\bm i}$ corresponds to performing a projective measurement in the eigenbasis $\{\ket{v_{\bm{i},\bm e}}\}_{\bm e}$ of $P_{\bm{i}}$, preparing an eigenstate randomly sampled from $\{\ket{v_{\bm{i},\bm e'}}\}_{\bm e'}$, and then multiplying the circuit output by $c_{\bm i, \bm e} c_{\bm{i},\bm e'}$ in classical post-processing.
Note that $\mathcal{E}_{P_{\bm i}}$ can also be written in the form of the post-processed MP map defined in Eq.~\eqref{def:virtual_mp}. 
While wire cutting separates a circuit along the time direction, gate cutting was later introduced in Ref.~\cite{Mitarai_2021} as an approach to separate a circuit along the spatial direction.

Both wire and gate cutting rely on the general framework of quasiprobability simulation (see Appendix~\ref{sec:QPD}).
In this framework, the sampling overhead is governed by the factor $\gamma^2:=(\sum_i |a_i|)^2$ of a given quasiprobability decomposition (QPD) $\sum_{i=1}^{m} a_i \mathcal{E}_i$:
for estimating an expectation value within a fixed additive error, a QPD with overhead $\gamma$ typically increases the sufficient number of circuit shots by a factor $\gamma^2$.
If a decomposition with overhead $\gamma_k$ is applied at locations $k=1,\dots,K$, the total sampling overhead scales as $\prod_{k=1}^K \gamma_k^2$. 
Moreover, the number $m$ of operations appearing in each decomposition affects the total execution time~\cite{harada2025doubly,PhysRevA.108.022615}.

To reduce this sampling overhead, many prior studies have pursued two broad approaches.
The first direction focuses on minimizing the number $K$ of cuts, using classical optimization techniques to select cutting locations so that the total number of cuts is as small as possible~\cite{10.1145/3445814.3446758,tomesh2023divide,brandhofer2023optimal}.
The second direction targets reducing the factor $\gamma$ and has mainly employed the following two types of approaches:
\begin{enumerate}
    \item Developing quasiprobability decompositions with smaller $\gamma$.
    \item Restricting the circuit structure to eliminate certain bases from the quasiprobability decomposition.
\end{enumerate}
In the following, we briefly review these two approaches and explain how our results relate to the existing methods.

\subsection{Developing quasiprobability decompositions with smaller sampling overhead}

One of the central directions in quantum circuit cutting is to design decompositions with small sampling overhead. In this subsection, we focus on wire cuts, which are most relevant to our work, and discuss how existing constructions relate to our results. 

\vspace{0.7em}

\noindent\textbf{Review of prior works} \hspace{0.3em}
As introduced above, the first wire-cutting protocol was proposed in Ref.~\cite{peng2021simulating}, using mid-circuit Pauli measurements as in Eq.~\eqref{apeq:pauli_decomp}. 
A key feature of this protocol is that no classical communication is used inside each MP channel: the state preparation is performed independently of the measurement outcome.  
Subsequent works reduced the factor $\gamma$ below the $4^n$ in Eq.~\eqref{apeq:pauli_decomp}, by allowing classical communication within each MP channel. 
This line of work on wire cuts began with Ref.~\cite{Lowe2023fastquantumcircuit}, which expressed the $n$-qubit identity channel as a linear combination of MP channels based on random Clifford measurements. 
This protocol achieves $\gamma=2^{n+1}+1$. 
Later, 
Ref.~\cite{brenner2023optimal} proved the lower bound $\gamma \ge 2^{n+1}-1$ when classical communication is allowed within each MP channel, and also provided a decomposition that achieves this bound using $n$ additional qubits.
This ancilla overhead motivated follow-up work developing alternative decompositions that preserve the optimal sampling overhead while reducing ancillary requirements~\cite{harada2025doubly,pednault2023alternative,PRXQuantum.6.010316}; see Table~\ref{table:comparison_of_sampling_overhead}.

\begin{table*}[t]
\renewcommand{\arraystretch}{1.7}
\begin{tabular}{cccccc}
\hline\hline
&\multirow{2}{*}{~Main primitive in MP procedures $\mathcal{M}_i$~}&\multirow{2}{*}{~LOCC~}& ~Number of~      & ~Multiplicative factor~&~\multirow{2}{*}{Optimality of $\gamma$} \\
&&      & ~ancilla qubits~ & ~$\gamma:=\sum_i|a_i|$~&\\
\hline
~~~\cite{peng2021simulating}~~~&~Pauli measurement~&~No~&~$0$~&~$4^n$~&~Optimal (LO)\\
~~~\cite{Lowe2023fastquantumcircuit}~~~&~Random Clifford measurement~&~Yes~&~$0$~&~$2^{n+1}+1$~&~--\\
~~~\cite{brenner2023optimal}~~~&~Quantum teleportation $\&$ Virtual Bell pairs~&~Yes~&~$n$~&~$2^{n+1}-1$~&~Optimal (LOCC)\\
~~~\cite{harada2025doubly}~~~&~Mutually unbiased bases (MUBs)~&~Yes~&~$0$~&~$2^{n+1}-1$~&~Optimal (LOCC)\\
~~~\cite{pednault2023alternative}~~~&~Random Clifford measurement~&~Yes~&~$0$~&~$2^{n+1}-1$~&~Optimal (LOCC)\\
~~~\cite{PRXQuantum.6.010316}~~~&~Random diagonal unitary $2$-designs~&~Yes~&~$0$~&~$2^{n+1}-1$~&~Optimal (LOCC)\\
\hline\hline
\end{tabular}
\caption{Comparison of existing quasiprobability decompositions for wire cuts of the form $\mathrm{id}^n=\sum_{i=1}^{m} a_i \mathcal{M}_i$, where $a_i\in \mathbb{R}$ and $\mathcal{M}_i$ represent post-processed MP maps.
Here ``LO'' denotes local operations without classical communication, whereas ``LOCC'' allows classical communication.
\label{table:comparison_of_sampling_overhead}
}
\end{table*}

\vspace{0.7em}

\noindent\textbf{Comparison with our result} \hspace{0.3em}
Prior work has established the lower bound $\gamma \geq 2^{n+1}-1$ for quasiprobability decompositions for wire cuts under the standard assumptions.
By contrast, Theorem~\ref{thm:existence_of_no_cost_cut} yields a decomposition with $\gamma=1$, which looks incompatible with the bound at first glance.
This arises from a difference in the formulation of the quasiprobability decomposition.
In the standard formulation adopted in earlier works~\cite{peng2021simulating,Lowe2023fastquantumcircuit,harada2025doubly,brenner2023optimal,pednault2023alternative,PRXQuantum.6.010316}, one demands an exact representation of the entire target operation.

\begin{dfn}[Quasiprobability decomposition]\label{dfn:QPd}
Let $\mathcal{T}$ be a target operation and let $S$ be a set of implementable operations.
A quasiprobability decomposition of $\mathcal{T}$ over $S$ is a representation 
\begin{equation}
     \mathcal{T}=\sum_i a_i \mathcal{E}_i,
\end{equation}
where $a_{i}\in \mathbb{R}$ and $\mathcal{E}_i\in S$.
We define $\gamma:=\sum_{i}|a_i|$ and refer to $\gamma^2$
as the sampling overhead.
\end{dfn}

In contrast, our approach only requires consistency at the level of observable expectation values. 
That is, we consider the expectation-value-level quasiprobability decomposition, defined as follows.

\begin{dfn}[Expectation-value-level quasiprobability decomposition]\label{dfn:exp_level_QPD}
Let $\mathcal{T}$ be a target operation and let $S$ be a set of implementable operations.
Fix a CPTP map $\Phi$, an observable $O$, and a quantum state $\rho$. 
For the task of estimating $\mathrm{tr}[O(\Phi\circ\mathcal{T})(\rho)]$, we call a representation an expectation-value-level quasiprobability decomposition if it satisfies
\begin{equation} 
    \mathrm{tr}[O (\Phi \circ \mathcal{T})(\rho)]=\sum_{i} a_i \mathrm{tr}[O (\Phi \circ \mathcal{E}_i)(\rho)],
\end{equation}
where $a_{i}\in \mathbb{R}$ and $\mathcal{E}_i\in S$. 
For such a decomposition, we define $\gamma_{\rm eff}:=\sum_{i}|a_i|$ and refer to $\gamma_{\rm eff}^2$ as the effective sampling overhead.
\end{dfn}
We emphasize that in many settings of interest, $\rho$ and/or $\Phi$ are not known classically. Thus, the expectation-value-level QPD should often be understood as an existence result about a decomposition that avoids the additional rescaling factor $\gamma_{\rm eff}$ in the evaluation phase.
Finding such a decomposition in practice may require an additional estimation step for $\rho$ and/or $\Phi$, whose sample complexity is not captured by $\gamma_{\rm eff}$.

In the context of the wire-cutting problem, previous work considers decompositions of the form
\begin{equation}
     \mathrm{id}^n=\sum_i a_i \mathcal{M}_i,
\end{equation}
where $a_i\in\mathbb R$ and $\mathcal M_i$ are MP procedures, possibly with classical post-processing weights absorbed as in Eq.~\eqref{def:virtual_mp}.
Within this standard notion of the quasiprobability decomposition (Definition~\ref{dfn:QPd}), the $\ell_1$-norm of the coefficients $a_i$ obey the lower bound $\gamma \geq 2^{n+1}-1$.
In contrast, Theorem~\ref{thm:existence_of_no_cost_cut} provides a decomposition in the sense of Definition~\ref{dfn:exp_level_QPD}:
\begin{equation} 
    \mathrm{tr}[O (\Phi \circ \mathrm{id}_n) (\rho)]=\sum_{i} a_i \mathrm{tr}[O (\Phi \circ \mathcal{M}_i)(\rho)],
\end{equation}
where $a_i\in\mathbb R$ and $\mathcal M_i$ are MP procedures, possibly with classical post-processing.
This theorem shows that there exists an
expectation-value-level decomposition with $\gamma_{\rm eff}=1$.
In words, this theorem gives a stronger form of this expectation-value-level
decomposition: once $\Phi$ and $O$ are fixed, the same MP channel works
for all input operators $X$, rather than only for a single fixed input
state $\rho$.

\subsection{Reduction of sampling overhead via structure-aware methods}

Another major line of research reduces the sampling overhead by leveraging the fact that, for certain fixed quantum circuits, many terms in the mid-circuit Pauli expansion contribute exactly zero to the target expectation value.
Below, we summarize the main prior works~\cite{10025537,10196555,10313822}, and show that these works can be reformulated in the language of expectation-value-level QPDs.

\vspace{0.7em}

\noindent\textbf{Review of prior works} \hspace{0.3em}
These works consider quantum circuits of the form shown in Fig.~\ref{fig:simple_circuit}.
The quantum circuit before the decomposition (Fig.~\ref{fig:simple_circuit}(a)) starts from the $n_{\rm tot}:=(n_A+n_{B}+n_{C})$-qubit initial state $\ket{0}^{\otimes n_{\rm tot}}$, followed by two unitary channels $\mathcal{U}_{1}(\cdot)=U_1\cdot U_1^{\dagger}$ and $\mathcal{U}_{2}(\cdot):=U_2 \cdot U_2^{\dagger}$ acting on registers $(A,B)$ and $(B,C)$, respectively.
Measurement is performed using a diagonal observable $O:=O_1 \otimes O_2$ where $O_{1}$ acts on $A$, $O_2$ acts on $(B,C)$, and $\| O_{1} \|_{\infty},\|O_2\|_{\infty} \leq 1$. 
The decomposed quantum circuit (Fig.~\ref{fig:simple_circuit}(b)) consists of two subcircuits obtained by inserting the Pauli expansion~\eqref{apeq:pauli_decomp} on register $B$ between $\mathcal{U}_{1}$ and $\mathcal{U}_{2}$. 
Using Eq.~\eqref{eq:pauli_decomp} for $\mathrm{id}^{n_B}$, the original expectation value
$\braket{O}:=\mathrm{tr}[O \mathcal{U}_{2} \circ \mathcal{U}_1 (\pure{0}^{\otimes n_{\rm tot}})]$ can be written as
\begin{equation}
    \braket{O} = \frac{1}{2^{n_{B}}} \sum_{P \in \mathcal{P}_{B}} \mathrm{tr}[(O_{1} \otimes P) \mathcal{U}_1(\pure{0}^{\otimes (n_{A}+n_{B})}) ] \, \mathrm{tr}[ O_{2} \mathcal{U}_2( P \otimes \pure{0}^{\otimes n_{C}})].
\end{equation}
where $\mathcal{P}_{B}:=\{I,\mathrm{X},\mathrm{Y},\mathrm{Z}\}^{\otimes n_{B}}$ is the set of $n_{B}$-qubit Pauli strings. 

\begin{figure*}[t]
\centering
\begin{center}
 \includegraphics[width=160mm]{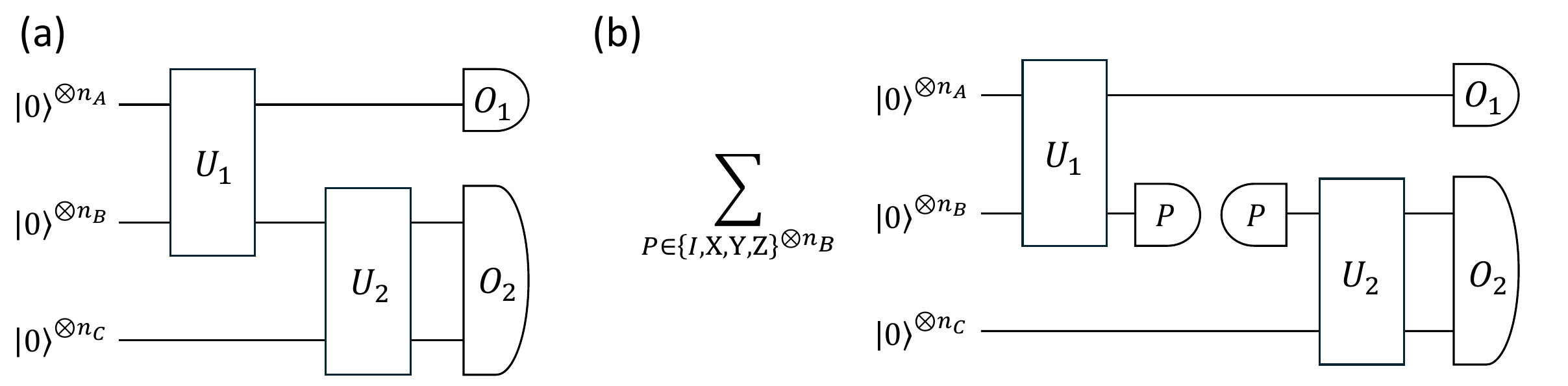}
\end{center}
\vspace{-10pt}
\caption{(a) Quantum circuit before the decomposition. The circuit consists of the $(n_A+n_{B}+n_{C})$-qubit initial state $\ket{0}^{\otimes (n_{A}+n_{B}+n_{C})}$, two unitary gates $U_{1}$ and $U_2$, and the diagonal observable $O=O_1 \otimes O_2$ where $\|O_1\|_{\infty},\|O_2\|_{\infty}\leq 1$. (b) Quantum circuit after the decomposition. By inserting the Pauli expansion~\eqref{eq:pauli_decomp} on register $B$ between $U_1$ and $U_2$, the circuit is separated into two subcircuits.}
\label{fig:simple_circuit}
\end{figure*}

The key idea in Ref.~\cite{10196555} is to identify cut locations for which the upstream contribution vanishes for some Pauli strings:
\begin{equation}\label{eq:golden_cutting}
    \mathrm{tr}\left[(O_{1} \otimes P) \mathcal{U}_1(\pure{0}^{\otimes (n_{A}+n_{B})}) \right] = 0
\end{equation}
for some $P\in \mathcal{P}_B$.
Whenever this condition holds, the corresponding downstream term does not contribute to the final estimate, and one can omit the executions needed to estimate $\mathrm{tr}[ O_{2} \mathcal{U}_2( P \otimes \pure{0}^{\otimes n_{C}})]$. This directly reduces the total number of circuit executions.
In Ref.~\cite{10196555}, such a location is called a \textit{golden cutting point} if Eq.~\eqref{eq:golden_cutting} holds for at least one $P\in\mathcal{P}_B$.
Note that such points do not necessarily exist and are generally not unique.
Moreover, one cannot know \textit{a priori} whether a chosen cut location is a golden cutting point. 
Therefore, the follow-up work~\cite{10313822} introduces a scheme to detect such Pauli strings, if they exist, during the execution of the upstream subcircuits.

A related idea is also adopted in Ref.~\cite{10025537}.
They consider the same circuit structure (Fig.~\ref{fig:simple_circuit}), but insert unitary operators $U_{\rm rco}$ and $U_{\rm rco}^{\dagger}$ immediately before and after the cutting location, and then apply the Pauli expansion.
By choosing $U_{\rm rco}$ such that the reduced operator on $B$ becomes diagonal, the upstream contributions vanish for all Pauli strings except those in $\{I,\mathrm{Z}\}^{\otimes n_{B}}$:
\begin{equation}\label{eq:ricco_original}
    \mathrm{tr}\!\left[(O_{1} \otimes P)\, (\mathcal{U}_{\rm rco}\circ \mathcal{U}_1)(\pure{0}^{\otimes (n_{A}+n_{B})}) \right] = 0,
    \quad \forall P\in \mathcal{P}_{B} \setminus \{I,\mathrm{Z}\}^{\otimes n_{B}},
\end{equation}
where $\mathcal{U}_{\rm rco}(\cdot):=U_{\rm rco}\cdot U_{\rm rco}^{\dagger}$.
Ref.~\cite{10025537} variationally optimizes such a $U_{\rm rco}$ and numerically demonstrates a reduction in the computational cost associated with the Pauli expansion.

\vspace{0.7em}

\noindent\textbf{Reformulation of prior works} \hspace{0.3em} 
We now explain how these structure-aware approaches fit into Definition~\ref{dfn:exp_level_QPD}.
Let
\begin{equation}
S_{\mathcal{P}_B} := \left\{ \mathcal{E}_{P_{\bm i}}(\rho)=\frac{1}{2^{n_{B}}} \mathrm{tr}\left[P_{\bm i} \rho\right] P_{\bm i}: P_{\bm i} \in \mathcal{P}_B=\{I,\mathrm{X},\mathrm{Y},\mathrm{Z}\}^{\otimes n_{B}} \right\}
\end{equation}
denote the set of operations used in the mid-circuit Pauli decomposition~\eqref{eq:pauli_decomp}. As remarked around Eq.~\eqref{apeq:pauli_decomp}, each $\mathcal{E}_{P_{\bm i}}$ can be implemented using a measure-and-prepare procedure with classical post-processing.
Then, the decompositions considered in Refs.~\cite{10196555,10313822} can also be understood under Definition~\ref{dfn:exp_level_QPD} as follows.

\begin{remark}[Refs.~\cite{10196555,10313822}]\label{task:gcp}
Consider the quantum circuit in Fig.~\ref{fig:simple_circuit}, and suppose the implementable operations are restricted to $S_{\mathcal{P}_B}$.
Let $\mathcal{P}'_{B}\subset \mathcal{P}_{B}$ be a subset of $n_B$-qubit Pauli strings, and define $S_{\mathcal{P}_B'}:=\{\mathcal{E}_{P}: P\in\mathcal{P}_B'\}\subset S_{\mathcal{P}_B}$. 
Then, Refs.~\cite{10196555,10313822} consider decompositions that satisfy
\begin{equation}
    \braket{O} = \sum_{\mathcal{M} \in S_{\mathcal{P}_B'} } \mathrm{tr}[O \mathcal{U}_{2} \circ \mathcal{M} \circ \mathcal{U}_1 (\pure{0}^{\otimes n_{\rm tot}})].\label{eq:gcp_subset}
\end{equation}
\end{remark}

Under this formulation, a golden cutting point corresponds to a cutting location for which there exists a strict subset $S_{\mathcal{P}'_{B}}\subsetneq S_{\mathcal{P}_{B}}$ satisfying Eq.~\eqref{eq:gcp_subset}.
At such a cutting location, $\gamma_{\rm eff}=|S_{\mathcal{P}'_{B}}|$ ($1 \le |S_{\mathcal{P}'_{B}}| < 4^{n_{B}}$) and an effective sampling overhead becomes
\begin{equation}
    \gamma_{\rm eff}^2 = |S_{\mathcal{P}'_{B}}|^2. 
\end{equation}

Similarly, the decomposition explored in Ref.~\cite{10025537} can be formulated under Definition~\ref{dfn:exp_level_QPD} as follows.

\begin{remark}\label{task:ricco}
Consider the quantum circuit shown in Fig.~\ref{fig:simple_circuit}, and suppose the implementable operations are restricted to
\begin{equation}
    S_{\mathcal{P}_B,U} := \left\{ \mathcal{E}_{U^{\dagger}P_{\bm i} U}(\rho) := \frac{1}{2^{n_B}} \mathrm{tr}\left[ U^{\dagger}P_{\bm i} U \rho\right] U^{\dagger}P_{\bm i} U : P_{\bm i} \in \mathcal{P}_B,\, U \in \Unit{\mathbb{C}^{d_B}} \right\}.
\end{equation}
Let $\mathcal{Z}:=\{I,Z\}^{\otimes n_B}\subset \mathcal{P}_B$, and let $\mathcal{P}'_{B}\subset \mathcal{P}_{B}$ be a subset of $n_B$-qubit Pauli strings. 
Then, Ref.~\cite{10025537} considers the decompositions that satisfy
\begin{equation}\label{eq:ricco_reformulated} 
    \braket{O} = \sum_{\mathcal{M} \in S_{\mathcal{P}'_B,U} } \mathrm{tr}[O \mathcal{U}_{2} \circ \mathcal{M} \circ \mathcal{U}_1 (\pure{0}^{\otimes n_{\rm tot}})].
\end{equation}
\end{remark}

Under this formulation, the difference between tasks in Remarks~\ref{task:gcp} and ~\ref{task:ricco} can be understood as the choice of the implementable operations. 
The set of implementable operations in the latter case is broader than the former, since $S_{\mathcal{P}_B}$ is obtained by fixing $U=I$ in $S_{\mathcal{P}_B,U}$.
Due to this broader class of implementable operations, it is guaranteed that there exists a decomposition 
\begin{equation}\label{eq:ricco_condition}
    \braket{O} = \sum_{P_{\bm i} \in \{I,Z\}^{\otimes n_{B}}} \mathrm{tr}[O \mathcal{U}_{2} \circ \mathcal{E}_{U^{\dagger}P_{\bm i}U} \circ \mathcal{U}_1 (\pure{0}^{\otimes n_{\rm tot}})].
\end{equation}
We remark that the unitary operator $U$ satisfying Eq.~\eqref{eq:ricco_condition} corresponds to $U_{\rm rco}$ in Eq.~\eqref{eq:ricco_original}. 
In this case, $\gamma_{\rm eff}=2^{n_B}$, and thus the effective sampling overhead becomes 
\begin{equation}
    (\gamma_{\rm eff})^2=4^{n_B}.
\end{equation}

Finally, applying Theorem~\ref{thm:existence_of_no_cost_cut} to the same circuit implies that, if one allows general MP channels with classical communication at the cutting location, there exist decompositions with $\gamma_{\rm eff}=1$, as shown in Proposition~\ref{prop:gamma1_in_out} below.

\begin{proposition}\label{prop:gamma1_in_out}
Let $\{\ket{j}\}_{j\in\{0,1\}^{n_B}}$ denote the computational basis on register $B$. Define MP channels $\mathcal{M}_{\rm in}$ and $\mathcal{M}_{\rm out}$ as
\begin{eqnarray}
    \mathcal{M}_{\rm in}(\bullet) &:=& \sum_{j}  \mathrm{tr}\left[ V_{\rm in} \pure{j} V_{\rm in}^{\dagger} \bullet \right] V_{\rm in} \pure{j} V_{\rm in}^{\dagger},\\
    \mathcal{M}_{\rm out}(\bullet) &:=& \sum_{j}  \mathrm{tr}\left[ V_{\rm out} \pure{j} V_{\rm out}^{\dagger} \bullet \right] V_{\rm out} \pure{j} V_{\rm out}^{\dagger}.
\end{eqnarray}
Here, $V_{\rm in}$ is a unitary that diagonalizes $Q_1:=\mathrm{tr}_A\!\left[
(O_1\otimes I_B)\,\mathcal U_1
\left(|0\rangle\langle0|^{\otimes(n_A+n_B)}\right)\right]$ and $V_{\rm out}$ is a unitary that diagonalizes $Q_2
:=\langle 0^{n_C}|U_2^\dagger O_2 U_2
|0^{n_C}\rangle$, i.e., $Q_1 = V_{\rm in} D_{\rm in} V_{\rm in}^{\dagger}$ and $Q_2 = V_{\rm out} D_{\rm out} V_{\rm out}^{\dagger}$ with diagonal $D_{\rm in},D_{\rm out}$. 
Then, both MP channels satisfy
\begin{eqnarray}\label{eq:equality_mp}
    \braket{O} = \mathrm{tr}\left[O \mathcal{U}_{2} \circ \mathcal{M}_{\rm in} \circ \mathcal{U}_1 (\pure{0}^{\otimes n_{\rm tot}})\right]
    = \mathrm{tr}\left[O \mathcal{U}_{2} \circ \mathcal{M}_{\rm out} \circ \mathcal{U}_1 (\pure{0}^{\otimes n_{\rm tot}})\right].
\end{eqnarray}
\end{proposition}

\begin{proof}[Proof of Proposition~\ref{prop:gamma1_in_out}]
Define a CPTP map $\Phi$ from $B$ to $(B,C)$ by $\Phi(X) = \mathcal{U}_{2}( X \otimes \pure{0}^{\otimes n_{C}})$.
Then, $\Phi^{\dagger}(O_2)=Q_2$ by definition.
Theorem~\ref{thm:existence_of_no_cost_cut} guarantees that for any input $X \in \Linear{\mathbb{C}^{d_B}}$
\begin{equation}
    \mathrm{tr}\left[O_2 \Phi(X)\right]
    = \mathrm{tr}\left[O_2 \Phi \circ \mathcal{M}_{\rm out}(X)\right].
\end{equation}
Setting $X=Q_1$ gives the second equality in Eq.~\eqref{eq:equality_mp}.

For $\mathcal{M}_{\rm in}$, since it dephases in an eigenbasis of $Q_1$, we have $\mathcal{M}_{\rm in}(Q_1)=Q_1$. Therefore,
\begin{eqnarray}
    \mathrm{tr}\left[O \mathcal{U}_{2} \circ \mathcal{M}_{\rm in}\mathcal{U}_1 (\pure{0}^{\otimes n_{\rm tot}}) \right] 
    &=& \mathrm{tr}_{BC}\left[ O_2 \mathcal{U}_2 \left( \mathcal{M}_{\rm in}\left( Q_1 \right) \otimes \pure{0}^{\otimes n_{C}} \right) \right]\\
    &=& \mathrm{tr}_{BC}\left[ O_2 \mathcal{U}_2 \left( Q_1 \otimes \pure{0}^{\otimes n_{C}} \right) \right]= \braket{O},
\end{eqnarray}
which completes the proof.
\end{proof}

Lastly, we mention a useful connection between the previous decomposition~\eqref{eq:ricco_condition} and $\gamma_{\rm eff}=1$ decompositions. 
In the following remark, we show that the unitary operator $U^{\dagger}$ in Eq.~\eqref{eq:ricco_condition} can be directly applied to construct a decomposition achieving an effective sampling overhead of unity.

\begin{remark}\label{remark:ricco_gamma_1}
The sum of operations $\mathcal{E}_{U^{\dagger}P_{\bm i}U}$ in Eq.~\eqref{eq:ricco_condition} is equal to a measure-and-prepare channel, i.e.,
\begin{equation}
    \sum_{P_{\bm i} \in \{I,Z\}^{\otimes n_{B}}} \mathcal{E}_{U^{\dagger}P_{\bm i}U}(\bullet) = \sum_{j \in \{0,1\}^{n_B}} \mathrm{tr}\left[ U^{\dagger} \pure{j} U \bullet \right] U^{\dagger} \pure{j}U.
\end{equation}
\end{remark}
\begin{proof}[Proof of Remark~\ref{remark:ricco_gamma_1}]
\begin{eqnarray}
    \sum_{P_{\bm i} \in \{I,Z\}^{\otimes n_{B}}} \mathcal{E}_{U^{\dagger}P_{\bm i}U}(\bullet) 
    &=& \frac{1}{2^{n_B}} \sum_{P_{\bm i} \in \{I,Z\}^{\otimes n_{B}}} \mathrm{tr}_1 \left[(U^{\dagger})^{\otimes 2} (P_{\bm i} \otimes P_{\bm i}) U^{\otimes 2} (\bullet \otimes I) \right]\\
    &=&  \mathrm{tr}_1 \left[ (U^{\dagger})^{\otimes 2} (I_1 \otimes I_2 + Z_1 \otimes Z_2)^{\otimes n_{B}} U^{\otimes 2} (\bullet \otimes I) \right]\\
    &=&  \mathrm{tr}_1 \left[ (U^{\dagger})^{\otimes 2} (\pure{0}_1\otimes \pure{0}_2 + \pure{1}_1\otimes \pure{1}_2)^{\otimes n_{B}} U^{\otimes 2} (\bullet \otimes I) \right]\\
    &=& \sum_{j \in \{0,1\}^{n_B}} \mathrm{tr}_1 \left[ (U^{\dagger})^{\otimes 2} (\pure{j} \otimes \pure{j}) U^{\otimes 2} (\bullet \otimes I) \right]\\
    &=& \sum_{j \in \{0,1\}^{n_B}} \mathrm{tr}\left[ U^{\dagger} \pure{j} U \bullet \right] U^{\dagger} \pure{j}U,
\end{eqnarray}
which completes the proof.
\end{proof}
\end{document}